\documentclass[a4paper]{lmcs}
\usepackage[utf8]{inputenc}
\usepackage{cite}
\usepackage{array}
\usepackage{url}
\usepackage{cmll}
\usepackage{amsmath}
\usepackage{amsthm}
\usepackage{amssymb}
\usepackage{wasysym,stmaryrd}
\usepackage{marvosym}
\usepackage[T1]{fontenc}
\usepackage{prooftree}
\usepackage[all]{xy}
\usepackage{tikz}
\usepackage[tikz]{bclogo}
\usepackage{graphicx}
\usepackage{cjhebrew}
\usepackage{mathpartir}
\usepackage{rotating}
\usepackage{fancybox}
\usepackage{mathrsfs}
\usepackage[justification=centering,labelfont=bf]{caption}

\usepackage{geometry} 
\theoremstyle{definition}
\newtheorem{theorem}{Theorem}
\newtheorem*{theorem*}{Theorem}
\newtheorem{definition}{Definition}
\newtheorem*{definition*}{Definition}

\newtheorem{lemma}[definition]{Lemma}
\newtheorem*{lemma*}{Lemma}

\newtheorem*{claim*}{Claim}
\newtheorem{example}{Example}
\newtheorem*{example*}{Example}
\newtheorem{notation}{Notation}
\newtheorem*{notation*}{Notation}
\newtheorem{proposition}[definition]{Proposition}
\newtheorem*{proposition*}{Proposition}

\newtheorem*{property*}{Property}
\newtheorem{observation}[definition]{Observation}
\newtheorem*{observation*}{Observation}
\newtheorem{remark}{Remark}
\newtheorem*{remark*}{Remark}



\usepackage{hyperref,verbatim}

\makeindex

\newcommand{\PRMK}{\mbox{\red{PRMK}}}

\newcommand{\pierre}[1]{\red{#1}\PRMK}
\newcommand{\repet}{\pierre{REPETITION}}

\newcommand{\fublainv}[1]{}
\newcommand{\fubla}[1]{\pierre{blahblahblah}}
\newcommand{\furef}[1]{\red{ref}}
\newcommand{\fucite}[1]{\red{cite}}

\newcommand{\igintro}[1]{\ignore{#1}}

\newcommand{\ighp}[1]{\ignore{#1}}
\newcommand{\chhp}[2]{#2} 

\newcommand{\skSk}{(k\cdot S_k)_{\kK}}

\renewcommand{\sup}{\mathop{\mathtt{sup}}}
\renewcommand{\inf}{\mathop{\mathtt{inf}}}

\newcommand{\hfills}{\hspace*{\fill}}

\long\def\ignore#1{\relax}

\newcommand{\fnsz}{\footnotesize}
\newcommand{\ssz}{\scriptsize}

\newcommand{\ssst}{\scriptscriptstyle}

\newcommand{\ie}{\textit{i.e.},\ }
\newcommand{\eg}{\textit{e.g.},\ }
\newcommand{\wrt}{w.r.t.\ }
\newcommand{\cf}{\textit{cf.}\ }
\newcommand{\resp}{resp.\ }
\newcommand{\via}{\textit{via}\ }
\newcommand{\vs}{\textit{vs.}\ }
\newcommand{\ih}{induction hypothesis}

\newcommand{\bcen}{\begin{center}}
\newcommand{\ecen}{\end{center}}

\newcommand{\Sec}{Sec.}

\newcommand{\tbul}{\noindent \textbullet~}

\newcommand{\leqs}{\leqslant}
\newcommand{\geqs}{\geqslant}

\newcommand{\ct}{\mathord{\cdot}}
\newcommand{\set}[1]{\{ #1 \}}
\newcommand{\mult}[1]{ [ #1 ]}
\newcommand{\ovl}[1]{\overline{#1}}

\newcommand{\eset}{\emptyset}

\newcommand{\emul}{\mult{\,}}
\newcommand{\tv}{o}
\newcommand{\inter}{\wedge}
\newcommand{\union}{\vee}
\newcommand{\rstr}[1]{|_{#1}}

\newcommand{\setpos}{\set{0,1,2}^*}

\newcommand{\phd}{\phantom{.}}

\newcommand{\subeq}{\subseteq}

\newcommand{\wdg}{\wedge}

\newcommand{\card}[1]{\# #1}

\newcommand{\Bohm}{B\"ohm}

\newcommand{\bbN}{\mathbb{N}}
\newcommand{\bbNfty}{{\mathbb{N}\cup\set{\infty}}}

\newcommand{\ttC}{\mathtt{C}}
\newcommand{\ttE}{\mathtt{E}}

\newcommand{\ttK}{\mathtt{K}}

\newcommand{\ttN}{\mathtt{N}}

\newcommand{\ttP}{\mathtt{P}}
\newcommand{\ttR}{\mathtt{R}}
\newcommand{\ttS}{\mathtt{S}}

\newcommand{\ttT}{\mathtt{T}}
\newcommand{\ttX}{\mathtt{X}}

\newcommand{\ttf}{\mathtt{f}}

\newcommand{\tti}{\mathtt{i}}

\newcommand{\ttx}{\mathtt{x}}

\newcommand{\scrD}{\mathscr{D}}

\newcommand{\scrR}{\mathscr{R}}
\newcommand{\scrRo}{\scrR_0}

\newcommand{\scrV}{\mathscr{V}}

\newcommand{\calH}{\mathcal{H}}

\newcommand{\secu}{^{\prime \prime}}
\newcommand{\asec}{a\secu} 
 
\newcommand{\csec}{c\secu}

\newcommand{\ovla}{\ovl{a}}
\newcommand{\ovlal}{\ovl{\al}}
\newcommand{\ovlalp}{\ovl{\al'}}

\newcommand{\Nmzo}{\bbN\setminus\set{0,1}}

\newcommand{\ttdom}{\mathtt{dom}}
\newcommand{\ttcodom}{\mathtt{codom}}
\newcommand{\ttdeg}{\mathtt{deg}}
\newcommand{\tttr}{\mathtt{tr}}

\newcommand{\ttpos}{\mathtt{pos}}

\newcommand{\ttax}{\mathtt{ax}}
\newcommand{\ttabs}{\mathtt{abs}}
\newcommand{\ttapp}{\mathtt{app}}
\newcommand{\ttBT}{\mathtt{BT}}

\newcommand{\ttAx}{\mathtt{Ax}}

\newcommand{\TypR}{\Types_\scrR}

\newcommand{\ttCh}{\mathtt{Ch}}

\newcommand{\io}{{i_0}}
\newcommand{\iI}{{i \in I}}
\newcommand{\iIa}{{i \in I_a}}
\newcommand{\jJ}{{j \in J}}
\newcommand{\kK}{{k \in K}}

\newcommand{\nN}{{n\in \bbN}}

\newcommand{\lam}{\lambda}
\newcommand{\Lam}{\Lambda}
\newcommand{\rew}{\rightarrow}
\newcommand{\redrew}{\mathbin{\red{\rew}}}
\newcommand{\rewsh}{\mathord{\rightarrow}}
\newcommand{\bred}{\rightarrow_{\!\beta}}
\newcommand{\bredfty}{{\bred^{\infty}}}
\newcommand{\hred}{\rightarrow_{\mathtt{h}}}

\newcommand{\hhred}{\rightarrow_{\mathtt{hh}}}

\newcommand{\lred}{\rightarrow_{\ell}}

\newcommand{\breda}[1]{\stackrel{#1}{\rightarrow}_{\!\beta}}

\newcommand{\beq}{\equiv_{\!\beta}}
\newcommand{\lx}{\lam x}
\newcommand{\ly}{\lam y}
\newcommand{\lz}{\lam z}

\newcommand{\lxrs}{(\lx.r)s}
\newcommand{\rsx}{r\subx{s}}
\newcommand{\ems}{\mult{\,}}
\newcommand{\est}{(\,)}
\newcommand{\al}{\alpha}

\newcommand{\gam}{\gamma}

\newcommand{\Gam}{\Gamma}

\newcommand{\Del}{\Delta}
\newcommand{\Deli}{\Del_i}

\newcommand{\Delk}{\Del_k}

\newcommand{\sig}{\sigma}
\newcommand{\sigi}{\sig_i}

\newcommand{\sigk}{\sig_k}

\newcommand{\epsi}{\varepsilon}
\newcommand{\fom}{f^{\infty}}
\newcommand{\rhoom}{\mult{\rho}_\om}

\newcommand{\arewa}{\mtv\rew \tv}
\newcommand{\marewa}{\mult{\arewa}}
\newcommand{\erewa}{\emul\rew \tv}
\newcommand{\merewa}{\mult{\erewa}}

\newcommand{\om}{\omega}
\newcommand{\Om}{\Omega}
\newcommand{\som}{{^{\om}}}

\newcommand{\arob}{\symbol{64}}

\newcommand{\hearts}{\heartsuit}
\newcommand{\clubs}{\clubsuit}
\newcommand{\sheart}{{\ssst \hearts}}
\newcommand{\sclub}{{\ssst \clubs}}

\newcommand{\TermV}{\mathscr{V}}
\newcommand{\TypeV}{\mathscr{O}}

\newcommand{\ju}[2]{#1 \vdash #2}

\newcommand{\hPi}{\hat{\Pi}}
\newcommand{\hPik}[1]{\hat{\Pi}_{_#1}}

\newcommand{\mtv}{\mult{\tv}}
\newcommand{\mrho}{\mult{\rho}}
\newcommand{\msig}{\mult{\sig}}
\newcommand{\mtau}{\mult{\tau}}
\newcommand{\msigi}{\mult{\sigi}_{\iI}}

\newcommand{\sSk}{(S_k)_{\kK}}
\newcommand{\sSpkp}{(S'_k)_{\kK'}}
\newcommand{\kSk}{(k\cdot S_k)}
\newcommand{\sSqk}[1]{(S^{#1}_k)_{\kK(#1)}}
\newcommand{\sSqkK}[1]{(S^{#1}_k)_{\kK(#1)}}

\newcommand{\Delf}{\Del_f}
\newcommand{\cu}{\mathtt{Y}}
\newcommand{\cuf}{\cu_f}

\newcommand{\zhnfo}{x\,t_1\ldots t_q}

\newcommand{\hnfo}{\lx_1\ldots x_p.\zhnfo}
\newcommand{\stacktq}{t_1\ldots t_q}
\newcommand{\lxrsstack}{\lxrs\,\stacktq}

\newcommand{\hreducibleo}{\lx_1\ldots x_p.\lxrsstack}

\newcommand{\Pex}{P_{\mathtt{ex}}}
\newcommand{\Piex}{\Pi_{\mathtt{ex}}}
\newcommand{\Sex}{S_{\mathtt{ex}}}
\newcommand{\sigex}{\sig_{\mathtt{ex}}}

\newcommand{\supf}{\,\!^f\!}
\newcommand{\supo}{\,\!^0\!}

\newcommand{\Types}{\mathtt{Typ}}
\newcommand{\STypes}{\mathtt{STypes}}
\newcommand{\Deriv}{\mathtt{Deriv}}

\newcommand{\DerivR}{\Deriv_\scrR}

\newcommand{\Axl}[1]{\mathtt{Ax}^{\lam}(#1)}
\newcommand{\AxPl}[1]{\mathtt{Ax}^{\lam}_P(#1)}
\newcommand{\Trl}[1]{\mathtt{Tr}^{\lam}(#1)}
\newcommand{\TrPl}[1]{\mathtt{Tr}^{\lam}_P(#1)}

\newcommand{\Rep}{\mathtt{Rep}} 

\newcommand{\RepPb}{\Rep_P(b)}
\newcommand{\RepPpb}{\Rep_{P'}(b)}

\newcommand{\leqfty}{\leqslant_\infty}

\newcommand{\init}{\mathtt{init}}

\newcommand{\dom}[1]{\ttdom(#1)}
\newcommand{\codom}[1]{\ttcodom(#1)}

\newcommand{\Ax}{\ttAx}
\newcommand{\Hd}{\mathtt{Hd}}
\newcommand{\Tl}{\mathtt{Tl}} 

\newcommand{\pos}[1]{\ttpos(#1)}

\newcommand{\Axa}{\Ax_{a}}
\newcommand{\AxP}{\Ax^P}
\newcommand{\AxPa}{\AxP_a}

\newcommand{\tr}[1]{ \tttr(#1)}  
\newcommand{\trP}[1]{\tttr^P(#1)}  

\newcommand{\Rt}{\mathtt{Rt}}

\newcommand{\Red}{\mathtt{Red}}
\newcommand{\Exp}{\mathtt{Exp}}
\newcommand{\Res}{\mathtt{Res}}

\newcommand{\rk}[1]{\mathtt{rk}(#1)}

\renewcommand{\deg}[1]{\ttdeg(#1)}
\newcommand{\cadout}[1]{\mathtt{cr}_{\mathtt{out}}(#1)}
\newcommand{\cadin}[1]{\mathtt{cr}_{\mathtt{in}}(#1)}

\newcommand{\cop}{\mathtt{cop}}
\newcommand{\cip}{\mathtt{cip}}

\newcommand{\pf}{\mathtt{pf}}

\newcommand{\rewfty}{\rew^{\infty}}
\newcommand{\Lamfty}{\Lam^{\infty}}
\newcommand{\Lamuuu}{\Lam^{111}}
\newcommand{\Lamzzu}{\Lam^{001}}

\newcommand{\Lamuzu}{\Lam^{101}}

\newcommand{\subb}[2]{[#1/#2]}
\newcommand{\subx}[1]{\subb{#1}{x}}

\newcommand{\subux}{\subb{u}{x}}

\newcommand{\tux}{t\subux}

\newcommand{\cd}{\mathtt{cd}}

\newcommand{\rdeg}{\mathtt{clev}}

\newcommand{\p}{\mathtt{p}}

\newcommand{\Approx}[1]{\mathtt{Approx}(#1)}

\newcommand{\ax}{\ttax}
\newcommand{\abs}{\ttabs}
\newcommand{\app}{\ttapp}

\newcommand{\Rst}{\ttR}

\newcommand{\ArgTr}{\mathtt{ArgTr}} 

\newcommand{\AT}{\mathtt{AT}}

\newcommand{\ArgPos}{\mathtt{ArgPos}}
\newcommand{\Cal}{\mathtt{Call}}

\newcommand{\QRes}{\mathtt{QRes}}

\newcommand{\red}[1]{\textcolor{red}{#1}}
\newcommand{\blue}[1]{\textcolor{blue}{#1}}

\newcommand{\violet}[1]{\textcolor{violet}{#1}}

\definecolor{greenone}{RGB}{0,200,140}

\definecolor{darkred}{RGB}{220,00,00}

\newcommand{\rewb}[1]{\stackrel{ #1}{\rightarrow}}
\newcommand{\code}[1]{\lfloor #1 \rfloor}

\newcommand{\ad}[1]{\mathtt{ad}(#1)}
\newcommand{\adr}[1]{\mathtt{adr}(#1)}
\newcommand{\cad}[1]{\mathtt{cr}(#1)}
\newcommand{\BT}[1]{\ttBT(#1)}

\newcommand{\trck}[1]{\,\red{[#1]}}

\newcommand{\posPr}[1]{~ \red{\langle#1 \rangle}}

\newcommand{\fA}{\supf A}
\newcommand{\fP}{\supf P}
\newcommand{\fPi}{\supf \Pi}

\newcommand{\fU}{\supf U}
\newcommand{\oB}{\supo B}

\newcommand{\Appfty}[1]{\mathtt{Approx}_{\infty}(#1)}

\newcommand{\rstra}{\rstr{a}}
\newcommand{\rstrb}{\rstr{b}}
\newcommand{\trb}{t\rstrb}
\newcommand{\tral}{t\rstr{\al}}

\newcommand{\tprb}{t'\rstrb}

\newcommand{\tra}{t\rstra}
\newcommand{\tras}{t\rstr{a_*}}

\newcommand{\tri}{\rhd}
\newcommand{\ra}{\mathring{a}}

\newcommand{\trra}{t\rstr{\ra}}
\newcommand{\rA}{\mathring{A}}
\newcommand{\rT}{\mathring{\ttT}}

\newcommand{\asc}{\mathtt{asc}}

\newcommand{\Asc}[1]{\mathtt{Asc}(#1)}

\newcommand{\twoarewa}{(2\cdot  \tv)  \rew \tv}

\newcommand{\juGtB}{\ju{\Gam}{t:B}}
\newcommand{\juGtpB}{\ju{\Gam}{t':B}}

\newcommand{\juRax}[2]{
\ju{#1:\mult{#2}}{#1:#2}
  }
\newcommand{\axxRo}[1]{\juRax{x}{#1}}

\newcommand{\juaxtt}{\axxRo{\tau}}

\newcommand{\juGtt}{\ju{\Gam}{t:\tau}}
\newcommand{\juGtpt}{\ju{\Gam}{t':\tau}}

\newcommand{\jufara}{\ju{f:\mult{\arewa}}{f:\arewa}}
\newcommand{\jufera}{\ju{f:\mult{\erewa}}{f:\erewa}}
\newcommand{\fara}{f:\arewa}

\newcommand{\juxkT}{\ju{x:(k\cdot T)}{x:T}}

\newcommand{\juCtt}{\ju{C}{t:T}}
\newcommand{\juCttsh}{\ju{C\!}{\!t\!:\!T}}
\newcommand{\juCtpt}{\ju{C}{t':T}}

\newcommand{\ttsupp}{\mathtt{supp}}
\newcommand{\ttbisupp}{\mathtt{bisupp}}
\newcommand{\supp}[1]{\ttsupp(#1)}
\newcommand{\bisupp}[1]{{\ttbisupp(#1)}}
\newcommand{\bisuppR}[1]{\ttbisupp_{\ttR}(#1)}

\newcommand{\supps}[1]{\ttsupp_{*}(#1)}

\newcommand{\eqr}{\equiv_{\scrR}}

\newcommand{\ttsz}{\mathtt{sz}}

\newcommand{\lra}{\leftrightarrow}             

\newcommand{\sz}[1]{\ttsz(#1)}
\newcommand{\fv}[1]{\mathtt{fv}(#1)}

\newcommand{\sep}{\hspace*{0.5cm}}   
\newcommand{\msep}{\hspace*{0.2cm}}  

\newcommand{\Dels}{\Del_{*}}

\newcommand{\Id}{\mathtt{I}}   

\newcommand{\pcons}{\mathtt{ptyp}}

\newcommand{\NFfty}{$\text{NF}_{\infty}$}

\newcommand{\prefleq}{\leqs}

\newcommand{\ttShp}{\ttS_{\mathtt{hp}}}

\newcommand{\Perm}[1]{\mathfrak{S}_{#1}}

\newcommand{\tthp}{\mathtt{hp}}
\newcommand{\psym}{h}

\newcommand{\hpt}{h} 

\newcommand{\HP}{\mathtt{HP}}

\newcommand{\PPP}{\mathtt{PPP}}

\newenvironment{proofsketch}
        {\bgroup \noindent \textit{Proof sketch.}}{\hfill \qedsymbol \egroup\medskip}

\newcommand{\ndrth}[4]{\node[draw] (#1) at (#2,#3) [thick,rounded corners=#4pt]
}
\newcommand{\ndrvth}[4]{\node[draw] (#1) at (#2,#3) [very thick,rounded corners=#4pt]
}
\newcommand{\slbox}[2]{{\begin{minipage}{#1cm}\textcolor{black}{#2}\end{minipage}} }
\newcommand{\btikz}{\bcen\begin{tikzpicture}}
\newcommand{\etikz}{\end{tikzpicture}\ecen}

\newcommand{\bltikz}{\begin{tikzpicture}}
\newcommand{\eltikz}{\end{tikzpicture}}

\newcommand{\trans}[3]{ 
  \begin{scope}[xshift=#1cm,yshift=#2cm]
     #3
  \end{scope}
}

\newcommand{\transh}[2]{ 
  \begin{scope}[xshift=#1cm]
     #2
  \end{scope}
}

\newcommand{\transv}[2]{ 
  \begin{scope}[yshift=#1cm]
     #2
  \end{scope}
}

\newcommand{\drawlabnode}[3]{  
     \trans{ #1 }{ #2 }{ 
       \draw (0,0) node {\small #3 };
       \draw (0,0) circle (0.23);
     }
}

\newcommand{\drawarob}[2]{   
  \trans{#1}{#2}{
       \draw (0,0) node {\small $\arob$};
       \draw (0,0) circle (0.23);
    }
}
\newcommand{\drawf}[2]{      
  \trans{#1}{#2}{  
    \draw (0,0) node {\small $f$};
    \draw (0,0) circle (0.23);
  }
}

\newcommand{\drawtri}[3]{   
  \trans{#1}{#2}{  
   \draw (0,0) --++ (0.55,0.9) --++ (-1.1,0) -- cycle ;
    \draw (0.07,0.6) node{\small#3} ; 
  }
}

\newcommand{\drawsmalltriin}[3]{   
  \trans{#1}{#2}{  
    \draw (0,0) --++ (0.35,0.65) --++ (-0.7,0) --++ (0.35,-0.65) ;
    \draw (0,0.38) node{\fnsz #3} ; 
  }
}

\newcommand{\drawlefttail}[2]{   
  \trans{#1}{#2}{
      \draw (0,0) --++ (-0.13,-0.18);
    }
  }

\newcommand{\drawrighttail}[2]{   
  \trans{#1}{#2}{
      \draw (0,0) --++ (0.13,-0.18);
    }
  }

\newcommand{\drawcuf}[2]{    
   \drawtri{#1}{#2}{\small $\cuf$}   
}

\newcommand{\blockunary}[3]{ 
  \drawlabnode{#1}{#2}{#3}
  \trans{#1}{#2}{
  \draw (0,0.23) --++ (0,0.44);
}
  }

\newcommand{\blockfacuf}[2]{   
  \trans{#1}{#2}{  
    \drawarob{0}{0}
    \drawf{-0.65}{0.9}
    \drawcuf{0.65}{0.9}
    \draw (-0.12,0.18) -- (-0.53,0.72); 
    \draw (0.12,0.18) -- (0.65,0.9);
  }
}
\newcommand{\blockfa}[2]{  
  \trans{#1}{#2}{  
    \drawarob{0}{0}
    \drawf{-0.65}{0.9}
    \draw (-0.12,0.18) -- (-0.53,0.72); 
    \draw (0.12,0.18) -- (0.53,0.72);
  }
}

\newcommand{\blocka}[2]{  
  \trans{#1}{#2}{  
    \drawarob{0}{0}
    \draw (-0.12,0.18) -- (-0.53,0.72); 
    \draw (0.12,0.18) -- (0.53,0.72);
  }
}

\newcommand{\blockfabis}[2]{
\trans{#1}{#2}{
    \draw (0.18,0.18) -- (0.72,0.72);  
    \draw (-0.9,0.9) node{$f$} ;
    \draw (-0.9,0.9) circle (0.25) ;
    \draw (-0.72,0.72) -- (-0.18,0.18) ;
    \draw (0,0) circle (0.25) ;
    \draw (0,0) node{$\arob$} ; }
  }
\newcommand{\blockcufbis}[2]{
\trans{#1}{#2}{
  \draw (0,0) -- (0.9,1.4) --++ (-1.8,0) -- cycle ;
  \draw (0,0.9) node{$\cuf$} ;
  \draw (-0.18,-0.18) -- (0,0) ;}
  }

\newcommand{\inputarewa}[2]{
  \trans{#1}{#2}{
 \draw (-0.9,0.7) node{$\red{\arewa}$}; 
   \red{\draw [>=stealth, ->] (-0.7,0.5)--(-0.25,0.25) ; }    
}}

\newcommand{\inputpharewa}[2]{
  \trans{#1}{#2}{
 \draw (-0.9,0.7) node{$\red{\mult{\phantom{\tv}}\rew\tv}$}; 
   \red{\draw [>=stealth, ->] (-0.7,0.5)--(-0.25,0.25) ;     
}}}

\newcommand{\outputtv}[2]{
  \trans{#1}{#2}{
\red{\draw [>=stealth, -> ] (0.3,-0.15)-- (0.75,-0.35) ; 
       \draw (0.95,-0.45) node{$\tv$} ;
    }
}}

\newcommand{\techrep}[1]{#1}

\begin{document}
\title{Sequences and Infinitary Intersection Types }

\author[P. Vial]{Pierre Vial}
\address{Inria Paris-Saclay, France}
\email{pierre.vial@inria.fr}
\keywords{intersection types, infinitary lambda-calculus, coinduction, non-idempotent intersection, sequence types, B\"ohm trees}

\maketitle

\begin{abstract}
 We introduce a new representation of \textbf{non-idempotent} intersection types, using \textbf{sequences} (families indexed with natural numbers) instead of lists or multisets. 
This allows scaling up intersection type theory to the infinitary $\lam$-calculus. 
\chhp{We thus characterize  \textbf{hereditary head normalization} (Klop's Problem) and we give a \textit{unique} type to all \textbf{hereditary permutators} (TLCA Problem \#20), which is not possible in a finite system.}{We thus characterize \textbf{hereditary head normalization}, which gives a positive answer to a question known as \textbf{Klop's Problem}.}  
On our way, we use non-idempotent intersection to retrieve some well-known results on infinitary terms. 

\end{abstract}












\section{Presentation}

Intersection types, introduceed by Coppo and Dezani~\cite{CoppoD78,CoppoD80,CoppoDV81}, are a powerful tool to describe and characterize the semantic properties of programs. 
For instance, compared to simple types, they type more terms, such as $\Del:=\lx.x\,x$, the self-application, and they  allow building models~\cite{BarendregtCoppoDezani83,AlessiBD06,PaoliniPR17} for the $\lam$-calculus. They also allow capturing observational equivalence (see~\cite{ManzonettoHDR17} for instance).
Various sets of $\lam$-terms have been characterized with (a form of) typability within intersection type systems. 

The question was naturally raised to know whether some sets of programs defined in terms of infinitary behaviors (for instance, as captured by \Bohm\ trees), such as the set of hereditary permutators~(TLCA Problem \#~20) 
could be characterized in such a system.
In particular, this led Klop and Dezani to formulate the following question~\cite{Tatsuta08}: can the set of  \textbf{hereditary head normalizing terms} (\ie the terms whose \Bohm\ tree does not contain $\bot$) be characterized by means of an intersection type system? This question was known as \textbf{Klop's Problem}.  
One may perhaps undestand better the definition of hereditary head normalizing terms by interpreting them as a generalization of weakly normalizing terms (terms which have a $\beta$-normal form) 
in that, $t$ is hereditary head  normalizing if, at any fixed depth $d$, $t$ can be reduced to a normal form up to subterms which occur below depth $d$. Such a term may not have a normal form: this occurs whenever redexes cannot be completely eliminated by reducing $t$, even if the redexes (of the reducts) of $t$ keep on increasing in depth. For instance, $\cuf:=\Delf\,\Delf$ (with $\Delf=\lx.f(x\,x)$) a variant of Curry fixpoint that will be the case study of this article, satisfies $\cuf \hred f(\cuf)$ (it operational semantics is given in Fig.~\ref{fig:cuf-red}). Thus, after $n$ steps of reductions, $\cuf$ gives $f^n(\cuf)$ (shortening $f(f(\ldots(\cuf))))$ with $n$ occurrences of $n$), which, on one hand, shows that $\cuf$ is not weakly normalizing, since its reducts all have a (unique) redex. On another hand, the only redex of $f^n(\cuf)$ occurs at depth $n$, which means that, above depth $n$, $f^n(\cuf)$ is an normal form. Asymptotically, when $n$ converges towards $\infty$, the computation of $\cuf$ outputs\label{disc:cuf-beginning} an infinite $\lam$-tree (corresponding to a infinite application of $f$, which we call $\fom$) which does not contain any redex: in other words, an infinite normal form. And this is the core reason why hereditary head normalization pertains to an infinitary behavior.

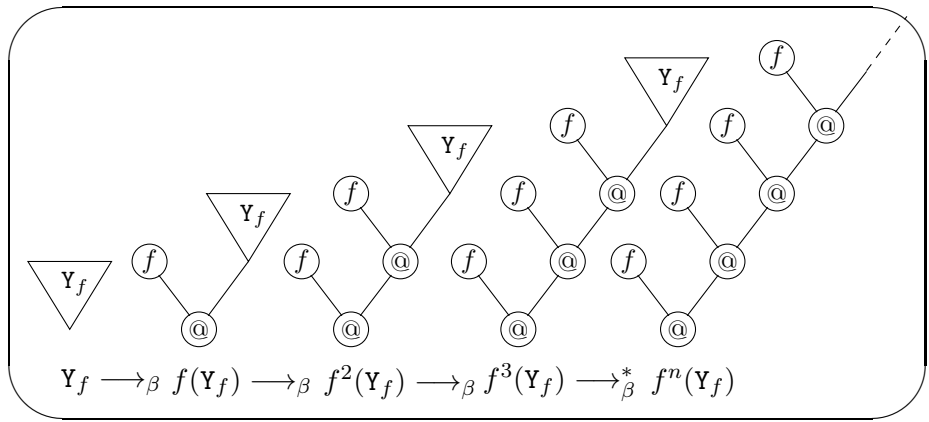
\begin{figure}[!h]
  \begin{center}
\ovalbox{
    \begin{tikzpicture}
\drawcuf{0}{0} 

\trans{1.7}{0}{
  \blockfacuf{0}{0}
}

\trans{3.7}{0}{
  \blockfa{0}{0}
  \blockfacuf{0.65}{0.9}
}

\trans{5.9}{0}{
  \blockfa{0}{0}
  \blockfa{0.65}{0.9}
  \blockfacuf{1.3}{1.8}
}

\trans{8}{0}{
  \blockfa{0}{0}
  \blockfa{0.65}{0.9}
  \blockfa{1.3}{1.8}
  \blockfa{1.95}{2.7}
  \draw [dashed] (2.48,3.43) --++ (0.53,0.72);
}

\draw (-0.25,-0.7) node [right]{$\cuf \longrightarrow_\beta$};
\draw (1.2,-0.7) node [right] {$f(\cuf) \longrightarrow_\beta$};
\draw (3.2,-0.7) node [right] {$f^2(\cuf) \longrightarrow_\beta$};
\draw (5.3,-0.7) node [right] {$f^3(\cuf) \longrightarrow^*_\beta$};
\draw (8.2,-0.7) node {$f^n(\cuf)$};
\end{tikzpicture}}
  \caption{$\cuf$, a Hereditary Normalizing Term}
  \label{fig:cuf-red}
\end{center}
\end{figure}

However, Tatsuta~\cite{Tatsuta08,Tatsuta08b} proved in 2008 that ``infinitary behaviors'' could not be characterized in an inductive, finitary setting. In particular, the answer to Klop's Problem is negative in the inductive case, as Tatsuta showed by:
\begin{itemize}
\item making the observation that, in an inductive type system, the set of typable terms is recursively enumerable.
\item proving that the set of hereditary head normalizing terms\chhp{ and that of hereditary permutator are}{is} \textit{not} recursively enumerable. 
\end{itemize}
In turn, this result raised the \textbf{question of whether such a characterization could be obtained by means of an \textit{infinitary/coinductive intersection type system}}.

In this article, we prove that it is possible and we give a type-theoreric characterization of hereditary head normalization using \textbf{non-idempotent intersection}. Non-idempotent intersection type systems were introduced in 1980 by Coppo, Dezani and Veneri~\cite{CoppoDV80a} before being rediscovered by Gardner~\cite{Gardner94} and de Carvalho in his PhD~\cite{Carvalho07,Carvalho18}, which emphasized they quantitative aspect. They have been since extensively studied (read~\cite{BucciarelliKV17} for a survey).
Non-idempotent intersection $\wdg$ does not satisfy the equality $A\wdg A = A$ as the set-theoretic intersection does ($X\cap X=X$).  
In its simplest form of this system, non-intersection is represented with \textbf{multisets} (the multiset $\mult{A,B,A}$ stands for the intersection $A\wdg B\wdg A$) and gives a typing system with only three rules, which is recalled in \Sec~\ref{ss:system-Ro-Klop}.
Then, reduction inside a derivation almost boils down to \textit{moving} parts of the initial derivation (see Fig.~\ref{fig:non-deter-srRo} and discussion in \Sec~\ref{ss:non-determinism}) and the proof of soundness (if a term is typable, then it is head normalizing) becomes trivial, since reduction decreases a positive integral measure on derivations.

To solve Klop's Problem, we will use a quantitative, resource-aware type system. 
\igintro{In the simplest form of this system, intersection is represented with multisets, as explained at the end of \Sec~\ref{ss:syntax-direction-discuss-intro}.  }
However, it turns out that a direct coinductive adaptation of \textit{infinite multisets} cannot work for two major reasons:
\begin{itemize}
\item It would lead to the possibility of typing some unsolvable terms\footnote{We actually proved that the coinductive version of $\scrRo$ types every term can typed in $\scrR$, the coinductive version of $\scrRo$~\cite{VialLICS18}}, like $\Om$. That is why a validity criterion is needed to discard irrelevant derivations, as in other infinitary frameworks \cite{santocanale01brics}. This validity criterion relies on the idea of \textbf{approximability} (every valid infinite derivation must be, in some sense, a limit of finite derivations), which we informally develop in \Sec~\ref{ss:sequence-tracking-klop-intro}, \ref{ss:typ-inf-nf-informal}, \ref{ss:degenerate}  and formally define in \Sec~\ref{s:approx}.
\item Moreover, as it turns out, multisets are not fit to formally express such a notion, because they do not allow pointing,  meaning for instance that we cannot distinguish the two occurrences of $\sig$ in $\mult{\sig,\tau,\sig}$. 
This motivates the need for rigid constructions: multisets
are then (coinductively) replaced by \textbf{sequences}, \ie families indexed by (non necessarily consecutive) natural numbers called \textbf{tracks} (\eg $(2\cdot \sig,3\cdot \tau,9\cdot \sig)$ instead of $\mult{\sig,\tau,\sig}$). We explain all this in \Sec~\ref{ss:sequence-tracking-klop-intro} below.
\end{itemize}

This leads us to define a type system that we call system $\ttS$, in which intersection is  represented by sequences of types (intersection is said to be sequential). Tracks, which act as \textit{identifiers}, constitute the main feature of system $\ttS$ presented here. 
With those pointers, any type can be tracked  through the rules of a whole typing derivation, and the notionsof approximation and of limit of a family of derivation can be defined.

 This contribution was presented in~\cite{VialLICS17} with some proof sketches.  Here, not only we provide full proofs (which are detailled in appendices), but in the core of the paper, we detail a crucial \textbf{argument of soundness} (if a term is suitably typable, then it has a finite or infinite normal form). We also explain in more details how infinite normal forms can be suitably typed.

 The expert reader on intersection types may be surprised by the following elements: while the subject reduction property (in its usual form) remains simple to prove,  soundness (typability implies normalization) and (infinitary) subject expansion are quite more technical than in the usual, finite non-idempotent case.
 The use of a validity criterion constrains us to prove subject reduction while considering residuals (\Sec~\ref{ss:one-step-sr-se} and \ref{ss:sr-proof}). Approximability (derivation are limits of finite ones) raises a problem with the left-hand side of judgments, which can be overcome by showing that they can be, in some sense, ignored (\Sec~\ref{ss:equinecessity-main}). Typing normal forms \textit{soundly} is also technical, because we must check that it can be done while respecting approximability (\Sec~\ref{ss:support-candidates} to \ref{ss:nf-approx}).
 At last, infinitary subject expansion is a subtle process, based on cutting infinite derivations into finite ones, proceeding to a substitution in a reduction path of infinite length followed by a finite number of expansion stepsn, and then considering a suitable join of the underlying constructions. It is incrementally explained in \Sec~\ref{ss:klop-prob-hp-statement}, \ref{ss:typ-inf-nf-informal} and \ref{s:infty-expansion}.

\subsection*{Structure of the paper.} This article is organized as follows: in \Sec~\ref{s:intro-klop}, we recall the mechanisms of intersection types on which we will base our approach and we present the main difficulties we shall have to overcome, along with a sketch of the tools we develop in this article. In \Sec~\ref{s:finite-inter-infinite-types}, we formally recall the \textit{infinitary} $\lam$-calculus and the \textit{finitary} non-idempotent intersection type system of Gardner-de Carvalho (system $\scrRo$) which use multisets to represent intersection.  \Sec~\ref{ss:system-Ro-Klop} is partly informal: we explain how we will be able to perform later in this article infinitary subject expansion (bring a typing backwards along an infinite sequence of reduction steps) and why this actually cannot be done with multiset intersection. In \Sec~\ref{s:rigid}, we formally define system $\ttS$ and sequence types. In \Sec~\ref{s:dynamics}, we study the \textit{one-step} dynamics of system $\ttS$: we define residuals and we prove a deterministic subject reduction property and the subject expansion property 
Approximability is defined in \Sec~\ref{s:approx} and the infinitary subject reduction property and the normalization of typable terms are proved. \Sec~\ref{s:normal-forms} is devoted to typing infinite normal forms and proving the infinitary subject expansion property. This concludes the study of Klop's Problem and of hereditary head normalization.\ighp{ In the short \Sec~\ref{s:charac-hp-S}, we use the previous contributions to characterize hereditary permutators in system $\ttS$. In the last section (\Sec~\ref{s:characterizing-hp-with-a-unique-type}), we define an extension of system $\ttS$, system $\ttShp$, that we prove to be sound and complete and to characterize hereditary permutators with a unique type.}

This brings us to the main contribution of this article, Theorem~\ref{th:charac-WN-S}: from \Sec~\ref{ss:typ-inf-nf-informal} to \Sec~\ref{s:normal-forms}, we solve Klop's Problem: we give a type theoretic characterization of hereditary head normalization. 
A general presentation of the content of the appendices, in which we give technical arguments that may be missing, is given on p.~\pageref{disc:app-pres}.

\section{Introduction}
\label{s:intro-klop}

In this section, we first recall basic notions on normalization, then we present some aspects of intersection types. We then devote our attention to Klop's Problem and the technical difficulties we had to overcome and the proof techniques we had to introduce to export intersection type theories to an infinitary setting and to characterize hereditary head normalization. Last, we present two main tools that we use in this article: sequences and approximations/truncations of derivations.

\subsection{Normalization}
\label{s:eval-strat-intro-klop}

\subsection*{Evaluation strategies}
We consider the following restrictions of $\beta$-reduction:
\begin{itemize}
\item The \textbf{head reduction}, denoted $\hred$, consists in reducing the \textbf{head redex} $\lxrs$ of a \textbf{head reducible} term $\hreducibleo$ (notice that any term that is not a HNF is head reducible). A term $t$ is \textbf{Head Normalizing (HN)} if the head reduction terminates on $t$ (and necessarily outputs a head normal form). 
\item The \textbf{leftmost-outermost reduction}, denoted $\lred$, consists in reducing the \textbf{leftmost-outermost redex} of a $\beta$-reducible term  $t$, \ie the one  whose $\lx$ is the leftmost one in $t$ (seen as string of characters). A term $t$ is \textbf{Weakly Normalizing (WN)} if leftmost-outermost reduction terminates on $t$ (and necessarily outputs a full normal form).
\item The \textbf{hereditary head reduction strategy} (formally defined in \Sec~\ref{ss:comput-bohm-trees-hp}) also aims at  computing a full normal form: for a given term $t$, it starts by applying the head reduction strategy on $t$ until it stops (otherwise, it loops). In that case, the current term is of the form $\hnfo$. Then one applies the head reduction strategy on each head argument $t_1,\ldots,\;t_q$ until they are all in head normal forms (if it does not loop for one of the $t_i$). Then one goes on with head reducing each new argument at constant applicative depth and so on. The strategy stops if a $\beta$-normal form is reached. Hereditary head reduction may be seen as a more balanced variant of leftmost-outermost reduction. Moreover, 
when  hereditary head reduction strategy does not stop but keeps on reducing deeper and deeper, it may be seen as asymptotically computing 
an infinite term, as it was recalled above. 
\end{itemize}

\subsection*{Having a normal form}
\chhp{
In rewriting theory~\cite{TeReSe}, \textit{normalization} also describes other forms of termination. The two important ones in this article are the following: }{Normalization may also be associated with reduction paths of unspecified forms (without referring to evaluation strategies). In this article, we mainly consider the two predicates:} 
\begin{itemize}
\item A $\lam$-term $t$  \textbf{has a head normal form} if there is a reduction path from $t$ to a head normal form.
\item A term $\lam$-term $t$ \textbf{has a full normal form} if there is a reduction path from $t$ to a full normal form.
\end{itemize}
Notice that having a head normal form seems to be more general/less constrained than
head normalization: the former pertains to the existence of a reduction path (to a head normal form) of an \textit{unspecified} form whereas the head reduction  reduces terms in a \textit{deterministic} way. 
Thus, if the head reduction strategy  terminates on a term $t$, then $t$ clearly has a head normal form.\ighp{ Likewise, if the leftmost-outermost strategy terminates on $t$, then $t$ has a full normal form.}  
Moreover, it is well-known of that the converse of these two implications is true, \ie if a term $t$ has a head \ighp{(\resp a full) }normal form then the head reduction strategy \ighp{(\resp the leftmost-outermost reduction strategy) }terminates on $t$.\ighp{ However, these two converse implications arenon-trivial and involve standardization, as we recall in the following paragraph.}


\igintro{
\subsection*{The Standardization Theorem and Reduction Strategies} 
The standardization theorem, due to Curry and Feys~\cite{CurryF58}, roughly states that if $t\bred^* t'$, then there is a reduction path from $t$ to $t'$ that never reduces on the left-hand side of a \textit{residual} of a redex that has already been contracted  (a proof may be found in chapter 11 of \cite{Barendregt85} or chapter 3 of the electronic version of \cite{Krivine93}). The fact that the head reduction strategy is complete for head normalization is usually proved as a corollary of the standardization theorem (or a slightly weaker result known as  pseudo-standardization). However, (pseudo-)standardization is quite non-trivial: one needs to prove that the confluence of the $\lam$-calculus (Church-Rosser property), the finite developments theorem and then use involved arguments to prove that series of developments can be commuted.

Besides, it is easy to prove the equivalence between the assertions ``the head reduction strategy is complete for head normalization'' and ``the leftmost-outermost reduction strategy is complete for weak normalization'', but it requires the Church-Rosser property. Yet, as we will see in \Sec~\ref{ss:inf-norm-bohm-intro}, the infinitary $\lam$-calculus only satisfies a weak (but not less difficult) form of confluence.  
}

\igintro{

\subsection*{Confluence in the infinitary $\lam$-calculus} The infinitary $\lam$-calculus is not confluent, for instance, $t:=(\lx.I(x\,x))(\lx.I(x\,x))\bred^2 \Om$, which only reduces to itself. Moreover, $t\hred I\,t\bred I(I\,t)\hred I(I(I\,t))\ldots$, so that $t$ asymptotically reduces to $I^{\om}:=I(I(I\ldots))$, which also only reduces to itself.  Thus, the two reducts $\Om$ and $I^{\om}$ of $t$ do not have a common reduct.

However, the infinitary $\lam$-calculus satisfies a weak form of confluence: in the example above, observe that $t$, $\Om$ and $I^{\om}$ are not head normalizing. Intuitively, they are \textit{meaningless} terms and in the \Bohm\ trees formalism, they would be replaced by the constant $\bot$. This is how confluence is retrieved: the infinitary $\lam$-calculus has the Church-Rosser property up to the identification of any pair of unsolvable subterm. This is weaker than true confluence, but the various proofs of this fact is actually quite involved and spans on  dozen of pages~\cite{KennawayKSV97,Joachimski04,Bahr18}, contrary to the confluence of the finite $\lam$-calculus.}

\subsection{Intersection Types}
\label{ss:intersection-overview-klop-intro}

We recall and discuss here some aspects of intersection types that will interest us in this article. All these aspects are illustrated in system $\scrRo$ in the forthcoming \Sec~\ref{ss:system-Ro-Klop}.

Normalization (termination) is a \textbf{dynamic} property: it pertains to the reduction of $\lam$-terms, \ie their \textit{evaluation}. For instance, proving that a term $t$ is head normalizing consists (in principle) in  giving a sequence of $\beta$-reduction steps starting at $t$ and ending with a head normal form.

On another hand, typing is \textbf{static} in the following sense: in a given type system, finding a derivation typing  a term $t$ or checking that a given tree of typing judgments is a correct derivation is done without having to perform reduction steps but just by considering the structure of the tree and verifying that the typing rules are correctly applied.

One of the desirable features of type systems is that some of them \textit{ensure} normalization, \eg if a term $t$ is typable in Girard's system F, then it is strongly normalizing~\cite{Girard72,GiraFont} (\ie no infinite reduction path starts at $t$). Thus, type systems provide \textit{static} proofs of \textit{dynamic} properties. In some systems, when we have managed to type a program $t$, we do not need to execute $t$  to know that $t$ is terminating.  
However, usually, when typability is related to termination, \eg in polymorphic or dependent type systems, it is a \textit{sufficient} condition for normalization but not a \textit{necessary} one: many normalizing terms are not typable. For instance, $\Del=\lx.x\,x$ is not typable in the simply typed $\lam$-calculus, and  $(\lz.  \ly. y(z\,\Id)(z\,\ttK))(\lx. x\,x)$, although strongly normalizing, 
is not typable in system F. 

The basic idea behind intersection type systems is allowing assigning several types to variables or to terms. They often \textit{characterize} (and not only ensure) normalization properties (head, weak, weak head, strong\ldots): in an intersection type system, typability is usually equivalent with (a notion of) normalization. Since the $\lam$-calculus is  Turing-complete, this makes  type inference undecidable in most intersection type  systems.
As an exception to this rule, intersection types were famously used by Kobayashi and Ong~\cite{KobayashiO09} to prove that Monadic Second Order logic is decidable for higher-order recursion schemes.\\


\noindent \textbf{Normalization theorems in intersection type theory.} %
Intersection type systems often give interesting properties pertaining to reduction strategies. To illustrate this, let us imagine for a moment that we have an intersection type system $\calH$ characterizing head normalization, \ie that satisfies the following theorem (let us call it $T_{\text{HN}}$):
\bcen
``\textit{Theorem $T_{\text{HN}}$:} for every $\lam$-term $t$, $t$ is is typable in system $\calH$ iff $t$ is head normalizing ''
\ecen 
In practice, this characterization theorem is proved in two steps:
\begin{itemize}
\item[$\Rightarrow$] A proof of ``If $t$ is $\calH$-typable (typable in system $\calH$), then the head reduction \textit{strategy} terminates on $t$''.
\item[$\Leftarrow$] A proof of ``If there is a reduction path from $t$ to a HNF, then $t$ is $\calH$-typable''.
\end{itemize}
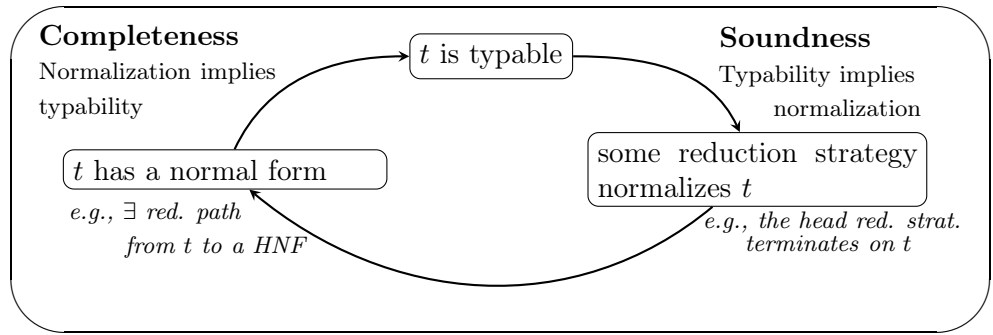
\begin{figure}
\ovalbox{
  \begin{tikzpicture}
\draw (4.5,1) node {\slbox{3}{\textbf{Soundness}\\
{\fnsz Typability implies\\\hspace*{0.6cm} normalization}}};
\draw (-3.5,1) node {\slbox{4.9}{\textbf{Completeness}\\
{\fnsz Normalization implies\\ typability}}};

    \node[draw] (RS) at (3.5,-0.3) [rounded corners=4pt] {
    \parbox{4.2cm}{some reduction strategy normalizes $t$}}; 
    \draw (2.65,-1) node [right] {\fnsz\textit{\eg the head red. strat.}};
    \draw (3.2,-1.28) node [right] {\fnsz\textit{terminates on $t$}};
    
  \node[draw] (TS) at (-3.5,-0.3) [rounded corners=4pt] {
    \parbox{4cm}{
    $t$ has a normal form}}; 
  \draw (-1.4,-1.1) node [left] {\slbox{4}{\fnsz\textit{\eg $\exists$ red. path\\\hspace*{0.6cm} from $t$ to a HNF}
      }};

    \node[draw] (TT) at (0,1.2) [rounded corners=4pt] {
      $t$ is typable};

    \draw [->,>=stealth,thick] (TT) to[out=0,in=115] (RS);
   \draw [->,>=stealth,thick] (RS) to[out=-140,in=-40] (TS);
   \draw [->,>=stealth,thick] (TS) to[out=65,in=180] (TT);

  \end{tikzpicture}
  }

\caption{Characterizing a notion of normalization with intersection types}
\label{fig:its-fundamental-diag}
\end{figure}

Observe that for $\Rightarrow$, we actually prove that a typable $t$ is terminating for the head reduction strategy, which is \textit{a priori} stronger than just the fact that $t$ has a head normal form\igintro{, as we saw in \Sec~\ref{ss:intro-klop-norm-and-red-strat}}. This proves that, in some sense, head reduction is \textit{complete} since it computes a head normal form when it exists. Since the termination of head reduction obviously implies having a head normal form, the implication $\Leftarrow$ gives us theorem $T_{\text{HN}}$ but also the following, non-trivial statement:
\bcen
\hspace*{-5cm} ``A term $t$ is head normalizing\\\hspace*{2.5cm} iff it is terminating for the head reduction strategy.''
\ecen
A remarkable aspect of this equivalence is that it is \textit{not} type-theoretic: it states that a certain reduction strategy is complete for a certain notion of normalization. But it was proved using types instead of syntactic techniques. The structure of this kind of proof is summarized in Fig.~\ref{fig:its-fundamental-diag}. The arrow on the right (typability implies normalization, or soundness) usually relies on subject reduction, whereas the arrow on the left (normalization implies typability, or completeness) relies on subject expansion (typing is stable under expansion).
To the best knowledge of the author, subject expansion is specific to intersection type system, unless it is hard-coded in the typing rules.
Indeed, to perform subject expansion (going from a typing of $\rsx$ to one of $\lxrs$), the variable $x$ can be assigned \textit{several} types and thus, be used as a placeholder for $s$, which can occur with several different types in $\rsx$.

\igintro{
\subsection*{Subject reduction and expansion} Now, let us explain more precisely how the implication $\Rightarrow$ and $\Leftarrow$ above are usually obtained.
\begin{itemize}
\item[$\Rightarrow$] Typability implies normalization:
\begin{itemize}
\item One proves \textbf{subject reduction}, meaning that typing is preserved under reduction (if $t\bred t'$ and $\juGtB$, then $\juGtpB$).
\item Using subject reduction and an additional argument (the nature and the difficulty of this additional argument is discussed below), one proves that some reduction strategy normalizes any typable term.
  \end{itemize}
\item[$\Leftarrow$] Proving that normalization implies typability is usually easier\footnote{In the inductive case. When considering infinite reduction paths, we will see that it is the main difficulty, \cf \Sec~\ref{ss:typ-inf-nf-informal}} and relies on two arguments:
\begin{itemize}
\item One proves \textbf{subject expansion}, meaning that typing is preserved under anti-reduction (if $t\bred t'$ and $\juGtpB$, then $\juGtB$).
\item One types the \textit{terminal states} (in our example,  head normal forms). Only an induction on the \textit{structure} of these terminal states is needed (see \Sec~\ref{ss:inter-from-syntax-intro-klop}). As we will see, this will be more or less true for \textit{infinite} normal forms as well.
\end{itemize}
\end{itemize}}%
\ighp{
This brings up several remarks:
\begin{itemize}
\item Typability implies normalization uses subject reduction along with an additional argument: this argument varies in nature depending on the type system.
\begin{itemize}
\item In the simply typed $\lam$-calculus, this is based upon the decrease of a suitable well-founded measure, usually a multiset containing the size of some types of the subterms in the subject of the derivation. This corresponds more or less to Gentzen's proof of cut-elimination~\cite{Gentzen34}. 
\item  In polymorphic type systems, one needs to use the reducibility candidates technique, due to Tait~\cite{Tait67} and extended by Girard to system F. This method sometimes requires dozens of pages.
\item In \textit{idempotent} intersection type system, one can use the reducibility technique or an argument \textit{à la Gentzen}~\cite{Valentini01}. 
\item Very interestingly, \textit{non-idempotent} intersection types provide the simplest possible proofs of normalization, because, in such a system, duplication of types is disallowed and derivations ``decrease'' in size at each reduction step,. As a consequence, reduction on typed terms must stop at some point. 
  \end{itemize}
\item Subject expansion does not hold in a type system without intersection: in particular, it is not satisfied in usual polymorphic or dependent type systems.
\end{itemize}

\begin{remark}
The proof scheme of Fig.~\ref{fig:its-fundamental-diag} applies to other notions of normalization (weak, weak head) but not to strong normalization, for which it needs to be adapted. Indeed, strong normalization is not preserved under anti-reduction, \eg $t':=y$ is strongly normalizing, $t:=(\lx.y)\Om\bred t'$ but $t$ is not strongly normalizing. The subject expansion property for intersection type systems characterizing strong normalization usually holds in a restricted form.
\end{remark}
}

\igintro{
\subsection{Intersection from the syntax perspective} 
\label{ss:inter-from-syntax-intro-klop}

Intersection types extend the syntax of the simply typed $\lam$-calculus by resorting to a new type constructor $\wdg$ (intersection). Naively, intersection would represent the set-theoretic intersection. Thus, $t:A\wdg B$ is derivable when $t:A$ and $t:B$ are derivable. For instance, since $\Id:A\rew A$ and $\Id:(A\rew B) \rew (A \rew B)$ are derivable, one has:
\bcen$\infer{\Id:A\rew A\hspace{1cm}\Id:(A\rew B) \rew (A \rew B)}{\Id: (A\rew A)\wdg (A\rew B)\rew (A\rew B)}\text{\small$\wdg$-intro}$\ecen

In turn, $\inter$ can be eliminated in the obvious way: if $t:A\wdg B$, we may assert $t:A$ as well as $t:B$. Note that intersection type systems also allow assigning several types to a variable in a context, \eg if $x$ is assigned $A\wdg (B\rew C)$ in a context, then $x$ may be used both as a term of type $A$ and as a term of type $B\rew C$.

Intuitively, intersection provides a finite form of polymorphism, \eg the type  $(A\rew A)\wdg  (A\rew B)\rew (A\rew B)$ can be thought as a double instance of the polymorphic type $\forall X.X\rew X$ (with $X=A$ and $X=A\rew B$). However, intersection is also less constrained than usual polymorphism \textit{à la} system F. For instance, the assignment $x:\tv \wdg (\tv' \rew \tv) \wdg (\tv \rew \tv\rew \tv')$ (with $\tv$ and $\tv'$ two distinct type variables) is sound, although the only polymorphic type that instantiates into  $\tv$,  $(\tv' \rew \tv)$ and $(\tv \rew \tv\rew \tv')$ is $\forall X.X$, \ie the type representing falsehood in system F.}

\igintro{ 
The unconstrained nature of $\wdg$ may give us intuitions on why normalization implies typability in an intersection type system:
\begin{itemize}
\item \textbf{All normal forms can be typed.} A simple structural induction allows typing every normal form.  Indeed, let $t=\hnfo$ be a normal form (so that $t_1,\ldots,t_q$ are also normal forms). Then, we can explain why a simple structural induction shows that $t$ is typable: assume that $t_1,\ldots,t_q$ are respectively typed with types $A_1,\ldots,A_q$ (induction hypothesis), then we type $t$ by assigning to the head variable of $t$ a suitable arrow type as follows:
$$ 
\infer[\abs^*]{
  \infer[\app]{
    \infer[\app^*]{
      \infer[\app]{\infer[\ax]{}{x:A_1\rew \ldots \rew A_q\rew B}
        \sep t_1:A_1     }{x\,t_1:A_2\rew \ldots \rew A_q \rew B}\sep t_2:A_2}
    {    x\,t_1\ldots t_{\!q\mathord{-}1}:A_q\rew B} \sep t_q:A_q}{
  \zhnfo:B}}{\hnfo:C_1\rew \ldots \rew C_p \rew B} $$
The star $*$ indicates several ($\geqs 0$) rules. 
For instance, let $t=(x\,(x\,y))\,(x\,\Id)$. Necessarily, $\Id$ has a type of the form $B\rew B$. We can then type $t$ as follows:
  \bcen $
  \infer{\infer{\infer{ }{x:C_1\rew C_2\rew C_3} \hspace{0.4cm}\infer{\infer{}{x:A\rew C_1}\hspace{0.3cm} \infer{\phd}{y:A}}{x\,y:C_1}   }{
    x\,(x\,y):C_2\rew C_3}
\hspace{0.5cm}
\infer{\infer{\phd}{x:(B\rew B)\rew C_2}\hspace{0.4cm}\stackrel{\vdots}{\Id:B\rew B}}{x\,\Id:C_2}}{(x\,(x\,y))\,(x\,\Id):C_3} 
$\ecen
Note we have assigned to $x$ the 3 types $A\rew C_1$ (to type $x\,y$), $(B\rew B)\rew C_2$ (to type $x\,\Id$) and $C_1\rew C_2\rew C_3$ (to type the whole term).\\
  
\item \textbf{Subject expansion holds.}
  Consider $\Del\,\Id\bred \Id\,\Id \bred \Id$ and say that the rightmost occurrence of $\Id$ has been typed with $A\rew A$ in the usual way. Then one can also type $\Id\,\Id$ with $A\rew A$ in the simply typed $\lam$-calculus by typing the occurrence of $\Id$ on the left-hand side of the application with $(A\rew A)\rew (A\rew A)$.
  $$
\infer{\Id:(A\rew A) \rew (A \rew A)\sep \Id:A\rew A}{\Id\,\Id:A\rew A}\app
  $$
  However, the term $(\lx.x\,x)\Id$ is not typable because $x$ should be assigned both the types $A\rew A$ and $A$ in $x\,x$, which is impossible and makes subject expansion fail in the simply typed $\lam$-calculus. In contrast, with the intersection operator, we can do so: $x\,x$ is typed with $A$ in the context $x:(A\rew A)\wdg A$, so that $\Del$ can be typed with $((A\rew A)\wdg A)\rew A$.
  $$
\infer{
\infer{  \infer{
    \infer{\phd}{x:(A\rew A)\rew (A\rew A)}\msep \infer{\phd}{x:A\rew A}}{
x\,x:A\rew A
}}{\lx.x\,x:(((A\rewsh A) \rewsh (A\rewsh A))\wdg (A\rewsh A))\rewsh (A\rewsh A)}\hspace*{0.3cm}\Id:(A\rewsh A)\rewsh (A\rewsh A)\hspace*{0.2cm} \Id:A\rewsh A }{(\lx.x\,x)\Id:A\rewsh A
}
$$

  In general, since intersection type systems allow assigning several types, a variable $x$ (bound in a redex) can be used as a (typed) placeholder for the occurrences of a same term $u$, even if $u$ occurs with different types. This explains why intersection types can satisfy subject expansion.

\end{itemize}
}

\ignore{
Indeed, a normal form $\hnfo$ (where $t_1,\ldots,t_q$ are also normal forms) is \textbf{fully} typed in an inductive way (meaning that we assume $t_1,\ldots,t_q$ fully typed) by assigning the head variable a suitable arrow type. 
$$ 
\infer[\abs^*]{
  \infer[\app]{
    \infer[\app^*]{
      \infer[\app]{\infer[\ax]{}{x:A_1\rew \ldots \rew A_q\rew B}
        \sep t_1:A_1     }{x\,t_1:A_2\rew \ldots B}\sep t_2:A_2}
    {    x\,t_1\ldots t_{\!q\mathord{-}1}:A_q\rew B} \sep t_q:A_q}{
  \zhnfo:B}}{\hnfo:C_1\rew \ldots \rew C_p \rew B} $$}

\subsection*{Partial typings.} The possibility to partially type terms will be crucial in the core of this article, when we need ``truncate'' infinite derivations typing infinite terms into finite derivations: 
many intersection type systems also feature a universal type $U$ which can be assigned to every term, so that a judgment of the form $t:U$ is meaningless. However, such a meaningless $U$ allows \textbf{partially} typing a term.  Actually, every head normal form becomes typable: if one assigns $U\rew \ldots \rew U\rew A$ (with $n$ arrows) and types the $t_i$ with $U$, then $\zhnfo$ has type $A$). 
The possibility to leave some arguments untyped is crucial for type systems to remain sound while ensuring semantic guarantees are less restrictive than strong normalization (\eg head or weak normalization). 
When considering such a universal type $U$, every term becomes typable and thus, typed judgments should be considered sound only when $U$ does not have ill-placed occurrences, called positive occurrences (technically, an occurrence is positive when it is nested in an even number of arrow domains).

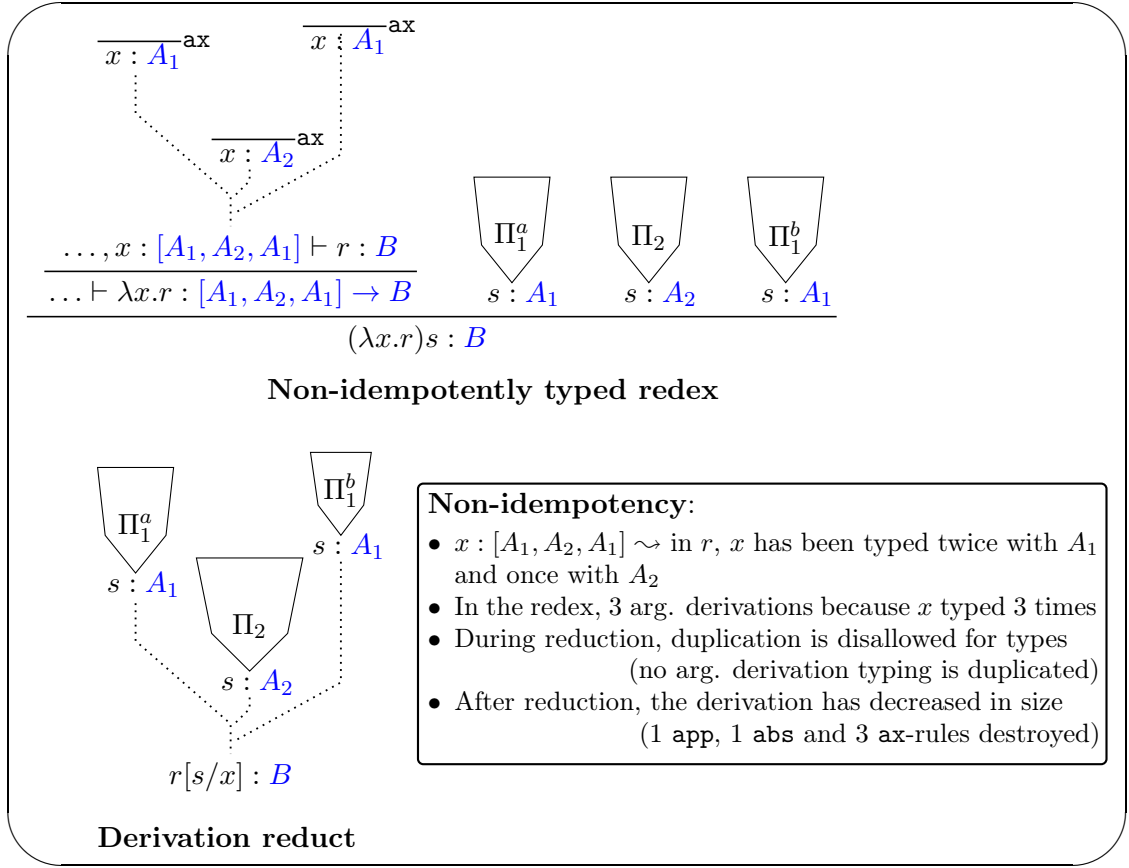
\begin{figure}
  
  \newcommand{\bigderiv}[3]{
\trans{#1}{#2}{
  \draw (0,0.6) node {#3};
  \draw [line width=0.7pt,dotted] (-1.05,2) --++ (0,-0.5) --++ (1.25,-1.1) -- (0.2,0);
\draw [line width=0.7pt,dotted] (0.45,0.75) --++ (0,-0.15) --++ (-0.2,-0.2);
\draw [line width=0.7pt,dotted] (1.65,2.55) --++ (0,-1.6) --++ (-1.4,-0.777);
}}

    \newcommand{\smallderiv}[3]{
\trans{#1}{#2}{
  \draw (0,0) --++ (0.5,0.5) --++ (0.2,1) --++ (-1.4,0) --++ (0.2,-1) -- cycle;
  \draw (0,0.6) node {#3};
}
    }

    \newcommand{\smallderivbis}[3]{
\trans{#1}{#2}{
  \draw (0,0) --++ (0.3,0.4) --++ (0.1,0.7) --++ (-0.8,0) --++ (0.1,-0.7) -- cycle;
  \draw (0,0.6) node {#3};
}
    }

        \newcommand{\smallderivter}[3]{
\trans{#1}{#2}{
  \draw (0,0) --++ (0.4,0.5) --++ (0.1,0.9) --++ (-1,0) --++ (0.1,-0.9) -- cycle;
  \draw (0,0.6) node {#3};
}
    }

\ovalbox{        
\begin{tikzpicture}


\bigderiv{-0.65}{1.6}{}

\draw [line width=0.6pt] (-2.2,4.05) --++ (1.1,0); 
\draw (-0.9,4.05) node {\small$\ax$};
\draw (-1.6,3.85) node {$x:\blue{A_{1}}$};

\trans{0}{-0.3}{

  \draw [line width=0.6pt] (-0.7,3.05) --++ (1.1,0); 
\draw (0.6,3.05) node {\small$\ax$};
\draw (-0.1,2.85) node {$x:\blue{A_{2}}$};

\draw [line width=0.6pt] (0.5,4.55) --++ (1.1,0); 
\draw (1.8,4.55) node {\small$\ax$};
\draw (1.1,4.35) node {$x:\blue{A_{1}}$};

}

\draw [line width=0.6pt] (-2.9,1) --++ (4.9,0);
  \draw (-0.45,1.3) node {$\ju{\ldots,x:\blue{\mult{A_1,A_2,A_1}}}{r:\blue{B}}$};
  \draw [line width=0.6pt] (-3.13,0.4) --++ (10.65,0);
  \draw (-0.45,0.7) node {$\ju{\ldots}{\lx.r:\blue{\mult{A_1, A_2, A_1}\rew B}}$};

\smallderivter{3.25}{0.85}{$\Pi^{a}_{1}$}
\draw (3.4,0.7) node {$s:\blue{A_{1}}$};

\smallderivter{5.05}{0.85}{$\Pi_{2}$}
\draw (5.2,0.7) node {$s:\blue{A_{2}}$};

\smallderivter{6.85}{0.85}{$\Pi^b_1$}
\draw (7,0.7) node {$s:\blue{A_1}$};

  \draw (2,0.1) node {$\lxrs:\blue{B}$};

\draw (3,-0.6) node  {\textbf{Non-idempotently typed redex}};;



\trans{0}{-7}{

\bigderiv{-0.65}{1.6}{}

  \smallderivter{-1.7}{4}{$\Pi^{a}_1$}

\draw (-1.6,3.85) node {$s:\blue{A_{1}}$};

\trans{0}{-0.3}{

  \smallderiv{-0.2}{3}{$\Pi_2$}

\draw (-0.1,2.85) node {$s:\blue{A_{2}}$};}

  \smallderivbis{1}{4.5}{$\Pi^b_{1}$}

\draw (1.1,4.35) node {$s:\blue{A_{1}}$};

\draw (-0.45,1.3) node {$\rsx:\blue{B}$};

\draw (-0.5,0.5) node  {\textbf{Derivation reduct}};;

}


\ndrth{rk}{2}{-1.8}{2}[below right]{\slbox{8.8}{
    \textbf{Non-idempotency}:
{\small 
    \begin{itemize}
\item 
  $x:\mult{A_1, A_2, A_1}$ $\leadsto$ in $r$, $x$ has been typed twice with $A_1$ and once with $A_2$
      \item In the redex,   3 arg. derivations
      because $x$ typed 3 times
    \item During reduction, duplication is disallowed for types\\
\hfills      (no arg. derivation typing is duplicated)
      \item After reduction, the derivation has decreased in size\\
\hfills        (1 $\app$, 1 $\abs$ and 3 $\ax$-rules destroyed)
        \end{itemize}}
    }};

\ignore{
\draw (2.9,4) node {\slbox{4}{\bcen \violet{\small no arg. deriv\\ has been duplicated\\[0.2cm] $\leadsto$ the deriv. has\\$\phantom{\leadsto}$ decreased in size}\ecen}};

\draw (5,2.7) node {\slbox{3}{\bcen \small
\violet{    3 arg. derivs
    because $x$ typed 3 times}\ecen
}};

  \red{
  \ndrvth{nonIdemTermin}{5}{4}{2}{\slbox{4.5}{\textbf{Non-idempotency}:\\
      \tbul \textbf{\red{duplication}} disallowed\\{\hfills  \fnsz \textit{(w.r.t. derivs)}}\\
      \tbul derivations \textbf{\red{decrease}} in size
  }};}}

\end{tikzpicture}
}
    \caption{Non-Idempotent Typing, Reduction and Decrease}
  \label{fig:intro-non-idem-typing-red-decrease}
  \end{figure}

\subsection*{Non-idempotent intersection}    
When intersection operator is not idempotent~\cite{CoppoDV80a,Gardner94,Carvalho18}, the types $A$ and $A\wdg A$ are not equivalent anymore. This expresses the fact that, from the quantitative point of view, using the type $A$ once or twice is not the same. Thus, non-idempotent intersection is closely related to Girard's Linear Logic~\cite{Girard87}. For instance, non-idempotent type systems keep information on how many times a variable has been assigned a type. A powerful consequence of this is that reduction cause derivations to decrease in size during reduction, meaning that if $\Pi$ concludes with $\juGtt$ and $t$ head-reduces to $t'$, then $\Pi'$, the derivation reduct of $\Pi$ concluding with $\juGtpt$, contains stricly less judgments than $\Pi$. This illustrated with Fig.~\ref{fig:intro-non-idem-typing-red-decrease}. This entails that it is trivial that a typed term is head normalizing.
 \igintro{ From the dynamic perspective, non-idempotency forbids the duplication of typing certificates during reduction:
\begin{itemize}
\item In a usual type system (system F, the calculus of inductive constrictions, on which \texttt{Coq} is based, idempotent intersection type) satisfying subject reduction, if the derivation $\Pi$ types $\lxrs$ is typed and $\Pi_s$ is the subderivation of $\Pi$ typing $s$, then  the derivation $\Pi'$ typing $\rsx$ contains $n$ copies of $\Pi_s$, where $n$ is the number of typed axiom rules typing $x$ in $r$ (this description is a bit too simple in the case of idempotent intersection types). This may cause a phenomenon of size explosion, which is the main reason why normalization proofs are difficult (see \Sec~\ref{ss:intersection-overview-klop-intro}).
\item In a non-idempotent intersection type system, no such duplication may occur. Consequently, the argument $s$ must be typed sufficiently many times in $\Pi$ so that no duplication is needed and each axiom rule typing $x$ is replaced by a pairwise distinct subderivation typing $s$. For instance, if $x$ has 3 typed occurrences of type $A$ in $\Pi$, then there are also in $\Pi$ 3 subderivations typing $s$ with $A$. Intuitively, this makes derivations big with a lot of redundant parts, but on  the other hand, the lack of duplication causes the derivations to grow smaller  along a reduction sequence (no size explosion). Consequently, reduction \textit{of a typed term} must stop at some point, \ie a normal form is reached. This strongly suggests indeed that the termination of typed terms is usually very easy to to prove in non-idempotent frameworks, and this is illustrated in  Fig.~\ref{fig:intro-non-idem-typing-red-decrease} \pierre{notation multiset}
\end{itemize}}%
\ighp{
Concretely, the decrease of derivation under reduction
corresponds to \textbf{weighted subject reduction} properties, \eg if $t\bred t'$ is a head step and $\juGtB$ is derivable, then $\juGtpB$ is derivable using strictly less judgments. Since there is no decreasing sequence of natural numbers of infinite length, then head reduction must stops when it is applied on a term which is typable in a non-idempotent system. The weighted aspect of proof reduction is illustrated in  Fig.~\ref{fig:intro-non-idem-typing-red-decrease}: a redex $\lxrs$ is typed. There are three axiom rules typing the variable $x$ of the redex, assigning twice type $A_1$ and once type $A_2$. This means that, in the environment of $r$, $x$ is assigned the type $A_1\wdg A_2\wdg A_1$, whereas in an \textit{idempotent setting}, the type of $x$ would have been collapsed with $A_1\wdg A_2$ and intuitively, the information giving the number of axiom rules typing $x$ would have been lost. Thus, $\lx.r$ has an arrow type with 3 types in its domain. Also because of non-idempotency, the argument $s$ must be typed three times, with matching types. When we fire the redex, no duplication of an argument derivation needs to take place, since $s$ has been typed the right number of times. We thus obtain a smaller derivation typing the reduct $\rsx$: the application and the abstraction rules of the redex have been destroyed, as well as the axioms typing $x$.}
\igintro{
\subsection*{Additional properties}
Actually,  the simplicity of normalization proof is not the only good feature of non-idempotent intersection, which explains that it is also interesting in an infinite setting (where the well-foundedness of $\bbN$ cannot be used in all generality).
It turns out that non-idempotent intersection types also enable simple combinatorial features, that we will thoroughly use:
\begin{itemize}
\item \textbf{Strictness.}
  Intersection is allowed only on the left-hand side of arrows, for instance $(A\wdg B)\rew C$ is allowed whereas $A\rew (B\wdg C)$ is not. We then do not consider introduction and elimination rules for $\wdg$, as the ones below. 
  See \Sec~3.2.1 of \cite{VialPhd} for more details.
  $$
\infer[\mathtt{intro}]{\ju{\Gam}{t:A}\msep \ju{\Gam}{t:B}}{\ju{\Gam}{t:A\wdg B}} \sep \sep \infer[\mathtt{elim-l}]{\ju{\Gam}{t:A\wdg B}}{\ju{\Gam}{t:A}}\sep \sep \infer[\mathtt{elim-r}]{\ju{\Gam}{t:A\wdg B}}{\ju{\Gam}{t:B}}
$$

\item \textbf{Syntax-direction.} \label{ss:syntax-direction-discuss-intro}
   To simplify the type system even more, non-idempotent intersection may be represented with multisets, \eg  the multiset $\mult{A,B,A}$ may represent the intersection $A\wdg B\wdg A$. This allows avoiding using a permutation rule stating that the types $A\wdg B \wdg A$ and $A\wdg A \wdg B$ are equivalent (as the one shown below), because $\mult{A,B,A}=\mult{A,A,B}$. 
  $$
  \infer[\mathtt{perm}]
        {\ju{\Gam,x:A_1\wdg \ldots \wdg A_n}{t:B}\sep
\sig \ \text{permutation} 
  }{\ju{\Gam,x:A_{\sig(1)}\wdg \ldots \wdg A_{\sig(n)}}{t:B}}
  $$ 
  Note that, without introduction and elimination rules for $\wdg$, such a permutation rule would be necessary to make $A\wdg B$ and $B\wdg A$ equivalent. 

\item \textbf{Relevance.} 
  By taking into account an exact record of how many times each type is assigned, non-idempotent intersection enable \textit{relevant} type systems, \ie type system without weakening (actually, strict \textit{idempotent} intersection type systems cannot satisfy both subject reduction an expansion, see \eg \Sec~3.2.2 of \cite{VialPhd} for more details). In a relevant type system, axiom rules have conclusions of the form $\ju{x:A}{x:A}$ (and not for instance
  $\ju{x:A\wdg B}{x:A}$ or   $\ju{x:A\wdg B\wdg C,y:D}{x:A}$), with \textit{one} assignment on the left-hand side. That is, an axiom rule typing an occurrence of $x$ just assigns \textit{one} type to this occurrence and nothing more.
  In particular, the typing context of an axiom rule can be identified by just looking at its right-hand side. By induction, in a relevant setting,
  typing contexts can be computed by inspecting axiom rules above them and thus,  do not need to be mentioned in derivations, because they can be computed by inspecting axiom rules (see Remark~\ref{rk:relevance-computation-contexts-Ro}). 
  However, it is more convenient to indicate them to write derivations. From the dynamic point of view, in a relevant derivation, reduction does not cause erasure of types, and in particular, erasable subterms are always left untyped.
\end{itemize}
With strictness and relevance, proof reduction   boils down to moving parts of the derivations, without duplication and erasure and  without structural rules popping up in the process (elimination,introduction, permutation). This gives rise to \textbf{syntax-directed} systems,
which feature only 3 rules (one for variables, one for abstractions and one for applications): such a system is given in \Sec~\ref{ss:system-Ro-Klop}.  
}

\subsection{Solving Klop's Problem\ighp{ and Hereditary Permutators}}
\label{ss:klop-prob-hp-statement}
Now that we have recalled some useful fact about intersection types, we explain in this section what are the difficulties we will meet and overcome in this article to export (non-idempotent) intersection type theory to an infinitary setting.\\



\ighp{
Another characterization problem arises in the study of the $\beta\eta$-invertible $\lam$-terms, pioneered by Curry and Feys~\cite{CurryF58} and consolidated by Dezani~\cite{Dezani76} who gave a characterization of \textit{weakly normalizing invertible} terms \wrt their normal forms.  This characterization was extended  by Bergstra and Klop~\cite{BergstraK80} for \textit{any} term: $\beta\eta$-invertible terms were proved to have \Bohm\ trees of a certain form, generalizing that given by Curry and Feys and suggesting to name them \textit{hereditary permutators}. Indeed, given a variable $x$, $t$ is a $x$-hereditary permutator iff \textit{coinductively}, $t$ reduces to $\lx_1\ldots x_n.x\,h_{\sig(1)}\ldots h_{\sig(n)}$ where $h_i$ is a $x_i$-hereditary permutator and $\sig$ is a permutation of $\set{1,\ldots,n}$ (the full definition is given in \Sec~\ref{ss:inf-lam-klop-journal}).  Hereditary permutators lacked a characterization with intersection types, so that \pierre{COINDUCTIVELY est ambigu}
 the problem of finding a type system assigning a \textit{unique} type to all hereditary permutators (and only to them) was inscribed in TLCA list of open problems by Dezani in 2006 (Problem \# 20).\\ }



\noindent \textbf{Soundness in an infinitary setting, and the need of a fine-grained notion of residuation.}
In the case of finitary non-idempotent intersection type, soundness is just based on the subject reduction property inducing the decrease of a positive integral measure and is comparatively very simple to prove. In the infinitary case, things are different: the soundness of a derivation $P$ is ensured provided $P$ is, in some sense, the limit object of \textit{finite} derivations (this is the core of the \textbf{approximability} criterion). To be more precise:
\begin{itemize}
\item Defining infinite derivations by just allowing infinitary multiplicites and nestings in multisets is unsound: for instance, the looping term $\Om$ becomes typable. This is why we need the approximability criterion (\Sec~\ref{ss:degenerate}). 
\item If a term $t$ is typed with an approximable derivation, which is a possibly infinite derivation that can be written as a limit (or the join) of finite derivations, then $t$ is necessarily typable in the finitary, usual setting, which ensures soundness (in the sense that, in the finite case, if $t$ is typable, then it is head normalizing).
\item But the necessity of handling the notion of approximability while reducing derivations (\ie reducing there subjects) constrain us to have a more meticulous description of the subject reduction property. Indeed, we need to show that approximability is stable under reduction. But for this, as it will be made clear in the paper, we need to have a precise, fine-grained notion of \textbf{residuation} (\Sec~\ref{ss:one-step-sr-se} and \ref{ss:sr-proof}).
\item At last, residuation raises a problem of its own.
As we shall see, residuation can be suitably defined only for the right-hand sides $t:T$ of typing judgments $\juCtt$ (where $C$ is the context, $t$ the subject and $T$ the type), whereas approximability pertains to their left-hand sides too. We explain how to overcome this problem with the crucial notion of \textbf{equinecessity} (\Sec~\ref{ss:equinecessity-main}).\\
\end{itemize}

\noindent \textbf{The infinite jump problem, or infinitary subject expansion.}
Another problem in the infinitary setting is ensuring \textit{infinitary subject expansion}:  actually, proving that typing is stable under one-step expansion is not really more difficult than it is in the finite case. However, a suitable notion of reduction path of infinite length (the so)-called \textbf{productive paths}, \Sec~\ref{ss:comput-bohm-trees-hp}) naturally arises when we deal with infinite terms and \Bohm\ tree~\cite{KennawayKSV97}.  As a side remark,  one-step subject reduction and even infinitary subject reduction along a productive path are also trivial (as long as we do not consider the question of approximability): the latter can be also thought as a limit process of a finite number of reduction steps. In other words, one may say that infinitary subject reduction is an asymptotic phenomenon.
However, the situation is different for infinitary subject expansion, because intuitively there is \textit{a priori} no way to pass \textit{incrementally}\footnote{To give an idea, if we consider a sequence of rational numbers $(x_n)_\nN$ which converges to a limit $x$ (which is rational or irrational), then passing from $x$ to any $x_n$ means going backward on $\infty-n$ (which is also $\infty$) many steps! In this respect, this is not incremental.} from a limit object to one of a finite rank (let us call that very informally the \textit{infinite jump problem}).\label{disc:infinite-jump-problem}  Infinitary subject expansion is not asymptotic in nature. But as, we shall see in \Sec~\ref{ss:typ-inf-nf-informal}, 
infinitary subject expansion  can be performed and  we may already give some elements of understanding on the involved techniques: on one hand, our notion of soundness---approximability---specifies that infinite derivations are \textit{static} limits of finite ones (by static, we mean that the definition of approximability is independent from any notion of reduction). On another hand, we have the usual notion of productive reduction paths of infinite length in the infinitary $\lam$-calculus, which defines a \textit{dynamic} notion of limit of a family of term: this limit is intuitvely the term which is asymptotically outputted after an infinite number of reduction steps, as in the example of $\cuf$ on page~\pageref{ss:system-Ro-Klop}.

Interestingly enough, the core technique used to perform infinitary subject expansion is also based on approximability, and actually, on confronting these two notions (static \vs dynamic) of limits: we use the fact that derivations are static limits of finite derivations to (1) perform infinite subject expansion on these finite derivations (for which the infinite jump probem has a solution, \Sec~\ref{ss:subj-subst}) (2) statically build a (possibly infinite) expanded derivation by taking the join of the resulting finite expanded derivations.

Thus, approximability is pivotal in two respects: (1) it is used to define which infinitary derivations are semantically valid and which are not (2) it is used to perform infinitary subject expansion.\\

\noindent \textbf{The problem of pointing.} A last technical ingredient to adapt intersection type theory to an infinitary setting is what we call \textbf{sequences} in this articles: multisets of elements which are decorated with suitable natural numbers which we call \textbf{tracks}. 
In non-idempotent intersection type theory, it is convenient (but not mandatory) to represent the intersection of types $A_1$,\ldots,$A_k$ with the multiset $\mult{A_1,\ldots,A_k}$: it allows not needing anymore structural rules and simplifying the type system. However, it has two drawbacks for us: it makes it impossible to have a \textbf{pointing mechanism} (in the multiset $\mult{A,B,A}$, it is impossible to point formally to one specific occurrence of $A$) and it gives rise to \textbf{non-determinism in proof reduction} as it is well-known (\Sec~\ref{ss:non-determinism}). 
Both problems, as we shall see, make it impossible to define approximability and ensure soundness. When elements are decorated with tracks (for instance, $\mult{A,B,A}$ may become $(2\cdot A,4\cdot B, 7\cdot A)$ or $(3\cdot A,2\cdot B,9\cdot A)$), we recover pointing and determinism. For instance, in $(2\cdot A,4\cdot B, 7\cdot A)$, one occurrence of $A$ is on track 2 (\ie is decorated with 2) and the other is on track 7. Actually, this tracking mechanism enables a fine-grained description of residuation on typing derivation and prove soundness, as we saw just above.

\subsection{A Glimpse at Sequences and Derivation Approximations}
\label{ss:sequence-tracking-klop-intro}
Now, let us give a glimpse at the notion of \textit{sequences} (which refines that of multisets) and of approximation of a derivation.\\

\noindent \textbf{Sequences and tracking.} 
Naive multisets do not allow defining a notion of position: we say that they do not allow \textbf{tracking}. To see this, consider the equality $\mult{A,B,A}=\mult{A,B}+\mult{A}$. There is no way to associate the occurrence of $A$ in $\mult{A}$, located in the right-hand side of the equality to one of the two occurrences of $A$ in $\mult{A,B,A}$, in the left-hand side. This shows that there is not way to point to a particular occurrence of $A$ inside $\mult{A,B,A}$. Actually, if we swap the two occurrences of $A$ in $\mult{A,B,A}$, we obtain
$\mult{A,B,A}$. If we do not do anything, we also obtain $\mult{A,B,A}$: the possible swap cannot be seen. It has actually no formal meaning.

To retrieve tracking and thus, the possibility to point at an element, we may annotate elements of multisets with pairwise distinct integers, that we call \textbf{tracks}. Such a decorated multiset is called a \textbf{sequence} and we use the constructor $(\_)$ for sequences instead of $\mult{\_}$ for multisets.
For instance, $(2\cdot A,3\cdot B,5\cdot A)$ is a sequence, there is \textbf{one occurrence} of $A$ \textbf{on track} 2 and another on track 5. Sequences come along with a disjoint union operator $\uplus$, \eg $(2\cdot A,3\cdot B,5\cdot A)= (2\cdot A,3\cdot B)\uplus (5\cdot A)$. Thanks to tracks, each occurrence of $A$ in one side of the equality are unambiguously associated with another on the other side, and in the left-hand side, we may point to the occurrence of $A$ on track 2 rather than the one on track 5. Disjoint union is not defined for two
sequences sharing a same track, \eg
$(2\cdot A,3\cdot B)\uplus (3\cdot B, 5\cdot A)$ is \textit{not} defined (because the two operands both use track 3, even though the same type $B$ occurs): we say that there is a \textbf{track conflict}.
Moreover, $\uplus$ is a \textbf{commutative} and associative, but \textbf{partial} operator. Observe that $(2\cdot A,3\cdot B,5\cdot A)$ may be written six different ways, \eg  $(5\cdot A,2\cdot A,3\cdot B)$ and $(2\cdot A,5\cdot A,3\cdot B)$: if we perform a swap in a sequence, we see it.

Our framework is deterministic, \eg there is a \textit{unique} canonical way to produce a derivation from another one when reducing a redex, contrary to system~$\scrRo$, and we can characterize infinitary semantics (more details in \Sec~\ref{ss:non-determinism}).

To sum up, in many ways, sequence types work like multiset types, \eg they give rise to a syntax-directed, strict and relevant typing system with only three rules\igintro{ (\cf \Sec~\ref{ss:syntax-direction-discuss-intro})}, but they allow tracking, constructing a suitable pointing mechanism inside derivations and thus, defining approximability, the validity criterion that will help us discard the unsound coinductive derivations. This will enable us characterizing hereditary head normalization\ighp{ and hereditary permutators}. \\

\noindent \textbf{Approximation.} As we have hinted at in \Sec~\ref{ss:klop-prob-hp-statement}, coinductive type grammar give rise to \textit{unsound} derivations, \eg derivations typing of $\Omega$, 
but we introduce a \textit{validity criterion} allowing to discard them. This criterion is called \textbf{approximability}.
 Intuitively, an infinite proof/typing derivation is approximable when it is obtained by superposing infinitely many \textit{finite} proofs, growing over and over, as in Fig.~\ref{fig:approx-deriv-intro}: the outer triangle represents an infinite proof/typing derivation $\Pi$ and the inner polygons represent finite proofs that ``fit'' in $\Pi$.  
Thus, an approximable derivation may be infinite, but it is asymptotically obtained from finite/sound proofs.
Equivalently, a derivation is approximable when it is the \textbf{join} of its \textit{finite} \textbf{truncations}, also called \textbf{approximations}. The crucial intuition behind this notion is that it allows ``reducing'' infinite derivations to finite ones along with the \textit{soundness properties} that come with them, such as head normalization.

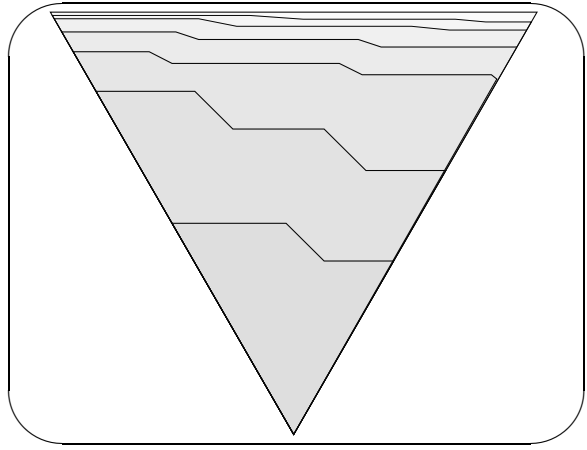
\begin{figure}[h]
  \begin{center}
    \ovalbox{
\hspace*{0.3cm}\begin{tikzpicture}

    \draw  (0,0) --++ (3.2,5.6) --++ (-6.4,0) -- cycle;
  \draw [fill,gray,opacity=0.04](0,0) --++ (3.2,5.6) --++ (-6.4,0) -- cycle;

  \draw [very thin] (0,0) --++ (-3.175,7/4*3.175) --++ (2.6,0) --++ (0.7,-0.05) --++ (2,0)  --++ (0.4,-0.03625) -- (4/7*5.47 ,5.47) -- cycle;
  \draw [fill,gray,opacity=0.04] (0,0) --++ (-3.175,7/4*3.175) --++ (2.6,0) --++ (0.7,-0.05) --++ (2,0)  --++ (0.4,-0.03625) -- (4/7*5.47 ,5.47) -- cycle;

  \draw [very thin] (0,0) --++ (-3.15,7/4*3.15) --++ (1.9,0) --++ (0.5,-0.1) --++ (2.3,0) --++ (0.5,-0.05) -- (4/7*5.3625,5.3625);
  \draw [fill,gray,opacity=0.04] (0,0) --++ (-3.15,7/4*3.15) --++ (1.9,0) --++ (0.5,-0.1) --++ (2.3,0) --++ (0.5,-0.05) -- (4/7*5.3625,5.3625);
  
  \draw (0,0) --++ (-3.05,7/4*3.05) --++ (1.5,0) --++ (0.3,-0.1) --++ (2.1,0) --++ (0.3,-0.1) -- (4/7*5.1375,5.1375);
  \draw [fill,gray,opacity=0.04] (0,0) --++ (-3.05,7/4*3.05) --++ (1.5,0) --++ (0.3,-0.1) --++ (2.1,0) --++ (0.3,-0.1) -- (4/7*5.1375,5.1375); 
  
  \draw (0,0) -- (-2.9,5.075) --++ (1,0) --++ (0.3,-0.155) --++ (2.2,0) --++ (0.3,-0.15) --++ (1.7,0) --++ (0.5/7,-0.4/7) -- cycle;
    \draw [fill,gray,opacity=0.04] (0,0) -- (-2.9,5.075) --++ (1,0) --++ (0.3,-0.155) --++ (2.2,0) --++ (0.3,-0.15) --++ (1.7,0) --++ (0.5/7,-0.4/7) -- cycle;

  \draw (0,0) -- (-2.6,4.55) --++ (1.3,0) --++ (0.5,-0.5) --++ (1.2,0) --++ (0.55,-0.55)-- (2,3.5) -- cycle ;
  \draw [fill,gray,opacity=0.04] (0,0) -- (-2.6,4.55) --++ (1.3,0) --++ (0.5,-0.5) --++ (1.2,0) --++ (0.55,-0.55)-- (2,3.5) -- cycle;

  \draw (0,0) -- (-1.6,2.8) --++ (1.5,0) --++ (0.5,-0.5) -- (1.314,2.3) -- cycle  ;
   \draw [fill,gray,opacity=0.04] (0,0) -- (-1.6,2.8) --++ (1.5,0) --++ (0.5,-0.5) -- (2.3*4/7,2.3) -- cycle  ;
\end{tikzpicture}\hspace*{0.5cm}}
\caption{An Approximable Derivation as an Infinite Superposition}
\label{fig:approx-deriv-intro}
\end{center}
\end{figure}

More precisely, derivations need not being infinite to characterize head normalization in the infinite calculus: infinite derivations are required to characterize hereditary head normalization in an infinitary calculus, because we need to type infinite normal forms without leaving any subterm untyped. Moreover, finite derivations allow us to prove operational properties of the infinite calculus in a relatively simple way (more details in \Sec~\ref{ss:subj-subst}). An important idea is that \textbf{any term which is approximably typable is finitarily typable and as such, is head normalizing}. The whole proof of correctness of the system (\wrt\ hereditary head normalization) is based on this observation.

Just to give a more precise idea of what is meant by superposition, the derivation fragment $\Pi_2$ below can be superposed upon $\Pi_1$ (we omit the typing assumptions and the possible premises of the top judgments):
\[
\Pi_1=\infer{\ttf:(A\wdg B)\rew A\rew C\sep \ttx:A\sep \ttx:B}{\ttf(\ttx):A\rew C}
\]
\[
\Pi_2=\infer{\ttf:(A\wdg B\wdg (B\rew C))\rew (A\wdg C) \rew D \sep \ttx:A\sep \ttx: B\sep \ttx:B\rew D}{
\ttf(\ttx):(A\wdg C)\rew C
  }
\]
Indeed, $\Pi_1$ is obtained from $\Pi_2$ by removing the symbols colored in red:
\[
\Pi_2=\infer{\ttf:(A\wdg B\red{\wdg (B\rew C)})\rew \red{(}A\wdg \red{C)} \rew D \sep \ttx:A\sep \ttx: B\sep \red{\ttx:B\rew D}}{
\ttf(\ttx):\red{(}A\wdg \red{C)}\rew C}
\]
Thus, $\Pi_1$ can be seen as a \textit{truncation} or an \textit{approximation} of $\Pi_2$. 
Infinite superposition strongly hints at the presence of complete lattices and complete partial orders, which is formalized in \Sec~\ref{s:approx}.

With this concept of approximation, we retrieve, in an infinitary setting, the main tools of finite intersection type theory that we sketched in \Sec~\ref{ss:intersection-overview-klop-intro}:
subject reduction, subject expansion, typing of normal forms\igintro{ (terminal states}). Actually, following the scheme of Fig.~\ref{fig:its-fundamental-diag}, we will prove that hereditary head reduction (\Sec~\ref{ss:inf-lam-klop-journal})   is \textit{asymptotically} complete for infinitary weak normalization.

\section{Finite Intersection and Infinite Terms}
\label{s:finite-inter-infinite-types}

 In \Sec~\ref{ss:inf-lam-klop-journal} and \ref{ss:comput-bohm-trees-hp}, we define the infinitary $\lam$-calculus and \Bohm\ trees. 
Gardner-de Carvalho's system $\scrRo$ is recalled in \Sec~\ref{ss:system-Ro-Klop}, in particular how weak normalization is characterized by considering \textit{unforgetful} derivations. In \Sec~\ref{ss:subj-subst}, we explain how system $\scrRo$ can still be applied to infinite terms. We encounter a first use of subject substitution, \ie replacing the subject $t$ of a derivation $\Pi$ by another $u$ which is equal to $t$ in the  typed parts of $\Pi$. We use this to perform an \textit{infinite} subject expansion on \textit{finite} derivations. 
 We present the main ideas to solve Klop's Problem in \Sec~\ref{ss:typ-inf-nf-informal} and \ref{ss:degenerate}: we informally explain (1) how infinitary subject expansion \textit{could} be performed by truncating derivations and taking the joins of directed families (2) how coinduction gives rise \textit{unsound} derivations, \eg derivations typing the unterminating term $\Om$ and how soundness can be retrieved with \textit{approximability}.

\subsection{Pointing in a $\lam$-term}
\label{s:pointing-lambda}
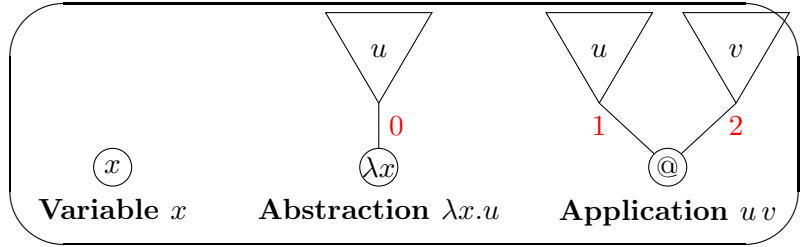
\begin{figure}[!h]
  \begin{center}
    \ovalbox{
  \begin{tikzpicture}

\draw (0,-2.6) node{\textbf{Variable $x$}} ;
\draw (0,-2.05) circle (0.25) ;
  \draw (0,-2.05) node{$x$} ;
\draw (3.5,-2.6) node{\textbf{Abstraction $\lambda x.u$}};
\draw (2.8,0) -- (4.2,0) -- (3.5,-1.2) -- (2.8,0) ;
\draw (3.5,-0.5) node{\textbf{$u$}} ;
\draw (3.5,-1.2) -- (3.5,-1.8) ;
\draw (3.5,-1.5) node[right]{\red{0}} ;
\draw (3.5,-2.05) circle (0.25) ;
\draw (3.5,-2.05) node{$\lambda x$} ;

\draw (7.3,-2.66)  node {\textbf{Application $u\, v$} } ;

\draw (8.2,-1.2) -- (8.9,0) -- (7.5,0) -- (8.2,-1.2) ;
\draw (7.48,-1.87) -- (8.2,-1.2) ;
\draw (8.2,-1.5) node {$\red{2} $} ;
\draw (8.2,-0.5) node {$v$} ;

\draw (6.4,-1.2) -- (5.7,0) -- (7.1,0) -- (6.4,-1.2) ;
\draw (7.12,-1.87) -- (6.4,-1.2) ;
\draw (6.4,-1.5) node {$\red{1}$} ; 
\draw (6.4,-0.5) node {$u$} ;

\draw (7.3,-2.05) node {$\arob$} ;
\draw (7.3,-2.05) circle (0.25) ;
\end{tikzpicture}}
\end{center}
\caption{Lambda Terms as Labelled Trees}
\label{fig:lam-terms-as-trees}
\end{figure}

Lambda terms can be seen as labelled trees following the pattern of Fig.~\ref{fig:lam-terms-as-trees}. 
We orient these labelled trees as typing derivation usually are: the root is put at the bottom of the figure (and not at the top),  because the  derivations of the type system\ighp{s} to come can be represented as refinements of  parsing trees of $\lam$-terms (see in particular Fig.~\ref{fig:app-node-S-with-edges} on p.~\pageref{fig:app-node-S-with-edges}). We say that $u$ is the \textbf{argument} of the application $t\,u$. 
This representation gives to the standard notions of \textbf{positions} in a $\lam$-term, represented as words on $\set{0,1,2}$ (the set of words on $\set{0,1,2}$ is denoted $\set{0,1,2}^*$), and \textbf{support} $\supp{t}$ of a term $t$ (the set of positions in $t$). If $b\in \supp{t}$, then $t(b)$ is the label at position $b$ in $t$ and $\trb$ is the subterm rooted at $b$. 
For instance, if $t:=\lx.x\,y$, then $\supp{t}=\set{\epsi,0, 0 \cdot 1,0\cdot 2}$ where $\epsi$ is the empty word, $t(\epsi)=\lx$, $t(0)=\arob$, $t(0\cdot 1)=x$, $t(0\cdot 2)=y$ and $t\rstr{0}=x\,y$. The \textbf{applicative depth} \texttt{ad} of a position $b$ in a term $t$ is the number of times that $b$ is nested inside the arguments of an application. 
For instance, if $u_1\,u_2\,u_3\,u_4$ are terms and $t=\lx.((y\,u_1)((y\,u_2)u_3))u_4$, $b_1=0\cdot 1^2\cdot 2$,  $b_2= 0\cdot 1\cdot 2\cdot 1\cdot 2$, $b_3=0\cdot 1\cdot 2^2$, $b_4=0\cdot 1$, then $t\rstr{b_i}=u_i$ and $\ad{b_1}=1$, $\ad{b_2}=\ad{b_3}=2$, $\ad{b_4}=0$.  In other words, if $b=b_0\cdot b_1 \ldots \cdot b_n$ (with $b_i$ integer), then $\ad{b}=\card{\set{i \in \set{0,\ldots,n}\,|\, b_i=2}}$, \ie $\ad{b}$ is the number of 2 that $b$ contains.

\subsection{Infinite Lambda Terms}
\label{ss:inf-lam-klop-journal}

    In this section, we present \Bohm\ trees (chapter 10 of Barendregt~\cite{Barendregt85}, in particular, Theorem 10.4.2.),
    a notion that was suggested by the proof of the \Bohm\ Separation Theorem, 
    and the construction of one of the infinitary calculi introduced in \cite{KennawayKSV97}. See also \cite{Czajka14,EndrullisHHP015,Bahr18} for alternative presentations.  
   This infinitary calculus is a natural framework which allows expressing \Bohm\ trees (more precisely, its normal forms are \Bohm\ trees without $\bot$), but it gives us more flexibility. For instance, the computation of an infinite normal form does not need to follow the path given by hereditary head normalization and we have a general notion of \textit{sound} computations of infinite length, that we call here \textit{productive reduction paths}. 
    We shall generalize some notions to the whole infinitary calculus (for instance, hereditary head normalization will pertain both to finite and infinite terms), since it gives a general framework where they can be expressed.

    We now generalize some of the notations of \Sec~\ref{s:finite-inter-infinite-types} in an infinitary setting. 
Let $\scrV$ be a countable set of term variables. The set $\Lamfty$ of infinitary $\lam$-terms is \textit{coinductively} defined by:
$$t, u~:=~ x \in \TermV~\|~(\lx.t)~\|~(t\arob u)$$ 
When there is no ambiguity, we usually just write $\lx.t$ and $t\,u_1\ldots u_n$ instead of $(\lx.t)$ and $(\ldots (t\arob u_1)\ldots \arob u_n)$. If $t$ is an infinitary term, then $\supp{t}$, the \textbf{support} of $t$ (the set of positions in $t$)  is defined
by $\supp{x}=\set{\epsi}$, $\supp{\lx.t}=\set{\epsi}\cup 0\cdot \supp{t}$ and $\supp{t\,u}=\set{\epsi}\cup 1\cdot \supp{t}\cup 2\cdot \supp{u}$ by \textit{coinduction}. 
The notations $t(b)$ and $\trb$ are defined by \textit{induction} on $b$. 

Let $t\in \Lamfty$. An \textbf{infinite branch} of $t$ is an infinite word $\gam$ over $\set{0,1,2}$ such all the  finite prefixes of $\gam$ are in $\supp{t}$. The notion of \textit{applicative depth} $\ad{\_}$ is straightforwardly extended to infinite branches:

\begin{example}\mbox{}
  \label{ex:lambfty-vs-lamzzu}
\begin{itemize}
\item 
\igintro{\pierre{ Let us formally define the term $\fom$from \Sec~\ref{ss:inf-norm-bohm-intro}}}
  Let us formally define $\fom$, the \Bohm\ tree of $\cuf$, 
  by the equation $\fom=f(\fom)$, \ie  $\fom:=f(f(f(f(\ldots))))$ is 
the labelled tree such that $\supp{\fom}=\set{2^n~|~n\in \bbN}\cup \set{2^n\cdot 1\,|\, n\in \bbN}$, $\fom(2^n)=\arob$ and $\fom(2^n\cdot 1)=f$ for all $\nN$. 
Then, $\fom$ has one infinite branch $2^\infty$, which is the infinite repetition of $2$: indeed, for all $n\in \bbN$, $2^n\in \supp{\fom}$. We have $\ad{2^\infty}=+\infty$.
\item
Likewise,  the infinite term ${}^\infty x\in \Lamfty$ defined by ${}^\infty x=({}^\infty x)x$, so that ${}^\infty x=(((\ldots) x)x)x$ satisfies $\supp{{}^\infty x}=\set{1^n\,|\,\nN} \cup \set{1 ^n \cdot 2\,|\, \nN}$, so $\supp{{}^\infty x}$ has the infinite branch $1^\infty$ (this indicates a leftward infinite branch), which satisfies $\ad{1^\infty}=0$ since 2 does not occur in $1^\infty$.
\end{itemize}
\end{example}

If $b\in \supp{t}$, the subterm (\resp the constructor) of $t$ at position $b$ is denoted $\trb$ (\resp $t(b)$), \eg if $t:=\lx.(x\,y)z$ and $b=0\cdot 1$ (\resp $b=0\cdot 4$), then  $\trb=x\,y$ and $t(b)=\arob$) (\resp $\trb=t(b)=z$).  

\begin{definition}[001-Terms]
\label{def:001-term-hp}
Let $t\in \Lamfty$. Then $t$ is a \textbf{001-term}, if, for all infinite branches $\gam$ in $\supp{t}$, $\ad{\gam}=\infty$. \techrep{The set of 001-terms is denoted $\Lamzzu$.}
\end{definition}

Once again, the vocable ``001-term'' comes from \cite{KennawayKSV97}. For instance, $\fom$ is a 001-term since its unique infinite branch $2^\infty$ satisfies $\ad{2^\infty}=+\infty$, whereas ${}^\infty x$ defined  above is \textit{not} a 001-term since its infinite branch $1^\infty$ satisfies $\ad{1^\infty}=0$.

\subsection{The computation of \Bohm\ trees}
\label{ss:comput-bohm-trees-hp}

Let $t,u \in \Lamfty$. The notation $\tux$ denotes the term obtained from $t$ by the \textit{capture-free} substitution of the occurrences of $x$ with $u$  (a formal definition can be found in \cite{KennawayKSV97} in the infinitary setting). One easily checks that if $t,u$ are 001-terms, then $\tux$ also is a 001-term. 
 The $\beta$-reduction $\bred$ is obtained by the contextual closure of $(\lx.t)u\bred \tux$ and  $t\breda{b} t'$ denotes the reduction of a redex at position $b$ in $t$, \eg $ \ly.((\lx.x)u)v \breda{0\cdot 1} \ly.u\,v$. Thus, the relation $t \breda{b} t'$ is defined by \textit{induction} on  $b\in \set{0,\,1,\,2}^*$. 

 $$\begin{array}{c@{\sep}c@{\sep}c@{\sep}c@{\sep}c}
   \infer{\phantom{t}}{(\lambda x.r)s\breda{\epsi} \rsx} &
 \infer{t\breda{b}t'}{\lx.t\breda{0\cdot b} \lx.t'} &
 \infer{t_1\breda{b} t_1}{t_1\,t_2\breda{1\cdot  b} t'_1\,t_2} &
 \infer{t_1\,t_2 \breda{2\cdot b} t_1\,t_2 '}{t_2\breda{b} t'_2} & 
 \infer{t\breda{b}t'}{t\bred t'}
\end{array}
 $$

The  notion of normal form generalizes to $\Lamzzu$, \eg a \textbf{001-Normal Form (001-NF)} is a 001-term that does not contain a redex. Thus, the set of 001-NF, denoted \NFfty, can be defined \textit{coinductively} by:
$$t,\,t_i~ ::=~  \hnfo~ (p,q\geqslant 0)$$
A term $t\in \Lamzzu$ is \textbf{head normalizing} if $t\hred^* \hnfo$ (with $t_1,\ldots,t_q\in \Lamzzu$), which is a head normal form\ighp{ of \textbf{arity} $p$}. A term $t\in \Lamzzu$ is \textbf{has a head normal form} if there is a \textit{finite} reduction path from $t$ to a head normal form $t'$ (\ie $t\bred t'$). In particular, the implication ``$t$  is head normalizing'' $\Rightarrow$ ``$t$ has a head normal form'' straightforwardly holds.  The converse implication (``$t$ has a head normal form $\Rightarrow$ $t$ is head normalizing'') also holds, although the proof differs from the finite cas~\cite{KennawayKSV97}, and we will give a simple semantic proof of this fact in \Sec~\ref{ss:subj-subst}.

\begin{definition}[\Bohm\ tree of a term]
  \label{def:bohm-tree-kj}
  Let $t$ be a 001-term.\\
  The \Bohm\ tree $\BT{t}$ of $t$ is \textit{coinductively} defined  by:
  \begin{itemize}
    \item $\BT{t}=\lx_1\ldots x_p.x\,\BT{t_1}\ldots \BT{t_q}$ if $t\hred^* \hnfo$.
    \item $\BT{t}=\bot$ if $t$ is not head normalizing.
  \end{itemize}
\end{definition}

For instance, $\BT{\Om}=\bot$ where $\Om=(\lx.x\,x)(\lx.x\,x)$,   if $t$ a 001-normal form then $\BT{t}=t$ and $\BT{\cuf}=\fom$.  
Intuitively, the computation of \Bohm\ trees is done by a possibly infinite series of head reductions at deeper and deeper levels. 
This corresponds to the asymptotic reduction strategy known as \textbf{hereditary head reduction}\ighp{ (\Sec~\ref{ss:intro-klop-norm-and-red-strat})}.

More formally, we define $\hhred$, a non-deterministic version of hereditary head reduction as follows: if $t$ is a reducible 001-term, we define $\adr{t}$, the applicative depth of reduction of $t$ as the minimal applicative depth of a redex inside $t$. There may be several such redexes, \eg if $t:=x\,(\Id\,\Id)\Omega$, then $\adr{t}=1$ and both $(\Id\,\Id)$ and $\Omega$ are redex occuring at applicative depth 1.

We then set $t\hhred t'$ whenever $t\bred t'$ and the contracted redex occurs at applicative depth $\adr{t}$, thus, with the same example, we have both $t\hhred x\,\Id\,\Om$ and $t\hhred t$. We can also define a (deterministic) parallel version of this reduction,  by specifying that, when $t$ is reducible, all redexes occuring at applicative depth $\adr{t}$ are simulateneously contracted: $x\,(\Id\,\Id)\,(\Id\,\Id)\,(x\,(\Id\,x))$ would give $x\,\Id\,\Id\,(x(\Id\,x))$.

\begin{observation}
  \label{obs:charac-hhn-ad}
Let $t$ be a 001-term. Then $t$ is hereditary head normalizing iff, for all $d\in \bbN$, there is a term $t'$ such that $\adr{t'}\geqs d$ and $t\hhred^* t'$.
\end{observation}

\ighp{
In the wake of Definition~\ref{def:bohm-tree-kj}, we may now define hereditary permutators.

\begin{definition}
  \label{def:hereditary-permutator} \mbox{}
  \begin{itemize}
  \item For all $x\in \TermV$, the sets $\HP(x)$ of \textbf{$x$-headed Hereditary  Permutators ($x$-HP)} ($x\in \TermV$) are defined by mutual \textit{coinduction}: 
    $$
\infer{   { \begin{array}{c}
      \hpt_1\in \HP(x_1)\ \ldots\ \hpt_n \in \HP(x_n)\sep (n\geqs 0, \sig \in \Perm{n},\  x_i \neq x,\ \text{$x_i$ pairwise distinct})\\  \text{and}\ h\hred^*\lx_1\ldots x_n.x\,\hpt_{\sig(1)}\ldots \hpt_{\sig(n)}
    \end{array}} }{h \in \HP(x)}$$
   
\item A \textit{closed hereditary permutator}, or simply, a \textbf{Hereditary Permutator (HP)} is a term of the form
  $\psym =\lx.h_0$ with $h_0\in \HP(x)$ for some $x$.
  \end{itemize}
\end{definition}

Thus, a $x$-headed hereditary permutator  is the head reduct of a hereditary permutator applied to the variable $x$ 

\begin{theorem}[Bergstra-Klop,\cite{BergstraK80}]
   A $\lam$-term $t$ is a hereditary permutator iff $t$ is invertible  modulo $\beta\eta$-conversion for the operation $\cdot$ defined by $u\cdot v=\lx.u\,(v\,x)$, whose neutral element is $I=\lx.x$.
\end{theorem}

Thus, $u$ is invertible when there exists $v$ such that $\lx.u(v\,x)=_{\beta\eta}=\lx.v(u\,x)=_{\beta\eta} I$. 
An extensive presentation of hereditary permutators and their properties is given in Chapter 21 of \cite{Barendregt85}.

Definition~\ref{def:hereditary-permutator} can be read as the specification of a set of terms whose \Bohm\ trees have a particular form.
}

\subsection*{Infinitary convergence.}
When $t$ is hereditarily head normalizing ($\BT{t}$ does not contain $\bot$), hereditary head reduction is productive in the sense that, whenever we wait sufficiently many evaluation steps, the term will not change under any fixed applicative depth (Observation~\ref{obs:charac-hhn-ad}). This is actually why hereditary head reduction computes something (in this case, the \Bohm\ tree of $t$). More generally, 
some reduction paths of an unspecified form have an infinite length but asymptotically produce a term:

\begin{definition}[Productive reduction paths]
\label{def:scrs-kj}
Let $t=t_0\breda{b_0} t_1\breda{b_1} t_2 \ldots t_n \breda{b_n} t_{n+1}\ldots$ be a reduction path of length $\ell \leqs \infty$.\\
Then, this reduction path is said to be \textbf{productive} if either it is of finite length ($\ell \in \bbN$), or $\ell=\infty$ and $\ad{b_n}$ tends to infinity (remember that $\ad{\cdot}$ is applicative depth). 
\end{definition}

A productive reduction path is called a \textit{strongly converging reduction sequence} in \cite{KennawayKSV97}, in which numerous examples are found. When $\BT{t}$ does not contain $\bot$, the hereditary head reduction strategy on a term $t$ gives a particular case of productive path.

\begin{lemma}[Limit of a productive path]
  \label{lem:limit-prod-paths-hp}
  Let $t=t_0\breda{b_0} t_1\breda{b_1} t_2 \ldots t_n \breda{b_n} t_{n+1}\ldots$ be a productive reduction path of infinite length.\\
  Then, there is a unique 001-term $t'$ such that, for every $d\geqs 0$, there is $N\in \bbN$ such that, for all $n\geqs N$, $\supp{t_n}\cap \set{b\in \setpos\,|\, \ad{b}\leqs d}=\supp{t'}\cap\set{b\in \setpos\,|\, \ad{b}\leqs d}$ and $t_n(b)=t'(b)$ on $\supp{t'}\cap \set{b\in \setpos\,|\,\ad{b}\leqs d}$.
\end{lemma}

The existence and the unicity of the term $t'$ in the statement of Lemma~\ref{lem:limit-prod-paths-hp} is easy to prove, and $t'$  is called the \textbf{limit} of the productive path. Intuitively, the lemma specifies that, when $t'$ is the limit of $(t_n)_{n\geqs 0}$, then $t'$ induces the same tree as $t_n$ at fixed applicative depth after sufficiently many reduction steps, as expected. 

\begin{notation}
We write $t\bredfty t'$ if $t\bred^* t'$ or $t'$ is the limit of a productive path starting at $t$. 
\end{notation}


For instance,
if $\Delf=\lx.f(x\,x)$, $\cuf=\Delf\,\Delf$ (with $f\in \TermV$), then $\cuf\breda{\epsi} f(\cuf)$, which gives the productive path $\cuf \breda{\epsi} f(\cuf)\breda{2}\ldots f^n(\cuf)\breda{2^n} f^{n+1}(\cuf)\ldots$ since $\ad{2^n}=n\longrightarrow \infty$. 
The limit of this path, which implements hereditary head reduction on $\cuf$, is $\fom$, \ie we have $\cuf\bredfty \fom$. This formalizes observations on the first page on this article. More generally, if $t$ is hereditary head normalizing, by Observation~\ref{obs:charac-hhn-ad}, hereditary head reduction on $t$ defines a productive reduction path whose limit is the normal form of $t$, whether this path is finite or not. As a counter-example, $t:= x\,\cuf\,\Om$ is not hereditary head normalizing and $\hhred$ does not stop on $t$, since $t\hhred x\,(f\,\cuf) \Om \hhred x\,(f\,\cuf)\Om$. However, the computation of $\cuf$ in $t$ gives a productive reduction path from $t$ to $x\,\fom\,\Om$.

Notice that $\Om\bred \Om\bred \ldots$ is not a productive path (although the term does not change)  since all reduction steps occur at null applicative depth.


By Observation~\ref{obs:charac-hhn-ad}, a term $t$ is hereditarily head normalizing iff $\hhred$ on $t$ defines a productive reduction path (whose limit is necessarily a 001-normal form by construction, whether this path is of finite length or not). Let us now have look at the alternative notion of infinite normalization.
 
\begin{definition}
\label{def:inf-wn-kj}
A 001-term $t$ has \textbf{an infinitary normal form} if there is a 001-normal form $t'$ such that $t\bredfty t'$ (for some productive path). 
\end{definition}

The reader may keep in mind that, when we say that $t$ has an infinitary normal form $t'$, $t'$ can be either finite or infinite.

It is known that $t$ has an infinitary normal form \ iff $t$ is hereditarily head normalizing, \ie iff the \Bohm\ tree of $t$ does not contain $\bot$. Notice that, if $\BT{t}$ does not contain $\bot$ (\eg when $t=\cuf$), then it is obvious that $t$ has an infinitary normal form by Definitions~\ref{def:bohm-tree-kj} and \ref{def:scrs-kj}, but the converse implication is not, especially in the infinite case~\cite{KennawayKSV97}. 
The statement is proved in \cite{KennawayKSV97} in a syntactical way, but Theorem~\ref{th:charac-WN-S} shall give an alternative  \textit{semantic} proof of this fact. Notice that, contrary to the finite case, this does not hold for leftmost-outermost reduction strategy. If $t :=\cuf\,\cuf$, whose \Bohm\ tree $t':= (\fom)\fom$ does not contain $\bot$, the leftmost-outermost strategy asymptotically computes $\fom\,\cuf$ and not $t'$: hereditary head reduction can be interpreted as a well-balanced version of l.o. reduction.

    
\begin{remark}[Variants of the infinitary $\lam$-calculus] The infinitary $\lam$-calculus has actually seven main variants, which were introduced by Klop and its collaborators~\cite{KennawayKSV97}.  Among those, only three behave well, \ie satisfy the weakened form of confluence. The infinitary calculus that is considered here is one of these 3 calculi, which is referred to as $\Lamzzu$. Its normal forms, as we shall recall in \Sec~\ref{ss:comput-bohm-trees-hp}, are the \Bohm\ trees (without $\bot$). The two other well-behaved infinitary calculi are $\Lamuuu$ and $\Lamuzu$ and their normal forms respectively correspond to the Berarducci trees~\cite{BerarducciI93,Berarducci96} and the Lévy-Longo trees~\cite{Levy75a,Longo83}.
\end{remark}

\subsection{The Finitary Type System $\scrRo$ and Unforgetfulness}
\label{ss:system-Ro-Klop} 

\newcommand{\ttocc}{\mathtt{occ}}
\newcommand{\occm}[1]{\ttocc_{\ominus}(\emul,#1)}
\newcommand{\occp}[1]{\ttocc_{\oplus}(\emul,#1)}
\newcommand{\occe}[1]{\ttocc_{\epsi}(\emul,#1)}
\newcommand{\occme}[1]{\ttocc_{-\epsi}(\emul,#1)}


In this section, we present a well-known variant system of Gardner-de Carvalho's using multisets to represent non-idempotent intersection. This system characterizes \textit{head }normalization. We state its main properties (subject reduction and expansion) and then we explain how it also enables characterizing \textit{weak} normalization, using \textit{unforgetfulness}.

System $\scrRo$  features \textit{non-idempotent} intersection types ~\cite{Carvalho07,Gardner94}, given by the following \textit{inductive} grammar:
$$\sigma,\,\tau \sep :: = \tv \sep |\sep \msigi \rew \tau $$
where the constructor $\ems$ is used for \textit{finite} multisets ($I$ is finite), and the type variable $\tv$ ranges over a countable set $\TypeV$. We write
$\mult{\sig}_n$ to denote the multiset containing $\sig$ with
multiplicity $n$. The multiset $\msigi$ is meant to be the non-idempotent 
intersection of the types $\sigi$, taking into account their
\textit{multiplicity}, and intersection is \textit{strict}, \ie does not occur on codomain of arrows\ighp{, as explained in \Sec~\ref{ss:syntax-direction-discuss-intro}}. 

In system $\scrRo$,  a \textit{judgment} is a triple $\ju{\Gam}{t:\tau}$, where $\Gam$ is a context, \ie a \textit{total} function from the set $\scrV$ of term variables to the set of multiset types $\msigi$, $t$ is a (finite or infinite) 001-term and $\tau$ is a type. The context $x:\msigi$ is the context $\Gam$ such that $\Gam(x)=\msigi$ and $\Gam(y)=\emul$ for all $y\neq x$. 
The multiset union +, satisfying \eg  $\mult{\tau,\sig}+\msig=\mult{\tau,\sig,\sig}$, is extended point-wise on contexts, \eg $\Gam+\Del$ maps $x$ on $\Gam(x)+\Del(x)$. We set $\dom{\Gam}=\set{x\in \TermV\,|\,\Gam(x)\neq \emul}$, the \textbf{domain} of the context $\Gam$. When $\dom{\Gam}\cap \dom{\Del}=\eset$,
 we may write $\Gam;\Del$ instead of $\Gam+\Del$.
The set of $\scrRo$-derivations is defined \textit{inductively} by the rules in Fig.~\ref{fig:system-Ro-klop}:
\begin{figure}[!h]
\bcen \ovalbox{$
  \begin{array}[t]{c}\infer[\ax]{ \phd}{ \ju{x:\mult{\tau} }{x:\tau}} \sep\sep \infer[\abs]{ \ju{ \Gam;x:\msigi}{t:\tau}}{\ju{\Gam}{\lx.t:\msigi\rew \tau}} \\[0.6cm]
    \infer[\app]{\ju{\Gam}{t:\msigi\rew \tau}\sep (\ju{\Deli}{u:\sigi})_{\iI} }{    \ju{\Gam+_{\iI} \Deli}{t\,u:\tau}}
    \end{array}$}\ecen
    \vspace*{-0.3cm}
    \caption{System $\scrRo$}
    \label{fig:system-Ro-klop}
\end{figure}

  We write $\Pi \tri \ju{\Gamma}{t: \tau}$ to mean that the \textit{finite} derivation $\Pi$ concludes with the judgment $\ju{\Gamma}{t:\tau}$ and $\tri \ju{\Gamma}{t: \tau}$ to mean that $\ju{\Gamma}{t:\tau}$ is derivable. 

  No weakening is allowed (\textbf{relevance}). In particular, the $\ax$-rule concludes with
  $\ju{x:\msig }{x:\sig}$ and not with $\ju{\Gam,x:\msigi}{x:\sig_{i_0}}$
  (for some $i_0\in I$). 
  Since contexts are total, if $\juGtt$ and $x\notin \dom{\Gam}$, then $\lx.t$ is always typable: rule $\abs$ gives $\ju{\Gam}{\lx.t:\emul\rew \tau}$.
  Relevance ensures that
  $\lambda x.x$ (resp. $\lx.y$) can be typed with $\mult{\tau}\rew \tau$ (resp. $\emul\rew \tau$), but \textit{not} with $\mult{\tau,\sig}\rew \tau$ (resp. $\mult{\tau}\rew \tau$).

  $$\infer[\abs]{\infer[\ax]{}{\ju{x:\mtau}{x:\tau}}}{\ju{}{\lx.x:\mtau\rew \tau}}\sep\sep\sep
  \infer[\abs]{\infer[\ax]{}{\ju{y:\mtau}{y:\tau}}}{\ju{x:\mtau}{\lx.y:\emul\rew \tau}}
  $$
  
  A straightforward induction shows that if $\juGtt$ is derivable, then $\dom{\Gam}$ is finite. 
  The definition below allows formulating quantitative properties of the system. 

\begin{remark}[Relevance and redundancy of context]
\label{rk:relevance-computation-contexts-Ro}
Contexts are actually superfluous in system $\scrRo$, since, by induction on $\Pi\tri \juGtt$, for all $x\in \scrV$, $\Gam(x)$ is the multiset of the types assigned to $x$ in axiom rules. All the cases are straightforward. In a non-relevant system, this does not hold, since $x$ may have been assigned types \via the weakening rule. However, we prefer to keep contexts explicit in system $\scrRo$ although they are superfluous, since they make the system easier to formulate.
\end{remark}

  \begin{definition}
\label{def:size-Ro-inf-klop}
    Let $\Pi$ be a $\scrRo$-derivation. Then the \textbf{size} of $\Pi$, denoted $\sz{\Pi}$, is the number of judgments in $\scrRo$.
  \end{definition}
  
  Observe that system $\scrRo$ is essentially finite (a derivation contains a finite number of judgments, judgments contains finite number of types, types are finite trees), although it may type infinite terms, which  can appear as \textit{untyped} arguments of applications, \eg $\fom$ ($:=f(\fom)$) in:
 $$\Pi'_1=\infer[\app]{\infer[\ax]{\phd}{\jufera}}{\ju{f:\mult{\erewa}}{\fom:\tv}}$$
The derivation $\Pi'_1$ is finite: it contains only two judgments, and finite types and contexts.

  System $\scrRo$ enjoys both \textbf{subject reduction} and \textbf{expansion}, meaning that types are invariant under (anti-)reduction (if $t\rew t'$, then $\tri \juGtt$ iff $\tri \ju{\Gam}{t':\tau}$). Notice the decrement of size in subject reduction statement.

  \begin{proposition}[Weighted Subject Reduction and Expansion in system $\scrRo$] \label{prop:sr-se-Ro-inf-klop}
    Assume $t\bred t'$.
    \begin{enumerate}
    \item If $\juGtt$ is $\scrRo$-derivable, so is $\juGtpt$.
    Moreover, if $\Pi\tri \juGtt$, then there is $\Pi'\tri \juGtpt$ such that
      $\sz{\Pi'}\leqs \sz{\Pi}$, with actually $\sz{\Pi'}<\sz{\Pi}$ when $t\hred t'$.
      \item If $\juGtpt$ is $\scrRo$-derivable, so is $\juGtt$.        
    \end{enumerate}    
  \end{proposition}
 
  \begin{proof}
 Both points are proved by \textit{induction} on the position of the reduction, using the following Substitution and Anti-Substitution Lemmas in the base case.
 \end{proof}

  \begin{lemma}[Substitution and Anti-Substitution in $\scrRo$]\mbox{}
    Let $r$ and $s$ be two 001-terms.
\begin{itemize}
\item If $\Pi \tri \ju{\Gam;x:\msigi}{r:\tau}$ and, $\forall\iI$, $\Phi_i\tri \ju{\Del_i}{s:\sigi}$, then there is a derivation $\Pi'\tri \ju{\Gam + (+_\iI \Del_i)}{\rsx:\tau}$ such that $\sz{\Pi'}=\sz{\Pi}+(+_\iI\sz{\Phi_i})-\card{I}$ where $\card{I}$ is the cardinality of $I$.
\item If $\Pi'\tri \ju{\Gam'}{\rsx:\tau}$, then there are derivations $\Pi \tri \ju{\Gam;x:\msigi}{r:\tau}$ and $\forall\iI$, $\Phi_i\tri \ju{\Del_i}{s:\sigi}$  such that $\Gam'=\Gam + (+_\iI \Del_i)$.
\end{itemize}
\end{lemma}

  \begin{proof}
Substitution is proved by structural induction on $\Pi$. Anti-substitution is proved by induction on $\Pi'$, with a case-analysis on the last constructor of $r$.
  \end{proof}

\noindent  The following result was proved in \cite{Carvalho07} for $\Lam$. We have slightly adapted it for $\Lamzzu$:

\begin{theorem}[de Carvalho]
\label{th:HN-Ro}
  A 001-term is HN iff it is typable in system $\scrRo$.
\end{theorem}

\begin{proof}
Assume that $\Pi\tri \juGtt$. Let $n=\sz{\Pi}$. Then weighted subject reduction entails that the head reduction strategy outputs a HNF in less than $n$ steps. Conversely, head normal forms are easily $\scrRo$-typable: to type $\hnfo$, we just need to assign $\emul \rew \ldots \rew \emul\rew \tv$ (with $q$ arrows) to the head variable $x$, so that $t_1,\ldots, t_q$ are left untyped. We then use subject expansion to conclude that every HN term is typable.
\end{proof}

\noindent \textbf{Characterizing finitary weak normalization.}  
 A term $t \in \Lamzzu$ is finitarily weakly normalizing when  leftmost-outermost reduction terminates on $t$ and outputs a \textit{finite} normal form $t'$ (note that if $u$ is infinite, then $(\lx.y)u$ is infinite but reduces to the finite normal form $y$ and is thus finitarily weakly normalizing in the sense that we consider here).

 To understand how it is done, observe first that system $\scrRo$ characterizes head normalization (and not weak normalization) because the arguments of a head variable can be left untyped: for instance, $x$ is assigned $\ems\rew \tau$, then $x\,t$ is typable with type $\tau$ for any term $t$ -- which is left untyped as a subterm of $x\,t$ -- even if $t$ is not head normalizing. 
 $$\infer[\app]{\infer[\ax]{\phd}{\ju{x:\mult{\emul \rew \tau}}{x:\emul \rew \tau}}}{\ju{x:\mult{\emul \rew \tau}}{x\,t :\tau}}$$
As we explained in \Sec~\ref{ss:intersection-overview-klop-intro} (Partial typings), the possibility to have untyped arguments is crucial to characterize head normalization because this does not forbid, \eg the head normal form $x\,\Om$ to be typed.

On the other hand, in order to characterize weak normalization, we must guarantee somehow that every subterm is typed, except the subterms that can be erased in a reduction paths, \eg the subterm $u$ in $(\lx.y)u$ should not be typed, since $(\lx.y)u\bred y$. To achieve this, $\ems$ should not occur at bad positions in a derivation $\Pi$. Actually, it is enough to only look at the judgment concluding
$\Pi$. We recall that  $\ems$ occurs negatively in $\ems\rew \tau$ and that $\ems$ occurs positively (resp. negatively) in $\msigi \rew \tau$ if $\ems$ occurs positively (resp. negatively) in $\tau$ or negatively (resp. positively) in some $\sig_i$, \ie the polarity of the occurrences of $\emul$ in a type is preserved in the codomain of arrows and inverted in the domain of arrows.


We say here that judgment $\juGtt$ is \textbf{unforgetful} when
$\ems$ occurs neither negatively in any $\Gam(x)$ nor positively in $\tau$. A term $t$ is unforgetfully typable if an unforgetful judgment $\juGtt$ is derivable. This is given by the following theorem~\cite{Carvalho18}, which is the non-idempotent counterpart of a similar characterization with idempotent intersection based on unforgetfulness:

\begin{theorem}
  \label{th:WN-Ro}
  A term $t$ is finitarily weakly normalizing iff it is unforgetfully typable in system $\scrRo$
\end{theorem}

A straightforward induction shows that, when  $\Pi\tri \juGtt$ is unforgetful, (1) if $t$ is reducible, then every redex of minimal applicative depth is typed, and (2) if $t$ is a normal form, 
 then every subterm of $t$ is typed in $\Pi$ (thus, no part of the normal form of a term $t$ is ``forgotten'' in an unforgetful derivation), which roughly justifies why the above theorem holds. In general, the unforgetfulness criterion gives a characterization of \textit{weak} normalization in any type system which characterizes \textit{head} normalization.

 For now (in particular for \Sec~\ref{ss:typ-inf-nf-informal}), it is enough to keep in mind that a \textbf{sufficient condition} of unforgetfulness is to be \textbf{$\ems$-free}: $t$ is weakly normalizing as soon as $\tri \juGtt$, where $\Gam$ and $\tau$ do not contain $\ems$. 
Throughout this article, we will frequently invoke system $\scrRo$ and this section to illustrate or motivate our choices.

\subsection{Infinitary subject reduction and expansion for finite derivations}
\label{ss:subj-subst}

Now, let us explain why subject reduction and expansion hold for infinite \textit{productive} reduction path (and not just finite ones). The fact that if $\tri \juGtt$ and $t\bredfty t'$ then there is $\tri \juGtpt$ (infinitary subject reduction) is actually an easy consequence of the definition of productive paths (Definition~\ref{def:scrs-kj}). In contrast, the fact that if $\Pi \tri \juGtpt$ and $t\bredfty t'$, then $\tri\juGtt$ (infinitary subject expansion) relies on the finiteness of the derivations in $\scrRo$ and a technical observation that we call subject substitution and that will be a key tool in \Sec~\ref{ss:typ-inf-nf-informal} and \ref{s:infty-expansion} to obtain full infinitary subject expansion. Let us prove the infinitary subject reduction property.

\begin{proposition}[Infinitary suject reduction (system $\scrRo$)]\label{prop:inf-subj-red-Ro}\mbox{}\\
  If $t\bredfty t'$ and $\tri_{\scrRo} \juGtt$, then $\tri_{\scrRo} \juGtpt$.
  \end{proposition}
  
  \begin{proof}
  Assume $\Pi \tri_{\scrRo} \juGtt$ and $\Pi =: \Pi_0 \breda{b_0} \Pi_1 \breda{b_1} \Pi_2 \ldots$ the derivations obtained by subject reduction along the given path $t\bredfty t'$.
    For every $k\in \bbN$, let $N_k$ be such that, for all 
    $n\geqs N_k$, $\ad{b_n}> k+1$. 
  
    Let $\Pi'$ be the derivation typing $t'$ defined as follows:
    under applicative depth $k$, $\Pi'$ is the same as $\Pi_n$ for any $n\geqs N_k$ except that the judgments of $\Pi'$ type $t'$ and its subterms instead of $t_n$ and its subterms.
  
    This definition is licit, because, under applicative depth $\Pi_{n}$ and $\Pi_{n'}$ do not differ whenever $n,n'\geqs N_k$, except for \textit{untyped} parts of their subjects, which also proves that $\Pi'$ is correct.
    \ignore{
    When $t\bredfty t'$ in a finite number of steps, this results from inductive application of Proposition~\ref{prop:sr-se-Ro-inf-klop}. Let us then assume that $t=t_0\breda{b_0} t_1\breda{b_1}t_2\bredfty t'$ is an infinite path and that $\Pi \tri_{\scrRo} \juGtt$.
  
    Let $k=\sz{\Pi}$ ($k<+\infty$!). Let $N$ such that, for all $n \geqs N$, $\ad{b_n}>k$ (such a $n$ exists, since the reduction path is productive). 
    By Proposition~\ref{prop:sr-se-Ro-inf-klop}, there is $\Pi_N\tri \ju{\Gam}{t_N:\tau}$ such that $\sz{\Pi_N}\leqs k$. 
  
   By Observation~\ref{obs:ad-less-than-size-Ro}, for all $n\geqs N$, $b_n\notin \hat{\Pi_N}$. Then, by induction on $n\geqs N$, for all $b\in \hat{\Pi}_N$, $t_n(b)=t_N(b)$. By Lemma~\ref{lem:limit-prod-paths-hp}, this implies that, for all $b\in \hat{\Pi}_N$, $t'(b)=t_N(b)$. By subject substitution (Lemma~\ref{lem:sub-subst-Rftyo}) applied to $\Pi_N$, there is a derivation $\Pi'\tri \ju{\Gam}{t':B}$. This concludes the proof.} 
  \end{proof}

  One notices that the proof of Proposition~\ref{prop:inf-subj-red-Ro} does not use anywhere the finiteness of the derivations (we perform subject substitution but only at fixed depth $k$ after a sufficient number of steps $N_k$). The statement and its proof shall still be valid in the infinitary system (Proposition~\ref{prop:infinite-subject-reduction}, which we prove more formally).

  But as we explained in \Sec~\ref{ss:klop-prob-hp-statement}, naively, subject expansion is not an asymptotic process. However, there is a way to by-pass the ``infinite jump problem'' by using the fact the derivations are finite: one may replace the limit $t'$ of a productive path  with some $t_n$ when $n$ is big enough, without compromising the correctness of the derivation.

Actually, in a derivation $\Pi$ whose subject is $t_1$, we may replace $t_1$ with any term $t_2$ which coincide with $t_1$ on the typed parts inside $\Pi$. We name this phenonemon \textit{subject substitution}.


\begin{example}
\label{ex:subj-sub-Ro}
Consider:
$$\Pi_1=\infer{\infer{}{\ju{x:\mult{\emul\rew \tv}}{x:\emul \rew \tv}}}{\ju{x:\mult{\emul\rew \tv}}{x(y\,z):\tv}}
\sep\sep\sep
\Pi_2=\infer{\infer{}{\ju{x:\mult{\emul\rew \tv}}{x:\emul \rew \tv}}}{\ju{x:\mult{\emul\rew \tv}}{x\,\Om:\tv}}
$$
We transform $\Pi_1$ into $\Pi_2$ by \textit{substituting} the subject $x(y\,z)$ with $x\,\Om$. This is licit because we replace $y\,z$, which is an untyped subterm of $x(y\,z)$ in the derivation $\Pi_1$.
\end{example}

To formulate subject substitution, we need to define the set of positions of the  typed subterms in a derivation $\Pi$:

\begin{definition}
  \label{def:typed-pos-Ro}
  Let $\Pi$ a derivation typing a term $t$. We define the set $\hPi\subset \supp{t}$ of the \textbf{typed positions} in $\Pi$ by the following induction on $\Pi$:
  \begin{itemize}
      \item If $\Pi$ ends with an axiom rule ($t=x$), then $\hPi=\set{\epsi}=\supp{t}$.
      \item If $\Pi$ ends with an $\abs$-rule ($t=\lx.t_0$) and $\Pi_0$ is its depth 1 subderivation, then $\hPi=\set{\epsi}\cup 0\cdot \hPik{0}$.
      \item If $\Pi$ ends with an $\app$-rule ($t=t_1\,t_2$), $\Pi_1$ is its left premise (concluding with $t_1:\msigi\rew \tau$) and the $(\Pi_i)_{\iI}$ (concluding respectively with $t_2:\sigi$) are its right premises, then
       $\hPi=\set{\epsi}\cup 1\cdot \hPik{1} \cup (\cup_{\iI} 2 \cdot \hPik{i})$
  \end{itemize}
\end{definition}

In the $\app$-case, if $I=\eset$, $t_2$ is not typed and 
we only have $\hPi=\set{\epsi}\cup 1\cdot \hPik{1} $ , as expected. In the typings of $x(y\,z)$ and $x\,\Om$ in Example~\ref{ex:subj-sub-Ro}, we have $\hat{\Pi}_1=\hat{\Pi}_2=\set{\epsi,1}$. The applicative depth of a typed position inside $\Pi$ cannot exceed the size of $\Pi$ (this can be proved by induction on $\Pi$):

\begin{observation}
  \label{obs:ad-less-than-size-Ro}
  Let $\Pi$ be a $\scrRo$-derivation. If $b\in \hPi$, then $\ad{b}<\sz{\Pi}$.
\end{observation}

\begin{lemma}[Subject Substitution] \fublainv{référer vers ce lemme ds \eg approx, sr infty}
  \label{lem:sub-subst-Rftyo}
  Let $t,u\in \Lamfty$.
If $\Pi \tri_{\scrRo} \juGtt$, $\hPi\subeq \supp{u}$ and, for all $b\in \hPi$, $u(b)=t(b)$, then there exists a $\scrRo$-derivation $\Psi$ such that $\Psi \tri \ju{\Gam}{u:\tau}$ and $\hat{\Psi}=\hPi$.
\end{lemma}

\begin{proof}
Straightforward by induction on $\Pi$.
\end{proof}

To obtain infinitary subject expansion in $\scrRo$, the idea is to perform subject substitution in the considered infinite productive reduction path. Indeed, if $\Pi$ types $t$ and $t$ and $u$ differ only for positions $b\in \setpos$ such that $|b|>\sz{\Pi}$ (where $|b|$ is the length of $b$), then Lemma~\ref{lem:sub-subst-Rftyo} can be applied on $\Pi$.

Contrary to infinitary subject reduction, the proof of  infinitary subject expansion below crucially relies on the fact that $\scrRo$-derivations are finite, in order to perform a subject substitution:

\begin{proposition}[Infinitary suject expansion (system $\scrRo$)]\mbox{}\label{prop:inf-subj-exp-Ro}\\
If $t\bredfty t'$ and $\tri_{\scrRo} \juGtpB$, then $\tri_{\scrRo} \juGtB$.
\end{proposition}

\begin{proof}
  When $t\bredfty t'$ in a finite number of steps, this is an obvious consequence of Proposition~\ref{prop:sr-se-Ro-inf-klop}. Let us then assume that $t=t_0\breda{b_0} t_1\breda{b_1}t_2\bredfty t'$ is an infinite path and that $\Pi'\tri_{\scrRo} \juGtpt$. Let $k=\sz{\Pi'}$.
  Let $N$ such that, for all $n \geqs N$, $\ad{b_n}>k$ (such a $n$ exists, since the reduction path is productive). By definition of the limit (Lemma~\ref{lem:limit-prod-paths-hp}), this implies that $\supp{t_N}\cap \set{b\in \set{0,1,2}^*\,|\, \ad{b}\leqs k} = \supp{t'} \cap \set{b\in \set{0,1,2}^*\,|\, \ad{b}\leqs k}$ and that, on this set, $t_N(b)=t'(b)$. In particular, by Observation~\ref{obs:ad-less-than-size-Ro}, $t_N(b)=t'(b)$ is well-defined and holds on $\hPi'$.

  By subject substitution (Lemma~\ref{lem:sub-subst-Rftyo}) applied to $\Pi'$ and $t_N$, there is a derivation $\Pi_N\tri \ju{\Gam}{t_N:\tau}$.
  
  By Proposition~\ref{prop:sr-se-Ro-inf-klop} applied $N$ times, there is $\Pi\tri \ju{\Gam}{t:\tau}$ such that $\sz{\Pi}\leqs k$. This concludes the proof.
\end{proof}

\begin{remark}\mbox{}
\begin{itemize}
  \item One may take a moment to consider the proof of Lemma~\ref{prop:inf-subj-exp-Ro} with that of Lemma~\ref{prop:inf-subj-red-Ro}. As noticed in \Sec~\ref{disc:infinite-jump-problem}, we need to overcome the ``infinite jump problem''. Here, it is pivotal that, under some applicative depth $k$, the subject $t'$ of the derivation to expand is left untyped. This allows us to ``jump'' from $t'$ to some $t_N$ for a sufficiently large $N$ which is similar to $t'$ under applicative depth $k$. Note that no such need arises while proving Lemma~\ref{prop:inf-subj-red-Ro}, since, for infinitary subject reduction, we define $\Pi'$, the derivation typing the limit, incrementally.
  \item   The lemmas of this section all apply to $\Lamfty$ (\Sec~\ref{ss:inf-lam-klop-journal}) and not only to $\Lamzzu$ with minor changes. 
  In particular, for any $t\in \Lamfty$, there is a reduction path from $t$ to a head normal form iff the head reduction terminates on $t$. 
Notice that some $t\in \Lamfty$ are \textit{headless}, \eg $^\infty x$ from Example~\ref{ex:lambfty-vs-lamzzu}, which has neither a head variable nor a head redex, because of its leftward infinite branch, but such a headless term is not typable in $\scrRo$.
\end{itemize}
\end{remark}

\subsection{Roadmap to solve Klop's Problem}
\label{ss:typ-inf-nf-informal}

Since hereditary head normalization is a form of infinitary weak normalization, we want to adapt the proof of Proposition~\ref{th:WN-Ro} which gives a type-theoretic characterization of \textit{finitary} weak normalization.\igintro{\pierre{ 
The two main ingredients follow the principles of \Sec~\ref{ss:inter-from-syntax-intro-klop}:}} To achieve this, we need:
\begin{itemize}
\item \textit{Typing (infinitary) normal forms} in unforgetful judgment (in this section, $\fom$).
\item \textit{Using a form of infinitary subject expansion} to obtain a derivation typing the expanded term (here, $\cuf$)
\end{itemize}
The second point (infinitary subject expansion) is delicate, but we give intuitions about how it may be achieved just below.  We will actually follow some of the ideas that were given in \Sec~\ref{ss:klop-prob-hp-statement}. 
This will allow us to present the key notions of \textbf{truncation} and \textbf{approximability}. Fig.~\ref{fig:expans-by-trunc-1} illustrates the main ideas of this section.\\

\noindent \textbf{A complete trial on $\cuf$.} The first point (typing 001-normal forms) is easier to grasp: notice that in a type system characterizing infinitary weak normalization, derivations should be able to type infinite normal forms \textit{fully}, \ie without leaving subterms untyped. For instance, the following derivation $\Pi'$, which fully types $\fom$, should be allowed (this derivation is also represented on top of Fig.~\ref{fig:expans-by-trunc-1}):

  $$\Pi'=\infer[\app]{\infer[\ax]{}{\ju{f:\mult{\arewa}}{f:\arewa}}\\
\sep \Pi'\tri \ju{f:\mult{\arewa}_{\omega}}{\fom:\tv}
  }{\ju{f:\mult{\arewa}_{\omega}}{\fom:\tv}}
$$

To define $\Pi'$, we have informally used an infinitary version of system $\scrRo$ that we call system $\scrR$ (system $\scrR$ is formally defined in Appendix~\ref{a:system-R}). System $\scrR$ allows infinite multisets (\eg $\mult{\tv}_{\omega}$ is the multiset in which $\tv$ occurs with an infinite multiplicity, so that $\mult{\tv}_{\omega}=\mult{\tv}+\mult{\tv}_{\omega}$), proofs of infinite depth and also infinitary nestings, \eg the type $\sig$ defined by $\sig=\mult{\sig}\rew \tv$, so that $\sig=\mult{\mult{\ldots}\rew\tv}\rew\tv$, is legal.\\

\begin{figure}
\thispagestyle{empty}
\begin{tikzpicture}

\transh{3.5}{

\draw (-2.7,5.8) node[right]{$\Pi'\rhd f:\,\red{ [ [ \tv]\rew \tv]_{\omega}} \vdash \fom:\red{\tv}$ } ; 

\draw (1,-1) node {$\Pi'$ is an \textit{infinite} derivation typing the \textit{infinite} term $\fom$};


  {\draw (4.3,0.8) node{\parbox{2.3cm}{
        \begin{center}Every Variable is Typed\end{center}} } ; }



    \inputarewa{2.7}{4.5}

    \outputtv{3.6}{3.6}

   \inputarewa{1.8}{3.6}

     
     \outputtv{2.7}{2.7}

     \inputarewa{0.9}{2.7}
  \outputtv{1.8}{1.8}
    
   \inputarewa{0}{1.8}
\outputtv{0.9}{0.9}

   \inputarewa{-0.9}{0.9}
   \outputtv{0}{0}


  \draw [style=dashed] (3.79,3.79) -- (4.9,4.9) ;

\blockfabis{3.6}{3.6}
    
  \blockfabis{2.7}{2.7}

\blockfabis{1.8}{1.8}

\blockfabis{0.9}{0.9}
    

\blockfabis{0}{0}
}


\transv{-7.6}{


\draw (5,6.15) node {$\Pi'$ can be truncated into \eg,};
\draw (-2,6) node [below right]{$\Pi'_3\tri \ju{f:\red{\mult{\arewa}_2+\mult{\erewa}}}{\fom:\red{\tv}}$};

  \draw (5,-0.6) node [below]{\parbox{11cm}{
     Both $\Pi'_3$ and $\Pi'_4$ are \textit{finite} derivations typing the \textit{infinite} term $\fom$}
    };



     \inputpharewa{0.9}{2.7} 
  \outputtv{1.8}{1.8}
    
   \inputarewa{0}{1.8}
\outputtv{0.9}{0.9}

   \inputarewa{-0.9}{0.9}
   \outputtv{0}{0}


  \draw [style=dashed] (3.79,3.79) -- (4.9,4.9) ;

\blockfabis{3.6}{3.6}


  \blockfabis{2.7}{2.7}


\blockfabis{1.8}{1.8}

\blockfabis{0.9}{0.9}
    

\blockfabis{0}{0}

\transh{7}{
\draw (-2.3,5.84) node [below right]{or};
  
\draw (-1.5,6) node [below right]{$\Pi'_4\tri \ju{f:\red{\mult{\arewa}_3+\mult{\erewa}}}{\fom:\red{\tv}}$};




   \inputpharewa{1.8}{3.6} 
     
     \outputtv{2.7}{2.7}

     \inputarewa{0.9}{2.7}
  \outputtv{1.8}{1.8}
    
   \inputarewa{0}{1.8}
\outputtv{0.9}{0.9}

   \inputarewa{-0.9}{0.9}
   \outputtv{0}{0}


  \draw [style=dashed] (3.79,3.79) -- (4.9,4.9) ;

\blockfabis{3.6}{3.6}


  \blockfabis{2.7}{2.7}


\blockfabis{1.8}{1.8}

\blockfabis{0.9}{0.9}
    

\blockfabis{0}{0}

  }

  }


\transv{-15.7}{
  \draw (5,6.4) node {\parbox{14cm}{
\begin{center}
The untyped parts of the subject can be substituted\\
      \eg $\fom$ can be replaced by $f^4(\cuf)$ in $\Pi'_3$ and $\Pi'_4$, yielding respectively\ldots\end{center} }};

\draw (-2,6) node [below right]{$\Pi^4_3\tri \ju{f:\red{\mult{\arewa}_2+\mult{\erewa}}}{f^4(\cuf):\red{\tv}}$};

\draw (5,-0.8) node [below]{\parbox{11cm}{
     Both $\Pi^4_3$ and $\Pi^4_4$ are \textit{finite} derivations typing the \textit{finite} term $f^4(\cuf)$}
    };



     \inputpharewa{0.9}{2.7} 
  \outputtv{1.8}{1.8}
    
   \inputarewa{0}{1.8}
\outputtv{0.9}{0.9}

   \inputarewa{-0.9}{0.9}
   \outputtv{0}{0}




  \blockcufbis{3.6}{3.6}

  \blockfabis{2.7}{2.7}


\blockfabis{1.8}{1.8}

\blockfabis{0.9}{0.9}
    

\blockfabis{0}{0}

\transh{7}{

\draw (-1.5,6) node [below right]{$\Pi^4_4\tri \ju{f:\red{\mult{\arewa}_3+\mult{\erewa}}}{f^4(\cuf):\red{\tv}}$};
  




   \inputpharewa{1.8}{3.6} 
     
     \outputtv{2.7}{2.7}

     \inputarewa{0.9}{2.7}
  \outputtv{1.8}{1.8}
    
   \inputarewa{0}{1.8}
\outputtv{0.9}{0.9}

   \inputarewa{-0.9}{0.9}
   \outputtv{0}{0}




  \blockcufbis{3.6}{3.6}

  \blockfabis{2.7}{2.7}


\blockfabis{1.8}{1.8}

\blockfabis{0.9}{0.9}
    

\blockfabis{0}{0}

  }

  }

\end{tikzpicture}    
\vspace*{-0.6cm}
\caption{Truncation and Subject Substitution}
\label{fig:expans-by-trunc-1}
\end{figure}

\noindent \textbf{Step 1 (truncations/approximations).} 
Notice that $\Pi'$ is of infinitary depth and \textit{fully} types $\fom$ (\ie no subterm of $\fom$ is left untyped in $\Pi'$).  Actually, $\Pi'$ is \textit{$\ems$-free} and in particular, it is unforgetful (recall the discussion surrounding Theorem~\ref{th:WN-Ro}).

Thus, $\Pi'$ is the kind  of derivation we want to expand, as to get a derivation $\Pi$ typing $\cuf$ and witnessing that this term is infinitarily weakly normalizing. Since $\cuf \bredfty \fom$ (infinite number of reduction steps), we are stuck in the ``infinite jumb problem'': indeed, the pivotal argument in Proposition~\ref{prop:inf-subj-exp-Ro} allowing to ``jump back'' from the limit of a productive path to a term of this path is that the subject is partially typed, which allows performing subject substitution. Here, not such thing occurs, since no subterm of $\fom$ is untyped in $\Pi'$.

But notice that $\Pi'$ can be truncated into the derivation $\Pi'_n$ below, for any $n\geqslant 1$. To define these derivations, we have set $\Gamma_n=f:\mult{\arewa}_{n-1}+[\erewa]$ and we write $x:\tau$ instead of $\juaxtt$ for $\ax$-rules:


%


  $$ \Pi'_n=
\infer[\app]{
  \infer[\ax]{}{\fara } \hspace*{-1cm}  \parbox[b]{4cm}{
$\begin{array}[b]{c}
  \infer[\app]{
      \infer[\ax]{}{\fara }
      \infer[\app]{\infer[\ax]{}{\ju{\Gam_1}{f:\erewa} }}{\ju{\Gam_1}{\fom:\tv}}}    
          {\ju{\Gamma_2}{\fom :\tv}}\\
  \vdots \\
\ju{\Gam_{n-1 }}{\fom :\tv}
\end{array}$}  \phantom{aa}}
{\ju{\Gam_n}{\fom : \tv}}
$$

Derivations $\Pi'_3$ and $\Pi'_4$ are represented in the middle of Fig.~\ref{fig:expans-by-trunc-1}. By \textbf{truncation} or \textbf{approximation}, we mean, as suggested in \Sec~\ref{ss:sequence-tracking-klop-intro}, that the finite derivation $\Pi'_n $ can be informally obtained from the infinite one $\Pi'$ by erasing some elements from the infinite multisets appearing in the derivation. We also informally write $\Pi'_1 \leqs \Pi'_2 \leqs \Pi'_3\ldots $ to mean that $\Pi'_n$ approximates $\Pi'_{n+1}$. Thus, $\Pi'_n\leqs \Pi'$ for all $\nN$.

Conversely, we see that $\Pi'$ is the graphical \textbf{join} of the $\Pi'_n$: $\Pi'$ is obtained by superposing suitably all the derivations $\Pi'_n$ on the same infinite sheet of paper.\\

\noindent \textbf{Step 2 (expand the finite approximations).}

\begin{itemize}
\item \textit{Substituting $\fom$.} Observe now that, although we do not know yet how to expand $\Pi'$, we can expand the $\Pi'_n$, because the $\Pi'_n$ are  $\scrRo$-derivations (thus, they are \textit{finite}): by subject expansion in system $\scrRo$ (Proposition~\ref{prop:sr-se-Ro-inf-klop}), for all $n\geqs 1$, there is a $\scrRo$-derivation $\Pi'_n$ concluding with $\ju{\Gam_n}{\cuf:\tv}$. In the particular case of $\Pi'$, we can detail the process: (1) $\Pi'_n$ leaves $\fom$ untyped beyond depth $n$ and (2) for any $k\geqs n$, $\fom$ and $f^k(\cuf)$ are similar below depth $n$, so using subject substitution, for any $k\geqs n$, $\Pi'_n$ yields a derivation $\Pi^k_n$ typing $f^k(\cuf)$.  We have represented $\Pi^4_3$ and $\Pi^4_4$ at the bottom of Fig.~\ref{fig:expans-by-trunc-1}.
\item \textit{Finite expansions.}
  Thus, $\Pi_n^k$ types $f^k(\cuf)$, the rank $k$ reduct of $\cuf$, so we can expand it $k$ times, obtaining a derivation $\Pi_n$. It can easily be observed that $\Pi_n$ does not depend on $k$, because $\Pi^k_n$ and $\Pi^n_n$ have the same typed parts.\\
\end{itemize}

\noindent \textbf{Step 3 (joining the finite expansions).}
We can then define $\Pi$, the expected expansion of $\Pi'$, as the join of all the derivations $\Pi_n$. This is justified by the following points:
\begin{itemize}
\item[(1)]
  for any $\nN$, we have $\Pi^n_0\leqs \Pi^n_1\leqs \ldots \leqs \Pi^n_n$ (see Fig.~\ref{fig:expans-by-trunc-1}, bottom)
\item[(2)]
  proof reduction and expansion are monotonic in this\footnote{But proof reduction is not monotonic in all cases, simply because, in general, it is non-deterministic in system~$\scrRo$. See \Sec~\ref{ss:non-determinism}} case. Thus, since $\Pi_0$, $\Pi_1$,\ldots, $\Pi_n$ are respectively obtained after $n$ expansion steps from $\Pi^n_0$, $\Pi^n_1$,\ldots, $\Pi^n_n$, we have $\Pi_0\leqs \Pi_1\leqs \ldots \leqs \Pi_n$. This intuitively explains why the $\Pi_n$ (typing $\cuf$) have an infinitary join $\Pi$, although the construction of this join is delicate for reasons to be presented in the next section.\\
\end{itemize}

\noindent \textbf{Summary.} 
Thus, given an infinite derivation $\Pi'$ typing $t'$ and a productive reduction path $t=t_0\bred t_1\bred \ldots \bredfty t'$, the main ingredients to perform infinitary subject expansion are:
\begin{itemize}
\item[1.] Approximating $\Pi'$ into a family of finite derivations $\fPi'$ also typing $t'$, so that $\Pi'$ is the (asymptotic) join of the $\fPi'$.
\item[2.a.] For each approximation $\fPi'$, replacing $t'$ with $t_n$ for a sufficiently large $n$ ($n$ depends on $\fPi'$). This gives $\fPi^n$, which is also finite, but types the term $t_n$ instead of $t'$.
\item[2.b.] For each $\fP'$, expanding $\fPi^n$ $n$ times. This gives $\fPi$, which type $t$.
\item[3.] Obtaining $\Pi$, the expansion of $\Pi'$, whose subject is $t$, by taking the join of the \textit{finite} $\fPi$ while $\fP'$ ranges over the approximations of $\Pi'$.
\end{itemize}

\subsection{Problems with infinitary typing and how to solve them}
\label{ss:degenerate}

Thus, the ideas of truncation, subject substitution and join guides us about how to perform $\infty$-subject expansion. The particular form of $\Pi_n$ and $\Pi$ does not matter (but they
 are given in Appendix \ref{s:expanding-pi-prime-n} for the curious reader). Let us just say here that the $\Pi_n$ involve a family of finite types $(\rho)_{n\geqslant 1}$ inductively defined by $\rho_1=\erewa$ and $\rho_{n+1}=\mult{\rho_k}_{1\leqslant k \leqslant n}\rew \tv$ and $\Pi$, their infinite join, involves an infinite type $\rho$ satisfying $\rho=\mult{\rho}_{\omega} \rew \tv$: if $t$ and $u$ are typed with $\rho$, then $t\,u$ may be typed with $\tv$.\\ 

\noindent \textbf{Problem 1 ($\Om$ is typable).}
Unfortunately, it is not difficult to see that the equality $\rho=\mrho_\om\rew \tv$ also allows  the non-head normalizing term $\Om=\Delta \Delta$ to be typed. 
$$\infer[\app]{\infer[\abs]{\infer[\app]{\infer[\ax]{}{\ju{x:\mrho}{x:\rho}}\sep \big(\infer[\ax]{}{\ju{x:\mrho}{x:\rho}}\big)_\om}{\ju{x:\mrho_\om}{x\,x:\tv}}}{\ju{}{\Delta:\rho}}\sep \big(\ju{}{\Delta:\rho}\big)_\om}{\ju{}{\Om:\tv}}$$

This last observation shows that the naive extension of the standard non-idempotent type system to infinite terms is unsound as non-head normalizing terms can be typed (actually, every term is typable in system $\scrR$~\cite{VialLICS18}: this means that every $\lam$-term has a non-empty denotation in the infinitary relational model). Therefore, we need to discriminate between sound derivations (like $\Pi$ typing $\cuf$) and unsound ones. For that, we define an infinitary derivation $\Pi$ to be \textbf{valid} or \textbf{approximable} when it is the join of all its \textit{finite} approximations (Figure~\ref{fig:approx-deriv-intro}). For instance, $\Pi'$ is approximable (it is the join of the $\Pi'_n$) and $\Pi$, defined as the join of the directed family $(\Pi_n)_{n\geqs 1}$ is approximable.
\\

\noindent \textbf{Problem 2 (validity is not definable with multisets).} It turns out that approximability cannot be formally defined in system $\scrR$, \ie with multisets as intersection types. We managed to define in the example of $\cuf$ and $\fom$ only because we explained the concept informally and in those terms, two equal subterms are assigned the same typing. Intuitively, approximability cannot be defined in system $\scrR$ because multisets disable tracking, as explained in \Sec~\ref{ss:sequence-tracking-klop-intro}. Another (informal) argument will be given later in \Sec~\ref{ss:non-determinism}: since proof reduction is not deterministic in system $\scrR$, the reduction dynamics of a $\scrR$-derivation $\Pi$ may be very different from that of finite its approximations. We give a formal argument, to long to be presented in the core of this article, in Appendix~\ref{a:rep-and-dynamics}, once we have presented system $\ttS$ and defined approximability in this system: we find two distinct $\ttS$-derivations $P_1$ and $P_2$, which both collapse on the same $\scrR$-derivation $\Pi$, but such that $P_1$ is approximable  (in system $\ttS$) whereas $P_2$ is not. This proves that approximability cannot be lifted from $\scrR$ to $\ttS$.


\section{Tracking types in derivations}
\label{s:rigid}

In this section, we define system $\ttS$, which is based on the key construct of
\textbf{sequence} (presented in \Sec~\ref{ss:sequence-tracking-klop-intro}) which we use to represent intersection in a \textit{rigid} way, \ie which enables tracking.
\begin{itemize}
\item We first coinductively define arrow types as function types whose domains are sequence of types. A notion of typing judgment naturally follows.
\item We then define derivations trees such that arguments of applications are typed with a sequence of judgments.
\end{itemize}




\subsection*{General notations} 
Some of these notations generalize or refine those of \Sec~\ref{s:finite-inter-infinite-types}.
\begin{itemize}
\item A \textbf{track} $k$ is any natural number. When $k\geqs 2$, $k$ is called an \textbf{argument track}.
  As for $\lam$-terms, 0 is dedicated to the constructor $\lx$, 1 is dedicated to the left-hand side of applications. Not only 2 but   all the $k\geqs 2$ to the possibly multiple typings of the arguments of applications.   For instance, in system $\ttS$, while typing an application $t\,u$, a subderivation on track $9$ (or on track 2, track 3 and so on) will be a subderivation typing the argument $u$. 
  This also explains why 0 and 1 will have a particular status in the definitions to come, and motivates the definition of \textit{collapse} below.
\item  The set of finite words on $\bbN$ is denoted with $\bbN^*$, $\epsi$ is the empty word, $a \cdot a'$ the concatenation of $a$ and $a'$.   The prefix order $\prefleq$ is defined on $\bbN^*$ by $a\prefleq a'$ if there is $a_0$ such that $a'=a\cdot a_0$, \eg $2\cdot 3 \prefleq 2\cdot 3\cdot 0 \cdot 1$. When the denotation is clear, we may just write $2031$ instead of $2\cdot 0 \cdot 3 \cdot 1$. 
\item The \textbf{applicative depth} $\ad{a}$ generalizes to any $a\in \bbN^*$: it still represents the number of nestings inside arguments,
  \ie  $\ad{a}$ is defined inductively by $\ad{\epsi}=0$, $\ad{a\cdot k}=\ad{a}$ if $k=0$ or $k=1$ and $\ad{a\cdot k}=\ad{a}+1$ if $k\geqs 2$.
\item  The \textbf{collapse}\label{txt:def-collapse} is defined on $\bbN$ by $\ovl{k}=\min(k,2)$ and on $\bbN^*$ inductively by $\ovl{\epsi}=\epsi$, $\ovl{a\cdot k}=\ovl{a}\cdot \ovl{k}$, \eg $\ovl{7}=2$, $\ovl{1}=1$ and $\ovl{2\cdot 3\cdot 0 \cdot 1}=2\cdot 2 \cdot 0 \cdot 1$. Intuitively, the collapse converts a position in system $\ttS$ into a position in a $\lam$-term (with only 0s, 1s or 2s) .\\
  This extends to words of infinite length,
  the collapse of $3^\om$ and $(1\cdot 2\cdot 3)^\om$ are $2^\om$ and $(1\cdot 2^2)^\om$.
\item  A \textbf{sequence} of elements of a set $X$ is a family $(k \cdot x_k)_{\kK}$ with $K\subeq \Nmzo$ and $x_k\in X$. In this case, if $k_0\in K$, $x_{k_0}$ is the element of $(k \cdot x_k)_{\kK}$ \textbf{on track} $k_0$.
  We often write $(k\cdot x_k)_{\kK}$ for $(x_k)_{\kK}$, which,  for instance, allows us to denote by $(2\cdot a,4\cdot b,5\cdot a)$ or $(4\cdot b,2\cdot a,5\cdot a)$  the sequence $(x_k)_{\kK}$ with $K=\set{2,4,5}$, $x_2=x_5=a$ and $x_4=b$. In this sequence, the element on track 4 is $b$. If $S=(k\cdot x_k)_\kK$ is a sequence,
we write $K=\Rt(S)$ ($\Rt$ stands for ``\textbf{roots}'', which are the \textit{top-level} tracks of $S$). 
Sequences come along with an (infinitary) \textbf{disjoint union} operator, denoted $\uplus$: let $S_j:=(k\cdot x_{k,j})_{\kK_j}$ be sequences for all $\jJ$ (where $J$ is a possibly infinite set):
\begin{itemize}
\item If the $(K_j)_{\jJ}$ are pairwise disjoint, then $\uplus_\jJ S_j$ is the sequence $(k\cdot x_k)_{\kK}$ where $K=\union_\jJ K_j$ and, for all $\kK$, $x_k=x_{k,j}$ where $j$ is the unique index such that $k\in K_j$. In that case, we say that the $S_j$ are \textbf{disjoint}.
\item If the $(K_j)_{\jJ}$ are not pairwise disjoint, then $\uplus_\jJ S_j$ is not defined.
\end{itemize}
When $J$ is finite, we use $\uplus$ as an \textit{infix} operator: $S_1\uplus S_2$ etc.

The operator $\uplus$ is partial, and infinitarily associative and commutative. For instance,  $(2\cdot a,\, 3\cdot b,\, 8\cdot a)\uplus (4\cdot a,\,9\cdot c)=(4\cdot a,\,9\cdot c) \uplus (2\cdot a,\, 3\cdot b,\, 8\cdot a)= (2\cdot a,\, 3\cdot b,\, 4\cdot a,\,8\cdot a,\,9\cdot c)$, but $(2\cdot a,\, 3\cdot b,\, 8\cdot a)\uplus (3\cdot b,\,9\cdot c)$ is not defined, because track 3 is in the roots of both sequences (\textbf{track conflict}).




\item A \textbf{tree} $A$ of $\bbN^*$ is a non-empty subset of $\bbN^*$ that is downward-closed for the prefix order ($a\leqslant a'\in A$ implies $a\in A$).
A \textbf{forest} is a set  of the form $A\setminus \set{\epsi}$ for some tree $A$ such that $0,\,1 \notin A$. Formally, a \textbf{labelled tree} $T$ (resp. \textbf{labelled forest} $F$) is a function to a set $\Sigma$, whose domain, called its \textbf{support} $\supp{T}$ (resp. $\supp{F}$), is a tree (resp. a forest). If $U=T$ or $U=F$, then $U\rstr{a}$ is the function defined on $\set{a_0\in \bbN^*\,|\, a\cdot a_0\in \supp{U}}$ and $U\rstr{a}(a_0)=U(a\cdot a_0)$. If $U$ is a labelled tree (resp. forest and $a\neq \epsi$), then $U\rstr{a}$ is a tree.
\end{itemize}

\subsection{Rigid Types}
\label{ss:types-S}


We start now to implement the ideas of \Sec~\ref{ss:sequence-tracking-klop-intro} to overcome the problems concluding \Sec~\ref{ss:degenerate}: to enable tracking (so that we may define approximability),  every argument derivation or type in the domain of an arrow must now receive a track for label.


Let $\TypeV$ be a countable set of types variables (metavariable $\tv$). 
The sets of (rigid) types $\Types^{111}$ (metavariables $T$, $S_i$, \ldots)
is \textit{coinductively} defined in Fig.~\ref{fig:S-types-def}.

\begin{figure}
  \begin{center}$ 
  \begin{array}{lll} 
     T,\,S_k & ::= &  \tv \sep \|\sep (k\cdot S_k)_\kK \rightarrow T \\[0.1cm]
  \end{array}$\\[-2ex] 
  \end{center}
  \caption{$\ttS$-Types}
  \label{fig:S-types-def}
  \end{figure}


 A \textbf{sequence type $F=(k\cdot S_k)_{k\in K}$} is a sequence of types in the above meaning and is seen as an intersection of the types $S_k$.
 We write $\est$ for the \textit{empty sequence}, \ie the sequence type whose support is empty. A sequence type of the form $(k\cdot T)$ is called a \textit{singleton sequence type}.  Note that, in $\skSk$, $K$ may be an infinite subset of $\bbN\setminus\set{0,1}$ and that the definition allows infinitely many nestings of sequences in types.

\begin{remark}
The equality between two types (resp. sequence types) may be defined by mutual coinduction: 
$F\rew T=F'\rew T'$ if $F=F'$ and $T=T'$ and $(T_k)_{k\in K}=(T'_k)_{k \in K'}$ if $K=K'$ and for all $k\in K,~ T_k=T_k'$.
\end{remark}


The \textit{support of a type} (resp. \textit{a sequence type}), which is a tree of $\bbN^*$ (resp. a forest),  is defined by mutual coinduction: $\supp{\tv}=\set{\epsi},~ \supp{F\rew T}=\set{\epsi}\cup \supp{F}\cup 1\cdot \supp{T}$ and $\supp{(T_k)_{k\in K}}= \cup_{k\in K} k\cdot \supp{T_k}$.

Thus, $\ttS$-types may be seen as labelled trees. 
For instance, 
$(7\ct \tv_1,3\ct \tv_2,2\ct \tv_1)\rew \tv$ is represented in Fig.~\ref{fig:pars-tree-type}

Thus, for types, track 1 is dedicated to the codomains of arrows.
\begin{figure}[!h]
  \begin{center}
  \ovalbox{
    \begin{tikzpicture}
  \draw (2,0) node {\small $\rew$};
  \draw (2,0) circle (0.18);
  
  \draw (3,1)  node{\small $\tv$};
  \draw (3,1) circle (0.18);
  \draw (2.87,0.87) -- (2.13,0.13) ;
  \draw (2.7,0.55) node {\small \red{$1$} };
  
  \draw (1.3,1) node{\small $\tv_1$} ;
  \draw (1.3,1) circle (0.18);
  \draw (1.91,0.161) -- (1.38,0.85) ;
  \draw (1.8,0.55) node {\small \red{$2$} };
  
  \draw (0.6,1) node{\small $\tv_2$} ;
  \draw (0.6,1) circle (0.18);
  \draw (1.86,0.13) -- (0.68,0.84);
  \draw (1.4,0.55) node {\small \red{$3$} };
  
  \draw (-0.4,1) node{\small $\tv_1$} ;
  \draw (-0.4,1) circle (0.18);
  \draw (1.83,0.08) -- (-0.3,0.85);
  \draw (0.1,0.55) node {\small \red{$7$} };
  \end{tikzpicture}}
  \end{center}
  \caption{The $\ttS$-type $(7\ct \tv_1,3\ct \tv_2,2\ct \tv_1)\rew \tv$ as a tree}
  \label{fig:pars-tree-type}
  \end{figure}

A type of $\Types^{111}$ is in the set $\Types^{001}$ if its support does not hold an infinite branch ending by $1^{\omega}$. For instance, this excludes the $T$ defined by $T=\est \rew T$ (\ie $T=\est \rew \est\rew \ldots$). This restriction means that we may only have finite series of arrows in a type. Indeed, 001-normal forms, even though they may be infinite, contain only finite series of abstraction nodes (\eg $\lx.\lx\ldots$ is not legal).

We just write $\Types$ for $\Types^{001}$ and consider only sequence types which hold only types from $\Types$.

When a family of sequence types $(F^i)_{\iI}$ is disjoint, then there is no overlapping of typing information between the $F^i$, and $\uplus_\iI F^i$ is defined.


\subsection{Rigid Derivations}
\label{s:deriv-S}

A \textbf{$\ttS$-context} $C$ is a \textit{total} function from $\TermV$ to the set of (001-)sequence types, for instance, $x:\sSk$ is the context which assigns $\sSk$ to $x$ and $\est$ to every other variable. 
The \textbf{domain} of $C$ is $\dom{C}=\set{x\in \TermV\,|\, C(x)\neq \est}$. We define the join of contexts pointwise. If $\dom{C}\cap \dom{D}=\emptyset$, we may write $C;D$ instead of $C\uplus D$.
A \textbf{judgment} is a triple of the form $\juCtt$, where $C$ is a context, $t$ a 001-term and $T$ is a (001-)type. A \textbf{sequence judgment} is a sequence of judgments  $(C_k\vdash t:T_k)_{k\in K}$ (notice that these judgments have the same subject $t$). For instance, if $5 \in K$, then the judgment on track 5 is $\ju{C_5}{t:S_5}$.

The set of \textbf{$\ttS$-derivations} (metavariable $P$) is defined
\textit{coinductively} by the rules of Fig.~\ref{fig:rules-system-S}.

\begin{figure}[!h]
  \bcen
  \ovalbox{
   $ \begin{array}{c}
      \infer[\ax]{\phd}{\juxkT}  \sep\sep\sep
    \infer[\abs]{\ju{C;x:\sSk}{t:T}}{\ju{C}{\lx.t:\sSk \rew T}}\\[0.7cm]
    \infer[\app]{\ju{C}{t:\sSk\rew T} \sep (\ju{D_k}{u:S_k})_{\kK}  }
          {\ju{C\uplus (\uplus_{\kK} D_k)}{t\,u:T}} 
  \end{array}$}
  \ecen
  \caption{System $\ttS$}
  \label{fig:rules-system-S}
  \end{figure}  

In the axiom rule, $k$ is called an \textbf{axiom track}. 
In the $\app$-rule, the contexts must be disjoint, so that no track conflict occurs. Otherwise, $\app$-rule cannot be applied. 

Derivations from system $\ttS$ may be seen as labelled trees: we define the \textbf{support} of  $P\rhd C\vdash t:T$ coinductively: $\supp{P}=\set{\epsi}$ if $P$ is an axiom rule, $\supp{P}=\{\epsi\}\cup 0\ct \supp{P_0}$ if $t=\lambda x.t_0$ and $P_0$ is the subderivation typing $t_0$, $\supp{P}=\set{\epsi}\cup 1\ct \supp{P_1} \cup_{\kK} k\ct \supp{P_k}$ if $t=t_1\,t_2,~ P_1$ is the left subderivation typing $t_1$ and $P_k$ the subderivation typing $t_2$ on track $k$. The $P_k$ ($k\in K$) are called \textbf{argument derivations}. Graphically, an $\app$-rule may be represented as in the top of Fig.~\ref{fig:app-node-S-with-edges}:
\begin{figure}[!b]
\bcen \ovalbox{\begin{tikzpicture}

\node (left) at (-2,1.2)  {$\ju{C}{t:(2\ct S_2,3\ct S_3,5\ct S_5)\rew T}$};
\node (arg2) at (2.4,1.2) {$\ju{D_2}{u:S_2}$};
\node (arg3) at (4.9,1.2) {$\ju{D_3}{u:S_3}$} ;
\node (arg5) at (7.4,1.2) {$\ju{D_5}{u:S_5}$}; 
\node (ccl) at (1,0) {$\ju{C\uplus D_2\uplus D_3 \uplus D_5 }{t\,u:T}$};

\draw (ccl) -- (left) node[midway,left]{\red{1}};
\draw (ccl) -- (arg2) node[midway,left]{\red{2}};
\draw (ccl) - -(arg3) node[midway,left]{\red{3}};
\draw (ccl) -- (arg5) node[midway,left]{\red{5}}; 

\draw (1.8,-0.7)  node {\textbf{Application node as a labelled tree }};

\draw (1.8,-2) node {
$\infer{
\ju{C}{t:(2\ct S_2,3\ct S_3,5\ct S_5)\rew T}\sep\msep \ju{D_2}{u:S_2}~\trck{2}\sep 
\ju{D_3}{u:S_3}~\trck{3}\sep 
\ju{D_5}{u:S_5}~\trck{5}}{
 \ju{C\uplus D_2\uplus D_3 \uplus D_5 }{t\,u:T}}$};

\draw (1.8,-3) node {\textbf{Compact representation}};
  \end{tikzpicture}
}
\ecen
\caption{Representing $\app$-nodes in system $\ttS$}
\label{fig:app-node-S-with-edges}
\end{figure}
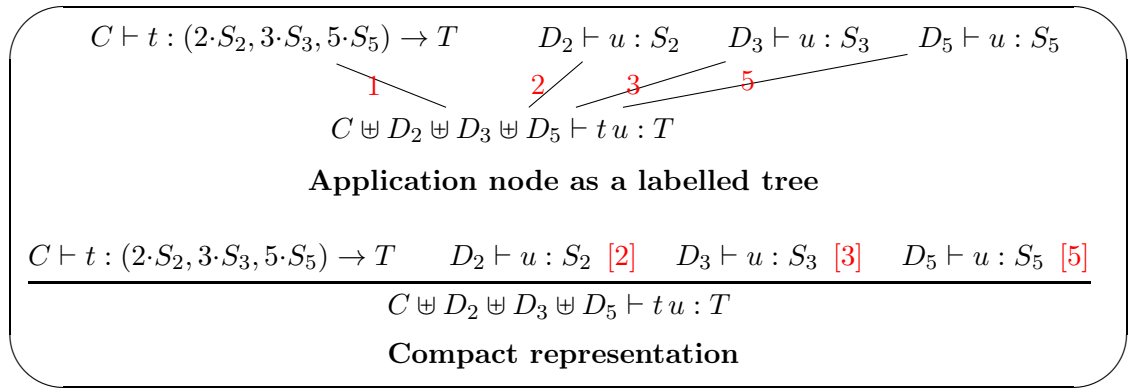
In this figure, argument judgments are on tracks 2, 4 and 5 (indicated in red) and the left-hand side judgment is on track 1 (compare with Fig.~\ref{fig:lam-terms-as-trees}). The same $\app$-rule has an alternative representation at the bottom of Fig.~\ref{fig:app-node-S-with-edges}, where \textit{argument} tracks are indicated between red square brackets inside the argument judgments, whereas track 1 is omitted.

\ignore{ 
We can define \textbf{isomorphisms of derivations}. It is formally done in
Appendix \ref{appIso}. Concretely, $P_1$ and $P_2$ are isomorphic, written
$P_1\equiv P_2$, if they type the same term, there is well-behaved labelled tree isomorphism between their support and use isomorphic types and contexts. In that case, we can define type isomorphisms that are compatible in some
sense with the typing rules in the two derivations $P_1$ and $P_2$.\\}

\begin{example}
\label{ex:S-deriv-1}
  In order to gain space, we do not write right-hand sides of axioms. We  set $\Sex=(8\ct o,3\ct o',2\ct o)\rew o'$. We indicate the track of argument derivations between brackets \eg $x:\!(2\ct \tv')\trck{3}$ means that judgment $\ju{x:\!(2\ct \tv')}{x\!:\!\tv'}$ is on track 3.
  $$\Pex=\infer[\abs]{\infer[\app]{\infer[\ax]{\phd}{x:\!\!(4\ct \Sex)  }\hspace{0.6cm}
      \infer[\ax]{\phd}{x:\!(9\ct \tv)}\trck{2} \hspace{0.5cm}  
\infer[\ax]{\phd}{x:\!(2\ct \tv')}\trck{3} \hspace{0.5cm}
\infer[\ax]{\phd}{x:\!(5\ct \tv)}\trck{8} 
    }{\ju{x:(2\ct \tv',4\ct (8\ct \tv,3\ct \tv',2\ct \tv)\rew \tv',5\ct \tv,9\ct \tv)}{xx:\tv'}}}{
  \ju{}{\lx.xx:(2\ct \tv',4\ct (8\ct \tv,3\ct \tv',2\ct \tv)\rew \tv',5\ct \tv,9\ct \tv)\rew \tv'}
              }$$

  In the $\ax$-rule concluding with $\ju{x:(5\ct o)}{x:o}$, the axiom track is 5. \end{example}

In  Example~\ref{ex:S-deriv-1} above,  $\supp{\Pex}=\set{\epsi,0,0\ct 1,0\ct 2,0\ct 3,0\ct 8}$ and we have $\Pex(0\ct 8)=\ju{x:(5\ct o)}{x:o}$. In particular, this judgment is on the \textbf{argument track 8} of the $\app$-rule at position 0.\\

\noindent \textbf{Collapse.} When we forget about tracks, a sequence naturally collapses on a multiset \eg $(3\ct a,5\ct b,8\ct a)$ collapses on $\mult{a,b,a}$. If this collapse is performed (coinductively)for all the sequences nested in derivations, then the derivations of $\ttS$ will collapse on derivations of $\scrR$. For instance, if we set $\sigex=\mult{o,o',o}\rew o'$, the derivation $\Pex$ of Example~\ref{ex:S-deriv-1} collapses on:
{$$\Piex=\infer{\infer{\infer{\phd}{x:\sigex  }\\
\infer{\phd}{x:\mult{o}} \\  
\infer{\phd}{x:\mult{o'}} \\
\infer{\phd}{x:\mult{o}} 
    }{\ju{x:\mult{o',\mult{o,o',o}\rew o',o,o}}{xx:o'}}}{
  \ju{}{\lx.xx:\mult{o',\mult{o,o',o}\rew o',o,o} \rew o'}
    }$$
}

\begin{remark}
  \label{rk:tracks-tedious}
The collapse turns out to be surjective: every $\ttS$-derivation is the collapse of a $\scrR$-derivation (it is non-trivial). Tracks can be interpreted as a very low-level specification for derivations, but as we said, they are necessary to express approximability and to prove the main interesting properties of system $\ttS$. However, while working with derivations and understanding some of the main intuitions of this article, it is still easier to just use multiset intersection and put tracks under the carpet.
\end{remark}

\noindent \textbf{Useful notations.} 
For the proofs of \Sec~\ref{s:normal-forms},  the following notations will be needed. The reader who does not want to read the proofs can skip them.
\begin{itemize}
  \item \textit{For applications:}
assume that $P$ types $t$, $a\in \supp{P}$ and $a=a_*\ct 1^{n}$. Then $\tras$ is of the form $\tra\,t_1\ldots t_n$
\begin{itemize}
\item The set $\ArgTr^i_P(a)$ (for $1\leqslant i \leqslant n$) contains the tracks of argument derivations typing the $i$-th argument $t_i$ below $a$.
\item The set $\ArgPos^i_P(a)$ contains the positions of those subderivations.
\end{itemize}
Formally, we set
 $1\leqslant i \leqslant n$, $\ArgTr^i_P(a)=\set{k\geqslant 2\,|\, a_*\ct 1^{n-i}\ct  k\in \supp{P}}$ and $\ArgPos^i_P(a)=a_*\ct 1^{n-i}\ct \ArgTr^i(a)=\set{a_*\ct 1^{n-i}\ct k\in \supp{P}\,|\, k\geqslant 2}$.  When $i$ is omitted, $i=1$ is assumed, \ie $\ArgTr_P(a)=\ArgTr^1_P(a)$. 
\item \textit{For variables and binding:} If $a\in A:=\supp{P}$ and $x\in \TermV$, we set $\AxP_a(x)=\set{a_0 \in A\,|\,a\leqslant a_0,\,t(a)=x,\nexists a_0',\,a\leqslant a_0'\leqslant a_0,\,t(a'_0)=\lx }$. Thus, $\AxP_a(x)$ is the set of positions of $\ax$-rules in $P$ above $a$ typing occurrences of $x$ that are not bound  at $a$.
  If $a_0\in A$ is an axiom, we write $\trP{a_0}$ for its associated axiom track.
\item We write $\Ax^P$ for the set of all axiom positions in $\supp{P}$.
\end{itemize}
Usually, $P$ is implicit and
we write only $\ArgTr^i(a)$, $\ArgPos^i(a)$, $\Axa(x)$ and $\tr{a}$.

For instance, in $\Pex$ (Example~\ref{ex:S-deriv-1}), $\ArgTr^1(0\ct 1)=\set{2,3,8}$ and $\ArgPos^1(0\ct 1)=\set{0\ct 2,0\ct 3,0\ct 8}$ (with $a_0=0$). Moreover, $\Ax_0(x)=\set{01,02,03,08}$, $\Ax_\epsi(x)=\eset$ and $\tr{01}=4$, $\trP{08}=5$, whereas  $\trP{0}$ is not defined. 


\section{Statics and Dynamics}
\label{s:dynamics}

In this section, we present \textit{bipositions}, which allow pointing inside $\ttS$-derivation (\Sec~\ref{ss:bisupp}). We use this notion to present and prove the subject reduction and the subject expansion properties (\Sec~\ref{ss:one-step-sr-se}). Proof expansion is actually defined \textit{uniformly}. We notice that proof reduction is deterministic in system $\ttS$. This leads us to define a notion of \textit{residuation} and to formulate a first argument showing that approximability is indeed \textit{not} definable while working with multiset intersection (\Sec~\ref{ss:non-determinism}).

\subsection{Bipositions and Bisupport}
\label{ss:bisupp}

In a rigid setting as system $\ttS$, we can identify and point to every part of a derivation, thus allowing to formulate many useful notions.

If $a\in \supp{P}$, then $a$ points to a judgment inside $P$ typing $t\rstr{\ovl{a}}$. We write this judgment $\ttC(a)\vdash t\rstr{\ovl{a}}:\ttT(a)$: we say $a$ is an \textbf{outer position} of $P$ ($\ovl{a}$ is defined on p.~\pageref{txt:def-collapse}). The context $\ttC(a)$ and the type $\ttT(a)$ should be written $\ttC^P(a)$ and $\ttT^P(a)$ but we often omit $P$.  From now on, we shall also write $\tra$ and $t(a)$ instead of $t\rstr{\ovl{a}}$ and $t(\ovl{a})$.

 In Example~\ref{ex:S-deriv-1}, $\Pex(01)=\ju{x:(4\cdot S)}{x:S}$, so $\ttC(01)=x:(4\cdot S)$ \ie $\ttC(01)(x)=(4\cdot S)$.  Since $S=(8\cdot o,3\cdot o',2\cdot o)\rew o'$, we have $\ttC(01)(x)(4)=\rew$, $\ttC(01)(x)(43)=o',\ \ttT(01)(\epsi)=\rew$, $\ttT(01)(1)=o'$. Likewise, $\Pex(03)=\ju{x:(2\cdot \tv')}{\tv'}$, so that $\ttC(03)=x:(2\cdot \tv')$ and $\ttT(03)=\tv'$. Thus, $\ttC(03)(x)(2)=\tv'$ and $\ttT(03)(\epsi)=o'$. We also have $\ttC(0)(x)=(2\cdot \tv',4\cdot (8\cdot \tv,3\cdot \tv',2\cdot \tv)\rew \tv',5\cdot \tv,9\cdot \tv)$, so that $\ttC(0)(x)(2)=\tv'$ and $\ttC(0)(x)(42)=\tv$. 

This  motivates the notion of \textbf{bipositions}: a biposition (metavariable $\p$) is a pointer into a type nested in a judgment of a derivation. 

\begin{definition}[Bisupport]\mbox{}
  \label{def:bisupp-S}
  \begin{itemize}
  \item A pair $(a,c)$ is a \textbf{right biposition} of $P$ if $a\in \supp{P}$ and $c\in \supp{\ttT^P(a)}$.
  \item A triple $(a,x,k\cdot c)$ is a \textbf{left biposition} if $a\in \supp{P},\, x\in \scrV$ and $k\cdot c\in \supp{\ttC^P(a)(x)}$.
\item The \textbf{bisupport} of a derivation $P$, written $\bisupp{P}$, is the set of its (right or left) bipositions.
  \end{itemize}
\end{definition}

We consider a derivation as a \textit{function} from its bisupport to the set $\TypeV\cup \set{\rew}$ and write now $P(a,c)$ for $\ttT^P(a)(c)$ and $P(a,x,k\cdot c)$ for $\ttC^P(a)(x)(k\cdot c)$. For instance, $(01,\epsi)$, $(01,1)$, $(03,\epsi)$ are right bipositions in $\Pex$ and $\Pex(01,\epsi)=\rew$, $\Pex(01,1)=\tv'$, $\Pex(03,\epsi)=\tv'$. Moreover, $(01,x,4)$, $(01,x,43)$ and $(0,x,42)$ are left bipositions of $\Pex$ and $\Pex(01,x,4)=\rew$, $\Pex(01,x,43)=\tv'$  and $\Pex(0,x,42)=\tv$.


\subsection{Quantitativity and Coinduction}
\label{ss:quant-coind-S}

\begin{figure*}
\begin{center}
\ovalbox{  \begin{tikzpicture}
  \draw (4,0) node{$\phd$};


  \draw (-0.7,1.03) --++ (0,0.38);
  \draw (-2.3,1.6) node [right] {\small $\ju{C;x:\sSk\!}{\!t\!:\!T} $};
  
  \draw (-0.8,1.25) node [right]{\small $\red{0}$};
  \draw (-2.6,0.9) node [right] {\small $\ju{C\!}{\!\lx.r\!:\!(S_k)_{\kK}\!\rew\! T}$};

  \draw (0.22,0.13)--(-0.57,0.67) ;
  \draw (-0.5,0.3) node [right] {\small $\red{1}$};
  \draw (-2.3,-0.1) node [right] {\small $\ju{C\uplus  D_3\uplus D_5\uplus D_8}{(\lx.r)s}:T \posPr{a}$};
  
  \draw (-0.7,5.2) node{\Large $P_r$};
  \draw (-0.7,1.8) --++(2.8,4) --++ (-5.6,0) --++ (2.8,-4) ;
  \draw (1.3,5.3) node {\ssz $\heartsuit$};  
  \red{
    \draw (2.3,4.9) node {$\posPr{a\ct 10\ct \al_{\scriptscriptstyle \heartsuit}}$};
    \draw [->,>=stealth] (1.8,5.1) --++(-0.35,0.2) ;
  }
  


  \draw (-2.5,4.6) node [right] {\small $\ovl{x:(2\ct S_2)}$ }; 
  \draw [dotted] (-1.9,4.4) --++(0,-0.21);
  \red{
    \draw (-2.8,3.5) node {$ \posPr{a\ct10\ct a_2}$};
    \draw [>=stealth,->](-2.4,3.7) --++ (0.1,0.7);
  }

  \draw (-1.65,3.1) node [right] {\small $\ovl{x:(7\ct S_7)}$};
  \draw [dotted] (-1,2.9) --++(0,-0.21);
  \red{
    \draw (0.9,2.7) node {$ \posPr{a\ct10\ct a_7}$};
    \draw [>=stealth,->](0.2,2.8) --++ (-0.3,0.1);
  }

  \draw (-0.6,4.3) node [right] {\small $\ovl{x:(3\ct S_3)}$};
  \draw [dotted] (0.3,4.1) --++(0,-0.21);
\red{
    \draw (1.7,3.8) node {$ \posPr{a\ct10\ct a_3}$};
    \draw [>=stealth,->] (1,3.9) --++ (-0.3,0.2);
  }
  

  \draw (1.5,1.1) --++ (0.5,0.8) --++(-1,0) --++ (0.5,-0.8) ;
  \draw (1.5,1.5) node{\small $P_2$};
  \draw (0.9,0.9) node [right] {\small $\ju{D_2\!}{\!s\!:\!S_2}$};
  \draw (0.4,0.13) -- (1.5,0.7);
  \draw (0.8,0.3) node [right] {\small $\red{2}$};

  \draw (3,1.1) --++ (0.5,0.8) --++(-1,0) --++ (0.5,-0.8) ;
  \draw (3,1.5) node{\small $P_3$};
  \draw (2.4,0.9) node [right] {\small $\ju{D_3\!}{\!s\!:\!S_3}$};
  \draw (1.1,0.13) -- (2.8,0.7);
  \draw (1.9,0.3) node [right] {\small $\red{3}$};

  \draw (4.6,1.1) --++ (0.5,0.8) --++(-1,0) --++ (0.5,-0.8) ;
  \draw (4.6,1.5) node{\small $P_7$};
  \draw (4,0.9) node [right] {\small $\ju{D_7\!}{\!s\!:\!S_7}$};
  \draw (2,0.13) -- (4.4,0.7);
  \draw (3.2,0.3) node [right] {\small $\red{7}$};

  \draw (4.36,1.77) node {\ssz $\clubsuit$}; 
  \draw (4.2,2.4) node {$\posPr{a\ct 7\ct \al_{\scriptscriptstyle \clubsuit}}$};
\red{  \draw [->,>=stealth] (4.2,2.2) --++ (0.15,-0.3) ;}

\draw (-1.8,-0.8) node [right] {\textbf{Subderivation typing the redex}};


\trans{-1.7}{0.4}{  

  \draw (9.8,3.6) node{\Large $P_r$};
  \draw (9.8,0.3) --++(2.8,4) --++ (-5.6,0) --++ (2.8,-4) ;
  \draw (11.8,3.8) node {\ssz $\heartsuit$};
    \red{
    \draw (12.3,3) node {$\posPr{a\ct \al_{\scriptscriptstyle \heartsuit}}$};
    \draw [->,>=stealth] (12,3.3) --++(-0.1,0.4) ;
  }

  \draw (7.2,-0.1) node [right] {\small $\ju{C\uplus  D_3\uplus D_5\uplus D_8}{r[s/x]}:T \posPr{a}$};

  \draw (8.6, 3.2) --++ (0.5,0.8) --++(-1,0) --++ (0.5,-0.8) ;
  \draw (8.6,3.6) node{\small $P_2$};
  \draw (8,3.1) node [right] {\small $\ju{D_2\!}{\!s\!:\!S_2}$ }; 
  \draw [dotted] (8.6,2.8) --++(0,-0.21);
  \red{
    \draw (7.7,2.2) node {$ \posPr{a\ct a_2}$};
    \draw [>=stealth,->](8.1,2.4) --++ (0.1,0.5);
  }

  \draw (9.5,1.8) --++ (0.5,0.8) --++(-1,0) --++ (0.5,-0.8) ;
  \draw (9.5,2.2) node {\small $P_7$};
  \draw (8.85,1.6) node [right] {\small $\ju{D_7\!}{s\!:\!S_7}$};
  \draw [dotted] (9.5,1.4) --++(0,-0.21);
  \draw (9.26,2.47) node {\ssz $\clubsuit$};
   \draw (8,1.4) node {$\posPr{a\ct a_7\ct \al_{\scriptscriptstyle \clubsuit}}$};
 \red{  \draw [->,>=stealth] (8.1,1.6) --++ (1.07,0.87) ;}

\red{
    \draw (11.4,1.2) node {$ \posPr{a \ct a_7}$};
    \draw [>=stealth,->](10.7,1.3) --++ (-0.3,0.1);
}
  
  \draw (10.8,3) --++ (0.5,0.8) --++(-1,0) --++ (0.5,-0.8) ;
  \draw (10.8,3.4) node {\small $P_3$}; 
  \draw (9.9,2.8) node [right] {\small $\ju{D_3\!}{s\!:\! S_3}$};
  \draw [dotted] (10.8,2.6) --++(0,-0.21);
\red{
    \draw (12.2,2.3) node {$ \posPr{a \ct a_3}$};
    \draw [>=stealth,->] (11.5,2.4) --++ (-0.3,0.2);
  }      
  
 \draw (7.5,-0.8) node [right] {\slbox{4.2}{\bcen \textbf{Subderivation typing\\ the reduct}\ecen}};

 }
\end{tikzpicture}}
  \caption{Subject Reduction and Residuals}
  \label{fig:SR-residuals}
\end{center}
\end{figure*}
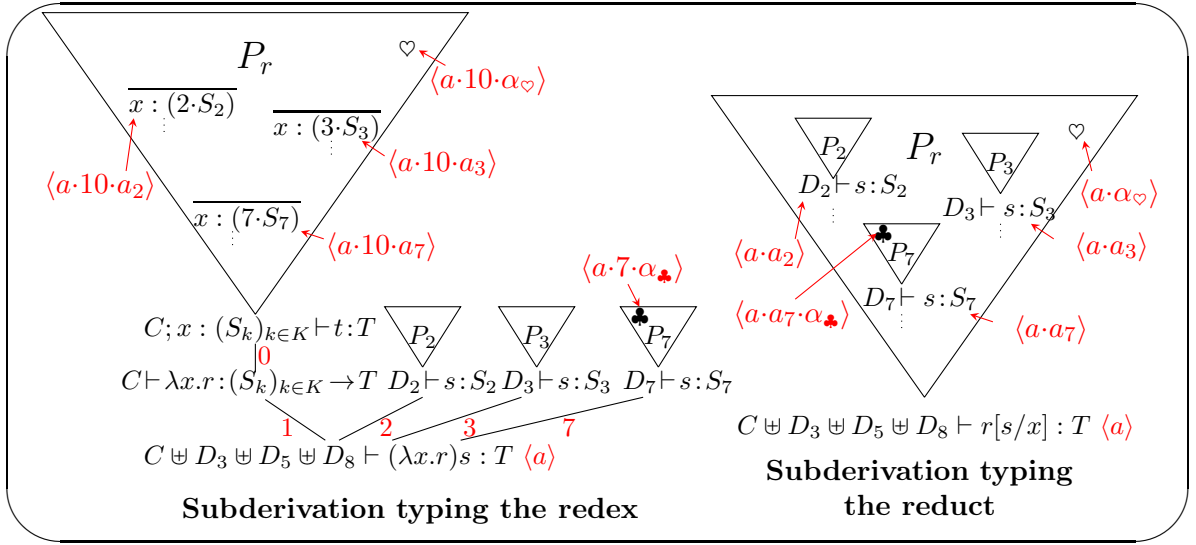

The presence of an infinite branch inside a derivation makes it possible that a type in a context is \textit{not} ``created'' in an axiom rule. For instance, we set, for all $k\geqslant 2,~ j_k=\ju{f:(i\cdot \twoarewa)_{i\geqslant k},\,x:(8\cdot\tv')}{\fom:o}$ and we coinductively define a family $(P_k)_{k\geqslant 2}$ of $\ttS$-derivations by 
$$P_k=\infer[\app]{\infer[\ax]{\phd}{\ju{f:k\cdot\twoarewa}{f:\twoarewa}} \sep
   P_{k+1}\tri j_{k+1}~\trck{2}
  }{\ju{f:(i\cdot \twoarewa)_{i\geqslant k},\,x:(8\cdot\tv')}{\fom:o}}  $$
We observe that the $P_k$ are indeed correct derivations of $\ttS$. 
Every type assignent to $f$ can be traced back to an $\ax$-rule.
However, notice that $x$ has been assigned a type (on track 8) although $x$ does not appear in the typed term $\fom$ and the part of the context assigned to $x$ cannot be traced back to any axiom rule typing $x$ with $o'$ (using axiom track $8$): we say that there is a \textbf{phantom axiom track} in the $P_k$. Thus, Remark~\ref{rk:relevance-computation-contexts-Ro} is not valid for the derivations $P_k$. 
This motivates the notion of \textit{quantitative} derivation, in which this does not happen:

\begin{definition}
  \label{def:quant-S-deriv}
   A derivation $P$ is \textbf{quantitative} when, for all $a\in \supp{P}$ and $x\in
\TermV$, $\ttC^P(a)(x)=\uplus_{a'\in \AxP_a(x)} (\trP{a'}\cdot \ttT^P(a'))$.
\end{definition}

This definition specifies that $P$ is quantitative when every context in $P$ can be computed from the $\ax$-rules. Notice that $P$ is always quantitative whenever the derivation is finite as observed in Remark~\ref{rk:relevance-computation-contexts-Ro}, or its subject is finite.

Now, assume $P$ is quantitative. 
Then $\Rt(\ttC(a)(x))=\set{\tr{a_0}\,|\, a_0\in \Axa(x)}$ and for all $a\in A,~x\in \TermV$ and $k\in \Rt(\ttC(a)(x))$, we write $\pos{a,\,x,\, k}$ for the unique position $a'\in \Axa(x)$ such that $\tr{a'}=k$. The reader who does not want to read the proofs can forget about this notation. 
  Actually, 
$\pos{a,\,x,\,k}$ can be defined by a downward induction on $a$ as follows:
\begin{itemize} 
\item If $a\in \Ax^P$, then actually $a\in \Ax(x)$ and $\tr{a}=k$ and we set 
  $\pos{a,\,x,\,k}=a$.
\item If $a\cdot 1\in A$, we set $\pos{a,\,x,\,k}=\pos{a\cdot \ell,\,x,\,k}$, where $\ell$ is the unique positive integer such that $k\in \Rt(\ttC(a\cdot \ell)(x))$. 
\item If $a\cdot 0\in A$, we set $\pos{a,\, x,\,k}=\pos{a\cdot 0,\,x, \, k}$ 
\end{itemize}

\subsection{One Step Subject Reduction and Expansion}
\label{ss:one-step-sr-se}

System $\ttS$ enjoys both subject reduction and expansion: if $t\rew^* t'$, then $\tri \juCtt$ iff $\tri \juCtpt$.

\begin{proposition}[One Step Subject Reduction]
\label{prop:subj-red-one-step-S}
  Assume $t\rew_{\beta}t'$ and $P\tri \juCtt$. Then there exists a derivation $P'$ such that $P'\tri \ju{C}{t':T}$.
\end{proposition}

\begin{proposition}[One Step Subject Expansion]
  \label{prop:subj-exp-one-step-S}
  Assume $t\rew_{\beta}t'$ and $P'\tri \ju{C}{t':T}$.
  Then there exists a derivation $P$ such that $P\tri \juCtt$.
\end{proposition}

These proposition are valid in general, whether the derivations are quantitative or not. However, in the following, we shall consider only the quantitative case (this allows us not handling phantom axiom tracks). They may also be proven really simply by using coinduction. However, in the forthcoming sections, we shall need to deal with a validity criterion, whose soundness \wrt reduction and expansion (Lemma~\ref{lem:approx-red}) can only be dealt with by having a suitable notion of residuation. More precisely, we shall need to have a good understanding of the action of reduction on bipositions.
We thus give an alternative proof of the subject reduction property using \textit{residuals}  
and defining directly the derivation reduct $P'$ in \Sec~~\ref{ss:sr-proof}. Details can be found in Appendix~\ref{a:equinecessity}. We also explain here why subject reduction may be seen as a \textit{deterministic} process in system $\ttS$ and why it is vital to ensure soundness (\Sec~\ref{ss:non-determinism}). Subject expansion is \textit{not} deterministic, but it may be processed \textit{uniformly}.

All this is illustrated by Fig.~\ref{fig:SR-residuals}: we assume that $\trb=\lxrs$ and $t \rewb{b} t'$, $P$ is a quantitative derivation concluding with $\juCttsh$. We also assume that $a\in P$ is such that $\ovl{a}=b$ (thus, $a$ is the position of a judgment typing the redex to be fired) and that there are exactly 3 $\ax$-rules typing $x$ above $a$, using axiom tracks 2, 3 and 7. Notice that $\ax$-rule typing $x$ on track 7 must be above $a\cdot 10$, so that its position is of the form $a\cdot 10\cdot a_7$. Likewise for the two other axioms.

Now, let us have a look at how reduction is performed inside $P$. We omit $\ax$-rules right-hand sides. We also indicate the position of a judgment between angle brackets, \eg $\posPr{a\cdot 10 \cdot a_3}$ means that judgment $\ju{x:(3\cdot S_3)}{x:S_3}$ is at position $a\cdot 10\cdot a_3$.


Notice how this transformation is \textit{deterministic}: for instance, assume $7 \in K$. There must be an axiom rule typing $x$ using axiom track 7, \eg $x:\, (7 \cdot S_7)\vdash x:\,S_7$ at position $a\cdot 10\cdot a_7$ and also a subderivation at argument track 7, namely, $P_7$ concluded by $s:S_7$
at position $a\cdot 7$. Then, when we fire the redex at position $b$, the subderivation $P_7$ \textit{must} replace the axiom rule on track 7, even if there may be other $k\neq 7$ such that $S_k=S_7$, in contrast to system $\scrR$ (see \Sec~\ref{ss:non-determinism}).

Thus, Proposition~\ref{prop:subj-red-one-step-S} canonically gives only one derivation $P'$ typing $t'$, so that we may also write $P\rewb{b} P'$.

Now, we observe that subject expansion cannot be deterministic in the same sense. When we pass from a derivation typing $\rsx$ to a derivation typing $\lxrs$, we create new axiom rules which type $x$.  
Those axiom rules must be assigned axiom tracks. 
But if for instance, 3 axiom rules are created above position $a$, there is no more reason to choose tracks 2, 3 and 7 than the tracks 8, 4, 38: the axiom tracks may be chosen arbitrarily, as long as they do not raise track conflicts. This explains the non-determinism of subject expansion.

However, as it will turn out in \Sec~\ref{s:infty-expansion}, we will need to expand simultaneously \textit{families} of derivations --- and this, infinitely many times. 
For that, we should find a way to perform subject expansion uniformly. Let then $\code{\cdot}$ be any \textit{injection} from $\bbN^*$ to $\bbN\setminus\set{0,1}$. 
We write $\Exp_b(P',\code{\cdot},t)$ for the \textit{unique} expansion of $P'$ such that $P \rewb{b} P'$ and, for all $a_0 \in \bbN^*$, if there is an $\ax$-rule typing $x$ created at position $a_0$, then the axiom track that has been assigned is $\code{a_0}$. 
Since $\code{\cdot}$ is injective, no track conflict may occur. We give all the details in Appendix~\ref{ss:uni-subj-exp}. The term $t$ must be indicated in the expression, because, for one $t'$ and one $b$, there may be several $t$ such that $t\rewb{b} t'$.

\begin{remark}
  \label{rmk:inf-quant-and-red}
If $P\breda{b}P'$ then $P$ is quantitative iff $P'$ is quantitative.
\end{remark}

Remark~\ref{rmk:inf-quant-and-red} can be understood by looking at Fig.~\ref{fig:SR-residuals}: indeed, if $P$ is not quantitative, there must be a phantom axiom track (\Sec~\ref{ss:quant-coind-S}) on an infinite branch. But whether this phantom axiom track visits $r$, $s$ or is outside $\lxrs$, the contracted redex, it will be found in $P'$. 
Conversely, if $P'$ is not quantitative, there is a phantom axiom rule on an infinite branch, which visits an occurrence of $s$, or visits $r$ but no substituted occurrence of $s$, or is outside $\rsx$. In all cases, this phantom axiom track will be found in $P$.\\






\noindent \textbf{A First Look at Residuals.} 
Deterministic subject reduction naturally allows defining the residuals of positions and \textit{right} bipositions after reduction, extending the notion of residuals for positions in $\lambda$-terms. 
We define them more precisely in \Sec~\ref{ss:sr-proof}, to sketch the formal details of the proof of the subject reduction property.

In Fig.~\ref{fig:SR-residuals}, $\hearts$ represents a judgment nested in $P_r$. Thus, its position must be of the form $a\cdot 10\cdot \al_{\sheart}$. After reduction, the $\app$-rule and $\abs$-rule at positions $a$ and $a\cdot 0$ have been destroyed and the position of this judgment $\hearts$ will be $a\cdot \al_{\sheart}$. We set then $\Res_b(a\cdot 10\cdot \al_{\sheart})=a\cdot \al_{\sheart}$.

Likewise, $\clubs$ represents a judgment nested in the argument derivation $P_7$ on track 7 above $a$
. Thus, its position must be of the form $a\cdot 7\cdot \al_{\sclub}$ where $a\cdot 7$ is the root of $P_7$. After reduction, $P_7$ will replace the $\ax$-rule typing $x$ on track 7, so its root will be at $a\cdot a_7$ (by definition of $a_7)$. Thus, after reduction, the position of judgment $\clubs$ will be $a\cdot a_7\cdot \al_{\sclub}$. We set then $\Res_b(a\cdot 7\cdot \al_{\sclub})=a\cdot a_7\cdot \al_{\sclub}$.

We can thus define the \textbf{residuals} of most positions $\al\in \supp{P}$, but not all, \eg $a\cdot 1$, that corresponds to the abstraction of the redex, is destroyed during reduction and does not have a residual. For \textit{right} bipositions, when $(a,c)\in \bisupp{P}$ and $a'=\Res_b(a)$ is defined, we set $\Res_b(a,c)=(a',c)$.

Note that defining residuals in system $\scrR$ would be impossible: system $\scrR$ lacks pointers and is not deterministic (\Sec~\ref{ss:non-determinism} and Fig.~\ref{fig:non-deter-srRo}).

\subsection{Safe Truncations of Typing Derivations and Determinism of Reduction}
\label{ss:non-determinism}



Now that we have described the deterministic reduction dynamics of system $\ttS$, we can explain why the non-deterministic system $\scrR$, based on multiset intersection, is unfit to express the notion of \textbf{approximability}, informally introduced at the end of \Sec~\ref{ss:typ-inf-nf-informal}.

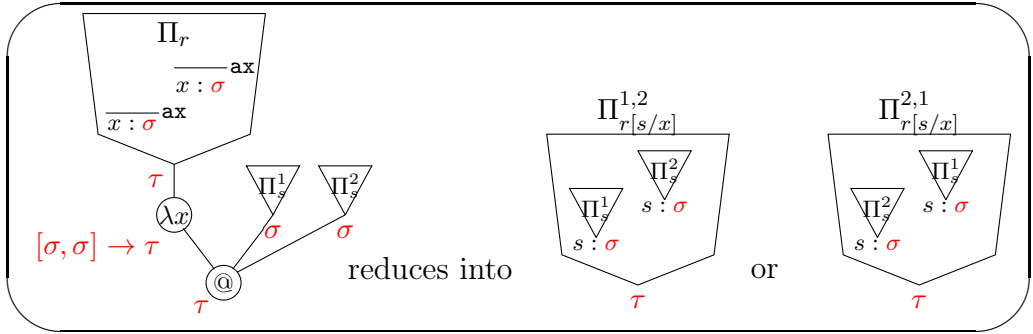
\begin{figure}
  \begin{center}
    \ovalbox{
  \begin{tikzpicture}

    {

 \trans{1.2}{0}{     
    \draw (-0.65,3.3) node {$\Pi_r$};
    \draw (-0.65,1.57) --++ (1,0.4) --++ (0.2,1.6) --++ (-2.4,0) --++ (0.2,-1.6) -- cycle;
    \draw (-0.65,1.57) node [below left] {$\red{\tau}$};
    \draw (-1.55,2.25) --++ (0.7,0);
    \draw (-0.65,2.25) node  {\fnsz $\ax$};
    \draw (-1.2,2.1) node  {\fnsz $x:\red{\sig}$};
    \draw (-0.65,2.85) --++ (0.7,0);
    \draw (0.25,2.85) node  {\fnsz $\ax$};
    \draw (-0.3,2.6) node {\fnsz $x:\red{\sig}$};
    
    \blockunary{-0.65}{0.9}{$\lx$}
    \draw (-0.7,0.8) node [below left] {$\red{\mult{\sig,\sig}\rew \tau}$};
    \blocka{0}{0}
    \draw (-0.05,-0.1) node [below left] {$\red{\tau}$};

\drawsmalltriin{0.65}{0.9}{$\Pi^1_s$}
\drawlefttail{0.65}{0.9}
\draw (0.65,0.9) node [below] {$\red{\sig}$};

\drawsmalltriin{1.6}{0.9}{$\Pi^2_s$}
\draw (0.18,0.14) -- (1.6,0.9);
\draw (1.6,0.9) node [below] {$\red{\sig}$};
 }
 }
\draw (2.7,0.2) node [right] {\large reduces into};

\trans{7}{3.4}{ 


\trans{0.3}{-5}{    
    \draw (-0.65,3.85) node {$\Pi^{1,2}_{\rsx}$};
    \draw (-0.65,1.57) --++ (1,0.4) --++ (0.2,1.6) --++ (-2.4,0) --++ (0.2,-1.6) -- cycle;
    \draw (-0.65,1.57) node [below] {$\red{\tau}$};
    \drawsmalltriin{-1.2}{2.2}{$\Pi^1_s$};
    \draw (-1.2,2.1) node  {\fnsz $s:\red{\sig}$};
    \draw (-0.3,2.6) node {\fnsz $s:\red{\sig}$};
    \drawsmalltriin{-0.3}{2.7}{$\Pi^2_s$}; 
}

\draw (1,-3.25) node [right] {\large or};

\trans{4}{-5}{    
    \draw (-0.65,3.85) node {$\Pi^{2,1}_{\rsx}$};
    \draw (-0.65,1.57) --++ (1,0.4) --++ (0.2,1.6) --++ (-2.4,0) --++ (0.2,-1.6) -- cycle;
    \draw (-0.65,1.57) node [below] {$\red{\tau}$};
    \drawsmalltriin{-1.2}{2.2}{$\Pi^2_s$};
    \draw (-1.2,2.1) node  {\fnsz $s:\red{\sig}$};
    \draw (-0.3,2.6) node {\fnsz $s:\red{\sig}$};
    \drawsmalltriin{-0.3}{2.7}{$\Pi^1_s$}; 
}

} 
  \end{tikzpicture}  }
  \caption{Non-Determinism of Subject Reduction (Multiset Intersection)}
  \label{fig:non-deter-srRo}
  \end{center}
\end{figure}

Let us consider a redex $t=(\lambda x.r)s$ and its reduct $t'=r[s/x]$. If a $\scrR$-derivation $\Pi$ types $t$, then $\Pi$ has one subderivation $\Pi_r$ of the form $\Pi_r\tri \ju{\Gam,\,x:\msigi}{r:\tau}$.

\ignore{
Also, for each $\iI$, $s$ has been given the type $\sigi$ inside some subderivation $\Pi_i$.  We can obtain a
derivation $\Pi'$ typing the term $t'$ by replacing the axiom rule
concluding with $x:\ju{\mult{\sigi}}{x:\sigi}$ with the derivation
$\Pi_i$.  
We then call $\Pi'$ a \textbf{derivation reduct}.

If a type $\sig$ occurs several times in $\msigi$---say $n$ times --- there must be $n$ axiom leaves in $\Pi$ typing $x$ with type $\sig$, but also $n$ argument derivations $\Pi_i$ concluding with $s:\,\sig$.  When an axiom rule typing $x$ and an argument
derivation $\Pi_i$ are concluded with the same type $\sigma$, we shall
informally say that we can \textbf{associate} them. This means that this
axiom rule can be substituted with that argument derivation $\Pi_i$ when
we reduce $t$ to produce a reduct derivation $\Pi'$ typing $t'$. There is not
only one way to associate the $\Pi_i$ to the axiom leaves typing $x$
(there can be as many as $n!$) and possibly many different derivation reducts.}

 Notice that proof reduction is non-deterministic in system $\scrRo$: let us consider a redex $t=\lxrs$ and its reduct $\rsx$. Assume that $t$ is typed so that $s$ is typed exactly \textit{twice} with the same type. For instance, say that $\Pi$ is a derivation typing $t$ which has a subderivation $\Pi_r\tri \ju{\Gam,x:\mult{\sig,\sig}}{r:\tau}$ (thus $\lx.r:\mult{\sig,\sig}\rew \tau$) and two subderivations $\Pi^1_s$ and $\Pi^2_s$ concluding with $\ju{\Del_i}{s:\sig}$ (one may assume $\Pi^1_s\neq \Pi^2_s$). This situation is represented in Fig.~\ref{fig:non-deter-srRo} and in particular, there are two different axiom rules typing $x$ with $\sig$ in $r$. Then, to obtain a \textit{derivation reduct} typing $t'=\rsx$ with $\tau$, there are two possibilities depending on which argument derivation $\Pi^i_s$ replaces which axiom rule typing $x$.

This makes a sharp difference with system $\ttS$: assume that $S_2=S_3=S_7=S$ so that argument derivation $P_2,\, P_3,\, P_7$ type $s$ with the same type $S$. Then  each $P_k$ ($k\in K$) will replace axiom rule at position $a\cdot 10\cdot a_k$ (see Fig.~\ref{fig:SR-residuals}) without other choice. In system $\ttS$, there is a unique (canonical) derivation reduct.

In contrast, let us dig into the following independent situations in system $\scrR$:
\begin{itemize}
\item Assume $\Pi^1_s$ and $\Pi^2_s$ (typing $s$), both concluding with the
  same type $\sig= \sig_1=\sig_2$, as in Fig.~\ref{fig:non-deter-srRo}.  Thus, we also have two axiom leaves \#1 and \#2 
  concluded by $\ju{x:\msig}{ x: \sig}$. In that case, during reduction, the axiom rules \#1 may be replaced with $\Pi^1_s$ (in that case, axiom \#2 is replace with $\Pi^2_s$) or with $\Pi^2_s$ (in that case, axiom \#2 is replaced with $\Pi^1_s$), because the types of axioms \#1 and \#1 match with those of the argument derivations $\Pi^1_s$ and $\Pi^2_s$.   
  When we truncate $\Pi$ into a
  finite $\supf \Pi$, the subderivations $\Pi^1_s$ and $\Pi^2_s$ are also
  cut into two derivations $\supf \Pi^1_s$ and $\supf \Pi^2_s$. In each
  $\supf \Pi^i_s$, $\sig$ can be cut into a type $\supf \sigi$.  When 
  $\Pi^1_s$ and $\Pi^2_s$ are different, it is possible that $\supf \sig_2 \neq \supf \sig_1$ for every finite truncation 
   of $\Pi$. Thus, it is possible that, for every truncation
  $\supf \Pi$, the axiom leaf \#1 does not match $\supf \Pi_2$:  indeed, a match between axiom rule and argument derivation that is possible in $\Pi$ could be impossible
  for any of its truncations. 
\item Assume this time $\sig_1\neq \sig_2$. When we truncate
$\Pi$ into a finite $\supf \Pi$, both $\sig_1$ and $\sig_2$ can 
be truncated into the same finite type $\supf \sigma$. In that case, we can associate $\supf \Pi_1$ with axiom \#2 and $\supf \Pi_2$ with axiom \#1
inside $\supf \Pi$ in $\supf \Pi$ (which is impossible in $\Pi$), thus producing a reduct derivation $\supf \Pi'$ typing 
$t'$, which has no meaning \wrt $\Pi$  ($\supf \Pi'$ would not be a truncation of any derivation reduct of $\Pi$).
\end{itemize}
That is why we need the \textit{deterministic} association between the argument derivations and the axiom rules typing each in system $\ttS$ (thanks to tracks),
so that the associations between them are preserved even when we truncate derivations. Ses not allow to formulate a well-fit notion of 
approximability for derivations that would be stable under 
(anti)reduction and hereditary for subterms. Whereas determinism ensures that proof reduction is monotonic in system $\ttS$, it is not monotonic in system $\scrR$.

Another, more formal argument proving that approximability cannot be defined in system $\scrR$ is presented in Appendix~\ref{a:rep-and-dynamics}:
 there exists a $\scrR$-derivation $\Pi$ and two $\ttS$-derivations $P_1$ and $P_2$ that both collapse on $\Pi$ such that $P_1$ is approximable (Definition~\ref{def:approximable} to come) and $P_2$ is not.

\subsection{Proving Subject Reduction with Residuals}
\label{ss:sr-proof}

In the beginning of \Sec~\ref{ss:one-step-sr-se}, we emphasized the necessity of describing the effect of reduction on bipositions.
In this section, we sketch a proof of the subject reduction property in the quantitative case using residuation. A detailed proof can be found in Appendix~\ref{a:equinecessity} (in the non-quantitative case), although  Fig.~\ref{fig:SR-residuals} synthetizes the core arguments. For that, we define and use the notion of residual of a biposition, which will be instrumental to prove that approximability is stable under reduction and expansion. The reader can skip this section if they want to and keep in mind the following  summary:
\begin{itemize}
\item Residuation is defined for almost all \textit{right} bipositions (but not for \textit{left} ones) following Fig.~\ref{fig:SR-residuals}.
\item It is actually not defined when the (right) biposition points inside the application, the abstraction or an occurrence of the bound variable of the contracted redex.  
\item Residuation is a partial injection, but it is surjective: every biposition of the derivation typing the reduct is the residual of a biposition in the initial derivation. 
\end{itemize}


\noindent \textbf{Hypotheses and Notations.}
We assume again that $P$ is a (quantitative or not) $\ttS$-derivation typing $t$ and $\trb =\lxrs$. We set $\tprb=\rsx$, so that $t\breda{b} t'$. 
%
We then use the following notations and conventions:
\begin{itemize}
\item Let $\beta \in \supp{t}$. 
  We set $\Rep_P(\beta)=\set{\al\in \supp{P}\,|\, \ovlal = \beta}$. If $\al \in \Rep_P(\beta)$, we say that $\al$ is a representative of $\beta$ (inside $P$).
\item The metavariable $a$ ranges over $\RepPb$, \ie over the representatives of the position of the contracted redex. The metavariable $\al$ will be used for any position in $\bbN^*$.
\item If $a\in \RepPb$, we set $\AxPl{a}:=\Ax_{a\cdot 10}(x)$ and $\TrPl{a}=\Rt(\ttC^P(a\cdot 10)(x))$. Since $P$ is assumed to be quantitative, $\TrPl{a}:=\set{\tr{\al_0}\,|\,\al_0\in \Axl{a}}$.
  Thus, $\AxPl{a}$ is the set of positions of the redex variable (to be substituted) above $a$ and $\TrPl{a}$ is the set of the axiom tracks that have been used for them. We often omit $P$ in the notations.
\end{itemize}
For instance, in Fig~\ref{fig:SR-residuals}, $\Axl{a}=\set{a\cdot 10 \cdot a_2,\; a\cdot 10 \cdot a_3, a\cdot 10 \cdot a_7}$ and $\Trl{a}=\set{2,3,7}$.
Since $P$ is quantitative, $\ttC(a\cdot 10)(x)$ must be of the form $\skSk$ where $K=\Trl{a}$, for any $a\in \RepPb$.


For $k\in \Trl{a}$, we write $a_k$ for the unique $a_k\in \mathbb{N}^*$ such that $\pos{a\cdot 10,\,x,\, k}=a\cdot 10\cdot a_k$: thus, $a\cdot 10\cdot a_k$ is the position of the axiom rule typing $x$ above $a$ using axiom track $k$.\\

\noindent \textbf{Residuation for positions.} Assume $\alpha \in A,~ \ovl{\alpha}\neq a,\, a\cdot 1,\, a \cdot 10\cdot a_k$ for no $a\in \RepPb$ and $k\in \TrPl{a}$. We follow Fig.~\ref{fig:SR-residuals} and the \textbf{residual position} of $\alpha$, written $\Res_b(\alpha)$, is defined as follows:
\begin{itemize}
\item  If $\alpha \geqslant a\cdot k\cdot \alpha_0$ for some $a\in \RepPb$  and $k\geqslant 2$ (paradigm $\clubs$), then $\Res_b(\alpha)=a\cdot a_k \cdot \alpha_0$
\item If $\alpha =a\cdot 10 \cdot \alpha_0$ for some $a\in \RepPb$ with $\al_0\notin a_k$ (paradigm $\hearts$), then  $\Res_b(\alpha)=a\cdot \alpha_0$
\item If $\ovl{a}\ngeqslant b$, $\Res_b(\alpha)=a$.
\end{itemize}
By case analysis, we notice that $t'(\Res_b(a))=t(a)$ for all $a\in \dom{\Res_b}$. We set $A'= \codom {\Res_b}$ (we call $A'$ \textbf{residual support} of $P$). \\

\noindent \textbf{Residuation for \textit{right} bipositions.} Now, whenever
$\alpha':= \Res_b(\alpha)$ is defined, the \textbf{residual
  biposition} of $\p:=(\al,\gam)\in \bisupp{P}$ is $\Res_b(\p)=(\alpha',\,\gamma)$.

We notice that $\Res_b$ is a \textit{partial injective}
function (both for positions and right bipositions). In particular,
$\Res_b$ is a bijection from $\dom {\Res_b}$ to $A'$ and we write
$\Res_b^{-1}$ for its inverse. We set, for all $\al'\in A'$, $\ttT'(\al)=\ttT(\Res^{-1}_b(\al'))$.\\

\noindent \textbf{Handling contexts.} For any $\alpha'\in A'$, let $\ttC'(\alpha')$ be the context defined by $\ttC'(\alpha'):=(\ttC(\alpha)\setminus x) \uplus(\uplus_{k\in \ttK(\al)} \ttC(\alpha\cdot k))$, where $\alpha:=\Res_b^{-1} (\alpha')$, $\ttK(\al):= \Rt(\ttC(\al)(x))$
and $\ttC'(\al')\setminus x$ is the context $C'$ such that $C'(x)=\est$ and $C'(y)=\ttC'(\al')(y)$ for all $y\neq x$. Intuitively, $\ttC'(\alpha')$ erases the types assigned to $x$ and replace them by the contexts assigned to the matching occurrences of $s$, as expected in the reduct, where the argument $s$ has replaced $x$. 
Notice that $\ttC'(\alpha)=\ttC(\alpha)$ for any $\alpha \in A$ such that $\ovl{\alpha} \ngtr b$, \eg $\ttC'(\epsi)= \ttC(\epsi)=C$. Moreover, since $\Res_b(\epsi)=\epsi$, $\ttT'(\epsi)=\ttT(\epsi)=T$.\\

\begin{remark}
  \label{rmk:residuals-right-only}\mbox{}
  \begin{itemize}
\item   Residuation is defined for \textit{right} bipositions only (and not for left ones). Indeed, there is not good way to define residiuals for left bipositions. This may be understood by having a look à Fig.~\ref{fig:SR-residuals} and remembering that left bipositions pertain to left-hand sides of judgments, \ie in the typing contexts.  When reduction is performed, a judgment nested in the argument $s$ may be moved at \textit{higher depth} in $r$. But typing context have ``descendants'' below them till the root of the derivation. For instance, if $y$ is free in $t$ and $y:\skSk$ appears somewhere in $s$, then we find $y:\skSk$ all the way down to the root of $t$. When $s$ is substituted when $t\bred t'$, $y:\skSk$ may have more descendants in $t'$ than in $t$. In these conditions, it is not possible to have quantitative notion of residuation for contexts, at least naively (it is actually possible, but only if we consider \textit{classes} of left bipositions, as in~\cite{VialLICS18})
\item However, altough residuation cannot be defined for left bipositions, morally, the definition $\ttC'$ specifies the effect of reduction on the \textit{left}-hand sides of judgments of derivation $P$.
  \end{itemize}  
  \end{remark}

We define now $P'$ as the labelled tree such that $\supp {P'}=A'$ and $P'(\alpha')$ is $\ju{\ttC'(\alpha')}{ t'\rstr{\alpha'}:\, \ttT'(\al')}$. We claim that $P'$ is a correct derivation concluded by $\juCtpt$: indeed, $\ovl{A'}\subset \supp {t'}$ stems from $\ovl{A}\subset \supp {t}$. Then, for any $\alpha'\in A'$ and $\alpha=\Res_b^{-1}(\alpha)$, 
$t'(\alpha')=t(\alpha)$ and the rule at position $\alpha'$
is correct in $P'$ because the rule at position $\alpha$ in $P$ is 
correct after we notice that:
\begin{itemize}
\item In the abstraction case, $t'(\alpha')=\lambda y$ implies $\ttC'(\alpha'\cdot 0)(y)=\ttC(\alpha\cdot 0)(y)$).
\item If $t'\rstr{\al'}$ is a substituted occurrence of $s$, then, $t(\al)=x$ and $P(\al)=\ju{x:(k\cdot \ttT(\al))}{x:\ttT(\al)}$ with $\al=a\cdot 10\cdot a_k$ for some $a\in \RepPb$ and the typing rules ensure that $\ttT(\al)=\ttT(a\cdot k)=\ttT'(\al')$.
\end{itemize}
Thus, $P'$ is a derivation concluding with $\juCtpt$. This proves that the subject reduction property holds in system $\ttS$.

\begin{observation}
  \label{obs:when-res-not-defined}
  As noted above, $\Res_b$ is a surjective and partial injective function from $\supp{P}$ to $\supp{P'}$ (thus, $\Res^{-1}(\p')$ is always a singleton set when $\p'\in \bisuppR{P'}$). Actually there are 3 cases whan $\Res_b(\al)$ is not defined: 
  \begin{itemize}
  \item $\al$ points to an $\ax$-rule typing the variable of the redex:  $t(\al)=x$. Equivalently, $\al=a\cdot 10\cdot a_k$ for some $a\in \RepPb$ and $k \in \Trl{a}$.
  \item $\al$ points to the contracted redex: $\tral=\lxrs$. Equivalently, $\al=a\in \RepPb$.
  \item  $\al$ points to the abstraction of the constracted redex: $\tral = \lx.r$. Equivalently, $\al = a\cdot 1$ for some $a\in \RepPb$.
  \end{itemize}
\end{observation}

Actually, the constructions presented here are mostly valid when $P$ is not assumed to be quantitative anymore.  In that case, $\TrPl{a}:=\Rt(\ttC^P(a\cdot 10)(x))$ may be different from $\ttX(a):=\set{\tr{\al_0}\,|\,\al_0\in \Axl{a}}$ ($\ttX(a)\subeq \TrPl{a}$ holds and $\TrPl{a}\setminus \ttX(a)$ corresponds to the phantom axiom tracks of $x$ above $a\cdot 10$). Then, $\ttC'(\alpha')$ is still defined as $(\ttC(\alpha)\setminus x) \uplus(\uplus_{k\in \Trl{a}} \ttC(\alpha\cdot k))$ with $\ttK(\al):= \Rt(\ttC(\al)(x))$ (thus, $\ttK(\al)$ may contain some phantom axiom tracks) but $\Res_b(a\cdot k\cdot \al)$ is not defined anymore when $k\in \Trl{a}\setminus \ttK(a)$

\section{Approximable Derivations and Unforgetfulness}
\label{s:approx}

We formally define \textit{approximations} and \textit{approximability} (\Sec~\ref{ss:lattices}). We prove that one-step proof reduction is monotonic. \Sec~\ref{ss:equinecessity-main} shows that very left-biposition is intrisically related to a right one, which is pivotal to prove soundness later (if a term is approximablity typable, then it is normalizing).
 In \Sec~\ref{ss:S-unforgetfulness}, we define the unforgetfulness criterion in system $\ttS$.  
We prove infinitary subject reduction and that the infinite proof reduction actually preserves approximability (\Sec~\ref{ss:inf-subj-red-S}). As a by-product, we prove that if a term is suitably typed, then it is hereditary head normalizing, which is one of the two implication of the fundamental equivalence between typability and infinitary normalization.

\subsection{Approximability and the Lattice of Approximations}
\label{ss:lattices}

In \Sec~\ref{ss:non-determinism}, we saw $\scrR$ is unfit to recover soundness through approximability.  Let us now work with system $\ttS$ only and formalize this notion.

As seen in \Sec~\ref{ss:typ-inf-nf-informal},  we must be able to truncate derivations (notion of approximation) and define the join of some families of derivations. This can be properly done in system $\ttS$:

\begin{definition}[Approximations] \mbox{}
\label{def:approximation} 
  \begin{itemize}
\item  Let $P$ and $P_*$ be two derivations typing a same term $t$. We say 
$P_*$ is an \textbf{approximation} of $P$, and we write $P_* \leqfty P$, if $\bisupp {P_*} \subseteq \bisupp {P}$ and for all $\p \in \bisupp {P_*},~ P_*(\p)=P(\p)$.
\item We write $\Appfty{P}$ for the set of approximations of a derivation $P$ and $\Approx{P}$ for the set of \textit{finite} approximations of $P$. 
\end{itemize}
\end{definition}

Thus, $P_* \leqfty P$ if $P_*$ is a \textit{correct} restriction of $P$ to a subset of $\bisupp{P}$ (\ie a restriction that gives an $\ttS$-derivation). We usually write $\fP$ for a \textit{finite} approximation of $P$ (\ie $\bisupp{\fP}$ is finite) and in that case only, we write $\fP \leqslant P$ instead of $\fP\leqfty P$. Actually, $\leqfty$ and $\leqslant$ are associated to lattice structures induced by the set-theoretic inclusion, union and intersection on bisupports :

\begin{theorem}
  \label{th:cpo}
  Let $t$ be a 001-term.
 The set of derivations typing $t$  is a directed complete semi-lattice \wrt $\leqfty$: 
  \begin{itemize}
  \item If $D$ is a directed set of derivations typing $t$:
\begin{itemize}
  \item The \textbf{join} $\sup D$ of $D$ is the function $P$ defined by $\dom{P}=\cup_{P_*\in D} \bisupp{P_*}$ and $P(\p) = P_*(\p)$ (for any $P_*\in D$ such that $\p \in \bisupp {P_*}$), which also is a derivation.
  \item  The \textbf{meet} $\inf D$ of $D$ is the function $P$ defined by $\dom{P}=\cap_{P_*\in D} \bisupp{P_*}$ and $P(\p) = P_*(\p)$ (for all $P_*\in D$), which also is a derivation.
\end{itemize}
  \item If $P$ is a derivation typing $t$, $\Appfty{P}$ is a complete lattice and $\Approx{P}$ is a lattice.
  \end{itemize}
\end{theorem}

\begin{proof}
See Appendix~\ref{a:lattices} for the details.
\end{proof}

\begin{remark}
Some $\ttS$-derivations do not have any finite approximations, \eg any $\ttS$-derivation typing $\Om$, because $\Om$ is not head normalizing.
\end{remark}

\noindent Approximation is compatible with reduction:

\begin{lemma}[Monotonicity of approximation] \label{lem:red-monot} \mbox{}
\begin{enumerate}
\item Reduction is monotonic: if $P_* \leqfty P,~ P_* \rewb{b} P'_*$ and $P\rewb{b}P'$, then $P'_* \leqfty P'$.
\item If $P\rewb{b} P'$, then, for any $P'_*\leqfty P'$, 
there is a unique $P_* \leqfty P$ such that $P'_* \rewb{b} P'_*$.
\item Uniform expansion is monotonic: let $\code{\cdot}:\bbN^* \rew \Nmzo$ an injection.
  If $P'$ types $t'$, $t\rewb{b} t'$ and  $P'_* \leqfty P'$, then $P_* \leqfty P$ with $P_*=\Exp_b( P'_*,\code{\cdot},t),~ P=\Exp_b(P',\code{\cdot},t)$.
\end{enumerate}
\end{lemma}

As we saw in \Sec~\ref{ss:non-determinism}, determinism of reduction in system $\ttS$ is crucial to ensure Point 1 of Lemma~\ref{lem:red-monot}, which in turn is pivotal to ensure that the notion of approximation is well-behaved \wrt conversion.


We define here our validity condition, \ie approximability: a derivation $P$ is approximable when it is the join of all its finite approximations.
Intuitively, this means that  $P$ is approximable if all its bipositions are 
meaningful, \ie can be part of a \textit{finite} derivation $\fP$ approximating $P$.

\begin{definition}
  \label{def:approximable}
A $\ttS$-derivation $P$ is said to be \textbf{approximable} if $P$ is the join of its \textit{finite} approximations, \ie $P=\sup_{\fP\leqs P} \fP$. 
\end{definition}

\noindent By definition of $\leqs$, we immediately obtain the following characterization of approximability:

\begin{lemma}
\label{lem:Bo-appox-charac}
  A $\ttS$-derivation $P$ is approximable iff, for all finite $\oB\subseteq \bisupp {P}$, there is a $\fP\leqslant P$ such that $\oB\subseteq \bisupp {\fP}$.
\end{lemma}

\subsection{Equinecessity: forgetting about left bipositions}
\label{ss:equinecessity-main}

Approximability is stable under reduction and expansion (forthcoming Lemma~\ref{lem:approx-red}) and as expected, it is a crucial property of system $\ttS$. However, there are some difficulties to overcome:
\begin{itemize}
\item The approximability condition (Lemma~\ref{lem:Bo-appox-charac}) pertains both to right and \textit{left} bipositions.  
\item Thus, this stability property shall involve residuals of bipositions.
\item But residuals are defined for \textit{right} bipositions only (Remark~\ref{rmk:residuals-right-only}).

\end{itemize}
There seems to be a catch. However, approximability may be reduced to \textit{right} bipositions only, as we  sketch in this section (full details are in Appendix~\ref{a:equinecessity}).

Assume for instance that $P$ is a derivation, that $\ju{x:\kSk}{x:S_k}$ is an axiom node in $P$ occuring at position $a$ and  $P_*\leqfty  P$ is a finite or not approximation of $P$.
Let us write $\ttT$ and $\ttT_*$ for $\ttT^P$ and $\ttT^{P_*}$, and so on. We have thus $S_k=\ttT(a)$ and $x:\kSk=\ttC(a)$.

Notice that $P_*$ may not contain that axiom rule (\ie both $a \in \supp{P_*}$ and $a\notin \supp{P_*}$ are possible \textit{a priori}): it occurs whenever this axiom has been cut off during truncation. Assume  that $a\in \supp{P_*}$:  
in particular, $(a,\epsi)\in \supp{P_*}$, since it points to the root of $\ttT(a)$. Actually, since  $P_*(a)=\ju{x:(k\cdot \ttT_*(a))}{x:\ttT_*(a)}$, 
$\supp{\ttC_*(a)(x)}=k\cdot \supp{\ttT_*(A)}$. This means that, for all $c\in \bbN$, $(a,x,k\cdot c)\in \bisupp{P_*}$ iff $(a,c)\in \bisupp{P_*}$. We say then that bipositions $(a,x,k\cdot c)$ and $(a,c)$ are \textbf{equinecessary \wrt derivation $P$} (the notion of equinecessity depends on $P$, the considered derivation). 

%
A crucial observation is that the cases of equinecessity we have observed above relate left \textit{bipositions} with right ones ($(a,x,k\cdot)$ is on the left and $(a\,c)$ on the right). Actually, we showed that, in an axiom node, any left biposition is equinecessary with a right-one (since $\supp{\ttC(a)(y)}=\eset$ if $y\neq x$ and $\supp{\ttC(a)}(x)=k\cdot \supp{\ttT(a)})$. This expresses the fact that, in an $\ax$-rule, the left-hand side and the right-hand side are the same (up to the axiom track in the context).

More generally, in a quantitative derivation (Definition~\ref{def:quant-S-deriv}), any left biposition is equinecessary with a right one. Indeed, by definition, in a quantitative derivation, contexts can be computed by looking at the axiom rules and any left biposition must come from the left-hand side of an axiom rule (since there are no phantom axiom tracks). We may thus observe:

\begin{observation}
\label{obs:red-to-right}
Let $P$ be a quantitative derivation. Every left biposition of $P$ is equinecessary with at least one right biposition.
\end{observation}

Observation~\ref{obs:red-to-right} is proved in details in Appendix~\ref{a:equinecessity}. 
But our troubles are not over yet: when $P$ types $t'$ and $t\bred t'$, some right bipositions $\p \in \bisupp{P}$ do not have a residual. If we want to transport approximability through reduction and expansion, we need to handle this. The idea is that any right biposition is equinecessary with one which has a residual.

By Remark~\ref{obs:when-res-not-defined}, $\Res_b(\p)$ (with $\p =(\al,c)$) is not defined when $t(\al)=x$, $\tral=\lxrs$ or $\tral=\lx.r$. We explain now why in each case, $\p$ and some biposition which has a residual are equinecessary.
\begin{enumerate}
  \item If $t(\al)=x$, say $\al=a\cdot 10\cdot a_k$, $P(\al)=\ju{x:\kSk}{x:S_k}$ with $S_k:=\ttT(a)$. Then the argument $s$ of the redex must also have type $S_k$. Namely, $P(a\cdot k)=\ju{D}{s:S_k}$ for some context $D$, and it is clear that biposition $\p=(a\cdot 10\cdot a_k,c)$ and biposition $(a\cdot k,c)$, which has a residual, are equinecessary.
  \item If $\tral=\lxrs$, then $\al=a\in \Rep_P(b)$. Since $r$ and $\lxrs$ must have the same type, it is clear 
that  $\p=(a,c)$ and $\p_0:=(a\cdot 10,c)$ are equinecessary.
Then, either $t(a\cdot 10)\neq x$ and $\p_0$ has a residual (we are done) or $t(a\cdot 10)=x$ and we use case 1 to conclude that $\p$ is equinecessary with a biposition which has a residual.
  \item If $\tral=\lx.r$, then $\al=a\cdot 1$ with $\ovla=b$ and $P(\al)=\ju{C}{\lx.t:\sSk\rew T}$ where $T:=\ttT(a)$ (which is the type of $\lxrs$) and the $\sSk$ have been assigned to $x$. Then either $\p=(\al,c)$ points to a symbol inside $T$ (we may then use case 2) or $\p$ points to a symbol in some $S_k$ (we may then use case 1) or $c=\epsi$, \ie $\p$ points to the top-level $\rightarrow$ in $\sSk \rew T$. In the latter case, $\p=(a\cdot 1,\epsi)$ is equinecessary to $\p_0:=(a\cdot 1,1)$ (if there is an arrow a position $\epsi$, there must me something at position 1). Since $\p_0$ points to a symbol in $T$ and that this case has already been handled, this shows also that $\p$ is equinecessary to a biposition which has a residual.
\end{enumerate}
We may thus state:

\begin{observation}
  \label{obs:red-to-right-res}
  Let $P$ be a quantitative derivation. Every left biposition of $P$ is equinecessary with at least one right biposition which has a residual.
\end{observation}

Observation~\ref{obs:red-to-right-res} is proved in detail in Appendix~\ref{a:approx-stable-conv}.



\noindent Quantitativity (Definition~\ref{def:quant-S-deriv}) is necessary to ensure approximability.

\begin{lemma}[Approximability and Quantitativity]\label{lem:approx-quant}\mbox{}\\
  If $P$ is not quantitative, then $P$ is not approximable.
\end{lemma}

\begin{proof}
  If $P$ is not quantitative, then $P$ contains some left biposition $\p:=(a,x,k\cdot c)$ that does not come from an $\ax$-rule, \ie $k$ is a phantom axiom track. 
  This phantom axiom track must have infinitely many equinecessary occurrences (along an infinite branch).
    An approximation $P_*\leqfty P$ that contains $\p$ has to contain all these occurrences and thus, cannot be finite. So $P$ cannot be approximable.
\end{proof}

\begin{lemma}[Approximability is stable under reduction and expansion]
\label{lem:approx-red}\mbox{}\\
If $P\rewb{b} P'$, then $P$ is approximable iff $P'$ is approximable.
\end{lemma}

\begin{proof} We use Lemma~\ref{lem:Bo-appox-charac}.
\begin{itemize}
  \item 
  Assume $P$ approximable. Let us show that $P'$ is also approximable.
  First, since $P$ is approximable, then $P$ is quantitative by Lemma~\ref{lem:approx-quant}, so $P'$ is also quantitative (Remark~\ref{rmk:inf-quant-and-red}). 
  Let $\oB'\subseteq \bisupp{P'}$ be a finite set of bipositions. We must prove that there exists $\fP' \leqs P'$ such that $\oB' \subeq \bisupp{\fP'}$.
  By Observation~\ref{obs:red-to-right}, there is a finite set $\oB'_0$ of right bipositions such that every $\p'\in \oB'$ is equinecessary with some $\p'_0 \in \oB'_0$.
  Let $\oB:= \Res^{-1}_b(\oB'_0)$. Since $\Res_b$ is surjective on right biposition (Observation~\ref{obs:when-res-not-defined}), $\oB'_0 \subeq \Res_b(\oB)$. 
  Since $P$ is approximable, there is $\fP\leqs P$ such that $\oB\in \bisupp{\fP}$. Let $\fP'$ defined by $\fP\breda{b}\fP'$, so that $\oB'_0\subeq \bisupp{\fP'}$. By Lemma~\ref{lem:red-monot} , $\fP'\leqs P'$.  By equinecessity, $\oB'\subeq \bisupp{\fP'}$. Thus, $P'$ is approximable.
\item  Assume that $P'$ is approximable.
Let us show that $P$ is also approximable.
  First, since $P'$ is approximable, then $P'$ is quantitative by Lemma~\ref{lem:approx-quant}, so $P$ is also quantitative (Remark~\ref{rmk:inf-quant-and-red}). 
  Let $\oB\subseteq \bisupp{P}$ be a finite set of bipositions. We must prove that there exists $\fP \leqs P$ such that $\oB \subeq \bisupp{\fP}$.
  By Observation~\ref{obs:red-to-right-res}, there is a finite set $\oB_0$ of right bipositions which have a residual, such that every $\p\in \oB$ is equinecessary with some $\p_0 \in \oB_0$.
  Let $\oB':= \Res_b(\oB'_0)$. Since $\Res_b$ is totally defined on $\oB_0$, $\oB_0 \subeq \Res^{-1}_b(\oB')$. 
  Since $P'$ is approximable, there is $\fP'\leqs P'$ such that $\oB'\in \bisupp{\fP'}$. By Lemma~\ref{lem:red-monot} (2nd point), there is $\fP\leqfty P$ such that $\fP \bred \fP'$, so that $\oB_0\subeq \bisupp{\fP}$. Since $\fP'$ is finite, $\fP$ is also finite. By equinecessity, $\oB\subeq \bisupp{\fP}$. Thus, $P$ is approximable.
\end{itemize}
\end{proof}

\begin{remark}[Root approximability]
  \label{rk:root-approx}
Approximability (Definition~\ref{def:approximable}) is a condition pertaining to a whole derivation $P$ in the sense that $P$ is approximable iff each biposition $\p$ in $P$ is inside some finite approximation of $P$ (Lemma~\ref{lem:Bo-appox-charac}). One may wonder whether it is possible to reformulate approximability such that it pertains only to the \textit{conclusion} of $P$, as some other criteria (\eg unforgetfulness, \Sec~\ref{s:finite-inter-infinite-types} and forthcoming \Sec~\ref{ss:S-unforgetfulness}). We prove in Appendix~\ref{a:approx-0-cex}
 that there is no way to do so, at least naively.
\end{remark}

\subsection{Unforgetfulness}
\label{ss:S-unforgetfulness}

We remember from \Sec~\ref{ss:system-Ro-Klop} that weak normalization for the finite calculus is characterized in system $\scrRo$ by means of \textit{unforgetful} derivations. In order to characterize weak normalization in $\Lamzzu$ (\Sec~\ref{ss:comput-bohm-trees-hp}), we shall adapt Theorem~\ref{th:WN-Ro} to system $\ttS$. This will yield Theorem~\ref{th:charac-WN-S}, one of the main results of this paper, stated as follows:

\begin{theorem*}  
  A term $t$ is weakly-normalizing in $\Lamzzu$ if and only if $t$ is typable by means of an approximable unforgetful derivation, and if and only if $t$ is hereditary head normalizing.
\end{theorem*}

To state and prove this theorem, we must first export the definition of unforgetfulness to system $\ttS$. We recall that the targets of arrows are regarded as positive and their sources as negative. The following definitions are straightforward adaptations from system $\scrRo$.

\begin{definition}
  Inductively:
  \begin{itemize}
  \item For all types $T$, $\est$ occurs negatively in $\est\rew T$.
  \item $\est$ occurs positively (resp. negatively) in $\sSk$ if there exists $\kK$ such that $\est$ occurs positively (resp. negatively) in $S_k$.
  \item $\est$ occurs positively (resp. negatively) in $\sSk\rew T$ if $\est$ occurs positively (resp. negatively) in $T$ or negatively (resp. positively) in $\sSk$.
  \end{itemize}  
\end{definition}

\begin{definition}[Unforgetfulness] \mbox{}
  \begin{itemize}
  \item A judgment $\juCtt$ is \textbf{unforgetful} when $\est$ does not occur positively in $T$ and  $\est$ does not occur negatively in $C(x)$ for any $x\in \TermV$.
  \item A derivation is \textbf{unforgetful} when it concludes with an unforgetful judgment.
    \end{itemize}
\end{definition}

We easily check by induction on $b\in \supp{t}$ that, if $t$ is a 001-NF and $P \tri \juCtt$ is unforgetful, then every subterm $\trb$ of $t$ is typed in $P$ (\ie $\forall b\in \supp{t},\ \exists a \in \supp{P},\ \ovla=b$).


\begin{lemma}
 \label{lem:uf-hereditary}
  If $P\rhd C \vdash t:\,T$ is an unforgetful derivation typing a head normal form 
$t=\lambda x_1\ldots x_p.x\,t_1\ldots t_q$, then, there are 
unforgetful subderivations of $P$ typing $t_1$, $t_2$,\dots, $t_q$. 
Moreover, if $P$ is approximable, so are they.
\end{lemma}

\begin{proof}
Whether $x=x_i$ for some $i$ or not, the unforgetfulness condition ensures that every argument of the head variable $x$ is typed, since $\est$ cannot occur negatively in its unique given type.
\end{proof}

\begin{lemma} \label{lem:typ-HN}
If $P \tri \juCtt$ is a \textit{finite} derivation, then 
$t$ is head normalizing. 
\end{lemma}

\begin{proof}\mbox{}
\begin{itemize} 
\item By the typing rules, the head redex---if it exists \ie if $t$ is head reducible, and not a head normal form---must be typed. 
\item When we reduce a typed redex, the number of rules of the derivation must
strictly decrease (at least one $\app$-rule and one $\abs$-rule
disappear). See Figure~\ref{fig:SR-residuals}. 
\item Since there is no infinite decreasing sequence of natural number, the head reduction strategy must halt at some point, meaning that a HNF is reached.
\end{itemize}
\end{proof}

\begin{proposition}
  \label{prop:uf-typ-implies-wn}
If a term $t$ is typable by a unforgetful approximable derivation,
then $t$ is hereditarily head normalizing.
\end{proposition}

\begin{proof}
  
Assume that $P\tri \juCtt$ is an approximable and unforgetful derivation. 
By Lemma~\ref{lem:typ-HN}, there is a head normal form $t':=\hnfo$ such that $t\hred^* t'$. By Proposition~\ref{prop:subj-red-one-step-S} and \label{lem:approx-red}, there is an approximable derivation $P'\tri \ju{C}{t':T}$. Since $P'$ and $P$ conclude with the same type and context, $P'$ also is unforgetful. By Lemma~\ref{lem:uf-hereditary}, there are approximable unforgetful derivations $P_1$, $P_2$,\ldots, $P_q$ respectively typing $t_1$, $t_2$,\ldots, $t_q$. Notice that the $t_i$ occur at applicative depth 1. 

This observation allows us to build by induction on $d\in \bbN$ a path $t =: t_0 \bred^* t_1 \bred^* t_2 \bred^* t_d\bred^*\ldots$, corresponding to hereditary head reduction (\Sec~\ref{ss:comput-bohm-trees-hp}), such that the reduction steps from $t_d$ to $t_{d+1}$ all occur at applicative depth $d$. This path stops at $t_d$ (\ie is of finite length), if $t_d$ is a normal form. It is of infinite length if no $t_d$ is a normal form. In both cases, it is by construction a productive reduction path. 
\begin{itemize}
\item If the path is finite, some $t_d$ is a normal form. Thus, $t$ is weakly normalizing.
\item If the path is of infinite length, let $t'$ be its limit. By definition of the limit, the specification of $t_d$ entails that, for all $d\in \bbN$, $t'$ does not have any redex under applicative depth $d$. Thus, $t'$ has no redex: $t'$ is a normal form, \ie $t$ is hereditarily head normalizing.
\end{itemize}
\end{proof}

\subsection{The Infinitary Subject Reduction Property}
\label{ss:inf-subj-red-S}


In this section, we prove subject reduction for productive reduction paths. The initial derivation may be approximable or not.

 Thus, when $t\bredfty t'$, we have to construct a derivation $P'$ typing $t'$ 
from a derivation $P$ typing a term $t$. The main intuition is as in \Sec~\ref{ss:subj-subst}: when a reduction is performed at applicative depth $n$, the contexts and types are not affected below depth $n$. Thus, a productive reduction path stabilizes contexts and types at any fixed applicative depth. It allows to define a derivation typing the limit $t'$.

The following (and straightforward) \textbf{subject substitution lemma}, similar to Lemma~\ref{lem:sub-subst-Rftyo}, is very useful while
working with productive paths. It states that we can freely change the untyped parts of a term in a typing derivation. 

\begin{lemma}
\label{lem:subject-substitution}
  Assume $P\rhd \juCtt$ and for all $a\in \supp {P},~ t(a) =t'(a)$ ($P$ is not necessarily assumed to be approximable).\\
Let $P[t'/t]$ be the \textit{labelled tree} obtained from $P$ by replacing $t$ by $t'$ (more precisely, $P[t'/t]$ is the labelled tree $P'$ such that $\supp{P'}=\supp{P}$ and, for all $a\in \supp{P},~ P'(a)=\ju{\ttC^P(a)}{t'\rstr{a}:\ttT^P(a)}$).\\
Then $P[t'/t]$ is a correct derivation.
\end{lemma}

Now, let us formally prove the infinitary subject reduction property. For that, we assume:
\begin{itemize}
\item $t\rightarrow^{\infty} t'$. For instance, say that $t=t_0 \rewb{b_0} t_1\rewb{b_1} \dots \rewb{b_{n-1}} t_n \rewb{b_n}t_{n+1}\rewb{b_{n+1}} \ldots $ with $b_n\in \{0,~ 1,~ 2\}^*$ and $\ad{b_n}\longrightarrow \infty$ and $t'$ is the limit of this path.
\item There is a 
  derivation $P \rhd \juCtt$ and  $A=\supp{P}$.\\[-0.15cm]
\end{itemize}

By following step by step the reduction path, 
$b_0,~ b_1,\ldots$, we get a sequence of derivations $P_n\rhd \ju{C}{t_n: T}$ of support $A_n$. We then write $\ttC_n,\; \ttT_n$ for $\ttC^{P_n}$ and $\ttT^{P_n}$. When performing $t_n\stackrel{b_n}{\rightarrow}t_{n+1}$, notice that $\ttC_n(a)$ and $\ttT_n(a)$ are not modified for any $a$ such that $b_n\nleqslant \ovla$ \ie $\ttC_n(a)=\ttC_{n+1}(a)$ and $\ttT_n(a)=\ttT_{n+1}(a)$ if $b_n\nleqslant \ovla$.

Let $a\in \bbN^*$ and $N\in \bbN$ be such that, for all $n\geqslant N,~ \ad{b_n}>\ad{a}$. There are two cases:
\begin{itemize}
\item $a\in A_n$ for all $n\geqslant N$. Moreover,
  $\ttC_n(a)=\ttC_N(a),~ \ttT_n(a)=\ttT_N(a)$ for all $n\geqslant N$, and $t'(a)=t_n(a)=t_N(a)$.  
\item $a\notin A_n$ for all $n\geqslant N$.
\end{itemize}

We set $ A'=\set{a\in \bbN^*\, |\, \exists N,\, \forall n \geqslant N, ~
a\in A_n}$. We define a labelled tree $P'$ whose support is $A'$ by
$P'(a)=\ju{\ttC_n(a)}{ t'\rstr{a}:\ttT_n(a)} $ and we set $\ttC'(a)=\ttC_n(a),\ttT'(a)=\ttT_n(a)$
for any $n\geqslant N(\ad{a})$ (where $N(\ell)$ is the smallest rank $N$ such that $\forall n \geqslant N,~ \ad{a_n}>\ell$) .

\begin{lemma}
  \label{lem:subj-red-1}
The labelled tree $P'$ is a derivation. 
\end{lemma}

\begin{proof} 
Let $a\in A'$ and $n\geqslant N(|a|+1)$. Thus, $t'(a)=t_n(a)$ and the 
types and contexts involved at node $a$ and its premises are the same in 
$P'$ and $P_n$. So the node $a$ of $P'$ is correct, because it is 
correct in $P_n$. 
\end{proof}

\begin{lemma}
  \label{lem:subj-red-2}
If $P$ is approximable, so is $P'$.
\end{lemma}

\begin{proof}
Let $\oB\subseteq \bisupp{P'}$. We set $\ell=\max\set{\ad{\p}\,|\, \p\in \oB}$ ($\ad{\p}$ is the applicative depth of the underlying $a\in \supp{P}$).
  

  By productivity, there is $N$ such that, $\forall n\geqslant N,~ \ad{b_n}\geqslant \ell+1$, and thus, $t_{N}(a)=t_n(a)=t'(a),~ \ttC_N(a)=\ttC_n(a)=\ttC'(a)$ and $\ttT_N(a)=\ttT_n(a)=\ttT'(a)$ for all $a$ such that $\ad{a}\leqslant \ell$ and $n\geqslant N$. In particular, $\oB\subseteq \bisupp{P_N}$.

  Since $P$ is approximable, by Lemma~\ref{lem:approx-red}, $P_N$ is approximable. So, there is $\fP_N \leqslant P$ such that $\oB\subseteq \bisupp{\fP_N}$.
  

  Let $d=\card{\bisupp{\fP_N}}$ and $N'\geqs N$ such that, for all $n\geqs N'$, $\ad{b_n}> d$. Let $\fP_{N'}$ the derivation defined by $\fP_N\breda{b_N} \ldots \breda{b_{N'-1}} \fP_{N'}$. Since $b\in \supp{\fP_{N'}}$ implies $\ad{b}\leqs N'<d<\ad{b_n}$ for any $n\geqs N'$, $t_{N'}$ and $t'$ do not differ on $\supp{\fP_{N'}}$ and we can apply Lemma~\ref{lem:subject-substitution}. Thus, let $\fP'$ be the derivation obtained by replacing $t_{N'}$ with $t'$ in $\fP$. We have $\oB \subeq \fP'\leqs P'$. Thus, $P'$ is approximable.
\end{proof}

\begin{proposition}[Infinitary Subject Reduction]
  \label{prop:infinite-subject-reduction}
  Assume $t\rew^{\infty} t'$ and $P\tri \juCtt$. Then there exists a derivation $P'$ such that $P'\tri \ju{C}{t':T}$.\\
  Moreover, if $P$ is approximable, $P'$ may be chosen to be approximable.
\end{proposition}

\begin{proof}
Consequence of Lemmas~\ref{lem:subj-red-1} and \ref{lem:subj-red-2}.
\end{proof}

\section{Typing Normal Forms and Subject Expansion}
\label{s:normal-forms}


In this section, we prove that (possibly infinite) normal forms are \textit{fully}  typable (\ie without leaving untyped subterms) in system $\ttS$ and that all the quantitative derivations typing a normal form are approximable (\ie the converse of Lemma~\ref{lem:approx-quant} is valid for normal forms).

Indeed, as hinted at in \Sec~\ref{ss:intersection-overview-klop-intro} and \Sec~\ref{ss:system-Ro-Klop}, a proof of \textit{Hereditarily Head Normalizing}  $\Rightarrow$ \textit{Typable} proceed by giving first an unforgetful typing of normal forms, and then, using a subject expansion property.

Actually, we will describe all the quantitative derivations typing any 001-normal form and prove them to be approximable. The latter part is an extra-difficulty imposed by working in the infinitary setting.  Then, we will prove an infinitary subject expansion property, which is enough to show the above implication.

Describing all the possible typings of  001-normal forms is based on the distinction  between what we call \textit{constrained} and \textit{unconstrained} positions in a normal form. We use this notions from  \Sec~\ref{ss:support-candidates} to \Sec~\ref{ss:nf-approx}.\\

\noindent \textbf{Constrained and Unconstrained Positions.}  
Infinitary normal forms (\Bohm\ trees without $\bot$) are coinductive assemblages of head normal forms (\cf \Sec~\ref{ss:comput-bohm-trees-hp}). It is then important to understand how a head normal form $t=\lx_1\ldots \lx_p.x\,t_1\ldots t_q$ may be typed. Let us have a look at figure \ref{fig:parsNF} and ignore for the moment the $\rdeg$ and positions (between chevrons) annotations: the head variable $x$ has been assigned an arrow type $\sSqk{1}\rew\ldots \rew \sSqk{q}\rew T$ whereas the first argument $t_1$ is typed with types $S^1_k$ ($k$ ranging over $K(1)$),\dots, the $q$-th argument $t_q$ is typed with types $S^q_k$ ($k$ ranging over $K(q)$).

In the term $t$, the subterms $x,\, x\,t_1,\, x\,t_1t_2,\ldots,\, x\,t_1\ldots t_{q\!-\!1}$, which are ``partial''  applications of the head variable $x$, 
are typed with arrow types, as well as the subterms $\lx_p.x\,t_1\ldots t_q,\,\lx_{p\!-\!1}x_p.x\,t_1\ldots t_q,\ldots$ and $\lx_1\ldots x_p.x\,t_1\ldots t_q$, which start with abstractions. By contrast, notice that subterm $x\,t_1\ldots t_n$, where $x$ is ``fully'' applied, has type $T$ and that this type $T$ may be \textit{any} type. So, we say that the type of subterm $x\,t_1\ldots t_q$ is \textbf{unconstrained} (in $t$), whereas for instance:
\begin{itemize}
\item $x\,t_1\ldots t_{q\!-\!1}$ has type $\sSqk{q}\rew T$: this type depends on the types $S^q_k$ given to the subterm $t_q$. We say informally that it \textit{calls} for the types of the subterm $t_q$.
  \item $\lx_p.x\,t_1\ldots t_q$ has type $\ttC(0^p)(x_p)\rew T$: this type depends on the types assigned to $x_p$ deeper in the term, which are stored in $\ttC(0^p)(x_p)$. We say informally that it \textbf{calls for the types} of variable the $x_p$.
\end{itemize}

\begin{figure*}
  \ovalbox{
    \begin{tikzpicture}
  

\draw (-3, 3) node {$\rdeg$};
\draw (-3,2.4) node {$q$};
\draw (-3,1.6) node {$q\!-\!\!1$};
\draw (-3,0.6) node {$1$};
\draw (-3,0) node {$0$};
\draw (-3,-0.9) node {$1$};
\draw (-3,-2) node {$p$};

  
  \draw (-2,2.4) node {\small $x$};
  \draw (-2,2.4) circle (0.23) ;

  \draw (-1.92,2.19) -- (-1.52,1.78);
  \draw (-1.8,1.9) node {\fnsz $\red{1}$};
  
  
  \draw (-1,2.45) node {\small $t_1$};
  \draw (-1,2.1) --++ (-0.4,0.6) --++ (0.8,0) --++ (-0.4,-0.6);
  
  \draw (-1.28,1.78) -- (-1,2.1);   

  \draw (-1.4,1.6) node {\small $\arob$};
  \draw (-1.4,1.6) circle (0.23);

  \draw [dotted] (-1.26,1.43) -- (-0.75,0.86);

  
 \draw (-0.2,1.62) node {\small $t_{\!q\!-\!1}$};
  \draw (-0.2,1.2) --++ (-0.4,0.6) --++ (0.8,0) --++ (-0.4,-0.6);
  
  \draw (-0.48,0.88) -- (-0.2,1.2);   

  
\draw (-0.6,0.7) node {\small $\arob$};
\draw (-0.6,0.7) circle (0.23) ;  

\draw (-0.48,0.52) -- (-0.16,0.16);  
\draw (-0.4,0.2) node {\fnsz $\red{1}$};


 \draw (0.4,0.92) node {\small $t_{q}$};
 \draw (0.4,0.5) --++ (-0.4,0.6) --++ (0.8,0) --++ (-0.4,-0.6);
  
 \draw (0.12,0.18) -- (0.4,0.5);  


\draw (0,0) node {\small $\arob$};
\draw (0,0) circle (0.23) ;

\draw (0,-0.23) --++ (0,-0.44);
\draw (-0.13,-0.42) node {\fnsz $\red{0}$};

\draw (0,-0.9) node {\ssz $\lambda \!\!\;x_{p}$};
\draw (0,-0.9) circle (0.23);

\draw [dotted] (0,-1.13) --++ (0,-0.64) ;

\draw (0,-2) node {\ssz $\lambda \!\!\;x_{1}$};
\draw (0,-2) circle (0.23);



\draw (-2.6,-3) node [right] {\textbf{Subtree of a normal form $t$}};

\trans{-3.7}{0}{
\draw (8.6,-3) node [right] {\textbf{Corresponding types}};




\draw (8.1,2.4) node {\small $\sSqk{1}\rew\! \ldots \!\rew \sSqk{q}\rew T\posPr{\ra_q}$};

  \draw (10.58,2.19) -- (10.98,1.78);
  
  

  \draw (12.1,2.4) node {\small $\sSqk{1}$}; 
  \draw (11.22,1.78) -- (11.5,2.15);   
  \draw (11.22,1.78) -- (11.65,2.15);
  \draw (11.22,1.78) -- (12,2.15);
  

\draw (8.7,1.6) node {\small $\sSqk{2}\rew\! \ldots \!\rew \sSqk{q}\rew T\posPr{\ra_{q\!-\!1}}$};
  
  \draw [dotted] (11.24,1.43) -- (11.75,0.86);

  

 \draw (13,1.4) node {\small $\sSqk{q\!-\!1}$}; 
  
 \draw (12.02,0.88) -- (12.3,1.2);
 \draw (12.02,0.88) -- (12.45,1.2);
 \draw (12.02,0.88) -- (12.8,1.2);


\draw (10.7,0.7) node {\small $\sSqk{q} \rew T\posPr{\ra_1}$};
 
\draw (12.02,0.52) -- (12.34,0.16);  


 \draw (13.5,0.7) node {$\sSqk{q}$};
 \draw (12.62,0.18) -- (12.9,0.5);
 \draw (12.62,0.18) -- (13.05,0.5);
 \draw (12.62,0.18) -- (13.4,0.5);  
 


\draw (12.25,0) node {$T\posPr{\ra}$};

\draw (12.5,-0.23) --++ (0,-0.44);

\draw (11.5,-0.9) node {\small $\ttC(\ra)(x_p)\rew T\posPr{\ra_{-\!1}}$};

\draw [dotted] (12.5,-1.13) --++ (0,-0.64) ;


\draw (10.2,-2) node {\small $\ttC(\ra_{1\!-\!p})(x_1) \rew \!\ldots\!\rew \ttC(\ra)(x_p)\rew T\posPr{\ra_{-\!p}}$};
}
\end{tikzpicture}}
    \caption{Typing Normal Forms}
    \label{fig:parsNF}
  \end{figure*}
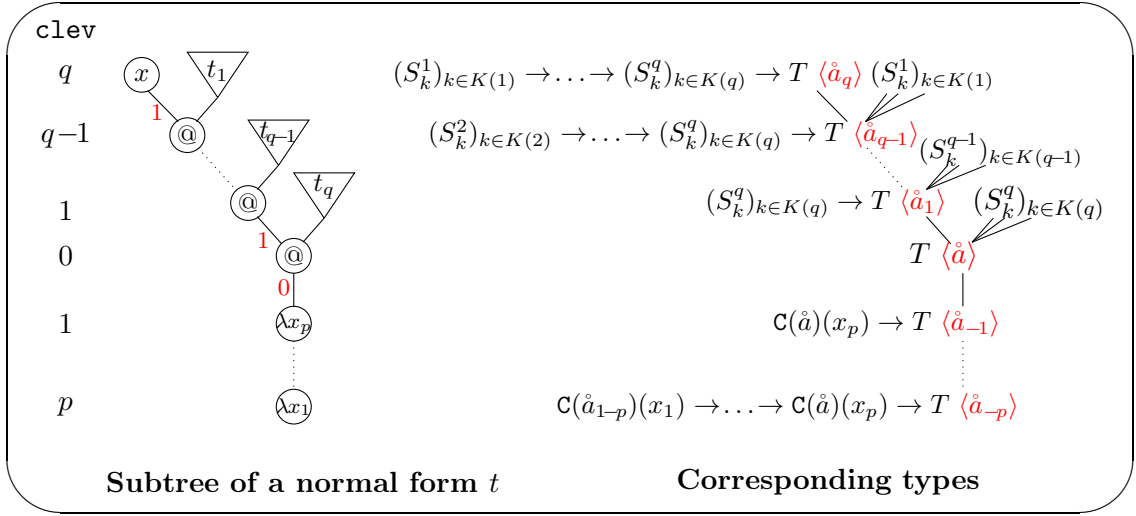

\subsection{Support Candidates}
\label{ss:support-candidates}

Before using the notion of constrained position above, we must devote our attention to the possible forms of the support $A:=\supp{P}$ of a derivation $P$ typing a 001-normal form $t$. This is done by considering the \textit{preorder} $\prec$ defined below, extending the prefix order.

If $\ovla,\ovla'\in\supp{t}$, we write $a\prec a'$ when there exists $a_0$ such that $a_0\leqslant a,~ a_0\leqslant a'$ and $\ad{a}=\ad{a_0}$. Then we observe that if $P$ is a derivation typing $t$, then $\supp{P}$ is downward closed for $\prec$ over $\supp{t}$, meaning by that, for all $a_1$ such that $\ovl{a_1}\in \supp{t}$ and $a_2\in \supp{P}$, then $a_1\prec a_2$ implies $a_1\in \supp{P}$.  Intuitively, $a\prec a'$ holds when $a$ and $a'$ are in the same sub-head normal form of $t$ or if $a'$ is nested in an argument derivation of the sub-head normal form containing $a$.

For instance, $a:=021031 \prec 021037=:a'$ since $a_0:=02103 \leqslant a, a'$ and $\ad{a_0} = \ad{a} = 2$. If $021037\in \supp{P}$, since $7$ is an argument track,  $t(02103)=\arob$ (\ie $a$ points to an $\app$-rule in $P$). Last, $02031$, which is this $\app$-rule left-hand side, should also be in $\supp{P}$, as well as every prefix of $021037$. And we have thus $021031\in \supp{P}$ as expected.

Conversely, a non-empty set $A$ downward closed for $\prec$ such that $\ovl{A}\subseteq \supp{t}$
will be called a \textbf{support candidate} for a derivation typing $t$ and we prove that, for all  support candidates $A$ associated with a normal form $t$, there is indeed a derivation $P$ typing $t$ such that $A=\supp{P}$. This shall be Lemma~\ref{l:NF-are-typable}.

For now, assume that there is a derivation $P$ typing the 001-normal form $t$. Each $a\in A$ is either an unconstrained position or \textit{related} to an unconstrained position, which we denote $\ra$, in the sense that $a$ is intuitively in the same head normal form than $\ra$ (inside $P$). The \textbf{constrain level} $\rdeg(a)$ of $a$ is defined as the graph distance between $a$ and $\ra$.
First, notice that, given $a\in A:=\supp{P}$, 
we have  $\tra=\lambda x_1\ldots x_n.u$, where $u$ is not an abstraction. The integer $n$ is sometimes called the \textbf{order} (as in~\cite{BucciarelliCFS16}) of $\tra$. We say then that $a$ is an order $n$ position. 
\begin{itemize}
  \item If $n\geqslant 1$, we say that $a$ is a \textbf{non-zero position} and we set $\ra=a\cdot 0^n$ (so that $\trra=u$) and $\rdeg(a)=n$.
  \item If $n=0$, we distinguish two subcases, by first defining $\ra$ as the shortest prefix $a_0\leqslant a$ such that $a=a_0\cdot 1^\ell$ (for some $\ell$) and setting $\rdeg(a)=\ell$:\\
    - If $\ell=0$, then we set $\ra=a$ and we say that $a$ is a \textbf{unconstrained position}.\\
    - If $\ell\geqslant 1$, we say that $a$ is a \textbf{partial (zero) position}.
\end{itemize}

As it has been observed above, $a$ is unconstrained when, intuitively, the type of the underlying subterm does not depend on deeper parts of the derivation. If $i\geqslant 0$, we write $\ra_i$ for $\ra\cdot 1^i$, and $\ra_{-i}$ for the rank $i$ prefix of $\ra$ (\eg $\ra=\ra_{\!-\!2}\cdot 0^2$ if $\trra$ is of order $\geqslant 2$). Thus, if $a\in \supp{P}$ and $d:=\rdeg(a)$, then $a=\ra_d$ if $a$ is a partial zero position and $a=\ra_{-d}$ if $a$ is a non-zero position.
More generally, following the discussion beginning \Sec~\ref{s:normal-forms}, we observe that
if $\rdeg(a)=d$, then $\ttT(a)$ is an arrow type $F_1\rew \ldots \rew F_d\rew \ttT(\ra)$, where $\ttT(\ra)$ is an unconstrained type. More precisely, using Fig.~\ref{fig:parsNF} to guide us:
\begin{itemize}
\item When $a$ is a non-zero position, \ie $\ra=a\cdot 0^d$ and $\tra$ is of the form $\lx_{q\!-\!d\!+\!1}\ldots \lx_p.x\,t_1\ldots t_q$ with $\trra=x\,t_1\ldots t_q$, then $F_1=\ttC(a\cdot 0)(x_{q\!-\!d\!+\!1}),\, F_2=\ttC(a\cdot 0^2)(x_{q\!-\!d\!+\!2}),\ldots,\, F_d=\ttC(a\cdot 0^d)(x_p)=\ttC(\ra)(x_p)$.
\item When $a$ is partial \ie $a=\ra\cdot 1^d$ and  $\trra$ is of the form $x\,t_1\ldots t_q$ with $\tra= x\,t_1\ldots t_{p\!-\!d}$, then $\trra=\tra\,t_{p\!-\!d\!+\!1}\ldots t_q$ and
  $F_1:=(k\cdot \ttT(\ra\cdot 1^{d-1} \cdot k))_{k \in \ArgTr^1(a)},\ldots,\,  F_d:=(k\cdot \ttT(\ra\cdot k))_{k \in \ArgTr^k(a)}$ (see end of \Sec~\ref{s:deriv-S} for notation $\ArgTr^i$): thus, $F_i$ is the sequence of types given to the $i$-th argument $t_{p\!-\!d\!+\!i}$ of $\tra$  w.r.t. position $a$.
  \end{itemize}

\subsection{Natural Extensions}
\label{ss:nat-ext}

Let $A$ be a support candidate for the 001-normal form $t$. We want to build a derivation $P$ such that $\supp{P}=A$. For that, we need to find types for every subterm $\tra$ when $a$ ranges over $A$. To obtain this, let us just assign first types at the unconstrained positions in $A$. Thus, let $\rT$ be  a function from $\rA$ to the set of types. We need to extend $\rT$ on $A$ so that we get a correct derivation $P$ typing $t$. We will call the resulting construction a \textit{natural extension}.

First, note that $\rdeg(a)$ may be defined in $A$ as it was for a derivation $P$ in the previous section (except that now, we only have $A$, a support candidate, instead of a whole derivation $P$): we set $a\in \rA$ iff ($a\cdot 0 \notin A$ and there is no $a'$ such that $a=a'\cdot 1$). Moreover, if $a\in A\setminus \rA$, then either there are unique $a'\in \rA$ and $n\geqs 1$ such that $a'=a\cdot 0^n$, or there are unique $a'\in \rA$ and $n\geqs 1$ such that $a=a'\cdot 1^n$. In both cases, we set $\ra=a'$ and $\rdeg(a)=n$. The notation $\ra_{\pm n}$ can easily be redefined.

As we have seen in \Sec~\ref{ss:support-candidates}, we must capture the way \textit{calls} are made in a derivation by a type to other types which are deeper in the derivation.
For that, to each $a\in \mathbb{N}^*$, we attribute an \textit{indeterminate} $X_a$, which shall be a placeholder for $\ttT(a)$ (when $\ttT(a)$ is built). 
Intuitively, $X_a$ calls for $\ttT(a)$, the type given to the subterm at position $a$.
For all $a\in A,\ x\in \TermV$, we set $A_a(x)=\set{a_0 \in A\,|\,a\leqslant a_0,\,t(a)=x,\nexists a_0',\,a\leqslant a_0'\leqslant a_0,\,t(a'_0)=\lx }$ so that we intend to have $\AxP_a(x)=A_a(x)$ (as well as $\trP{a_0}=\code{a_0}$) when $P$ is built.

Combining all the above intuitions, we set, for all $a\in A,\,x\in \TermV,~ \ttE(a)(x)=(\code{a_0}\cdot X_{a_0})_{a_0\in A_a(x)}$ (thus, $\ttE(a)(x)$ calls for the types given to $x$ in $\ax$-rules above $a$). If $a\in A$ is partial or unconstrained, $d=\rdeg(a)$ (\ie $a=\ra\cdot 1^d$) and $1\leqslant i \leqslant d$, we define the sequence $\Rst_i(a)$ by $\Rst_i(a)=(k\cdot X_{\ra_{d-i}\cdot k})_{k\in \ArgTr_A^i(a)}$ with $\ArgTr^i_A$ defined as for derivations by $\ArgTr^i_A=\set{k\geqslant 2\,|\, a_*\cdot 1^{d-i}\cdot  k\in A}$. The intuition is that $\Rst_i(a)$ calls for the types given to the argument of the $i$-th application below $a$.
\begin{itemize}
\item If $a\in A$ is a non zero position, \eg $\tra$ is of the form $\lambda x_1\ldots x_d.\trra$. We then set:
  $$\Cal(a)=\ttE(a\cdot 0)(x_1)\rew\ldots \rew \ttE(\ra)(x_d)\rightarrow \rT(\ra)$$
\item If $a\in A$ is partial, we set:
  $$\Cal(a)=\Rst_1(a)\rew \ldots \rew \Rst_n(a)\rew \rT(\ra)$$
\item If $a\in A$ is unconstrained, we set $\Cal(a)=\rT(a)$. 
\end{itemize}
We then extend $\rT$ (defined  on unconstrained positions only, for the time being) to $A$ by the following coinductive definition: for all $a\in A$, $$\ttT(a)=\Cal(a)[\ttT(a')/X_{a'}]_{a'\in \bbN^*}$$ For all $a\in A$, we define the contexts $\ttC(a)$ by:
$$\ttC(a)(x)=\ttE(a)(x)[\ttT(a')/X_{a'}]_{a'\in A_a(x)}$$

Those definitions are well-founded, because whether $a$ is non zero position or a partial one, every occurrence of an $X_\al$ in $\ttE(a)$ or $\Cal(a)$  is at depth $\geqslant 1$ and the coinduction is \textit{productive}. Finally, let $P$ be the labelled tree whose support is $A$ and such that, for $a\in A$, $P(a)$ is $\ju{\ttC(a)}{\tra: \ttT(a)}$.


\begin{lemma}
\label{l:NF-are-typable}
  The labelled tree $P$ defined above is a derivation.
\end{lemma}


\begin{proof}
  Let $a\in A$. Whether $t(a)$ is $x$, $\lambda{x}$ or $\symbol{64}$, we check the associated rule has been correctly applied.  Roughly, this comes from the fact that the variable $X_{a'}$ is ``on the good track'' (\ie $\code{a'}$) in $\ttE(a)(x)$, as well as in $\Rst_i(a)$, thus allowing to retrieve correct typing rules.
\end{proof}

We then call the derivation $P$ built above the \textbf{natural extension} of the pair $(A,\,\rT)$. Natural extension give all the possible quantitative derivations typing a normal form. For our purpose, they also give:

\begin{lemma}
  \label{lem:NF-uf-typable} 
A normal form $t$ is unforgetfully typable.
\end{lemma}

\begin{proof}
We set $A=\supp {t}$ and $\rT(a)=\tv$ for each unconstrained position (where
$\tv$ is a type variable). Then, the extension $P$ of $(A,\,\rT)$ is an unforgetful derivation typing $t$.
\end{proof}



\subsection{Approximability of Normal Derivations} 
\label{ss:nf-approx}

In this section, we prove that (1) for derivations typing  001-normal forms, it is enough to be quantitative to be valid  and that (2) as a corollary, all 001-normal forms are approximably and unforgetfully typable:
when we truncate a derivation, we must keep all equinecessary positions. Thus, when we approximate $P$ into $\fP$, if we keep $\p\in \bisupp{P}$, we must keep all the bipositions $\p'$ which are equinecessary with $\p$ and we need to prove that this can be finitarily done with any \textit{finite} subset $\oB$ of the bisupport of a derivation typing a normal form.
To prove this, we define the rank of a position as follows:

\begin{definition}
  \label{def:rank-max-width-depth}
  Let $a\in \bbN^*$. The \textbf{rank} of $a$, denoted $\rk{a}$ is defined by
  $\rk{a}=\max(\ad{a},\max(a))$. 
\end{definition}

\noindent Thus, $\rk{a}$ bounds the width and the applicative depth of $a$. 
For instance, $\rk{0^6\cdot 2^5} = 5$ and $\rk{2\cdot 8\cdot 3}=8$.

\begin{lemma}
  \label{lem:q-deriv-NF-approx}
  If $P$ is a quantitative derivation typing a 001-normal form $t$, then $P$ is approximable.
\end{lemma}

\begin{proofsketch}
  Let $P\tri \juCtt$ be a quantitative derivation typing a 001-normal form $t$.
  Let us prove that $P$ is approximable. For that, we will build, 
  for any finite subset $\supo B$ of $\bisupp {P}$,  a finite derivation $\supf P \leqslant P$ containing $\oB$. We will proceed as follows:
\begin{itemize} 
\item We choose a finite support candidate $\supf A\subseteq A$ of $t$ \ie we will discard all positions in $A$ but finitely many.
\item Then, we choose, for each $\ra\in \fA$, a finite approximation $\supf \ttT(\ra)$ of $\ttT(\ra)$. 
\end{itemize} 
The natural extension of $(\supf A,\supf T)$ will be a finite derivation $\supf P\leqslant P$ typing $t$ such that $\oB \subeq \supp{\fP}$. \\

Actually, we define $P_n$, the \textbf{rank $n$ truncation} of $P$ as follows:
\begin{itemize}
\item We define $A_n$ by discarding every position $a\in A$ such that $\ad{a}>n$ or $a$ contains a track $> n$, \ie we set  $A_n=\set{a\in A\,|\, \rk{a}\leqs n}$. Since $t\in \Lamzzu$, $A$ does not have infinite branch of finite applicative depth and thus, $A_n$ is a \textit{finite} set of positions.
\item For each $\ra\in \rA_n$, we define $\rT_n(\ra)$ by discarding every $c\in \supp{\ttT(\ra)}$ such that $\ad{c}>n$ or $c$ has a track $>n$, \ie we restrict define $\rT_n$ so that
  $\supp{\rT_n(\ra)}=\set{c\in \supp{\ttT_n(\ra)}\,|\, \rk{c}\leqs n  }$). Since $\ttT_n(\ra)\in \Types$ (and not in $\Types^{111}-\Types^{001}$), $\rT_n(\ra)$ is a \textit{finite} type.
\end{itemize}
We define now $P_n$ as the natural extension of $(A_n,\, \rT_n)$.  Using the quantitativity of $P$, we may prove then that, for all $\oB\subseteq \bisupp{P}$, there exists a large enough $n$ such that $\oB \subseteq \bisupp{P_n}$. The idea is the following: as we have seen, each biposition $\p$ may ``call'' a chain of deeper bipositions. However, the set of bipositions called by $\p$ is finite and we may define the \textbf{called rank} $\cad{\p}$ of $\p$ as the maximal applicative depth of a biposition called by $\p$. The called rank of $\p$ can be understood as the maximal rank of a biposition which is equinecessary with $\p$. A crucial point is that (1) since $P$ is quantitative and types a normal form, for all $\p\in \bisupp{P}$, the called rank of $\p$ is finite (2) $P_n$ contains all the bipositions of $P$ whose called rank is $\leqs n$ (see Appendix~\ref{a:normal-forms} for a detailed proof of these two claims). Then, since $\oB$ is finite, we define $n$ as $\max(\set{\cad{\p}\,|\,\p \in \oB})$  and we have $\oB\subeq \bisupp{P_n}$. This proves that $P$ is approximable.
\end{proofsketch}

\noindent  Lemmas~\ref{lem:NF-uf-typable} and \ref{lem:q-deriv-NF-approx} give immediately the following proposition:

\begin{proposition}
\label{prop:NF-uf-approx-typable}
Every 001-normal form is approximably and unforgetfully typable in system $\ttS$.
\end{proposition}

\subsection{The Infinitary Subject Expansion Property}
\label{s:infty-expansion}


In \Sec~\ref{ss:inf-subj-red-S}, we defined the derivation $P'$ resulting
from a productive reduction path from any (approximable or not) derivation
$P$. Things do not work so smoothly for subject expansion when we try
to define a good derivation $P$ which expands a derivation $P'$
typing the limit of a productive reduction path. Indeed, approximability play a central role \wrt expansion. Assume that:
\begin{itemize} 
\item $t\bredfty t'$. Say by means of the productive reduction path $t=t_0 
\breda{b_0} t_1\breda{b_1} \ldots t_n \breda{b_n} t_{n+1} \bred \ldots$ with 
$b_n\in \set{0,~ 1,~ 2}^*$ and $\ad{b_n}\longrightarrow \infty$.
\item $P'$ is an approximable derivation of $C'\vdash t':\,T'$.
\item We fix an arbitrary injection $a\mapsto \code{a}$ from $\bbN^*$ to $\bbN\setminus\set{0,1}$.
\end{itemize} 
We want to show that there exists a derivation $P$ concluding with $\ju{C'}{t:T'}$. We will follow closely the ideas presented in \Sec~\ref{ss:typ-inf-nf-informal}, but now we can implement them formally. A complete proof is given in Appendix~\ref{a:expans}.


 Let $\fP' \leqs P'$ be a finite approximant of $P'$. 
Since $\fP'$ is finite, for a large enough $n$, $t'$ can be replaced by $t_n$ inside $\fP'$, according to Lemma~\ref{lem:subject-substitution}: we set $\fP_n=\fP'[t_n/t']$, which is a finite derivation typing $t_n$. But when $t_n$ is typed instead of $t'$, we can perform $n$ steps of $\code{\cdot}$-expansion (starting from $\fP_n$) to obtain a finite derivation $\fP$ typing $t$, that we denote $\Exp(\fP')$ ($\Exp(\fP')$ implicitly depends on $\code{\cdot}$ and the productive path, but we omit them in the notation).

Notice then that the definition of $\Exp(\fP')$  is sound because $\Exp(\fP')$ does not depend on the value of $n$, as long as $n$ is great enough (meaning that no reduction takes place in $\supp{\fP'}$ after rank $n$, which ensures in particular
that $t_k$ and $t'$ coincide on $\supp{\fP'}$), because, when we expand a redex in an untyped part of a derivation, it is as if we perform subject substitution.

Now, let us prove that the set $\scrD := \set{\Exp(\fP')\;|\;  \fP' \in \Approx{P'}}$ (the set of the expansions of the finite approximations of $P'$) is a directed set.
Assume $\fP'_1 \leqs \fP'_2\leqs P'$. Let $n$ big enougth so that no reduction takes place in $\supp{\fP'_2}$ after rank $n$. 
We set $\fP_{i,n}:=\fP'_1[t_n/t']$ for $i=1,2$, so that $\Exp(\fP_i)$ is obtained from $\fP_{i,n}$ after $n$ steps of $\code{\cdot}$-expansion. 
By monotonicity of uniform expansion (Lemma~\ref{lem:red-monot}), we have $\Exp(\fP_1)\leqs \Exp(\fP_2)$. Thus, $\Exp$ is monotonous.
Since $\Approx{P'}$ is directed and $\Exp$ is monotonous, $\scrD$ is directed.

Then, by Theorem~\ref{th:cpo}, we define $P$ as the join of the $\fP$ when  $\fP'$ ranges over $\Approx{P'}$. It is immediate to see that the context and the type concluding $P$ are the same as those of $P'$, because $\fP'$ and $\Exp(\fP')$ have the same contexts and types for all $\fP'\leqs P'$. 
This yields a derivation satisfying the desired properties.

\begin{proposition}
  \label{prop:infinite-subject-expansion}
  Assume $t\bredfty t'$ and $P'\tri \ju{C'}{t':T'}$.\\
  If $P'$ is approximable, then there exists an approximable derivation $P$ such that $P\tri \ju{C'}{t:T'}$.
\end{proposition}

Since infinitary subject reduction and expansion (for productive reduction paths) preserve unforgetful derivations, it yields our main characterization theorem :

\begin{theorem}[Infinitary Normalization Theorem]
  \label{th:charac-WN-S}\mbox{}\\
  Let $t$ be a 001-term. Then $t$ has a (possibly infinite) normal form in $\Lamzzu$ if and only if $t$ is typable by means of an approximable unforgetful derivation, and if and only iff $t$ is hereditarily head normalizing.
\end{theorem}

\begin{proof}
We follow the proof scheme of Fig.~\ref{fig:its-fundamental-diag}.
  \begin{itemize}
    \item If $t$ is typable by means of an approximable and unforgetful derivation, then $t$ is HHN by Proposition~\ref{prop:uf-typ-implies-wn}.
    \item If $t$ is HHN, then $t$ is obviously 001-WN.
    \item If $t$ is 001-WN, let $t'$ be the 001-normal form of $t$. By Lemma~\ref{lem:NF-uf-typable}, there is an approximable and unforgetful derivation $P'$ concluding with $P'\tri \juCtpt$ for some $C$ and $T$. By Proposition~\ref{prop:infinite-subject-expansion}, there is an approximable derivation $P\tri \juCtt$. Since $P'$ is unforgetful, $P$ also is. 
  \end{itemize}
 \end{proof}

\begin{remark}
Theorem~\ref{th:charac-WN-S} proves that hereditary head reduction is asymptotically complete for infinitary weak normalization: it computes the 001-normal form of a 001-term, whenever there is a productive reduction path from this term to a 001-normal form.
\end{remark}

\section{Conclusion}
\label{conclusion}




We have provided a type-theoretic characterization of the set of hereditary head normalizing terms (the terms whose \Bohm\ tree does not contain $\bot$). This characterization is stated in  Theorem~\ref{th:charac-WN-S}. Hereditary head normalization was proved \textit{non}-recursive enumerable by Tatsuta~\cite{Tatsuta08}, and thus, impossible to characterize with an inductive type system.

To obtain this characterization, we had then to build a \textit{coinductive} type system, system $\ttS$,  based on intersection types, which ensure not only soundness (the implication: if a term is typable, then it has some expected semantic property) as in most type systems, but also completeness (the converse implication: if a term has some expected semantic property, then it is typable). More precisely, our characterization uses \textit{non-idempotent} intersection type theory, introduced in \cite{CoppoDV80a} and developped on by Gardner and de Carvalho~\cite{Gardner94,Carvalho18}.

System  $\ttS$ enjoys the main properties of non-idempotent intersection type system: it quantitative (in the finite-case), syntax-directed, relevant and has a simple reduction combinatorics. Its main novel features are:
\begin{itemize}
\item It is infinitary, meaning that types and derivations are \textit{coinductively} generated.
\item  It is endowed a validity condition called \textit{approximability}. An infinitary derivation is approximable when it is the join of its finite \textit{approximations}. Approximability is used to discard \textit{unsound coinductive} derivations.
\item It is based on \textit{sequence types}, which are annotated multisets  allowing \textit{tracking} and retrieving \textit{determinism} (thus, in this respect, system $\ttS$ is rigid). Sequence types are necessary to define approximability.
\end{itemize}

%

This contribution shows that non-idempotent intersection type theory extends to infinitary settings. The characterization presented here has been extented to obtain one of the set of \textit{hereditary permutators}~\cite{VialFSCD19}, another set of $\lam$-terms that Tatsuta~\cite{Tatsuta08b} proved \textit{not} to be characterizable in an inductive way, because it also captures an asymptotic behavior.

These results the way for other semantic characterizations in other infinitary settings, in particular while considering Lévy-Longo trees and Berarducci trees, which provide semantics refining \Bohm\ trees.

It also seems that approximability, which is defined here in terms of directed sets, should admit a categorical generalization and could help define infinitary models in a more generic way.

\bibliographystyle{alpha}
\bibliography{biblio-full.bib}

\newpage

\tableofcontents

\newpage

\appendix

\subsection*{Contents of the appendices.}\mbox{}\label{disc:app-pres}
\begin{itemize}
\item In Appendix~\ref{s:expanding-pi-prime-n}, we give the expansions $\Pi$ and $\Pi_n$ (typing $\cuf$) of the derivations $\Pi'$ and its finite truncations $\Pi'_n$ typing $\fom$, which we studied in \Sec~\ref{ss:typ-inf-nf-informal}.
\item In Appendix~\ref{a:lattices}, we detail the proof of Theorem~\ref{th:cpo} which states that the set of $\ttS$-derivations typing a given $t$  is a directed complete semi-lattice.
\item Appendix~\ref{a:equinecessity}: we detail more the mechanisms of residuation  and equinecessity in \Sec~\ref{a:qres-Shp-formal}, which we introduced  in \Sec~\ref{s:equinecessary-bip}. This help us prove subject reduction and expansion (without hypothesis of quantitativity, contrary to \Sec~\ref{ss:sr-proof}).
In \Sec~\ref{ss:uni-subj-exp}, we prove uniform subject expansion in system $\ttS$ (Proposition~\ref{prop:subj-exp-one-step-S}). 
In \Sec~\ref{a:approx-0-cex}, using equinecessity, we prove that approximabilty cannot be expressed as a predicate that pertains only to the root judgment of a derivation, as claimed  in Remark~\ref{rk:root-approx}.
\item In Appendix~\ref{a:normal-forms}, we explain why every \textit{quantitative} derivation $P$ typing a \textit{normal form} is indeed approximable. For that, we prove the two crucial claims  of   the proof  of Lemma~\ref{lem:q-deriv-NF-approx}: (1) all bipositions of such a derivation $P$ have finite called rank and (2) the bipositions of rank $\leqs n$ define an approximation of $P$.
\item In Appendix~\ref{a:system-R}, we build $\scrR$, the infinitary version of system $\scrRo$ (\Sec~\ref{ss:system-Ro-Klop}) based on intersection types as infinite multisets. This help us formally prove in \Sec~\ref{a:rep-and-dynamics} that approximability cannot be defined when intersection types are multisets (whereas \Sec~\ref{ss:non-determinism} gives only suggestion why it is impossible).
\end{itemize}

\section{Expanding the $\Pi'_n$ and $\Pi'$}
\label{s:expanding-pi-prime-n}

\newcommand{\biglabnode}[3]{  
     \trans{ #1 }{ #2 }{ 
       \draw (0,0) node {\small #3 };
       \draw (0,0) circle (0.35);
     }
}

\begin{figure}
\begin{tikzpicture}

  \trans{-3.5}{8}{
\draw (2,-1) node {$\rho_1=\erewa$};
  
\drawlabnode{2}{0}{$\rew$}

\drawlabnode{3}{1.2}{$\tv$}
\draw (2.86,1.02) -- (2.16,0.16) ;
  
\transh{7}{
\draw (2,-1) node {$\rho_{n+1}=\mult{\rho_k}_{1\leqs k \leqs n} \rew \tv$};
  
  \drawlabnode{2}{0}{$\rew$}

\drawlabnode{3}{1.2}{$\tv$}
\draw (2.86,1.02) -- (2.16,0.16) ;

\drawlabnode{1.3}{1.2}{$\rho_n$}

\draw [dotted] (0,1.2) --++ (1,0);

\draw (1.89,0.183) -- (1.41,1) ;
\drawlabnode{-0.4}{1.2}{$\rho_2$}

\draw (1.83,0.13) -- (-0.28,1);

\drawlabnode{-1.4}{1.2}{$\rho_1$}

\draw (1.81,0.08) -- (-1.27,1);}}

\draw (0.4,6) node {The derivation $\Pi_n$ below is obtained from $\Pi^k_n$ (Fig.\;\ref{fig:expans-by-trunc-1}) after $k$ steps of expansion.};
\draw (0.4,5.5) node {The subderivation $\Psi_n$ types $\Delf$ with $\rho_n$. };


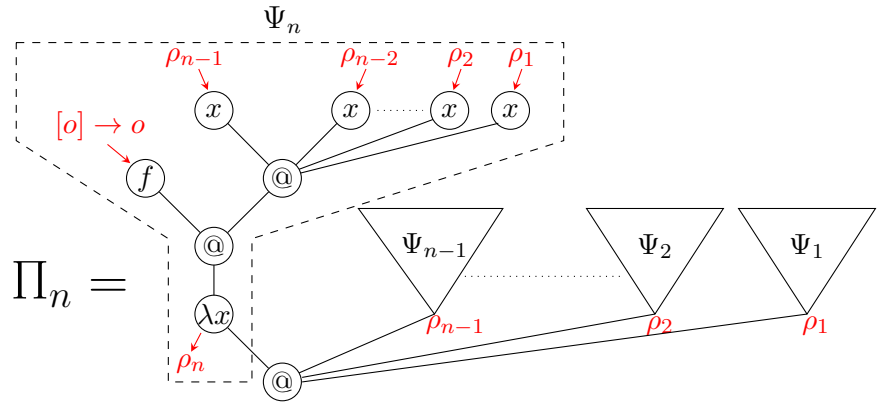
      \draw (0,4.8) node {\large $\Psi_n$};
\draw [dashed] (-3.5,4.5) --++ (7.2,0) --++ (0,-1.3) --++ (-4.1,-1.3) --++ (0,-1.9)--++ (-1.1,0) --++ (0,1.9) --++ (-2,1.3)-- cycle ;
      
      \draw (-2.8,1.3) node {\huge $\Pi_n =$ } ;
      
      \red{
      \draw [->,>=stealth] (-2.3, 3.15) -- (-2,2.9 ) ;
      \draw (-2.4,3.4) node{$[\tv]\rew \tv$} ;
      }


      \red{
        \draw [->,>=stealth] (-1.1,4.15) -- (-0.98, 3.87) ;
        \draw (-1.15,4.3) node{$\rho_{n-1}$} ;
      }

      \red{
        \draw [->,>=stealth] (1.1,4.15) -- (0.98,3.87) ;
        \draw (1.15,4.3) node{$\rho_{n-2}$} ;        
      }
      
      \red{
        \draw [->,>=stealth] (2.35,4.15) -- (2.25,3.87) ;
        \draw (2.35,4.3) node{$\rho_2$} ;
      }

   \red{
        \draw [->,>=stealth] (3.15,4.15) -- (3.05,3.87) ;
        \draw (3.15,4.3) node{$\rho_1$} ;
      }
      
      \red{
      \draw [->,>=stealth] (-1.07,0.65) -- (-1.2,0.4) ;
      \draw (-1.2,0.25) node{$\rho_n$} ;
      }

      \red{
        \draw (1.75,0.75) node[right]{$\rho_{n-1}$};
        \draw (4.65,0.75) node[right]{$\rho_2$} ;
        \draw (6.7,0.75) node[right]{$\rho_1$} ;
      }
      

      \draw (0,0) node{$\arob$} ;
      \draw (0,0) circle (0.25) ;
      \draw (-0.18,0.18) -- (-0.72,0.72) ;
      \draw (-0.9,0.9) node{$\lx$} ;
      \draw (-0.9,0.9) circle (0.25) ;   
      
      \draw (-0.9,1.15) -- (-0.9,1.55) ;
      \draw (-0.9,1.8) node{$\arob$} ;
      \draw (-0.9,1.8) circle (0.25) ;

      \draw (-1.08,1.98) -- (-1.62,2.52) ;
      \draw (-1.8,2.7) node{$f$} ;
      \draw (-1.8,2.7) circle (0.25) ;
      
      \draw (-0.72,1.98) -- (-0.18,2.52) ;
      \draw (0,2.7) node{$\arob$} ;
      \draw (0,2.7) circle (0.25) ; 
      
      \draw (-0.18,2.88) -- (-0.72,3.43) ;
      \draw (-0.9,3.6) node{$x$} ;
      \draw (-0.9,3.6) circle (0.25) ;

      \draw (0.18,2.88) -- (0.72,3.43) ;
      \draw (0.9,3.6) node{$x$} ;
      \draw (0.9, 3.6) circle (0.25) ;

      \draw [dotted] (1.25,3.6) -- (1.85,3.6) ;

      \draw (0.22,2.82) -- (1.98,3.48) ; 
      \draw (2.2,3.6) node{$x$} ;
      \draw (2.2,3.6) circle (0.25) ;

      \draw (0.24,2.77) -- (2.84,3.43) ; 
      \draw (3,3.6) node{$x$} ;
      \draw (3,3.6) circle (0.25) ;

      \draw (0.22,0.12) -- (2,0.9) ;
      \draw (2,0.9) -- (1,2.3) -- (2.9,2.3) -- cycle ;
      \draw (2,1.8) node{$\Psi_{n-1}$} ;

      \draw [dotted] (2.4,1.4) -- (4.5,1.4) ;
      
      \draw (0.26,0.06) -- (4.9,0.9) ;
      \draw (4.9,0.9) -- (5.8,2.3) -- (4,2.3) -- cycle ; 
      \draw (4.9,1.8) node{$\Psi_2$} ;

      \draw (0.26,0) -- (6.9,0.9) ;
      \draw (6.9,0.9) -- (7.8,2.3) -- (6,2.3) -- cycle ; 
      \draw (6.9,1.8) node{$\Psi_1$} ;
      
\draw (0.4,-0.7) node {\textit{Remark:} $\Pi_1$ is a bit different, it just assigns $\erewa$ to $f$.};
      
      \transv{-3.9}{
        \draw (0.4,2.1) node {\parbox{12cm}{
            \begin{center}
              Intuitively, the family $(\rho_n)_{n\geqs 1}$ is ``directed'' ($\rho_n$ is a truncation of $\rho_{n+1}$).\\
                The infinite type $\rho$ below is the ``join'' of this family.
                \end{center}
            }
        };
        
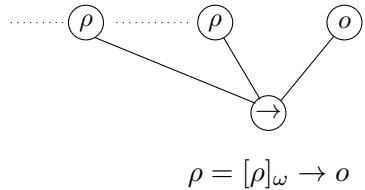
\draw (2,-0.8) node {$\rho=\mult{\rho}_{\om} \rew \tv$};
  
  \drawlabnode{2}{0}{$\rew$}

\drawlabnode{3}{1.2}{$\tv$}
\draw (2.86,1.02) -- (2.16,0.16) ;

\drawlabnode{1.3}{1.2}{$\rho$}

\draw [dotted] (0,1.2) --++ (1,0);

\draw (1.89,0.183) -- (1.41,1) ;
\drawlabnode{-0.4}{1.2}{$\rho$}

\draw (1.83,0.13) -- (-0.28,1);

\draw [dotted] (-0.7,1.2) --++(-0.7,0) ;

}

\transv{-11.9}{
 \draw (0.4,6) node {\parbox{12cm}{
            \begin{center}
              Since  $(\rho_n)_{n\geqs 1}$ is ``directed'',  $(\Psi_n)_{n\geqs 0}$ and then, $(\Phi_n)_{n\geqs 0}$ are also directed ($\Psi_n$ and $\Phi_n$ are truncations of $\Psi_{n+1}$ and $\Phi_{n+1}$ resp.).\\
Their \resp joins are then $\Psi$ and $\Pi$ below.
            \end{center}
            }
        };


      \draw (0,4.8) node {\large $\Psi$};
\draw [dashed] (-3.5,4.5) --++ (7.2,0) --++ (0,-1.3) --++ (-4.1,-1.3) --++ (0,-1.9)--++ (-1.1,0) --++ (0,1.9) --++ (-2,1.3)-- cycle ;
      
      \draw (-2.8,1.3) node {\huge $\Pi =$ } ;
      
      \red{
      \draw [->,>=stealth] (-2.3, 3.15) -- (-2,2.9 ) ;
      \draw (-2.4,3.4) node{$[\tv]\rew \tv$} ;
      }


      \red{
        \draw [->,>=stealth] (-1.1,4.15) -- (-0.98, 3.87) ;
        \draw (-1.15,4.3) node{$\rho$} ;
      }

      \red{
        \draw [->,>=stealth] (1.1,4.15) -- (0.98,3.87) ;
        \draw (1.15,4.3) node{$\rho$} ;        
      }
      
      \red{
        \draw [->,>=stealth] (2.35,4.15) -- (2.25,3.87) ;
        \draw (2.35,4.3) node{$\rho$} ;
      }

      
      \red{
      \draw [->,>=stealth] (-1.07,0.65) -- (-1.2,0.4) ;
      \draw (-1.2,0.25) node{$\rho$} ;
      }

      \red{
        \draw (1.8,0.75) node[right]{$\rho$};
        \draw (4.7,0.75) node[right]{$\rho$} ;
      }
      

      \draw (0,0) node{$\arob$} ;
      \draw (0,0) circle (0.25) ;
      \draw (-0.18,0.18) -- (-0.72,0.72) ;
      \draw (-0.9,0.9) node{$\lx$} ;
      \draw (-0.9,0.9) circle (0.25) ;   
      
      \draw (-0.9,1.15) -- (-0.9,1.55) ;
      \draw (-0.9,1.8) node{$\arob$} ;
      \draw (-0.9,1.8) circle (0.25) ;

      \draw (-1.08,1.98) -- (-1.62,2.52) ;
      \draw (-1.8,2.7) node{$f$} ;
      \draw (-1.8,2.7) circle (0.25) ;
      
      \draw (-0.72,1.98) -- (-0.18,2.52) ;
      \draw (0,2.7) node{$\arob$} ;
      \draw (0,2.7) circle (0.25) ; 
      
      \draw (-0.18,2.88) -- (-0.72,3.43) ;
      \draw (-0.9,3.6) node{$x$} ;
      \draw (-0.9,3.6) circle (0.25) ;

      \draw (0.18,2.88) -- (0.72,3.43) ;
      \draw (0.9,3.6) node{$x$} ;
      \draw (0.9, 3.6) circle (0.25) ;

      \draw [dotted] (1.25,3.6) -- (1.85,3.6) ;

      \draw (0.22,2.82) -- (1.98,3.48) ; 
      \draw (2.2,3.6) node{$x$} ;
      \draw (2.2,3.6) circle (0.25) ;


      \draw [dotted] (2.5,3.6) -- (3.1,3.6) ; 
       
      \draw (0.22,0.12) -- (2,0.9) ;
      \draw (2,0.9) -- (1,2.3) -- (2.9,2.3) -- (2,0.9) ;
      \draw (2,1.8) node{$\Psi$} ;

      \draw [dotted] (2.4,1.4) -- (4.5,1.4) ;
      
      \draw (0.22,0.12) -- (4.9,0.9) ;
      \draw (4.9,0.9) -- (5.8,2.3) -- (4,2.3) -- (4.9,0.9) ; 
      \draw (4.9,1.8) node{$\Psi$} ;

      \draw [dotted] (5.4,1.4) -- (6,1.4) ;

  }
      \end{tikzpicture}
\caption{Expanding the $\Pi^k_n$ and $\Pi'$}
\label{fig:expanding-the-Pi^k_n$}
\end{figure}
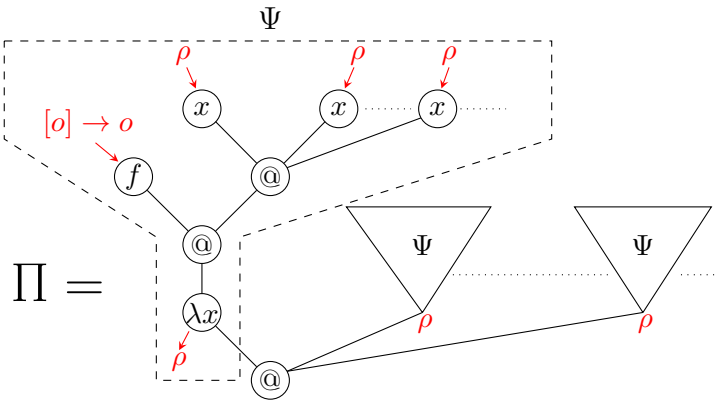

In this Appendix, we build the derivation $\Pi_n$ which are discussed in \Sec~\ref{ss:typ-inf-nf-informal} and \ref{ss:degenerate}

We define a family $(\rho_n)_{n\geqs 1}$ by induction:
  \begin{itemize}
    \item $\rho_1=\erewa$
    \item $\rho_{n+1}=\mult{\rho_k}_{1\leqs k \leqs n} \rew \tv$
\end{itemize}
  Thus, a term $t$ typed with $\rho_{n+1}$ can be fed with an argument $u$ typed with $\rho_1,\ldots,\rho_{n}$ to give the term $t\,u$ of type $\tv$.

We also set, for all $n\geqs 1$,  $\Gam_n=x:(\marewa_{n-1}+\erewa)$.
  
  
  We set:
\[\Psi_1=
\infer{
\infer{
  \infer{\phd}{\jufera}}{
    \ju{f:\merewa}{f(x\,x):\tv}
}}{
  \ju{f:\merewa}{\Delf:\rho_1}
  }
\]
and, for all $k\geqs 2$:
\[\Psi_k=
\infer{
\infer{
  \infer{\phd}{\jufara} \\ 
 \infer{\infer{\phd}{\ju{x:\mult{\rho_{k-1}}}{x:\rho_k}} \\
      \big( \infer{\phd}{\ju{x:\mult{\rho_i}}{x:\rho_i}}\big)_{1\leqs i <k-2}  }{
      \ju{x:\mult{\rho_i}_{1\leqs i \leqs k-1}}{x\,x:\tv}}
  }{
   \ju{f:\marewa;x:\mult{\rho_i}_{1\leqs i \leqs k}}{f(x\,x):\tv}
}}{
   \ju{f:\marewa}{\Delf:\rho_{k}}
  }
\]
We may then define:
\[\Pi_1=\infer{\Psi_1\tri \ju{f:\merewa}{\Delf:\merewa} }{\Gam_1:\cuf}
\]
and, for all $n\geqs 2$:
\[
\Pi_n=
\infer{\Psi_n\tri \ju{f:\marewa}{\Delf:\rho_n} \hspace{0.5cm} (\Psi_k \tri \ju{f:\mult{\mult{\ldots}\rew \tv}}{\Delf:\rho_k})_{1\leqs k\leqs n-1} }{\ju{\Gam_n}{\cuf:\tv}}
\]

Fig.~\ref{fig:expanding-the-Pi^k_n$} gives a more graphic representation of the $\Pi_n$.  By induction, on $n$, we prove that:

\begin{lemma*}
For all $n\geqs 1$, the derivation $\Pi_n$ is obtained from $\Pi^n_n$ by $n$ steps of subject expansion.
\end{lemma*}

The last part of this appendix is more informal. We explain how the derivation $\Pi'$ typing $\fom$ can be expanded into a derivation $\Pi$ typing $\cuf$ by taking the infinite reduction sequence $\cuf \rew^\infty \fom$ backwards. The derivations $\Pi_n$ play a key role and $\Pi$ will be obtained by taking their join them.  First, we observe that, for all $k<n$:
\begin{itemize}
\item The type $\rho_k$ is a truncation of $\rho_n$ since 
  $\mult{\rho_i}_{1\leqs i < k}$ is obviously a truncation of $\mult{\rho_i}_{1\leqs i < n}$  (so that $\rho_k=\mult{\rho_i}_{1\leqs i < k}\rew \tv$ is a truncation $\rho_n=\mult{\rho_i}_{1\leqs i < n}\rew \tv$). Thus, the family $(\rho_n)$ is ``directed''.
\item The derivation $\Psi_k$ is a truncation of $\Psi_n$ thanks to the first observation.
\item The derivation $\Pi_k$ is a truncation of $\Pi_n$.
\end{itemize}
Those observations are corroborated by drawing $\rho_n,\ \Psi_n,\ \Pi_n$ for small values of $n$.
We then observe that the type $\rho$ coinductively defined by $\rho=\rhoom\rew \tv$ is the ``join'' of the types $\rho_n$ (the family $(\rho_n)$ is directed and every finite truncation of $\rho$ is ``included'' in a $\rho_n$ for some great enough $n$).
We define
\[\Psi=
\infer{
\infer{
  \infer{\phd}{\jufara} \\ 
 \infer{\infer{\phd}{\ju{x:\mult{\rho}}{x:\rho}} \\
      \big( \infer{\phd}{\ju{x:\mult{\rho}}{x:\rho}}\big)_{\om}  }{
      \ju{x:\mult{\rho}_\om}{x\,x:\tv}}
  }{
   \ju{f:\marewa;x:\mult{\rho}_\om}{f(x\,x):\tv}
}}{
   \ju{f:\marewa}{\Delf:\rho}
  }
\]
Then $\Psi$ is the ``infinitary join'' of the $\Psi_n$, notably because the $\rho$ is the infinitary join of the $\rho_n$.
\[
\Pi=
\infer{\Psi_n\tri \ju{f:\marewa}{\Delf:\rho} \hspace{0.5cm} (\Psi_k \tri \ju{f:\marewa}{\Delf:\rho})_{\om} }{\ju{\Gam}{\cuf:\tv}}\]
We notice likewise that, intuitively, $\Pi$ is the ``infinitary join'' of the $\Pi_n$, and $\Pi$ conclude, as expected, with $\ju{\Gam}{\cuf:\tv}$.

\ignore{
\subsection*{An example of Unsound Derivation}

Using the type $\rho$, we set:
\[\Pi_{\Del}:=
\infer{
\infer{
\infer{\phd}{\axxRo{\rho}}\hspace{0.5cm}
\bigg( \infer{\phd}{\axxRo{\rho}}\bigg)_\om}{
\ju{x:\rhoom}{x\,x:\tv}
}}{
\ju{}{\Del:\rho}
  }
\]
Then $\Om$ can be typed with $\tv$:
\[\Pi_\Om :=
\infer{\Pi_\Del \tri \ju{}{\Del:\rho}\hspace{1cm}
  (\Pi_\Del \tri \ju{}{\Del:\rho})_\om }
      {\ju{}{\Om:\,}
        }
\]
}


\section{Lattices of (finite or not) approximations}
\label{a:lattices}

\chhp{
In this Appendix, we generalize Theorem~\ref{th:cpo} to $\ttShp$ and we prove it, \ie we show that the set of $\ttShp$-derivations typing a given term $t$ is a d.c.p.o.}{
  In this Appendix, we prove Theorem~\ref{th:cpo}, \ie we show that the set of $\ttS$-derivations typing a given term $t$ is a d.c.p.o.
}

\begin{notation*}
  Let $P$ be a $\ttS$-derivation. The set of approximations of $P$ is denoted $\Appfty{P}$ and the set of finite approximations of $P$ is denoted $\Approx{P}$.
\end{notation*}

\begin{definition}
\label{def:app:tail-head-arrow}
If $U=F\rew T$, we set $\Tl(U)=F$ and $\Hd(U)=T$ (\textbf{tail} and \textbf{head}).  
\end{definition}

\noindent We define now approximations for types, contexts and context sequences.

\begin{definition}[Approximations of context sequences]\mbox{}
  \label{def:approx-contexts-etc}
  \begin{itemize}
  \item Let $T$ and $T'$ two $\ttS$-types. We write $T \leqs T'$ if $\supp{T}\subeq \supp{T'}$ and for all $a\in \supp{T}$, \chhp{$T(a)\leqs T'(a)$}{$T(a)=T'(a)$}.
  \item Let $\sSk$ and $\sSpkp$. We write $\sSk \leqs \sSpkp$ if $K\subeq K'$ and for all $k\in K$, $S_k\leqs S_k'$.
  \item Let $C$ and $C'$ two $\ttS$-contexts. We write $C \leqs C'$ if, for all $x\in \TermV$, $C(x)\leqfty C'(x)$.  
\item 
  Let $(C_j)_{\jJ}$ and $(C'_j)_{\jJ'}$ two context families. We write
  $(C_j)_{\jJ}\leqfty  (C'_j)_{\jJ'}$ if $J\subeq J'$ and, for all $\jJ$, $C_j\leqfty C'_j$.
\end{itemize}
\end{definition}

We say a family  $(C_j)_{\jJ}$ of contexts is a \textbf{compatible context family} when $\uplus_{\jJ} C_j$ is defined (no track conflict). 
We easily prove:

\begin{lemma}
\label{lem:Shp-lattice-types-contexts}
In system $\ttS$, the sets of types, sequence types, contexts and compatible context families are c.p.o. 
\end{lemma}

Moreover, to prove that the set of $\ttS$-derivations typing a same term
is also a c.p.o., we need:

\begin{lemma}
\label{lem:prove-lattice} \mbox{}
\begin{itemize}
\item If $(k_i\cdot T_i)_\iI$ is a direct family of singleton sequence types, then, for all $\iI$, $k_i=k$ for some $k\in \Nmzo$, $(T_i)_\iI$ is directed and $\sup_\iI (k_i\cdot T_i)=(k\cdot \sup_\iI T_i)$.
\item Let $(T_i)_\iI$, $(T'_i)_\iI$ and  $(F_i)_\iI$ a directed families of $\ttS$-types, of $\ttS$-arrow types and of $\ttS$-sequence types respectively.
\begin{itemize}
\item If $\forall \iI$, $\Hd(T'_i)=T_i$, then $\Hd(\sup_\iI T'(i))=\sup_\iI T(i)$
\item If $\forall \iI$, $\Tl(T'_i)=F_i$, then $\Tl(\sup_\iI T'(i))=\sup_\iI F_i$.
\end{itemize}
\item If $(S_i)_\iI$ and $(T_i)_\iI$ are directed families such that $(S_i,T_i)\in \PPP_{d_i}$ (for some $d_i$) for all $\iI$, $S=\sup_\iI S_i$, and $T=\sup_\iI T_i$, then $(S,T)\in \PPP_d$ for $d=\sup_\iI d_i \in \bbN\cup \set{\infty}$.  
\item Let $(C_i)_\iI$ be a directed sequence of contexts and $x\in \TermV$. Then have $(\sup_\iI C_i)(x)=\sup_\iI (C_i(x))$.
\ighp{\item Let $((C^i_k)_{\kK_i})_\iI$ a directed family of compatible context families. Then $(C_k)_\kK$, the supremum of this family, is a sequence of compatible context such that $K=\cup_\iI K_i$, for all $k\in K$, $C_k=\sup_{\set{i\in I\,|\, k\in K_i}} C^i_k$ and $\uplus_\kK C_k = \sup_\iI (\uplus_{\kK_i} C^i_k)$.}
\end{itemize}
\end{lemma}

\subsection{Meets and Joins of Directed Derivations Families}

\begin{lemma}
\label{lem:sup-deriv}
Let $(P_i)_{i\in I}$ be a directed (and non-empty) family of $\ttS$-derivations typing the same term $t$.\\
We define $P$ by $\bisupp {P}=\cup_{\iI} \bisupp {P_i}$ and \chhp{$P(\p)=\sup_\iI P_i(\p)$}{$P(\p)=P_i(\p)$ for any $\iI$ such that $\p\in \bisupp{P_i}$ ($P_i(\p)$ does not depend on $i$)}.\\
Then $P$ is a $\ttS$-derivation and we write $P=\sup_{\iI} P_i$.\\
If $I$ is finite and all the $P_i$ are, then $P$ is also finite.
\end{lemma}

\begin{proof} 
  We set $A_i=\supp{P_i}$, $B_i=\bisupp{P_i}$ for all $\iI$, $A=\cup_\iI A_i$, $B=\bisupp{P_i}$, so that $B=\dom{P}$.
  For all $a\in A$, let $I_a$ denote the set $\set{i\in I\,|\, a\in \supp{P_i}}$.  We write $\ttT^i$ and $\ttC^i$ instead of $\ttT^{P_i}$ and $\ttC^{P_i}$ and so on.

Let us check now that $P$ is a correct $\ttS$-derivation (such that $A=\supp{P}$): this is actually enough to prove the whole statement since, if $P$ is correct, any derivation $P'\geqs P_i$ for all $\iI$ will clearly verify $P'\geqs P$ by definition of $\leqs$. 

For all $a\in A$, let $\ttC(a)$ denote the restriction of $P(a)$ on  $\TermV\times \bbN^*$ and let $\ttT(a)$ denote the restriction of $P$ on $\bbN^*$.
Thus, we have $\ttT(a)=\sup_{\iIa} \ttT^i(a)$ and $\ttC(a)=\sup_{\iIa} \ttC^i(a)$.
 By Lemma~\ref{lem:Shp-lattice-types-contexts}, for all $a\in A$, $\ttT(a)$ and $\ttC(a)$ are a $\ttS$-type and context respectively.
 
For any $a\in A$ such that $t(a)=\arob$, we set $\AT(i)=\set{k\geqs 2\,| a\cdot k\in A}$, \ie  $\AT(i)=\cup_{\iI}\ArgTr^1_{P_i} (a)$  ($\AT(i)$ is simply $\ArgTr^1_P(a)$, but this latter notation is not proved to be licit yet, as long as $P$ is not proved to be a correct $\ttS$-derivation!). 

We prove the correctness of $P$ according to which constructor is $t(a)$, for any $a\in A$.
\begin{itemize}
\item Case $t(a)=x$. Let $i,i'\in I$ such that $a\in \supp{P_i}$ and $P_i\leqs P_{i'}$. Thus, $P_i(a)$ and $P_{i'}(a)$ have the form:
$$\begin{array}{c}
P_{i_0}(a) = \ju{x:k_{i_0}\cdot \ttT^{i_0}(a)}{x:\ttT^{i_0}(a)}\sep (i_0\in \set{i,i'}) \\
\end{array}$$
with $k_{i_0} =\tttr^{i_0}(a)$ and  $\ttC^{i_0}(a)=x:k_i\cdot \ttT^{i_0}(a)$. Since $P_i \leqs P_{i'}$, $k_i=k_i'$. \\
By Lemma~\ref{lem:prove-lattice}, $\ttC^i(a)=\sup_{\iI} (\tttr^i(a)\ct \ttT^i(a))=k\cdot \tt$, where $k$ is defined by $\tttr^i(a)$ for any $i\in I$ such that $a\in \supp{P}$. Thus, $P(a)=\ju{x:k\ct \ttT(a)}{x:\ttT(a)}$, \ie $a$ is a correct axiom node in $P$.\\
\item Case $t(a)=\lx$. Then $\tra=\lx.u$ for some $u$. Let $i,i'\in I$ such that $a\in \supp{P_i}$ and $P_i\leqs P_{i'}$.\chhp{ First note that one of the following cases holds (a) $\ttT^i(a)$ is an arrow type (this implies that 
  $\ttT^i(a)(0)=\rew$ and that $P_i$ is an $(\abs)$-node) (b) $\ttT^i(a)=\pcons_{d_i}$ (for $d_i\in \bbNfty$) (and $a$ is an $(\tthp_{d_i})$-node in $P_i$).}{ Thus,
    $\ttT^i(a)$ is an arrow type (this implies that 
    $\ttT^i(a)(0)=\rew$ and that $P_i$ is an $(\abs)$-node).
  }
Since $P_i\leqs P_{i'}$, we also have $\ttT^{i'}(a)(0)=\rew$. 
Thus, $P_i(a)$ and $P_{i'}(a)$ have the form:
$$P_{i_0}(a)= \infer[\abs]{\ju{\ttC^{i_0}(a\cdot 0)}{u:\ttT^{i_0}(a\cdot 0)}}{\ju{\ttC^{i_0}(a)}{\lx.u:\ttC^{\io}(a\cdot 0)(x) \rew \ttT^{\io}(a\cdot 0)}} \sep  (i_0\in \set{i,i'})$$
with $\ttC^{i_0}(a\cdot 0)=\ttC^{i_0},x:\ttC^{\io}(a\cdot 0)$, $\ttT^{i_0}(a)=\ttC^{\io}(a\cdot 0) \rew \ttT^{i_0}(a\cdot 0)$. Since $P_{i}\leqs P_{i'}$, we have
$\ttC^{i}(a)\leqs \ttC^{i'}(a)$,
$\ttT^{i}(a\cdot 0)\leqs \ttT^{i'}(a\cdot 0)$ and 
$\ttC^{i}(a\cdot 0))\leqs \ttC^{i'}(a\cdot 0)(x)$.
\\
By Lemma~\ref{lem:prove-lattice}, the above equalities implies that $\ttC(a\cdot 0)=\ttC(a),x:\ttC(a\cdot 0)(x)$ and $\ttT(a)=\ttC(a\cdot 0)\rew \ttT(a\cdot 0)$. Thus, $a$ is a correct $\abs$-node in $P$. 
\ighp{\item[(b)] $\ttT^{i}(a)=\pcons_{d_i}$. Since $P_i\leqs P_{i'}$, $T^{i'}(a)=\pcons_{d'}$ for some $d'_i\in \bbNfty$ with $d_i\leqs d'_i$ and $P_i(a)$, $P_{i'}(a)$ have the form:
$$
P(a)=\infer[\tthp_{d_\io}]{\ju{\ttC(a\cdot 0)}{u:\ttT(a\cdot 0)}\sep\msep
(\ttC^\io(a\cdot 0)(x),\ttT^\io(a\cdot 0)) \in \PPP_{d_\io}}{\ju{\ttC(a)}{\lx.u:\pcons_{d_\io}}}
$$
with $\ttT^\io(a)=\pcons_{d_\io}$, $\ttC^\io(a\cdot 0)=\ttC^\io(a),x:\ttC^\io(a\cdot 0)(x)$. Since $P_i\leqs P_{i'}$, we have $\pcons_{d_i}\leqs \pcons_{d_{i'}}$, $\ttC^i(a)\leqs \ttC^{i'}(a)$ and $\ttC^i(a\cdot 0)(x)\leqs \ttC^{i'}(a\cdot 0)$.\\
By Lemma~\ref{lem:prove-lattice}, the above equalities implies that $\ttC(a\cdot 0)=\ttC(a),x:\ttC(a\cdot 0)(x)$, $(\ttC(a\cdot 0)(x),\ttT(a\cdot 0)\in \PPP_d$ for some $d$ and $\ttT(a)=\pcons_d$ (for this same $d$). Thus, $a$ is a correct $\tthp_d$-node in $P$.}
\item Case $t(a)=\arob$:  then $\tra=u\,v$ for some $u$ and $v$. Moreover, $P_i(a)$ and $P_{i'}(a)$ have the form:
$$
P_\io(a)=\infer[\app]{\ju{\ttC^\io(a\cdot 1)}{u:\ttT^\io(a\cdot 1)}
\sep (\ju{\ttC^\io(a\cdot k)}{v:\ttT^\io(a\cdot k)}\msep   \trck{k})_{k\in \ArgTr^1_{P_\io}(a)}}{\ju{\ttC^\io(a)}{u\,v:\ttT^\io(a)}}
$$
where 
$\Hd(\ttT^\io(a\cdot 1))=\ttT^\io(a)$, $\Tl(\ttT^\io(a \cdot 1))=(k \cdot \ttT^\io(a\cdot k))_{k \in \ArgTr^1_{P_\io}(a)}$ and $\ttC^\io(a)=\uplus_{k\in \set{1}\cup \ArgTr^1_{P_\io}} \ttC^\io(a\cdot k)$. Since $P_i\leqs P_{i'}$, $\ttT^i(a \cdot 1) \leqs \ttT^{i'}(a\cdot 1)$ and $(\ttC^i(a\cdot k))_{\set{1}\cup\ArgTr^1_{P_i}(a)}\leqs (\ttC^{i'}(a\cdot k))_{\set{1}\cup\ArgTr^1_{P_{i'}}(a)}$. \\
By Lemma~\ref{lem:prove-lattice}, the above equalities implies that $\Hd(\ttT(a\cdot 1))=\ttT(a)$, $\Tl(\ttT(a\cdot 1))=\uplus_{k\in \AT(a)}(k\cdot \ttT(a\cdot k))$  and  $\ttC(a)=\uplus_{k\in \set{1}\cup \AT(a)} \ttC(a\cdot k)$.
 Thus, $a$ is a correct $\app$-node in $P$. 
\end{itemize} 
\end{proof}

\begin{lemma}
\label{lem:inf-deriv}
Let $(P_i)_{i\in I}$ be a non-empty family of $\ttS$-derivations typing the same term $t$, such that $\forall i,\,j \in I,~ \exists P\in \Deriv,~ P_i,\,P_j\leqfty P$.\\
We define $P$ by $\bisupp{P}=\cap_{\iI} \bisupp {P_i}$\chhp{ and $P(\p)=\inf_{\iI} P_i(\p)$ for any $\iI$}{and $P(\p)=P_i(\p)$ for any $\iI$ ($P_i(\p)$ does not depend on $i$)}.\\
Then, $P$ defines a correct  derivation (which is finite if one of the $P_i$ is finite). We write $P=\inf_{\iI} P_i$.
\end{lemma}

\begin{proof}
We prove this in the same fashion as Lemma~\ref{lem:sup-deriv}.
\end{proof}

The previous lemmas define the join and the meet of derivations (under the same derivation) as their set-theoretic union and intersection. More precisely, the statement below\ighp{ (valid for system $\ttS$)} entails Theorem~\ref{th:cpo}:


\begin{theorem*}
 The set of $\ttS$-derivations typing a same term $t$ endowed with $\leqfty$ is a directed complete semi-lattice.
  \begin{itemize}
  \item If $D$ is a directed set of derivations typing $t$:
\begin{itemize}
  \item The \textbf{join} $\sup D$ of $D$ is the function $P$ defined by $\dom{P}=\cup_{P_*\in D} \bisupp{P_*}$ and $P(\p) = \sup_{P_* \in D}P_*(\p)$, which also is a derivation.
  \item  The \textbf{meet} $\inf D$ of $D$ is the function $P$ defined by $\dom{P}=\cap_{P_*\in D} \bisupp{P_*}$ and $P(\p) = \inf_{P_*\in D}(\p)$, which also is a derivation.
  \end{itemize}
  \item If $P$ is a $\ttShp$-derivation typing $t$, $\Appfty{P}$ is a complete lattice and $\Approx{P}$ is a lattice.
\end{itemize}
\end{theorem*}

\section{Reduction and Approximability}
\label{a:equinecessity}

In this Appendix, (1) we prove 
Observation~\ref{obs:red-to-right} (every biposition is equincessary with a residuable right-position), which is crucial to ensure
that approximability in system $\ttS$ is stable under reduction and expansion (Lemma~\ref{lem:approx-red}) (2) we release the hypothesis of quantitativity and we prove that subject reduction and expansion are monotonic   (3) we prove the infinitary subject expansion property and at last (4) we prove that approximability cannot be defined only in terms of the concluding judgment of a derivation, as announced in Remark~\ref{rk:root-approx}.

%
%
%

\noindent In this appendix, right bipositions play a particular role, so we set:

\begin{notation*}
Let $P$ be a $\ttS$-derivation.   The set of \textbf{right bipositions} in $P$ is denoted $\bisuppR{P}$ (\ie $\bisuppR{P}=\bisupp{P}\cap \bbN^* \times \bbN^*$).
\end{notation*}

\subsection{Quasi-Residuation in System $\ttS$}
\label{a:qres-Shp-formal}

In this section, we define the notion of quasi-residuation in system $\ttS$. This extends \Sec~\ref{ss:one-step-sr-se}.\\


\noindent \textbf{Hypotheses.} 
For the remainder of this section, we assume that $P\tri \juCtt$ is a   $\ttS$-derivation and $\trb=\lxrs$, $t\breda{b} t'$, so that $\tprb=\rsx$.  We follow Fig.~\ref{fig:SR-residuals} and we reuse some notations of \Sec~\ref{ss:one-step-sr-se}, \ref{ss:non-determinism} and \Sec~\ref{ss:sr-proof}\ighp{, which do not need to be changed for system $\ttS$, including $\RepPpb$, $\TrPl{a}$}.  
We set, for all $a\in \supp{P}$, $\ttX(a)=\set{\trP{\al}~|~ \al \in \AxP_{a\cdot 10}(x)}$. We have $\ttX(a)\subeq \TrPl{a}$, and if $P$ is quantitative, $\ttX(a)=\TrPl{a}$ holds.

For all $k\in \ttX(a)$, we also denote by $a_k$ the unique $\al\in \bbN^*$ such that $a\cdot 10 \cdot \al\in \Axl{a}$ and $\trP{a\cdot10\cdot \al}=k$, \ie $a\cdot 10\cdot a_k$ is the unique axiom assigning track $k$ to $x$ above $a$ (note that $a_k$ implicitly depends on $P$).

\subsection*{Residuation} 
Let $\al\in \supp{P}$, we define $\Res_b(\al)$ case-wise.
\begin{itemize}
\item 
  If $\ovlal \not\geqslant b$, $\Res_b(\al)=\set{\al}$.
\item If $\al=a\cdot 1\cdot 0\cdot \al_0$ for some $a\in \RepPb$ and $t(a)\neq x$, $\Res_b(a)=a\cdot \al_0$
\item If $\al= a \cdot k\cdot \al_0$ for some $a \in \RepPb$ and $k\in \ttX(a)$, then $\Res_b(\al)=a\cdot a_k \cdot \al_0$.
\item In any other case, $\Res_b(a)$ is not defined.\\
\end{itemize}

\noindent \textit{Right bipositions.} Let $(\al,\gam)\in \supp{P}$. If $\al'=\Res_b(\al)$ is defined, then we set $\Res_b(\al,\gam)=(\al',\gam)$. In any other case, $\Res_b(\al,\gam)$ is not defined.

\subsection*{Quasi-Residuation} Let $\al\in \supp{P}$. We define $\QRes_b(\al)$, the \textbf{quasi-residual} of $\al$.
\begin{itemize}
\item \textit{Extension of $\Res_b$:} if $\Res_b(\al)$ is defined, then $\QRes_b(\al)=\Res_b(\al)$.
      \item \textit{Variable of the redex:} if $\al=a\cdot 10\cdot a_k$ for some $a\in \RepPb$ and $k\in \ttX(a)$, then $\QRes_b(\al)=a\cdot a_k$.
    \item \textit{Root of the redex:} if $\al=a\in \RepPb$, then $\QRes_b(\al)=\al$.
\end{itemize}
Thus, the only case in which $\QRes_b(\al)$ is not defined is when $\ovl{\al}=b\cdot 1$. \\

\noindent \textit{Right bipositions.} Let $(\al,\gam)\in \supp{P}$, we define $\QRes_b(\al,\gam)$ case-wise:
\begin{itemize}
\item If $\al'=\QRes_b(\al)$ is defined, then $\QRes_b(\al,\gam)=(\al',\gam)$.
\item  \textit{Abstraction of the redex:} If $\al=a\cdot 1$ for some $a\in \RepPb$ (\ie $t(\al)=\lx$):
  \begin{itemize}
  \item
    if $\gam=1\cdot \gam_0$ for some $\gam_0\in \bbN^*$, then $\QRes_b(a\cdot 1,\gam)=(a,\gam_0)$.
  \item If $\gam=k\cdot \gam_0$ for some $k\geqs 2$ and $\gam_0\in \bbN^*$, then $k\in \ttK(a)$ and $\QRes_b(a\cdot 1,\gam)=(a\cdot a_k,\gam)$.
  \item If $\gam=\epsi$, then $\QRes_b(a\cdot 1,\epsi)=(a,\epsi)$. 
  \end{itemize}
\end{itemize}
Note that $\Res_b$ and $\QRes_b$ implicitly depend on $P$ in the argument of the redex, since the mapping of (bi)positions in the argument derivations typing $s$ depend on axiom tracks inside $P$, but we omit $P$ in the notation since there is

\subsection*{First properties} We observe by case-analysis that $\Res_b$ is injective. For all $\al'\in \Res_b(\supp{P})$ (where $\al=\Res^{-1}_b(\al)$), we set $\ttT'(\al')=\ttT(\al)$ and
$\ttC'(\al')=\ttC(\al)$ if $\ovlal\ngtr b\cdot 10$ and 
$\ttC'(\al')=\ttC(\Res_b(\al'))\setminus x) \uplus (\uplus_{k\in \ttK(\al)} \ttC(a\cdot k) )$ if $\ovlal\geqs b\cdot 1$, $\ttK(\al)=\Rt(\ttC(\al)(x))$ and $a$ is the unique prefix of $\al$ such that $\ovla=b$. This generalizes the construction of \Sec~\ref{ss:sr-proof}. 
We define $P':=\Red_b(P)$ (and we write $P\breda{b} P'$) as the tree labelled with $\ttS$-judgments by $\supp{P'}=\Res_b(\supp{P})$ and for all $\al'\in \supp{P'}$, $P'(\al')=\ju{\ttC'(\al')}{t'\rstr{\al'}:\ttT'(\al)}$.

\noindent Quasi-residuation preserves types. Residuation preserves constructors:

\begin{lemma}
  \label{lem:res-cons-qres-typ}
  Assume $P\breda{b} P'$. Let $\al \in \supp{P}$
  \begin{itemize}
  \item If $\al'=\Res_b(\al)$ is defined, then $t'(\al')=t(\al)$.
  \item If $\al'=\QRes_b(\al)$ is defined, then $\ttT'(\al')=\ttT(\al)$.
  \end{itemize}
\end{lemma}

\noindent Lemma~\ref{lem:res-cons-qres-typ} is pivotal to prove:

\begin{lemma}[Subject Reduction]
  \label{lem:app-sr-S}\mbox{}\\
Assume $P \breda{b} P'$. Then the labelled tree $P'$ is a correct $\ttS$-derivation.
\end{lemma}

\begin{proof}
We reason by case analysis to prove that each node of $P'$ is a correct typing rule of system $\ttS$, using Lemma~\ref{lem:res-cons-qres-typ} and the definition of $\ttC'$. For $\al'\in \supp{P'}$, a special care must be given when  $\al'$ or one of its children corresponds to $b$ or an occurrence of $s$.
\end{proof}

\begin{lemma} 
  \label{lem:inj-surj-res-Shp}
  Assume that $P$ is a \textit{quantitative} $\ttS$-derivation and $P\breda{b} P'$.
\begin{itemize}
\item $\Res_b$ is a bijection from a subset of $\bisuppR{P}$ to $\bisuppR{P'}$.
\item $\QRes_b$ is a total surjective function from $\bisuppR{P}$ to $\bisuppR{P'}$. 
Moreover, for all  $\p'\in \bisupp{P'}$, $p$ has at most six antecedents by $\QRes_b$
\end{itemize}
\end{lemma}

In short, the last point of Lemma~\ref{lem:inj-surj-res-Shp} holds, because $\QRes_b(\al,\gam)$ is defined in 6 cases, each one being injective. A derivation typing a redex of the form $(\lx.x)s$ shows that this number may be reached. 

\begin{remark}
Note that, for all $\p\in \bisuppR{P}$ and $\p':=\QRes_b(\p)$, we have $\ttP(\p)=\ttP'(\p')$ except maybe when $\p=(a\cdot 1,\epsi)$ for some $a\in \RepPpb$.
\igintro{\pierre{Actually, for that reason, we could have left $\QRes(a\cdot 1,\epsi)$ undefined, but it is more convenient that $\QRes_b$ should be total function for the proof of Lemma~\ref{app:lem:Shp-approx-stable-red-exp} (approximability is stable by conversion).}}
\end{remark}

\noindent \textbf{Monotonicity.} Observe that $P_1\leqfty P_2$ iff $P_1$ and $P_2$ type the same term, $\supp{P_1}\leqfty \supp{P_2}$ and for all $a\in \supp{P_1}$, $\ttC_1(a)\leqfty \ttC_2(a)$ (Definition~\ref{def:approx-contexts-etc}) and $\ttT_1(a)\leqfty \ttT_2(a)$.
  In particular, if $P_1\leqfty P_2$ type $t$, $t\breda{b} t'$, then $\Res_b(\supp{P_1})\subeq \Res_b(\supp{P_2})$ and $\ttT'_1(\Res_b(\al))=\ttT_1(\al)\leqfty \ttT_2(\al) = \ttT'_2(\Res_b(\al'))$. Last, $\ttC'_i(\al')(y)=(\tr'_i(\al'_0) \cdot \ttT'(\al'_0))_{\al'_0\in \Ax^{P_i}_(\al')(y)}$. Since  $\Ax^{P_1}_{\al'}(y)\subeq \Ax^{P_2}_{\al'}(y)$ for all $\al'\in \supp{P'_1}$ (because $\supp{P_1}\subeq \supp{P_2}$), we have $\ttC'_1(\al')(y) \leqfty \ttC'_2(\al')(y)$ for all $\al'\in \supp{P'_1}$ and $y\in \TermV$. Thus, $\Res_b$ is monotonic on the set of $\ttS$-derivations typing $t$.

\subsection{Equinecessary bipositions}
\label{s:equinecessary-bip}

In this section, we use the notion of \textit{equinecessity} to prove that we can forget about left bipositions while working with approximability, as we suggested in the beginning of \Sec~\ref{a:equinecessity}.

\begin{definition} 
Let $P$ a \textit{quantitative} $\ttS$-derivation and $\p_1,\,\p_2$ two bipositions of $P$.
\begin{itemize}
\item We say $\p_1$ \textbf{subjugates} $\p_2$ if, for all $P_*\leqfty 
P$, $\p_1\in P_*$ implies $\p_2\in P_*$.
\item We say $\p_1$ and $\p_2$ are \textbf{equinecessary} (written $\p_1 \lra \p_2$) if, for all  $P_*\leqfty P$, $\p_1\in P_*$ iff $\p_2\in P_*$.
\item Let $B_1,B_2\subeq \bisupp{P}$. We also write $B_1\lra B_2$ if, for all $P_* \leqfty P$, $B_1 \subeq \bisupp{P}$ iff $B_2\subeq \bisupp{P}$.
\end{itemize}
\end{definition}

Note that subjugation and equinecessity are implicitly defined \wrt $P$.
There are many elementary equinecessity cases that are easy to observe.
We need only a few ones and we define $\asc(\p)$ and $\Asc{\p}$ (standing for ``ascendance'')
so that  $\p \lra \asc(\p)$ and $\p \lra \Asc{\p}$ for all $\p\in \bisupp{P}$.
\begin{itemize}
\item $\asc(\p)$ is defined for any $\p\in \bisupp{P}$ which
is not in an axiom leaf.
\begin{itemize}
	\item \textit{Left bipositions:} $\asc(a,\,x,\,k\cdot c)=(a\cdot \ell,\,x,\, k\cdot c)$, where 
$\ell\geqslant 0$ is the unique integer such that $(a\cdot \ell,\,x,\, k\cdot c)\in
\bisupp{P}$.
	\item \textit{Right bipositions ($\abs$):} 
        if $t(\ovla)=\lambda x$, $\asc(a,\,\epsi)=(a\cdot 0,\,\epsi),~ 
\asc(a,\,1\cdot c)=(a\cdot 0,\, c)$ and $\asc(a,\, k \cdot c)=(a\cdot 0,\,
x,\, k\cdot c)$ if $k\geqslant 2$.
	\item \textit{Right bipositions ($\app$)}: if $t(\ovla)=\arob$, $\asc(a,\,c)=(a \cdot 1,\, 1\cdot c)$.
        \end{itemize}
\item $\Asc{\p}$ is a \textit{right biposition} and is defined as the highest right biposition related to $\p$ by $\asc$.
\begin{itemize}
\item \textit{Right bipositions:} if $\p=(a,c)$, let $h$ be maximal such that $\asc^h(\p)$ is defined. In that case, we set $\Asc{\p}=\asc^h(\p)$.
\item \textit{Left bipositions:} if $\p=(a,x,k\cdot c)$, let $h$ be maximal (if it exists) such that $\asc^h(\p)$ exists. In that case, $\asc^h(\p)$ is of the form $(a_0,x,k\cdot c)$ with $t(a_0)=x$. We set then $\Asc{\p}=(a_0,c)$.
\end{itemize}
\end{itemize}
Since $t\in \Lamzzu$ (and not in $\Lamuuu\setminus\Lamzzu$), $\Asc{\p}$ is defined for any right biposition $\p$. If $P$ is quantitative, then $\Asc{\p}$ is also defined for any left biposition.\\

An examination of the $\app$-rule shows that, if $t(a)=\arob$, for all $k\geqslant 2$ and $c\in \bbN^*$:
$$(a\cdot 1,k\cdot c)\lra (a\cdot k,c)\hspace*{1cm} (E1)$$
Indeed, if $\tra=u\,v$ is typed---say that $u:\skSk\rew T$ and $(v:S_k \trck{k})_{k\in K}$---the domain of the arrow typing $u$ correspond to the types assigned to the argument $u$ (if $c\in \supp{S_k}$, then $(a\cdot 1,k\cdot c)$ and $(a\cdot k,c)$ are both in $\bisupp{P}$).\\

Now, assume that $P$ is still quantitative, $t\rstr{b}=\lxrs$, $a \in \bisupp{P}$, $\ovla=b$, $k\in \TrPl{a}$. Thus, $t\rstr{a\cdot 1}=\lx.r:\skSk\rew T$ and $(s:S_k \trck{k})_{k\in K}$.
Let $c\in \supp{S_k}$ and $\p:=(a\cdot 1,k\cdot c)\in \bisupp{P}$. Notice that $\asc{p}=(a\cdot 10,x,k\cdot c)$ and $\Asc{\p}=(a\cdot 10 \cdot a_k,c)$. Thus: $$ (a\cdot 1,k\cdot c) \lra (a\cdot 10 \cdot a_k,c)\hspace*{1cm}(E2)$$


In particular, by $(E1)$ and $(E2)$:
$$(a\cdot 10\cdot a_k,\,c) \lra (a\cdot k,\, c)$$
This relation specifies that the types of the occurrences of $x$ match the types of the argument $s$ of the redex, \ie if $c\in \supp{S_k}$, both $(a\cdot k,c)$ (pointing in $s:S_k$) and $(a\cdot 10 \cdot a_k,c)$ (pointing in $x:S_k$) are in $\bisupp{P}$.

\begin{lemma}
  \label{lem:obs-equi-bip-res}
Assume that $P$ is quantitative, $P\breda{b} P'$, $\tra=\lxrs$ and 
$\p\in \bisupp{P}$. If $\Res_b(\p)$ is not defined, one of the following cases occur:
\begin{itemize}
\item If $\p$ is a left-biposition $(\al,y,k\cdot c)$ with $\Asc{\p}=:(\al_*, c)$:
\begin{itemize}
\item If $y\neq x$, then $\p \lra \p_0:= \Asc{\p}$ and $\Res_b(\p_0)$ is defined.
\item If $y = x$, then $\al_* = a\cdot 10 \cdot a_k$ for some $a$ and $k$ such that $\ovla=b$ and $k\in \TrPl{a}$.
Then 
$\p \lra \p_0:= (a\cdot 1,k\cdot c)$ by $(E3)$ and $\Res_b(\p_0)$ is defined.
\end{itemize}
\item If $\p$ is a right-biposition $(\al,c)$  with $\Asc{\p}=:(\al_*,c_*)$, then $\al=a\cdot 10\cdot \al_0$ for some $a$ such that $\ovla=b$: 
\begin{itemize}
 \item If $t\rstr{a_*}=y\neq x$, then $\p \lra \p_0:= \Asc{\p}$ and $\Res_b(\p_0)$ is defined.
 \item If $t\rstr{a_*}=x$, let $k:= \trP{a_*}$ and $\p_0:=(a\cdot k,c_*)$. Then $\p\lra \p_0$ and $\Res_b(\p_0)$ is defined by $(E3)$.
\end{itemize}
\end{itemize}
\end{lemma}

\begin{proof}\mbox{} The cases resorting to $(E3)$ are justified in the statement. Let us check that we are exhaustive.
Notice first that, when $\p$ is a left-biposition, $\Asc{\p}$ is always defined since $P$ is quantitative. Moreover, when $(\al_*,c):=\Asc{\p}$, then $t(\al_*)=y$ for some $y\in \TermV$ ($\al_*$ points to an $\ax$-rule). 
  \begin{itemize}
   \item If $y \neq x$, $\Res_b(\al_*)$ and thus $\Res_b(\al_*,c)$ are both defined.
   \item If $y=x$, $\Asc{\p}$ is of the form $(a\cdot 10\cdot a_k,c)$ for some $a$ satisfying $\ovla=b$. But, by $(E3)$, $\Asc{\p}$ is equinecessary to $(a\cdot k,c)$ (pointing inside $s$), which always has a residual by definition of $\Res_b$.
  \end{itemize}  
\end{proof}

One then notices that Observation~\ref{obs:red-to-right-res}, which is crucial to prove correctness (Lemma~\ref{lem:approx-red}), 
is a corollary of Lemma~\ref{lem:obs-equi-bip-res}.

\subsection{Uniform Subject Expansion}
\label{ss:uni-subj-exp}

We prove now the \textit{uniform} subject expansion property for system $\ttS$. By \textit{uniform}, we mean that we use a function $\code{\cdot}:\bbN\rew \bbN^*$ to decide the value of the axiom tracks in the axiom rules created during anti-reduction. Uniformity is crucial to ensure that subject expansion is monotonic. 

Assuming $\trb=\lxrs$ and $t\breda{b} t'$, we will built a derivation $P$ typing $t$ from a derivation $P'$ typing $t'$. The constructions is based on the following observation: there are three kinds of positions in $t'$: (1) positions outside the reduct $\rsx$ (2) positions inside $r$ (but outside an occurrence of $s$) (3) positions inside an occurrence of $s$.
Things are a bit less clear in $t$: we have (1) positions outside the redex $\lxrs$ (2) positions inside $r$ (3) positions inside $s$ (4) positions corresponding to the application or the abstraction of the redex $\lxrs$. Moreover, inside $r$, one should distinguish positions corresponding to $x$ (which is destroyed during reduction) and the other positions. All this has been illustrated with Fig.~\ref{fig:SR-residuals}.\\

\noindent \textbf{Hypotheses.} Let $P'\tri \juGtpt$ a quantitative $\ttS$ derivation with $\tprb=\rsx$ and $t$ the term such that $\trb=\lxrs$ and $t\breda{b} t'$.
Let $\code{\cdot}$ be an injective function from $\Nmzo$ to $\ttN^*$.
We set $A'=\supp{P'}$.\\

\noindent \textbf{Building the support of $P$.} We first build $A:=\supp{P}$ and a quasi-residuation function $QR_b$ from $A$ to $A'$.

We write $\ttT'$, $\ttC'$ and $\tttr'$ for $\ttT^{P'}$, $\ttC^{P'}$ and $\tttr^{P'}$. We set $B=\supp{t}$, $B'=\supp{t'}$ and $A'=\supp{P}$. Let $B_x=\set{\beta\in B\,|\, t(\beta)=x}$ (we assume Barendregt convention), $B'_{s\epsi}=\Res_b(B)$, $B'_s={\beta'_0\cdot \beta'_1\in B'\,|\,\beta'_0\in B'_{s\epsi}}$ and $A_b=\RepPpb$, 
$A'_{s\epsi}=\set{\al'\in A'\,|\, \ovlalp \in B'_{s\epsi}}$ and $A'_s=\set{\al'\in A'\,|\,\ovlalp \in B'_s}$. Thus, $A'_{s\epsi}$ (\resp $A'_s$) is the set of positions (in $P'$) of the  occurrences of $s$ (\resp inside an occurrence of $s$) resulting from substitution  We set $A'_{s\tti}=A'_s \setminus A_{s\epsi}$ for the set of positions \textit{internal} to an occurrence of $s$.\\

\noindent \textit{Outside the redex and inside $r$.} 
First, we set $A^0_b=\set{\al'\in A'\,|\,\ovlalp\not\geqs b}$, $A'_r=\set{al'\in A'\,|\, \ovlalp\geqs b,\, \ovlalp \notin A'_{s\tti} }$ (note that we have $\ovl{\al'_0}\in \supp{r}$) 
and $A_r=\set{a\cdot 10\cdot \al'_0 \,|\, a\in \RepPpb,\, \al'_0 \in \bbN^*,\, a\cdot \al'_0 \in A'_r}$ .
Thus, $A^0_b$ is the set of position outside the reduct in $t'$ and 
$A_r$ corresponds to the set of positions inside the subterm $r$ in $t$.
We define $QR_b$ as the identity on $A^0_b$ and by $QR_b(a\cdot 10\cdot \al'_0)=a\cdot \al'_0$ from $A_r$ to $A'_r$.

We split $A_r$ in two: let $A^x_b=\set{a\cdot 10\cdot \al'_0\,|\, a\in \RepPpb\, a\cdot \al'_0\in  A'_{s\epsi}}$ and $A^1_b=A_r\setminus A^x_b$. Intuitively, $A^x_b$ corresponds to the occurrences of $x$ in $t$ (which are in $r$) whereas $A^1_b$ corresponds to the other positions inside the subterm $r$ in $t$.\\


\noindent \textit{Inside $s$.}
We set
$A_s=\set{a\cdot \code{a\cdot 10\cdot \al'_0}\cdot \al_1\,|\, a\in \RepPpb,\,a\cdot \al_0\in A'_{s\epsi},\, a\cdot \al'_0\cdot \al_1\in A'}$. Intuitively, $A_s$ will correspond to the position of $s$ in $P$.
Indeed, $a\cdot \al'_0$ (\resp $a\cdot 10\cdot \al'_0$) corresponds to an occurrence of $s$ in $t'$ (\resp of $x$ in $t$)  whereas $a$ points to an $\app$-rule typing the application of the redex $\lxrs$ below.
As we explained above,  the natural number  $k:=\code{a\cdot 10\cdot  \al_0}$ gives us the axiom track that the occurrence of $x$ at position $a\cdot 10\cdot \al_0$ should be assigned while typing $t$. Then, we should add an argument at $\app$-node $a$ on track $k$. Hence, $a\cdot \code{a\cdot \al_0}$.
We define  $QR_b$ from $A_s$ to $A'_s$ by mapping $a\cdot \code{a\cdot 10\cdot \al'_0} \cdot \al_1$ on $a\cdot \al'_0 \cdot \al_1$ (note that $a\cdot \al'_0\in A'_s$). Observe that this induces a bijection from $A_s$ to $A'_s$.

We may now define $A$, which will be the support of $P$, as a disjoint union:
$$A:=A^0_b\cup A^x_b \cup A^1_b  \cup A_s \cup \RepPpb \cup \RepPpb \cdot 1 $$
We define $QR_b$ on $\RepPpb$ as the identity and we leave it undefined on $\RepPpb\cdot 1$ (which corresponds to the abstraction $\lx.r$). Since $A'=A^0_b\cup A'_r \cup A'_s$,  we have that $QR_b$ is a function from
$A\setminus \RepPpb\cdot 1 $ to $A'$.\\

\noindent \textbf{Building types and contexts.} We set, for all $\al \in A\setminus \RepPpb\cdot 1$, $\ttT(\al)=\ttT'(\al')$ with $\al'=QR_b(\al)$.
If $a\in \RepPpb$, we set $\ttT(a\cdot 1)=\ttT(a\cdot k)_\kK \rew \ttT(a)$ where $K=\set{k\geqs 2\,|\, a\cdot k\in A}$.

In order to define the contexts in $P$, we must observe that, in $r$ (inside $\lxrs$), $x$ is a placeholder for occurrence of $s$.

Since $P'$ is quantitative, we have, for all $\al'\in A'$, $y\in \TermV$, $\ttC'(\al')(y)=  (\tttr'(\al'_0)\cdot \ttT'(\al'_0))_{\al'_0 \in A'(\al',y)}$  where $A'(\al',y):=\Ax^{P'}_{\al'}(y)$.

Let $\al\in A_r$ and $\al':=QR_b(\al)$. For all $y\in \TermV$, 
we set $A'_0(\al',y):=A'(\al',y)\setminus A'_s$, the set of positions of axioms typing $y$ (in $P'$) above $\al'$ which are not in any occurrence of $s$.
We set $A(x,\al)=\set{\al_0\in A^x_b\,|\, \al_0\geqs \al}$ (the positions of $x$ in $A$ above $\al$). If $\al_0\in A(x,\al)$, then the axiom rule typing $x$ at position $\al_0$ is a placeholder for a typed occurrence of $s$ at position $QR_b(\al_0)$ in $P'$.
We set $A'_x(\al',y):=\set{\al'_0\in A'_{s\epsi}\,|\, \al'_0\geqs \al'}$, which is the set of positions  in $P'$ pointing to occurrences of $s$ above $\al'$. 
We then define $\ttC(\al)$ case-wise:
\begin{itemize}
\item If $\al \in A_r$, then $\ttC(\al)=(\tttr'(\al'_0)\cdot \ttT'(\al'_0))_{\al'_0\in A'_0(al',y)}$. for all $y\neq x$ and $\ttC(\al)(x)=(\code{\al_0} \cdot \ttT'(QR_b(\al_0)))_{\al_0\in A(\al,x)}$. This latter definition is correct since $\code{\cdot}$ is injective.
\item 
  If $\ovlal \not\geqs b$ (\ie $\al\in A^0_b$) or $\al \in A_s$, then $\ttC(\al)=\ttC'(\al')$
\item If $\al =a\cdot 1$ for some $a\in \RepPpb$, then we set $\ttC(\al)=(\tttr'(\al'_0)\cdot \ttT'(\al'_0))_{\al'_0\in A'_0(al',y)}$. for all $y\neq x$ and $\ttC(\al)(x)= \est$
  \end{itemize}
We define $\Exp_b(P',\code{\cdot},t)$ as the tree $P$ labelled with $\ttS$-judgments such that $\supp{P}=A$ and for all $\al \in A$, $P(\al)=\ju{\ttC(\al)}{t\rstr{\al}:\ttT(\al)}$.
For instance, the construction of $A_r$ depends on $t$. The construction of $A_s$ depends on $\code{\cdot}$.

We have, for all $\al\in A \setminus \RepPpb$ and $\al'=QR_b(\al)$:
$$\ttT(\al)=\ttT'(\al')\sep \sep \text{(type preservation)}$$

\noindent \textbf{Correction and monotonicity.}
Let $\al\in A$. We set $\ttCh(\al):=\set{\al \cdot k\in A\,|\, k\in \bbN  }$ ($\ttCh$ stands for ``children''). Thus, $\ttCh(\al)$ corresponds to the premises of $\al$ in $A$.
Let $\al'\in A'$. We set likewise $\ttCh'(\al'):=\set{\al'\cdot k\in A'\,|\, k\in \bbN}$. A case analysis shows that, if $\al\in A^0_b\cup A^1_b \cup A_s$ and $\al'=QR_b(\al)$, then $QR_b$ induces a bijection from $\ttCh(\al)$ to $\ttCh'(\al')$. By type preservation, this implies that for all $\al \in  A^0_b\cup A^1_b \cup A_s$, $\al$ is given by a correct $\ttS$-rule.

We must also check that, for $\al \in A^x_b\cup \RepPpb \cup \RepPpb\cdot 1$, $\al$ is correct: this easily follows from the definition of $\ttC$ and $\ttT$. Thus, $P$ is correct and concludes with $\ju{\ttC(\epsi)}{t:\ttT(\epsi)}$, \ie $\ju{C}{t:T}$, as expected.

To prove monotonicity, \ie $P'_1\leqfty P'_2$ implies
$\Exp_b(P'_1,\code{\cdot},t)\leqfty \Exp_b(P'_2,\code{\cdot},t)$, one needs to remark the two following crucial points: (1) the construction of $A_s$, which depends on $\code{\cdot}$, is monotonic (2) if $t(\al)=x$, then 
$(\code{\al}\cdot \ttT_1(\al))\leqfty (\code{\al}\cdot \ttT_2(\al))$
because we assign the same axiom rule in the two singleton sequence types. This ensures that $\ttC_1(\al)\leqfty \ttC_2(\al)$ in all cases.\\

\begin{remark}\mbox{}
  \label{rk:exp-ad-hoc}
  \begin{itemize}
  \item The function $\code{\cdot}$ does not need to be total on $\Nmzo$, nor injective for the above construction to work. It only needs to be defined for the positions of the axiom rules that are created (\ie $\code{\cdot}$ must be defined on $A^x_b$), so that an axiom track is assigned for each axiom rule that is created by expansion. And for all $\RepPpb$, $\code{\cdot}$ needs to be injective on $\set{\al \in A^x_b~|~\al \geqs a}$ (not on $\bbN^*$!), which is enough to ensure that no track conflict occurs.
  \item Thus, the monotonicity of expansion thus admits a variant, without having to specify a function $\code{\cdot}$: if $\fP' \leqs P$ type $t'$ and $t\breda{a}t'$, then there exists $\fP\leqs P$ typing $t$ such that $P\breda{b} P'$ and $\fP\breda{b}\fP$. For that, consider the function $\code{\cdot}$ defined by $\code{\al}=\trP{\al}$ for all $\al \in \supp{P}$ such that $t(\al)$ is the variable of the redex. 
  \item   
  One way wonder how the expansion $P$ depends on the choice of $\code{\cdot}$. It actually does not matter that much, because all the possible expansions of $P'$ are (in some sense) \textit{isomorphic}, whatever the values of the new axiom tracks are. 
\end{itemize}  
\end{remark}

\subsection{Quantitativity is stable under reducation and expansion.}
\label{a:approx-stable-conv}

We can now prove that quantitatitivity is stable under reduction and expansion, as stated in Remark~\ref{rmk:inf-quant-and-red}. This means that we did not lose any generality after \Sec~\ref{s:approx} by considering only quantitative derivations.

\begin{lemma}
\label{lem:quant-stable-red-exp}
If $P\breda{b}P'$, then $P$ is quantitative iff $P'$ is quantitative.
\end{lemma}

\begin{proof}\mbox{}
\begin{itemize}  
\item  Assume that $P\tri \juCtt$ is not quantitative, $\trb=\lxrs$ and that $t$ respects Barendregt convention. Thus, there are $\al \in \supp{P}$,  $y\in \TermV$ and $k\geqs 2$ such that, for all $d\geqs 0$, there is $\al_0\geqs a$ such that
  $|\al_0| \geqs d$,  $k\in \supp{\ttC(\al_0)(y)}$. Thus, $B:=\set{\al_0\geqs a~|~k\in \supp{\ttC(\al_0)(x)}}$, the set of ascendants of $(\al_0,y,k)$, defines (the postfix of) an infinite branch of $\supp{P}$. There are three cases: (1) $B$ visits an occurrence of $s$ (\ie there is $\al_0\in B$ such that $\ovl{\al_0}\geqs b \cdot 1$) (2) $B$ visits $r$ but no occurrence of $s$ (3) $B$ does not visit the redex $\lxrs$. By case analysis and using the definition of $\ttC$, we obtain an infinite branch $B'$ in the derivation reduct $P'$ which is in an occurrence of $s$ in case (1), in $r$ but does not visit an occurrence of $s$ in case 2, is outside $\rsx$ in case (3). In case (1), we may replace $(\al,y,k)$ by one of its ascendant and assume that $\al$ is in an occurrence of $s$  without losing generality. Likewise, in case (2), we may assume that $(\al,y,k)$ is in $r$.  
  \item The converse implication is proved by following back the steps of the above proof and considering the definition of $\ttC$.
\end{itemize}
\end{proof}

\subsection{Proof of the infinitary subject expansion property}
\label{a:expans}

We prove now infinitary subject expansion for system $\ttS$ \chhp{(Proposition~\ref{prop:inf-subj-exp-Shp-hp})}{Proposition~\ref{prop:infinite-subject-expansion}}. \ighp{First, subject substitution holds:

\begin{lemma}
\label{lem:subject-substitution-Shp}
Let $P\tri \juCtt$ be a $\ttS$-derivation such that for all $a\in \supp {P},~ t(a) =t'(a)$ ($P$ is not necessarily assumed to be approximable).\\
Let $P[t'/t]$ be the \textit{labelled tree} obtained from $P$ by replacing $t$ by $t'$ (more precisely, $P[t'/t]$ is the labelled tree $P'$ such that $\supp{P'}=\supp{P}$ and, for all $a\in \supp{P},~ P(a)=\ju{\ttC(a)}{t'\rstr{a}:\ttT(a)}$).\\
Then $P[t'/t]$ is a correct derivation.
\end{lemma}
}

As mentioned in \Sec~\ref{ss:one-step-sr-se}, performing an 
expansion of a term inside a derivation requires that we choose new axiom tracks. We will 
do this \textit{uniformly}, \ie we fix an injection 
$\code{\cdot}$ from $\mathbb{N}^*$ to $\mathbb{N}-\{0,\,1\}$ and any 
axiom rule created at position $a$ will use the axiom track value 
$\code{a}$. As we have seen (\Sec~\ref{ss:uni-subj-exp}), this ensures that expansion is monotonic.\\

\noindent \textbf{Hypotheses.} We consider an \textit{approximable} $\ttS$-derivation $P'\tri \ju{C'}{t':T'}$ and a productive reduction path $t=t_0 \breda{b_0} t_1 \breda{b_1} t_2 \ldots t_n \breda{b_n} t_{n+1}\rewfty t'$. Thus, $\ad{b_n} \rew \infty$. We and set $A'=\supp{P'}$.\\

Assume $\fP'\leqslant P'$. Let $N \in \mathbb{N}$ such that, for all 
$n\geqslant N,~ b_n \notin \ovl{\supf A'}$ with $\supf A'=\supp {\fP'}$. For $n\geqslant N$, we write $\fP'(n)$ for the derivation 
replacing $t'$ by $t_n$ in $\fP'$. This derivation is correct 
according to the subject substitution lemma (\chhp{Lemma~\ref{lem:subject-substitution-Shp}}{Lemma~\ref{lem:subject-substitution}}), 
since $t_n(\ovl{a})=t'(\ovl{a})$ for all $a\in \supf A'$.\\

We then write $\fP'(n,~ k)$ (with $0\leqslant k \leqslant n$) the 
derivation obtained by performing $k$ expansions (\wrt our 
reduction sequence and $\code{\cdot}$). Since $b_n$ is not in $A$, we 
observe that $\fP'(n+1,\, 1)= \fP'(t_n)$. Therefore, for all 
$n\geqslant N,~ \fP'(n,\, n)=\fP'(N,\, N)$. Since we could 
replace $N$ by any $n\geqslant N$, $\fP$ is morally $\supf 
P'(\infty,\, \infty)$. We write $P=\init(P')$ to refer to this 
deterministic construction (which implicitly depends on $\code{\cdot}$).\\


We set $\scrD=\{\init(\fP')~ |~ \fP'\leqslant P'\}$. Let us 
show that $\scrD$ is a directed set.

Let $\fP'_1,\, \fP'_2\leqslant P'$. We set 
$\fP'=\sup(\fP'_1,\fP'_2)$, so that $\fP'$ is also finite. Let $N$ be great enough so that $\forall n\geqslant N$, $b_n \notin \ovl{\supf A'}$ with $\supf A'=\supp {\fP'}$.

We have $\fP'_i\leqslant \fP'$, so $\fP_i'(N) \leqslant \fP(N)$. Thus, by monotonicity of \textit{uniform} expansion, $\fP'_i(N,\,N)\leqslant \fP'(N,\,N)$ \ie $\init(\fP_i)\leqslant \init(\fP)$.

Since $\scrD$ is directed, we can set  $P=\sup_{\set{\fP'\,|\,\fP'\leqslant P'}} \init( \fP')$. Since for any $\fP'\leqslant P$ and the associated usual notations, 
$\supf \ttC(\epsi)=\supf \ttC'(\epsi)$, $\supf \ttT(\epsi)=\supf \ttT'(\epsi)$ and, by Lemma~\ref{lem:prove-lattice}, $\ttC(\epsi),~ \ttC'(\epsi),~ \ttT(\epsi),~ \ttT'(\epsi)$ are the respective infinite joins of $\supf \ttC(\epsi)$, $\supf \ttC'(\epsi)$, $\supf \ttT(\epsi)$, $\supf \ttT'(\epsi)$ when $\fP'$ ranges over $\Approx{P'}$,
we conclude that $\ttC(\epsi)=\ttC'(\epsi)=C'$ and $\ttT(\epsi)=\ttT'(\epsi)=T'$. In particular, $\ju{C'}{t:T'}$ is approximably derivable. This concludes the proof of the infinitary subject expansion property.


\subsection{Equinecessity and Bipositions of Null Applicative Depth}
\label{a:approx-0-cex}

One may wonder whether every biposition (in a quantitative derivation $P$) is equinecessary with a \textbf{root biposition} \ie a biposition that is located in the judgment concluding $P$. Such a biposition is of the form $\p=(\epsi,c)$ with $c\in \supp{\ttT^P(0)}$ or $\p=(\epsi,x,k\cdot c)$ with $k\cdot c\in \supp{\ttC^P(0)(x)}$. 

This would imply that, for a derivation $P$ to be approximable, it is enough to have: ``$P$ is quantitative and, for all  $\oB\subeq \bisupp{P}$ finite set of \textit{root} bipositions, there exists $\fP\leqs P$ such that $\oB\subeq \bisupp{\fP}$''. We call this condition \textbf{root approximability}. 
$\p\in \oB$,

This is actually true in the finite case and so, for the approximable derivation. However, this is not true for \textit{any} derivation. We exhibit a counter-example of this conjecture in this section \ie a derivation that root approximable but not approximable. We present this counter example with a $\scrR$-derivation whereas it should be a $\ttS$-derivation (since approximability is only an informal notion for System $\scrR$), but it is easier to understand that way. Corresponding $\ttS$-derivations are not difficult to define from our presentation.

The idea is to use a productive reduction path $t_0\rew t_1\rew t_2\rew \ldots \rewfty t'$ such that no reduction step is erasing but there is a variable $x\in \fv{t_0}=\fv{t_1}=\ldots$ such that $x\notin \fv{t'}$ (\ie there is an asymptotic erasure).

Let $\Dels=\lx.(\lz.y(x\,x\,z)$ and $t=\Dels\,\Dels$, so that $t\rew \lz.y(t\,z)$ (see Fig.\;\ref{fig:1-cex-approx}). We have $t\beq \cu\,(\lam tx.(\lz.y(tx)))$.

\begin{figure}
  \begin{center}
    \ovalbox{
\begin{tikzpicture}

\drawtri{0}{-0.9}{t}
\draw (2,-0.45) node {$\longrightarrow$};
\transh{4}{
\drawlabnode{-0.65}{0.9}{$y$}
\drawtri{0}{1.8}{$t$}
\drawrighttail{0}{1.8}
\drawlabnode{1.3}{1.8}{$z$}
\blocka{0.65}{0.9}
\blocka{0}{0}
\blockunary{0}{-0.9}{$\lz$}
}
\end{tikzpicture}}
\end{center}
\caption{Reduction of $t$}
\label{fig:1-cex-approx}
\end{figure}
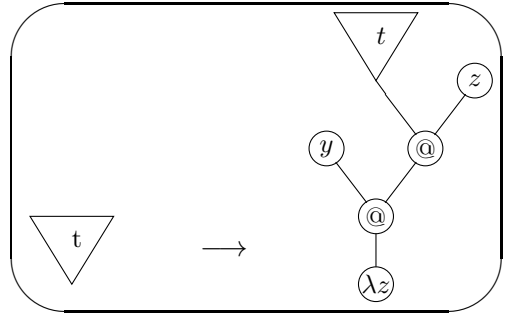
Note that $t\,f$  converges to the term $t'=y\som$, which does not contain $f$. Indeed, $t\,f\rew (\lz.y(t\,z))f \rew y(t\,f)$ (non-erasing steps).

There are two derivations $\Psi_1$ and $\Psi_2$ respectively concluding with  $\ju{y:\mult{\arewa}_\om}{\lz.y(t\,z):\arewa}$  and  $\ju{y:\mult{\arewa}_\om}{\lz.y(t\,z):\erewa}$. They are obtained from $\Psi'_1$ and $\Psi'_2$ from Fig\;\ref{fig:1-cex-approx} by a one-step expansion:
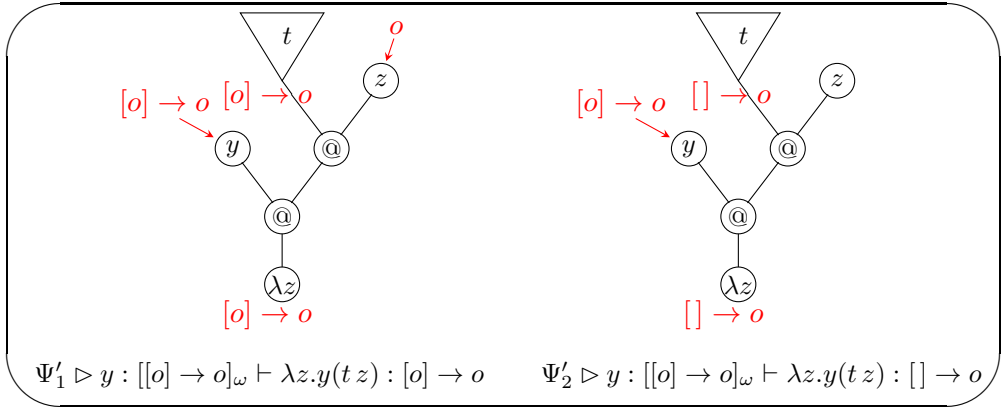
\begin{figure}
\begin{center}
\ovalbox{
  \begin{tikzpicture}
\inputarewa{-0.65}{0.8}
\drawlabnode{-0.65}{0.9}{$y$}
\draw (-0.2,1.6) node {$\red{\arewa}$};
\drawtri{0}{1.8}{$t$}
\drawrighttail{0}{1.8}
\draw (1.5,2.5) node {$\red{\tv}$};
\red{
\draw [>=stealth,->] (1.47,2.36) --++ (-0.1,-0.3);}
\drawlabnode{1.3}{1.8}{$z$}
\blocka{0.65}{0.9}
\blocka{0}{0}
\blockunary{0}{-0.9}{$\lz$}
\draw (-0.2,-1.3) node {$\red{\arewa}$};
\draw (-0.3,-2.1) node {\small {$\Psi'_1\tri \ju{y:\mult{\arewa}_\om}{\lz.y(t\,z):\arewa}$}};

\transh{6}{
\inputarewa{-0.65}{0.8}
\drawlabnode{-0.65}{0.9}{$y$}
\draw (-0.1,1.6) node {$\red{\erewa}$};
\drawtri{0}{1.8}{$t$}
\drawrighttail{0}{1.8}
\drawlabnode{1.3}{1.8}{$z$}
\blocka{0.65}{0.9}
\blocka{0}{0}
\blockunary{0}{-0.9}{$\lz$}
\draw (-0.2,-1.3) node {$\red{\erewa}$};
\draw (0.3,-2.1) node {\small {$\Psi'_2\tri\ju{y:\mult{\arewa}_\om}{\lz.y(t\,z):\erewa}$}};
}

\end{tikzpicture}}
\end{center}
\caption{Two Derivations typing $t':=\lz.y(t\,z)$}
\label{fig:2-cex-approx}
\end{figure}

Let $\Pi_{\Om}$ a derivation concluding with $\ju{}{\Om:\tv}$. This derivation is unsound and is intuitively not approximable: it is impossible to find a  finite derivation (of $\scrRo$) concluding with $\ju{}{\Om:\tv}$.
Then using $\Psi_1$ and $\Pi_{\Om}$ and an $\app$-rule we can build a derivation $\Pi_1$ concluding with  $\ju{y:\mult{\arewa}_\om}{t\,\Om:\tv}$, in which the subterm $\Om$ is typed. We can also build a derivation $\Pi_2$ from $\Psi_2$ that also concludes with  $\ju{y:\mult{\arewa}_\om}{t\,\Om:\tv}$. This times, the subterm $\Om$ is not typed.

It is not difficult to see that $\Pi_2$ is intuitively approximable, whereas $\Pi_1$ is root approximable but not fully approximable. Roughly speaking, the subderivation of $\Pi_1$ typing $\Om$ is the only non-approximable part of $\Pi_1$. 
Indeed:
\begin{itemize}
\item $\Pi_2$ is not approximable since it contains a subderivation typing the mute term $\Om$.
\item Every (finite) approximation of $\Pi_2$ is an approximation of $\Pi_1$. Thus, the join of the finite approximations of $\Pi_1$ is actually $\Pi_2$.
\end{itemize}
This proves that \textit{root} approximability is not equivalent to approximability, and that root approximability actually accepts derivations which are partially unsound.


\section{Approximability of normal derivations}
\label{a:normal-forms}


\noindent In this Appendix, we describe all the quantitative derivations typing normal forms in system $\ttS$ and we prove that all of them are approximable.
\ighp{

Actually, this statement could be generalized to system $\ttShp$ it is somewhat clear that every quantitative derivation typing a normal form in system $\ttShp$ is also approximable, but:
\begin{itemize}
\item This would be more tedious, because it is a bit more complicated to describe natural extensions in system $\ttShp$: indeed, we must consider all the subpositions of $\supp{t}$ which correspond to an element of $\HP_d$ for some $d\in \bbN \cup{\infty}$.
\item We do not need it to obtain the main theorem of characterization in system $\ttShp$ (Theorem~\ref{th:hp-pcons-hp}).
\end{itemize}
}

%
%
%
%
%
%

\subsection{Called rank of a position inside a type in a derivation}

 We recall that the rank $\rk{a}$ of $a\in \bbN^*$ is defined by $\rk{a}=\max (\ad{a},\max(a))$ (Definition~\ref{def:rank-max-width-depth}). We reuse the notation for constrain levels $\rdeg$ defined in \Sec~\ref{ss:support-candidates} and the notation $\Cal(a)$ define in \Sec~\ref{ss:nat-ext}.\\

For each $a$ in $A$ and each position $c$ in $\Cal(a)$ such that
$\Cal(a)(c)\neq X_{a'}$ (for all $a'$), we define the numbers $\cadout{a,c}$ and $\cadin{a,c}$ by:
\begin{itemize}
\item When $a$ is a unconstrained node, $\cadout{a,c}$ is $\rk{a}$ and $\cadin{a,c}=\rk{c}$.
\item When $a$ is a non-zero position: the respective values of $\cadout{a,c}$ and $\cadin{a,c}$ for the positions colored in red are $\rk{a}$ and $\rk{c}$.
$$\ttE(\ra)(x_1)\redrew \ttE(\ra)(x_2)\redrew \ldots \redrew \ttE(\ra)(x_n)\red{\rightarrow \ttT(\ra)}$$
where $n=\rdeg(a)$.
In particular, $\cadout{a,1^i}=\ad{a}$ and $\cadin{a,1^i}=0$ for $1\leqs i \leqs n$..
\item When $a$ is partial: the respective values of $\cadout{a,c}$ and $\cadin{a,c}$ for the positions 
colored in red are $\rk{a}$ and $\rk{c}$.
$$\Rst_1(a)\redrew \ldots
\redrew \Rst_n(a)\red{\rightarrow \ttT(\ra)}$$
where $n=\rdeg(a)$.
In particular, $\cadout{a,1^i}=\ad{a}$ for $1\leqs i \leqs n$ and $\cadout{a,1^n\cdot c}=\ad{a}$.
\end{itemize}

For each $a\in A$ and each position $c$ in $\ttT(a)$, we define
$\cadout{a,c}$ and $\cadin{a,c}$ by extending $\mathtt{cr}_\mathtt{ out}$ and $\mathtt{cr}_{\mathtt{in}}$ via substitution (this is formally done in the next section).

Again, for each $a\in A$, each variable $x$ and each position
$c$ in $\ttC(a)(x)(c)$, we define $\cadout{a,x,c}$ and $\cadin{a,x,c}$  by extending $\mathtt{cr}_\mathtt{out}$ via  substitution.

The definition of $\cadout{a,c}$, $\cadout{a,x,c}$, $\cadin{a,c}$ and $\cadin{a,x,c}$)  are sound, because in $\ttE(a)(x)$ and $\Rst_i(x)$, the $X_{a'}$ occur at depth $\geqs 1$ in $\Cal(a')$.

\begin{definition}
If $\p \in \bisupp{P}$, we define \textbf{called rank} of $\p$  by $$\cad{\p}=\max(\cadout{\p},\, \cadin{\p})$$
\end{definition}

\subsection{Truncation of rank $n$}

We present more formally the definitions of the last section and  we recall that $\rk{a}=\max (\ad{a},\,\max(a))$ for $a\in \bbN^*$.\\

As we saw above, a quantitative derivation typing a normal form can be reconstructed from its supports and the types assigned in the unconstrained positions. Thus, in a quantitative derivation $P$ typing a normal form $t$ such that $\supp{P}=A$, every $\p\in \bisupp{P}$ (with $\p=(a,c)$ or $\p=(a,x,c)$) must come from (so to say) an $(a',c')$ with $a\geqs a'$. We call $a'$ the  \textit{calling outer position} and $c'$ the \textit{calling inner position} of $\p$. Let us define them formally now. 
\\

For all $a\in A$ and $k\in \mathbb{N}$, we set, by induction on $k$, $\ttT^0 (a)=X_a$ and 
$\ttT^{k+1}(a)= \ttT^k(a)[\Cal(a')/X_{a'}]_{a'\in \bbN^*}$ and for all $k \in \mathbb{N}$, we set $\supps {\ttT^k(a)}=\set{c\in \supp {\ttT^k(a)}\,|\, \forall a'\in \bbN^*,~ \ttT^k(a)(c)\neq X_{a'}}$. Thus, $\supp{\ttT(a)}=\cup_{k\in \bbN} \supps{\ttT^*(a)}$ for all $a\in A$.\\

If $c\in \supp {\ttT(a)}$, there is a minimal $k\in \bbN$ such that $c\in \supps{\ttT^k(a)}$.
We denote it $\cd(a)(c)$ (call-depth of $c$ at pos. $a$).

When $k=\cd(a)(c)$, there are unique $c'\in \supp {\ttT(a)},\; c\secu\in \bbN^*$ and $a'\in A$ such that
$c=c'\cdot \csec,\ \ttT^{k-1}(a)(c')=X_{a'}$ (we have necessarily $a\leqslant a'$ since $a'$ is called by $a$), $\csec\in \supps{\Cal(a')}$ and $\Cal(a')(\csec)=\ttT(a)(c)$. 
We write $a'=\cop(a,c)$ (\textit{calling outer position} of $c$ at position $a$), $c'=\cip(a,c)$ (\textit{calling inner position} of $c$ at position $a$)
and $\csec=\pf(a,c)$ (\textit{postfix} of $c$ at position $a$). 
Then, we set $\cadout{a,c}:= \rk{a'}$ and $\cadin{a,c}=\rk{\csec}$.
Finally, we may define $\cad{a,c}$ by $\cad{a,c}=\max(\cadout{a,c},\cadin{a,c})$. Implicitly, $\cad{a,c}$ depends on $P$, but $P$ is omitted from the notation.


\begin{itemize}
\item We set $A_n=\{a\in A\, |\, \rk{a}\leqslant n\}$.
\item For all $\ra\in \rA_n$, we define $\rT_n(\ra)$ by removing
all the positions $c$ such that $\rk {c} > n$, \ie $\rT_n(\ra)$ is the restriction of $\rT(\ra)$ on $\set{c\in \supp{\rT(\ra)}\,|\, \rk{c}\leqslant n}$.
\end{itemize}
Since $t$ is in $\Lamzzu$ (and not in $\Lamuuu\setminus \Lamzzu$), $A_n$ is finite. Since, for all $a\in A$, $\ttT(a)$ is in $\Types$ (and not in $\Types^{111}\setminus\Types$), the type $\rT_n(\ra)$ is finite for all $\ra\in \rA_n$. 

We define $P_n$ as the natural extension of $(A_n,\rT_n)$. We retrieve contexts $\ttC_n$ and types $\ttT_n$ such that, for all $a\in A_n$, $P_n=\ju{\ttC_n(a)}{\tra:\ttT_n(a)}$.\\

\begin{lemma}
\label{l:cut}
For all $k\in \mathbb{N},\; a\in A,\; c\in \bbN^*$, we have $\rk{a}\leqslant n$ and 
$c\in \supp {\ttT_n^k(a)}$ iff $c\in \supp {\ttT^k(a)}$ and
$\cad {a, c} \leqslant n$. \\
In that case, $\ttT^k(a)(c)=\ttT_n^k(a)(c)$.
\end{lemma}

\begin{proof}
By a simple but tedious induction on $k$. The notations $\cop_n$, $\cip_n$, $\pf_n$ and $\ttT^k_n$ for the \textit{ad hoc} definitions of calling outer positions and so on for $P_n$.\\

\begin{itemize}
\item Case $k=0$: \\
$\Rightarrow$: if $\rk{a}\leqslant n$ and $c\in \supp {\ttT_n^0(a)}$, then $\rk{a}\leqslant n$ and $c=\epsi$. 
By definition, $\cadout{a,\epsi}=\rk{a}\leqslant n$ and 
$\cadin{a}(\epsi)=\rk{\epsi}=0$. Thus, $\cad {a, c}\leqslant n$.\\
$\Leftarrow$: conversely, if $c\in \supp {\ttT^0(a)}$ and $\cad {a,\, c}\leqslant n$, we have likewise $c=\epsi$ and $\cop(a,\epsi)=a$, so $\rk{a}=\cadout{a,c}\leqslant n$. So $c=\epsi\in \ttT_n^0(a)$.\\
\item Case $k+1$: 
$\Rightarrow$: we assume that $\rk{a}\leqs n$ and $c\in \supp{\ttT^k_n(a)}$.
We have two cases, whether $c\in \supps{\ttT^k_n(a)}$ or not. If $c\in \supps{\ttT^k_n(a)}$, then $\ttT^{k+1}_n(a)(c)=\ttT^k_n(a)(c)$. By \ih, we have $\ttT^k_n(a)(c)=\ttT^k(a)(c)$. This implies $\ttT^{k+1}(a)(c)=\ttT^k(a)(c)$ and also, $\ttT^{k+1}_n(a)(c)=\ttT^{k+1}(a)(c)$. Thus, let us assume $c\notin \supps{\ttT^k_n(a)}$.\\

We have once again two cases, whether $c\in \supp{\ttT^k_n(a)}$ or not.

Assume first that $c\in \supp{\ttT^k_n(a)}$. Since $c\notin \supps{\ttT^k_n(a)}$, we have $\ttT_n^k(a)(c)=X_{a'}$ where $a'=\cop(a,c)$. Then $\ttT^{k+1}_n(a)(c)=\ttT^=\Cal_n(a')(\epsi)$. Moreover, the\ \ih applies to $(a,c)$ and gives $\cad{a,c}\leqslant n$ and $\ttT^k(a)(c)==\ttT^k_n(X_{a'})$. Thus, $\ttT^{k+1}(a)(c)=\Cal(a')(\epsi)$ and $\rk{a'}=\cadout{a,c}\leqslant n$. Since $a\leqs a'$, we have $\rk{a}\leqslant n$.
Since $\rk{\epsi}=0$, we have $\Cal_n(a')(\epsi)=\Cal(a')(\epsi)$. So $\ttT^{k+1}(a)(c)=\ttT^{k+1}_n(a)(\epsi)=\Cal_n(a')(\epsi)$.

We assume now that $c\notin \supp{\ttT^k_n(a)}$. Since $c\in \supp{T^{k+1}_n(a)}\setminus \supp{\ttT^k_n(a)}$, there are unique $a'\in A_n$, $c''\in \supp{T^k_n(a)}$ and $\csec \in \bbN^*$ such that $c=c'\cdot \csec$, $\ttT^k_n(a)(c)=X_{a'}$.
In particular, $\ttT^k_n(a)(c)=\ttT^k_n(a')(\csec)$. We have two subcases:
\begin{itemize}
\item Case $\ttT^{k+1}_n(a)(c)\neq X_{\asec}$ for all $\asec$. In that case, $a'=\cop_n(a,c)$, $c'=\cip_n(a,c)$ and $\csec=\pf_n(a,c)$. 
We have $\csec\in \supps{\Cal_n(a')}\leqs \supp{\Cal(a')}$. This entails in particular $\ttT^{k+1}(a)(c)=\Cal_n(a')(c)=\ttT^{k+1}_n(a,c)$.\\
Moreover, we have $\cadout{a,c}=\rk{a'}\leqslant n$ since $a\in A_n$ and $\cadin{a,c}=\rk{\csec} \leqslant n$ since $\csec \in \supp{\ttT_n(a')}$. Thus, $\cad{a,c}\leqslant n$.
\item Case $\ttT^{k+1}_n(a)(c)=X_{\asec}$ for some $\asec \geqs a'$. In that case, $a\secu=\cop_n(a,c)$, $c=\cip_n(a,c)$ and $\pf_n(a,c)=\epsi$. Due to the form of $\Cal_n(\_)$, $\secu=1^i \cdot \ell$ for some $j<\rdeg (a')$ and $\ell =\code{a'}\geqs 2$. By \ih, $\ttT^k(a)(c')=\ttT^k_n(a)(c')$, so $\ttT^k(a)(c')=X_{a'}$. Since $c\secu = 1^i\cdot \ell \in \supp{\Cal_n(a')}$ and $\Cal_n(a')\leqs \Cal(a')$, we also have $\Cal(a')(c\secu)=X_{a\secu}$. Thus, $\ttT^{k+1}(a)(c)=X_{a\secu}=\ttT^{k+1}_n(a)(c)$.\\
Moreover, $\cadout{a,c}=\rk{\asec}\leqs n$ since $\asec \in A_n$, and $\cadin{a,c}=\rk{\epsi}=0$, so $\cad{a,c}\leqs n$.\\
\end{itemize}

$\noindent \Leftarrow$: conversely, if $a\in A,\ c\in \supp {\ttT^{k+1}(a)}$ and $\cad{a, c}\leqslant n$, we assume that $a\notin \supps {\ttT^k(a)}$ (if $a\in \supps{\ttT^k(a)}$, the case is straightforwardly handled by \ih).

Assume first that $c\in \supp{\ttT^k(a)}$ and $\ttT^k(a)=X_{a'}$. Then $\rk{a'}=\cadout{a,c}\leqslant n$, so that $a'\in A_n$, and $\ttT^{k+1}(a)=\Cal_n(a')(\epsi)$. By \ih on $(a,c)$, $a\in A_n$,  $c\in \supp{\ttT^k_n(a)}$ and $\ttT^k_n(a)(c)=X_{a'}$, so $\ttT^{k+1}_n(\epsi)=\Cal_n(a')(\epsi)$. Since $\Cal_n(a')\leqs\Cal(a)$, we have $\Cal_n(a')=\Cal(a')(\epsi)$, which is equal to $\ttT^{k+1}(a)(c)$ since $\ttT^{k}(a)(c)=X_{a'}$. Thus, $\ttT^{k+1}_n(a)(cc)=\ttT^{k+1}(a)(c)$.

We assume now that $c\notin \supp{\ttT^k(a)}$.
 Since $c\in \supp{T^{k+1}(a)}\setminus \supp{\ttT^k(a)}$, there are unique $a'\in A$, $c'\in \supp{T^k(a)}$ and $\csec \in \bbN^*$ such that $c=c'\cdot \csec$, $\ttT^k(a)(c)=X_{a'}$. By \ih on $c'$, $\rk{a}\leqs n$, $c'\in \supp{T^k_n(a)}$ and $\ttT^k_n(a)(c')=\ttT^k(a)(c)=X_{a'}$.
In particular, $\ttT^k(a)(c)=\ttT^k(a')(\csec)$. We have two subcases:
\begin{itemize}
\item If $\ttT^{k+1}(a)(c)\neq X_{\asec}$ for all $\asec$, then $a'=\cop(a,c)$, $c'=\cip(a,c)$ and $\csec = \pf(a,c)$.  Moreover, $\rk{\csec}=\cadin{a,c}\leqs \cad{a,c}\leqs n$, so $\secu\in \supp{\Cal_n(a')}$. By definition of $\ttT^{k+1}_n(a)$, $\ttT^{k+1}=\Cal_n(a')(c')$. Since $\Cal_n(a')\leqs \Cal(a')$, we have $\Cal_n(a')(c')=\Cal(a)(c')$, which is equal to $\ttT^{k+1}(a)(c)$. Thus, $\ttT^{k+1}_n(a)(c)=\ttT^{k+1}(a)(c)$.
\item If $\ttT^{k+1}(a)(c)=X_{\asec}$ for some $a\secu\geqs a'$. In that case, $\asec=\cop(a,c)$, $c=\cip(a,c)$ and $\pf(a,c)=\epsi$. 
Due to the form of $\Cal(\_)$, $\secu=1^j \cdot \ell$ for some $ j< \rdeg (a')$ and $\ell =\code{a'}\geqs 2$. By definition of $\Cal_n(\_)$, we also have $1^j\in \Cal_n(a')$ and $\Cal_n(a')(\csec)=X_{\asec}$ ($\Cal_n(a')$ and $\Cal(a')$ may only differ on positions $\geqs 1^{\rdeg(a')}$ by construction). By definition of  $T^{k+1}_n$, we have $T^{k+1}_n(a)(c)=\Cal_n(a')(X_{\asec})$. In particular, $T^{k+1}_n(a)(c)=\ttT^{k+1}(a)(c)$.
\end{itemize}
\end{itemize}
\end{proof}

\subsection{A Complete Sequence of Derivation Approximations}

We now prove Lemma~\ref{lem:q-deriv-NF-approx}:

\begin{lemma*}
If $P$ is a quantitative derivation typing a normal form $t$, then $P$ is approximable.
\end{lemma*}

\begin{proof}
Let $\oB\subset \bisupp{P}$ a finite subset. We set $n=\max\set{\cad{\p}\,|\, \p\in B}$. Then $\oB\subseteq \bisupp{P_n}$.

In order to conclude, it is enough to prove that $P_n$ is a \textit{finite} derivation. For that, we notice that $A_n$ is finite and, for all $a\in A_n$, $\Cal_n(a)$ is finite and $\ttT_n^{k+1}(a)=T_n^{n+1}(a)$ for all $k\geqslant n+1$, since $\ttT^{n+1}_n(a)$ does not contain any $X_{a'}$. In particular, $\ttT_n(a)$ is finite. Since $A_n$ and all the $\ttT_n(a)$ ($a$ ranging over $A$) are finite, $P_n$ is finite.
\end{proof}

\section{Infinitary Multiset Types}
\label{a:system-R} 


We present here a definition of type assignment system $\scrR$,
which is an infinitary version of Gardner-de Carvalho's system $\scrRo$.\\

Intuitively, a multiset is a sequence in which we have the positions (the tracks) of the elements has been collapsed, \eg the multiset $\mult{\tv_1,\tv_2,\tv_1}$ may be seen indifferently as the collapse of the sequences $(3\cdot \tv_1,5\cdot \tv_2,7\cdot \tv_2)$, $(2\cdot \tv_1,8\cdot \tv_2,10\cdot \tv_2)$ or $(3\cdot \tv_2,5\cdot \tv_1,7\cdot \tv_2)$. Since there are nesting of multisets in the finite types from System $\scrRo$, the types of $\scrRo$ may be seen as an inductive/nested collapse of finite types of System $\ttS$.

In order to build an infinitary version of $\scrRo$ whose types will feature infinitary nestings of multisets, it is natural to define those types as coinductive collapse of sequences.\\

\subsection{Types}  
Let $U_1$ and $U_2$ be two labelled trees or two sequences of labelled trees (in the sense of \Sec~\ref{s:rigid}
). A \textbf{01-stable isomorphism} from $U_1$ to $U_2$ is a bijection $\phi$ from $\supp{U_1}$ to $\supp{U_2}$ such that:
\begin{itemize}
\item $\phi$ is monotonic for the prefix order.
\item If $a\cdot k\in \supp{U_1}$ with $k=0$ or $k=1$, then $\phi(a\cdot k)=\phi(a)=k$
\item For all $a\in \supp{U_1}$, $U_2(\phi(a))=U_1(a)$.
\end{itemize}
Thus, a 01-stable isomorphism is either an isomorphism of labelled trees or an isomorphism of sequences, which preserves tracks 0 and 1. We write $U_1\equiv U_2$ when there is a 01-stable isomorphism from $U_1$ to $U_2$.
When $U$ is a tree, the definition implies that $\phi(\epsi)=\epsi$. 
If $U_1$ and $U_2$ are types (resp. sequence types), $\phi$ is called a \textbf{\textbf{type isomorphism}} (resp. \textbf{sequence type isomorphism}).\\

Alternatively, we can define $U_1\equiv U_2$ for types and sequence types by coinduction, without reference to 01-stable isomorphism:
\begin{itemize}
\item $\tv \equiv \tv$
\item $\sSk\equiv \sSpkp$ if there is a bijection $\rho:K\rew K'$ such that, for all $\kK$, $S_k\equiv S'_{\rho(k)}$.
\item $\sSk\rew T\equiv \sSpkp\rew T'$  if $\sSk\equiv \sSpkp$ and $T\equiv T'$.
\end{itemize}

The set $\TypR$ of types of System $\scrR$ is defined as the quotient set $\Types/\equiv$.

If $U$ is a $\ttS$-type or a sequence type, its equivalence class is written 
$\overline{U}$. We may now define coinductively the notation of the collapses of $\ttS$-types:
\begin{itemize}
\item The equivalent class of a sequence type $F=\sSk$ is the 
multiset type written $\mult{\overline{S_k}}_{k\in \Rt (F)}$
\item  We write $\ovl{F}\rew  \ovl{T}$ for $\ovl{F\rew T}$.
\item If $\tv$ is a type variable, $\ovl{\tv}$ is the singleton $\set{\tv}$. In that case, we just write (abusively) $\tv$ instead of $\ovl{\tv}$ or $\set{\tv}$.
\end{itemize}

Countable sum $+_\iI \overline{F^{i}}$ is defined on 
the set of multisets types by using an arbitrary bijection $j$ from the pairwise disjoint countable sum $\bbN \times (\bbN\setminus\set{0,1})$ to $\bbN\setminus\set{0,1}$.

Let $(F^i)_{\iI}=(\sSqkK{i})_{\iI}$ a countable family of sequence type such that $I\subseteq \bbN$. We define $\stackrel{j}{+}_{\iI} F^i$ as the sequence type $F=\sSk$ such that $K=\set{j(i,k)\,|\, i\in I,\,\kK(i)}$ and, for all $\kK$, $S_k$ is $S^i_{k_0}$ where $(i,k_0)$ is the unique pair such that $i\in I,\; k_0\in K(i)$ and $k=j(i,k_0)$. We may then prove that (1) if for all $\iI$, $F^i \equiv G^i$, then $\stackrel{j}{+}_{\iI} F^i\equiv \stackrel{j}{+}_{\iI} G^i$ (2) the class of $\stackrel{j}{+}_{\iI} F^i$ does not depend on $j$. Thus, 
we may define countable sum operator + on multiset types. It is routine work to prove that $+$ infinitarily associative and commutative on multiset types, as expected.\\

\subsection{Typing rules} 
An \textbf{$\scrR$-context} is a total function from the set of term variables  $\TermV$ to the set of infinitary multiset types. An \textbf{$\scrR$-judgment} is a triple of the form $\juGtt$ where $G$ is a $\scrR$-context, $t$ a 001-term and $\tau$ a $\scrR$-types. The set of \textbf{semi-rigid} derivations is the set of trees (labelled with $\scrR$-judgments) defined coinductively by the following rules:
\begin{center}
$$\infer[\ax]{\phd}{\ju{x:[\tau]}{x:\tau  }}
\sep\sep
\infer[\abs]{\ju{\Gam;x:\msigi}{t: \tau } }
{\ju{\Gam}{\lambda x.t:\, \msigi\rew\tau }}$$
$$\infer[\app]{\ju{\Gam}{t: \msigi\rew \tau} \sep
(\ju{\Delk}{u:\sigk )_{\kK}}}
{\ju{\Gam +_{\kK} \Delk}{t\,u:\, \tau}}$$
\end{center}

Why do we say that these derivations are \textit{semi-rigid}? Because the argument of derivations are still placed on argument tracks $\kK$, while obviously, they should not matter when we work with multisets.
Let $P_1$ and $P_2$ be two semi-rigid derivations.  
We define the set of $\scrR$-derivation as the quotient set of that of semi-derivation by the relation $\equiv$.

An element of $\DerivR$ is usually written $\Pi$, whereas an element of $\Deriv_*$ is written $P$ (as for system $\ttS$). Notice the derivation $\Pi$ and $\Pi'$ of \Sec~\ref{ss:typ-inf-nf-informal} and Fig.~\ref{fig:expans-by-trunc-1} are objects of $\Deriv$.

\subsection{Quantitativity and Coinduction in System $\scrR$}
\label{s:unquant}
Let $\Gam$ be \textit{any} context. Exactly as in \Sec~\ref{ss:quant-coind-S}, we can use the infinite branch of $\fom$
 to give the following variant of derivation $\Pi'$ from \Sec~\ref{ss:typ-inf-nf-informal}, which respects the rules
of system $\scrR$:
  $$
\infer[\app]{\infer[\ax]{}{\jufara}\sep
\Pi'_{\Gam}\tri \ju{f:\mult{\arewa}_{\omega} + \Gam}{
\fom:\tv}}{ \ju{f:\mult{\arewa}_{\omega} +\Gam }{\fom :\tv}}
$$


The notion of quantitativity (Definition~\ref{def:quant-S-deriv}) can be adapted while considering multisets:

\begin{definition}\mbox{}
\begin{itemize}
\item A semi-rigid derivation $P$ is \textbf{quantitative} if, for all $a\in \supp{P},~  \Gam(a)(x)= \mult{\tau(a')}_{a'\in \AxPa{x}}$.
\item A $\scrR$-derivation $\Pi$ is \textbf{quantitative} if any of its semi-rigid representatives is (in that case, all of them are quantitative).
\end{itemize}
\end{definition}

In the next subsection, we show that a derivation $\Pi$ from system $\scrR$ can have both quantitative and not quantitative, approximable and not approximable representatives in System\;$\ttS$. It once again shows that rigid constructions allow a more fine-grained control than system $\scrR$ does on derivations.

\subsection{Representatives and Dynamics}
\label{a:rep-and-dynamics}

A $\ttS$-derivation $P$   \textbf{represents} a derivation $\Pi$ if the semi-rigid derivation $P_*$
defined by $\supp {P_*}=\supp {P}$ and $P_*(a)=\ju{\ovl{C(a)}}{\tra:\,\ovl{T(a)}}$, is a semi-rigid representative of $\Pi$.
We write $P_1 \eqr P_2$ when two $\ttS$-derivations $P_1$ and $P_2$ both represent the same $\scrR$-derivation $\Pi$. 

\begin{proposition}
If a $\ttS$-derivation $P$ is quantitative, then the $\scrR$-derivation 
$\ovl{P}$ is quantitative in system $\scrR$.
\end{proposition}

\noindent Using natural extensions (\Sec~\ref{ss:nat-ext}), it easy to prove:


\begin{proposition}
If $\Pi$ is a quantitative derivation typing a 001-normal form, then, there is 
a quantitative rigid derivation $P$ such that $\overline{P}=\Pi$.
\end{proposition}

\begin{proof}
Let $P(\ast)$ be a semi-rigid derivation representing $\Pi$. We set $A=\supp{P(\ast)}$ and for all full position $a\in A$, we choose a representative $T(\mathring{a})$ of $\tau(a)$. We apply then the canonical construction of \Sec~\ref{ss:nat-ext}, which yields a rigid derivation $P$ such that $P_*=P(\ast)$ (we show that, for all $a\in A$, $T(a)$ represents $\tau(a)$).
  \end{proof}

We can actually prove that every quantitative $\scrR$-derivation can be  represented with a $\ttS$-quantitative rigid derivation and that we can endow  it with every possible infinitary reduction choice~\cite{VialLICS18}.
However, a quantitative derivation can also have a not quantitative rigid 
representative, as we see below with $\Pi'$ and $\tilde{P}'$.


We omit again the right side of axiom rules, \eg
$f:(\twoarewa)_2 $ stands for $\ju{f:(\twoarewa)_2}{f:\twoarewa}$. Moreover, outer tracks are indicated between red square brackets.\\

\begin{figure*}[!t]
 {
$$P'_2=
\infer{
\infer{\phd}{f\!:\!(2\ct \twoarewa) \trck{1}}\
\infer{
  \infer{\phd}{f\!:\!(3\ct \twoarewa)  \trck{1}}\
        P'_4\tri \ju{f\!:\!(k\ct \twoarewa)_{k\geqslant 4}}{
            \fom: \tv} \trck{2}}{
  \ju{f:(k\ct \twoarewa)_{k\geqslant 3}}{\fom:\tv \trck{2}}}}
{\ju{f:(k\ct \twoarewa)_{k\geqslant 2}}{\fom: \tv}}
$$}

{
$$\tilde{P}'_2=
\infer{
  \infer{\phd}{ f\!:\!(2\ct \twoarewa) \trck{1}}\msep
  \infer{
    \infer{\phd}{f\!:\!(4\ct \twoarewa) \trck{1}}\msep
    \tilde{P}'_5\tri \ju{f\!:\!(k\ct \twoarewa)_{k=3 \vee k\geqslant 5}}{\fom:\tv\trck{2}}}{
    \ju{f:(k\ct \twoarewa)_{k\geqslant 3}}{\fom:\tv} \trck{2}}}
{ \ju{f:(k\ct \twoarewa)_{k\geqslant 2}}{\fom:\tv}}
$$}

{
$$\tilde{P}'_k=
\infer{
  \infer{\phd}{ f\!:\!(k\ct \twoarewa) \trck{1}}\sep
    \tilde{P}'_{k+1}\tri \ju{f\!:\!(i\ct \twoarewa)_{i=3 \vee i \geqslant k+1}}{\fom:\tv}\trck{2}  }
{\ju{f:(i\ct \twoarewa)_{i\geqslant k}}{\fom:\tv}}\sep \sep \text{(for $k\geqs 5$)}
$$ }

\caption{Representing $\Pi'$}
\label{fig:two-representatives-fom}
\end{figure*}

\tbul~ \newcommand{\Rtv}{{R_{\tv}}}
We define $\Rtv$ by $\Rtv:=(k\cdot \Rtv)_{k\geqs 2}\rew \tv$ and
Let $P_k$ ($k\geqslant 2$) and $P$ be the following $\ttS$-derivations:  
{
$$
  P_k=
 \infer{
\infer{\infer{\phd}{f\!:\!(k\ct \twoarewa) 
  \trck{1} }  \
\infer{
  \infer{\phd}{x:\,(2\ct \Rtv) \trck{1}}
  (\infer{\phd}{x:(i\ct \Rtv) \trck{i-1}})_{i\geqs 3}}{
   \ju{x:(i\ct \Rtv)_{i\geqslant 2}}{x\,x:\tv} \trck{2}}}{
  \ju{f:(k\ct \twoarewa)}{ f(x\,x):\tv}\trck{0}}}
{\ju{f:(k\ct \twoarewa)}{\Delta_f:\Rtv}}
$$}

{
$$
P=
\infer{P_2 \trck{1}\sep (P_k \trck{k-1})_{k\geqslant 3}
}{\ju{f:(k\ct \twoarewa)_{k\geqslant 2}}{\Del_f \Del_f}}
$$
}

\tbul ~ Let $\tilde{P}_k$ ($k\geqslant 2$) and $\tilde{P}$ be the following $\ttS$-derivations:

{
  $$\tilde{P}_k=
\infer{
  \infer{
    \infer{\phd}{f:(k\ct \twoarewa)~  \trck{1}} \
   \infer{      
     \infer{\phd}{x:(3\ct \Rtv)~  \trck{1}}
     \infer{\phd}{x:(2\ct \Rtv)~ \trck{2}}
     (\infer{\phd}{x:(i \ct \Rtv)~ \trck{i-1}})_{i\geqslant 4}
   }{\ju{x:(i\ct \Rtv)_{i\geqslant 2}}{x\,x:\tv} \trck{2}}}
        {\ju{f:(k\ct \twoarewa) }{ f(x\,x):\tv} \trck{0}}}
{\ju{f:(k\ct \twoarewa)}{\Del_f:\Rtv}}
$$}

{
$$
\tilde{P}=
\infer{\tilde{P}_2 \trck{1} \sep (\tilde{P}_k \trck{k-1})_{k\geqslant 3} }{
  \ju{f:(\twoarewa)_{k\geqslant 2} }{\Delta_f \Delta_f}}$$
}

\tbul ~ The rigid derivations $P$ and $\tilde{P}$ both represent
$\Pi$. Intuitively, subject reduction in $P$ will consist in taking the
first argument $P_3$, placing it on the first occurrence of $x$ in
$f(x\,x)$ (in $P_2$) and putting the other $P_k$ ($k\geqslant 4$) in the
different axiom rules typing the second occurrence of $x$ in the same
order, \ie, for $k \geqs 3$, $P_k$ is moved from track $k-1$ to track $k-2$. There is a simple decrease on the track number and we can
go this way towards $P'_2$ typing $\fom$.

The rigid derivation $\tilde{P}$ process the same way, except it will always skip $\tilde{P}_3$ ($\tilde{P}_3$ will remain on track 2). Morally, we 
perform subject reduction ``by-hand'' while avoiding to ever place $\tilde{P}_3$ in head position, \ie $\tilde{P}_3$ is never ``consumed'' by reduction.

A computation shows that infinitary reductions performed in $P$ and $\tilde{P}_2$ yield respectively to $P'_2$ and $\tilde{P}'_2$ of Fig.~\ref{fig:two-representatives-fom}.

Thus, $P'$ and $\tilde{P}'$ both represent $\Pi'$ (from \Sec~\ref{ss:typ-inf-nf-informal}), but $P'$ is quantitative whereas $\tilde{P}'$ is
not (the track $3$ type assigned to $f$ does not end in an axiom leaf). This proves  that quantitativity is not stable up to the limit of a productive reduction path and that $\tilde{P}$ is not approximable by Proposition~\ref{prop:infinite-subject-reduction}.

Moreover, by Proposition~\ref{prop:NF-uf-approx-typable}, $P'_2$ is approximable. By infinitary subject expansion, $P_2$ also is (actually, one may also use the $\ttS$-representatives of $\Pi_n$ and $\Pi'_n$ in Fig.~\ref{fig:expans-by-trunc-1} to approximate $P$ and $P'$ arbitrarily).\\

\noindent \textbf{Conclusion.} Thus, $\Pi$ and  $\Pi'$ have both approximable and not approximable representatives. Actually, $\Pi$ has a representative which is quantitative but not approximable.  This proves that approximability cannot be defined in system $\scrR$.

\end{document}